\documentclass[final,3p,times]{elsarticle}

\usepackage{amssymb}

\usepackage{graphicx}
\usepackage{times}
\usepackage[linesnumbered,ruled,algonl,vlined]{algorithm2e}
\usepackage{hyperref}
\usepackage{latexsym}
\usepackage{amsmath}

\usepackage[numbers]{natbib}

\newcommand{\squishlist}{
   \begin{list}{$\bullet$}
    { \setlength{\itemsep}{0pt}      \setlength{\parsep}{3pt}
      \setlength{\topsep}{3pt}       \setlength{\partopsep}{0pt}
      \setlength{\leftmargin}{1.5em} \setlength{\labelwidth}{1em}
      \setlength{\labelsep}{0.5em} } }

\newcommand{\squishlisttwo}{
   \begin{list}{$\bullet$}
    { \setlength{\itemsep}{0pt}    \setlength{\parsep}{0pt}
      \setlength{\topsep}{0pt}     \setlength{\partopsep}{0pt}
      \setlength{\leftmargin}{2em} \setlength{\labelwidth}{1.5em}
      \setlength{\labelsep}{0.5em} } }

\newcommand{\squishend}{
    \end{list}  }

\newcommand{\eat}[1]{}
\newtheorem{definition}{Definition}[section]

\newtheorem{lemma}{Lemma}[section]
\newtheorem{heuristic}{Heuristic}[section]
\newproof{proof}{Proof}[section]

\begin{document}

\begin{frontmatter}



\title{Probabilistic Voronoi Diagrams for Probabilistic Moving Nearest Neighbor Queries}


\author{Mohammed Eunus Ali\footnote{The corresponding author}, Egemen Tanin, Rui Zhang, and Ramamohanarao Kotagiri}

\address{Department of Computer Science and Software Engineering\\University of Melbourne, Victoria, 3010, Australia\\Tel.: +61 3 8344 1350\\Fax: +61 3 9348 1184\\\{eunus,egemen,rui,rao\}@csse.unimelb.edu.au}

\begin{abstract}
A large spectrum of applications such as location based services
and environmental monitoring demand efficient query processing on
uncertain databases. In this paper, we propose the probabilistic
Voronoi diagram (PVD) for processing moving nearest neighbor
queries on uncertain data, namely the probabilistic moving nearest
neighbor (PMNN) queries. A PMNN query finds the most probable
nearest neighbor of a \emph{moving query point} continuously. To process PMNN queries efficiently, we
provide two techniques: a pre-computation approach and an incremental
approach. In the pre-computation
approach, we develop an algorithm to efficiently evaluate PMNN
queries based on the pre-computed PVD for the entire data set. In
the incremental approach, we propose an incremental probabilistic safe
region based technique that does not require to pre-compute the
whole PVD to answer the PMNN query. In this incremental approach,
we exploit the knowledge for a known region to compute the lower bound of the probability of an object being the nearest
neighbor. Experimental results show that our approaches significantly
outperform a sampling based approach by orders of magnitude in
terms of I/O, query processing time, and communication overheads.
\end{abstract}


\end{frontmatter}

\section{Introduction}
\label{sec:intro}
Uncertainty is an inherent property in many
database applications that include location based
services~\cite{goci04:mo}, environmental
monitoring~\cite{Madden03.sigmod}, and feature extraction
systems~\cite{liu06.icpr}. The inaccuracy or imprecision of data
capturing devices, the privacy concerns of users, and the
limitations on bandwidth and battery power introduce uncertainties
in different attributes such as the location of an object or the
measured value of a sensor. The values of these attributes are
stored in a database, known as an uncertain database.

In recent years, query processing on an uncertain database has received significant attention from the research community due to its wide range of applications. Consider a location based application where the location information of users may need to be pre-processed before publishing due to the privacy concern of users. Alternatively, a user may want to provide her position as a larger region in order to prevent her location to be identified to a particular site. In such cases, locations of users are stored as uncertain attributes such as regions instead of points in the database. An application that deals with the location of objects (e.g., post office, hospital) obtained from satellite images is another example of an uncertain database. Since the location information may not be possible to identify accurately from the satellite images due to noisy transmission, locations of objects need to be represented as regions denoting the probable locations of objects. Likewise, in a biological database, objects identified from microscopic images need to be presented as uncertain attributes due to inaccuracies of data capturing devices.

In this paper, we propose a novel concept called \emph{Probabilistic
Voronoi Diagram} (PVD), which has a potential to efficiently process nearest neighbor (NN) queries on an uncertain database. The PVD for a given set of uncertain objects
${o_{1},o_{2},...,o_{n}}$ partitions the data space into
a set of \emph{Probabilistic Voronoi Cells} (PVCs) based on the
probability measure. Each cell $PVC(o_i)$ is a region in the data
space, where each data point in this region has a higher
probability of being the NN to $o_i$ than any other object.

A nearest neighbor (NN) query on an uncertain database, called a
Probabilistic Nearest Neighbor (PNN) query, returns a set of
objects, where each object has a non-zero probability of being the
nearest to a query point. A common variant of the PNN query that
finds the most probable NN to a given query point is also called a
top-1-PNN query. Existing research focuses on efficient processing
of PNN
queries~\cite{Cheng03.sigmod,Cheng04.tkde,Cheng.ICDE10,Kriegel07.dasfaa}
and its
variants~\cite{beskales08.vldb,chris07:efficienttop-k,Soliman07:top-kquery}
for a \emph{static query point}. In this paper, we are interested
in answering Probabilistic Moving Nearest Neighbor (PMNN) queries
on an uncertain database, where data objects are \emph{static},
the query is \emph{moving}, and the future path of the moving
query is \emph{unknown}. A PMNN query returns the most probable
nearest object for a moving query point continuously.

A straightforward approach for evaluating a PMNN query is to use
a sampling-based method, which processes the PMNN query as a
sequence of PNN queries at sampled locations on the query path.
However, to obtain up-to-date answers, a high sampling rate is
required, which makes the sampling-based approach inefficient due
to the frequent processing of PNN queries.

To avoid high processing cost of the sampling based approach
and to provide continuous results, recent approaches for
continuous NN query processing on a \emph{point data set} rely on
safe-region based techniques, e.g., Voronoi
diagram~\cite{okabe00:voronoi}. In a Voronoi diagram based
approach, the data space is partitioned into disjoint Voronoi
cells where all points inside a cell have the same NN. Then, the
NN of a query point is reduced to identifying the cell for the
query point, and the result of a moving query point remains valid
as long as it remains inside that cell. Motivated by the
safe-region based paradigm, in this paper we propose a Voronoi
diagram based approach for processing a PMNN query on a set of
uncertain objects.

Voronoi diagrams for uncertain
objects~\cite{Cheng.ICDE10,evans08.CCCG} based on a simple
distance metric, such as the minimum and maximum distances to
objects, result in a large neutral region that contains those
points for which no specific NN object is defined. Thus, these are
not suitable for processing a PMNN query. In this paper, we propose
the PVD that divides the space based on a probability measure
rather than using just a simple distance metric.

\eat{The key idea of our approach is to develop a \emph{Probabilistic
Voronoi Diagram} (PVD) for a given set of uncertain objects
${o_{1},o_{2},...,o_{n}}$. The PVD partitions the data space into
a set of \emph{Probabilistic Voronoi Cells} (PVCs) based on the
probability measure. Each cell $PVC(o_i)$ is a region in the data
space, where each data point in this region has a higher
probability of being the NN to $o_i$ than any other object. }

A naive approach to compute the PVD is to find the top-1-PNN for
every possible location in the data space using existing static
PNN query processing
techniques~\cite{Cheng03.sigmod,Cheng04.tkde,beskales08.vldb},
which is an impractical solution due to high computational
overhead. In this paper, we propose a practical solution to compute
the PVD for a set of uncertain objects. The key idea of our approach is to efficiently compute the probabilistic bisectors between two neighboring objects that forms the basis of PVCs for the PVD.

After computing the PVD, the most probable NN can be determined by simply identifying the PVC
in which the query point is currently located. The result of the
query does not change as long as the moving query point remains in
the current PVC. A user sends its request as soon as it exits the
PVC. Thus, in contrast to the sampling based approach, the PVD
ensures the most probable NN for every point of a moving query path is
available. Since this approach requires the
pre-computation of the whole PVD, we name it the \emph{pre-computation
approach} in this paper.

The pre-computation approach needs to access all the objects from
the database to compute the entire PVD. In addition, the PVD needs
to be re-computed for any updates (insertion or deletion) to the
database. Thus the pre-computation approach may not be suitable
for the cases when the query is confined into a small region in
the data space or when there are frequent updates in the database.
For such cases, we propose an incremental algorithm based on the
concept of local PVD. In this approach, a set of surrounding
objects and an associated search space, called \emph{known
region}, with respect to the current query position are retrieved
from the database. Objects are retrieved based on their
probabilistic NN rankings from the current query location. Then, we compute
the local PVD based only on the retrieved data set, and develop  a
\emph{probabilistic safe region} based PMNN query processing
technique. The probabilistic safe region defines  a region for an
uncertain object where the object is guaranteed to be the most
probable nearest neighbor. This probabilistic safe region enables a user to
utilize the retrieved data more efficiently and reduces the
communication overheads when a client is connected to the server
through a wireless link. The process needs to be repeated as soon
as the retrieved data set cannot provide the required answer for
the moving query point. We name this PMNN query processing
technique the \emph{incremental approach} in this paper.

In summary, we make the following contributions in this paper:

\begin{itemize}
    \item We formulate the Probabilistic Voronoi Diagram (PVD) for uncertain objects and propose techniques to compute the PVD.
    \item We provide an algorithm for evaluating PMNN
    queries based on the pre-computed PVD.
    \item We propose an incremental algorithm for evaluating PMNN queries based on the concept of local PVD.
    \item We conduct an extensive experimental study which shows that our PVD based approaches outperform the sampling based approach significantly.
\end{itemize}

The rest of the paper is organized as follows.
Section~\ref{sec:ps} discusses preliminaries and the problem
setup. Section~\ref{sec:rw} reviews related work. In
Section~\ref{sec:pvd}, we formulate the concept of PVD and present
methods to compute it, focusing on one and two dimensional spaces.
In Section~\ref{sec:pmnn}, we present two techniques:
pre-computation approach and incremental approach for processing
PMNN queries. Section~\ref{sec:exp} reports our experimental
results and Section~\ref{sec:conc} concludes the paper.

\section{Preliminaries and Problem Setup}
\label{sec:ps} Let $O$ be a set of uncertain objects in a
$d$-dimensional data space. An uncertain object $o_{i} \in O$, $1
\leq i \leq |O|$, is represented by a $d$-dimensional uncertain
range $R_{i}$ and a probability density function~$(pdf)$
$f_{i}(u)$ that satisfies $\int_{R_{i}}f_{i}(u)du = 1$ for $u \in
R_{i}$. If $u \notin R_{i}$, then $f_{i}(u)=0$.\eat{ In this
paper, our goal is to introduce the concept of probabilistic
Voronoi diagram (PVD). Even the PVD for uncertain objects with
uniform pdf is not trivial, and hence} We assume that the pdf of
uncertain objects follow uniform distributions for the sake of
easy explication. Our concept of PVD is applicable for other types
of distributions. We briefly discuss PVDs for other distributions
in Section~\ref{subsec:discus}).\eat{For simplicity, we assume
uniform distributions for the pdf of uncertain objects to
illustrate the concept of PVD.} For uniform distribution, the pdf
of $o_{i}$ can be expressed as $f_{i}(u)=\frac{1}{Area(R_{i})}$
for $u\in R_i$. For example, for a circular object $o_i$, the
uncertainty region and the pdf are represented as
$R_{i}=(c_i,r_i)$ and $f_{i}(u)=\frac{1}{\pi r_{i}^2}$,
respectively, where $c_i$ is the center and $r_i$ is the radius of
the region. We also assume that the uncertainty of objects remain
constant.

An NN query on a traditional database consisting of a set of data
points (or objects) returns the nearest data point to the query
point. An NN query on an uncertain database does not return a
single object, instead it returns a set of objects that have
non-zero probabilities of being the NN to the query point. Suppose
that the database maintains only point locations $c_1$, $c_2$, and
$c_3$ for objects $o_1$, $o_2$, and $o_3$, respectively (see
Figure~\ref{fig:circle_prb}). Then an NN query with respect to $q$
returns $o_2$ as the NN because the distance $dist(c_2,q)$ is the
least among all other objects. In this case, $o_1$ and $o_3$ are
the second and third NNs, respectively, to the query point $q$. If
the database maintains the uncertainty regions $R_{1}=(c_1,r_1)$,
$R_{2}=(c_2,r_2)$, and $R_{3}=(c_3,r_3)$ for objects $o_1$, $o_2$,
and $o_3$, respectively, then the NN query returns all three
$(o_1,p_1),(o_2,p_2),(o_3,p_3)$ as probable NNs for the query
point $q$, where $p_1>p_2>p_3>0$ (see
Figure~\ref{fig:circle_prb}).

A Probabilistic Nearest Neighbor (PNN) query~\cite{Cheng03.sigmod}
is defined as follows:

\begin{definition}
\label{def:2_1} (\textit{PNN}) Given a set $O$ of uncertain
objects in a $d$-dimensional database, and a query point $q$, a
PNN query returns a set $P$ of tuples $(o_{i},p_{i})$, where
$o_i\in O$ and $p_{i}$ is the non-zero probability that the
distance of $o_i$ to $q$ is the minimum among all objects in $O$.
\end{definition}

The probability $p(o_i,q)$ (or simply $p_{i}$) of an object
$o_{i}$ of being the NN to a query point $q$ can be computed as
follows. For any point $u\in R_i$, where $R_i$ is the uncertainty
region of an object $o_i$, we need to first find out the
probability of $o_i$ being at $u$ and multiply it by the
probabilities of all other objects being farther than $u$ with
respect to $q$, and then summing up these products for all $u$ to
compute $p_i$. Thus, $p_i$ can be expressed as follows:

\begin{equation}
\label{eq:pnn}
    p_{i} = \int_{u\in R_{i}}f_{i}(u)\big(\prod_{j\neq i}
    \int_{v\in R_{j}}P(dist(v,q)>dist(u,q))dv\big)du,
\end{equation}

where the function $P(.)$ returns the probability that a point
$v\in R_j$ of $o_j$ is farther from a point $u\in R_i$ of $o_i$.

\begin{figure}[htbp]
    \centering
        \includegraphics[width=1.8in]{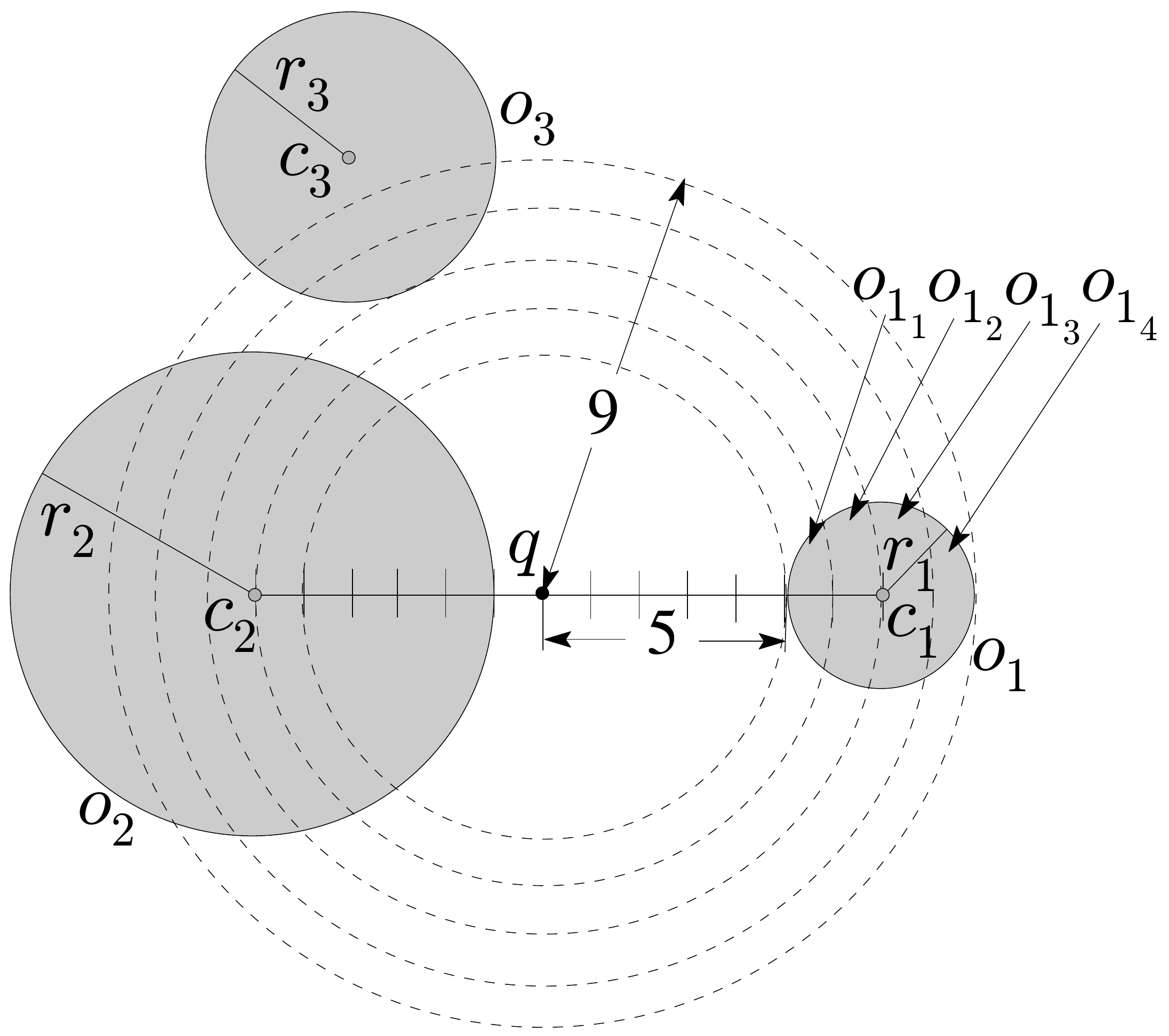}
    \caption{An example of a PNN query}
    \label{fig:circle_prb}
\end{figure}

Figure~\ref{fig:circle_prb} shows a query point $q$, and three
objects $o_1$, $o_2$, and $o_3$. Based on Equation~\ref{eq:pnn},
the probability $p_1$ of object $o_1$ being the NN to $q$ can be
computed as follows. In this example, we assume \emph{a discrete space} where the radii of three objects
are 5, 2, and 3 units, respectively, and the minimum distance of
$o_1$ to $q$ is 5 units. Suppose that the dashed circles $(q,5)$,
$(q,6)$, $(q,7)$, $(q,8)$, and $(q,9)$ centered at $q$ with radii
5, 6, 7, 8, and 9 units, respectively, divide the uncertain region
$R_1$ of $o_1$ into four sub-regions $o_{1_1}$, $o_{1_2}$,
$o_{1_3}$, and $o_{1_4}$,  where $o_{1_1}=(c_1,r_1)\cap(q,6)$,
$o_{1_2}=(c_1,r_1)\cap(q,7)-o_{1_1}$,
$o_{1_3}=(c_1,r_1)\cap(q,8)-(o_{1_1}\cup o_{1_2})$,
$o_{1_4}=(c_1,r_1)\cap(q,9)-(o_{1_1}\cup o_{1_2}\cup o_{1_3})$;
similarly $R_2$ is divided into six sub-regions $o_{2_1}$,
$o_{2_2}$, $o_{2_3}$, $o_{2_4}$, $o_{2_5}$, and $o_{2_6}$; $R_3$
is divided into three sub-regions $o_{3_1}$, $o_{3_2}$, and
$o_{3_3}$.

Then $p_1$ can be computed by summing: (i) the probability of
$o_1$ being within the sub-region $o_{1_1}$ multiplied by the
probabilities of $o_2$ and $o_3$ being outside the circular region
$(q,6)$, (ii) the probability of $o_1$ being within the sub-region
$o_{1_2}$ multiplied by the probabilities of $o_2$ and $o_3$ being
outside the circular region $(q,7)$, (iii) the probability of
$o_1$ being within the sub-region $o_{1_3}$ multiplied by the
probabilities of $o_2$ and $o_3$ being outside the circular region
$(q,8)$, and (iv) the probability of $o_1$ being within the
sub-region $o_{1_4}$ multiplied by the probabilities of $o_2$ and
$o_3$ being outside the circular region $(q,9)$.

As we have discussed in the introduction, in many applications a
user may often be interested in the most probable nearest
neighbor. In such cases, a PNN only returns the object with the
highest probability of being the NN, also known as a
\emph{top-1-PNN query}. In this paper, we address the probabilistic
moving NN query that continuously reports the most probable NN for
each query point of a moving query.

\eat{
We define this Probabilistic
Moving Nearest Neighbor (PMNN) query as follows:

\begin{definition}
\label{def:2_2} (\textit{PMNN}) Given a set $O$ of uncertain
objects in a $d$-dimensional database, and a moving query point
$q$, a PMNN query returns object $o_{i}$ having the highest
probability of being the nearest neighbor for every position of
$q$.
\end{definition}
}

From Equation~\ref{eq:pnn}, we see that finding the most probable
NN to a static query point is expensive as it involves costly
integration and requires to consider the uncertainty of other
objects. Hence, for a moving user that needs to be updated with
the most probable answer continuously, it requires repetitive
computation of the top object for every sampled location of the
moving query. In this paper, we propose PVD based approaches for
evaluating a PMNN query.

In this paper, we propose two techniques: a pre-computation
approach and an incremental approach to answer PMNN queries. Based
on the nature of applications, one can choose any of these
techniques that suits best for her purpose.  Moreover, both of our
techniques fit into any of the two most widely used query
processing paradigms: \emph{centralized paradigm}, and
\emph{client-server paradigm}. In the centralized paradigm the
query issuer and the processor reside in the same machine, and the
total query processing cost is the main performance measurement
metric. On the other hand, in the client-server paradigm, a client
issues a query to a server that processes the query, through
wireless links such as mobile phone networks. Thus, in the
client-server paradigm the performance metric includes both the
communication cost and the query processing cost.

In the rest of the paper, we use the following functions:
$min(v_{1}, v_{2},...,v_n)$ and $max(v_{1}, v_{2},...,v_n)$ return
the minimum and the maximum, respectively, of a given set of
values $v_{1}$, $v_{2}$,...,$v_n$; $dist(p_1,p_2)$ returns the
Euclidian distance between two points $p_1$ and $p_2$;
$mindist(p,o)$ and $maxdist(p,o)$ return the minimum and maximum
Euclidian distances, respectively,  between a point $p$ and an
uncertain object $o$.

We also use the following terminologies. When the possible range
of values of two uncertain objects overlap then we call them
\emph{overlapping objects}; otherwise they are called
\emph{non-overlapping objects}. If the ranges of two objects are
of equal length then we call them \emph{equi-range objects};
otherwise they are called \emph{non-equi-range objects}.
%

\section{Background}
\label{sec:rw}

\eat{Previous work on uncertain databases have studied query types
such as range queries~\cite{Cheng03.sigmod,cheng04:threshold}; NN
queries~\cite{Cheng04.tkde,Cheng03.sigmod,Kriegel07.dasfaa,traj09:trajectory};
top-$k$
queries~\cite{beskales08.vldb,chris07:efficienttop-k,Soliman07:top-kquery};
and skyline queries~\cite{pei07:skyline}.}

In this section, we first give an overview of existing PNN query
processing techniques on uncertain databases that are closely
related to our work. Then we present existing work on Voronoi
diagrams.

\subsection{Probabilistic Nearest Neighbor}
\label{subsec:pnn} Processing PNN queries on uncertain databases
has received significant attention in recent years.
In~\cite{Cheng03.sigmod}, Cheng~et~al. proposed a numerical
integration based technique to evaluate a PNN query for
one-dimensional sensor data. In~\cite{Cheng04.tkde}, an I/O
efficient technique based on numerical integration was developed
for evaluating PNN queries on two-dimensional uncertain moving
object data. In~\cite{Kriegel07.dasfaa}, authors presented a
sampling based technique to compute PNN, where both data and query
objects are uncertain. Probabilistic threshold NN queries have
been introduced in~\cite{Cheng08.icde}, where all objects with
probabilities above a specified threshold are reported.
In~\cite{traj09:trajectory}, a PNN algorithm was presented where
both data and query objects are static trajectories, where the
algorithm finds objects that have non-zero probability of any
sub-intervals of a given trajectory. Lian et
al.~\cite{lian08:pgnn} presented a technique for a group PNN query
that minimizes the aggregate distance to a set of static query
points.

The PNN variant, top-$k$-PNN query reports top $k$ objects which
have higher probabilities of being the nearest than other objects
in the
database~\cite{beskales08.vldb,chris07:efficienttop-k,Soliman07:top-kquery}.
Among these works,
techniques~\cite{chris07:efficienttop-k,Soliman07:top-kquery} aim
to reduce I/O and CPU costs independently.
In~\cite{beskales08.vldb}, the authors proposed a unified cost
model that allows interleaving of I/O and CPU costs while
processing top-$k$-PNN queries. This method~\cite{beskales08.vldb}
uses lazy computational bounds for probability calculation which
is found to be very efficient for finding  top-$k$-PNN.

Any existing methods for static PNN
queries~\cite{Cheng03.sigmod,Cheng04.tkde,Kriegel07.dasfaa} or its
variants~\cite{beskales08.vldb,chris07:efficienttop-k,Soliman07:top-kquery}
can be used for evaluating PMNN queries which process the PMNN
query as a sequence of PNN queries at sampled locations on the
query path. Since in this paper we are only interested in the most
probable answer, we use the recent
technique~\cite{beskales08.vldb} to compute top-1-PNN for
processing PMNN queries in a comparative sampling based approach
and also for the probability calculation in the PVD.

Some techniques~\cite{dai05:existprob,yiu09:existprob} have been
proposed for answering PNN queries (including top-$k$-PNN) for
existentially uncertain data, where objects are represented as
points with associated membership probabilities. However, these
techniques are not related to our work as they do not support
uncertainty in objects' attributes. Our problem should also not be
confused with maximum likelihood classifiers~\cite{thomas97:ml}
where they use statistical decision rules to estimate the
probability of an object being in a certain class, and assign the
object to the class with the highest probability.

All of the above mentioned schemes assume a static query point for
PNN queries. Though, continuous processing of NN queries for a
moving query point on a \emph{point data set} was also a topic of
interest for many years~\cite{tao02:tp}, we are the first to
address such queries on an \emph{uncertain data set}. In this
paper, we propose efficient techniques for probabilistic moving
NN queries on an uncertain database, where we continuously report
the most probable NN for a moving query point.


\subsection{Voronoi Diagrams}
\label{subsec:vd}

The Voronoi diagram~\cite{okabe00:voronoi} is a popular approach
for answering both static and continuous nearest neighbor queries
for two-dimensional \emph{point data}~\cite{zhang:sq}. Voronoi
diagrams for extended objects (e.g., circular
objects)~\cite{kar01:VDcircle} have been proposed that use
boundaries of objects, i.e., minimum distances to objects, to
partition the space. However, these objects are not uncertain, and
thus,~\cite{kar01:VDcircle} cannot be used for PNN queries.

Voronoi diagrams for uncertain objects have been proposed that can
divide the space for a set of sparsely distributed
objects~\cite{Cheng.ICDE10,evans08.CCCG}. Both of these approaches are based on the distance metric, where $mindist$ and $maxdist$ to objects are used to calculate the boundary of the Voronoi edges.

The Voronoi diagram of~\cite{evans08.CCCG} can be described as follows.

Let $R_1,R_2,...,R_n$ be the regions of a set $O$ of uncertain
objects $o_1,o_2,...,o_n$, respectively. Then a set of sub-regions
or cells $V_1,V_2,...,V_n$ in the data space can be determined
such that a point in $V_i$ must be closer to any point in $R_i$
than to any point in any other object's region. For two objects
$o_i$ and $o_j$, let $H(i,j)$ be the set of points in the space
that are at least as close to any point in $R_i$ as any point in
$R_j$, i.e.,
\begin{displaymath}
H(i,j)=\{p\|\forall x\in R_i \forall y\in R_j~dist(p,x)\leq
dist(p,y)\},
\end{displaymath}
where $p$ is a point in the data space.

Then, the cell $V_i$ of object $o_i$ can be defined as follows:
\begin{displaymath}
V_i=\cap_{j\neq i}H(i,j).
\end{displaymath}
The boundary $B(i,j)$ of $H(i,j)$ can be defined as a set of
points in $H(i,j)$, where $p\in B(i,j)$ and
$maxdist(p,o_i)=mindist(p,o_j)$. If the regions are circular, the
boundary of object $o_i$ with $o_j$ is a set of points $p$ that
holds the following condition:
\begin{displaymath}
dist(p,c_i)+r_i=dist(p,c_j)-r_j,
\end{displaymath}
where $c_i$ and $c_j$ are the centers and $r_i$ and $r_j$ are the
radii of the regions for objects $o_i$ and $o_j$, respectively.

\begin{figure}[htbp]
    \centering
        \includegraphics[width=1.5in]{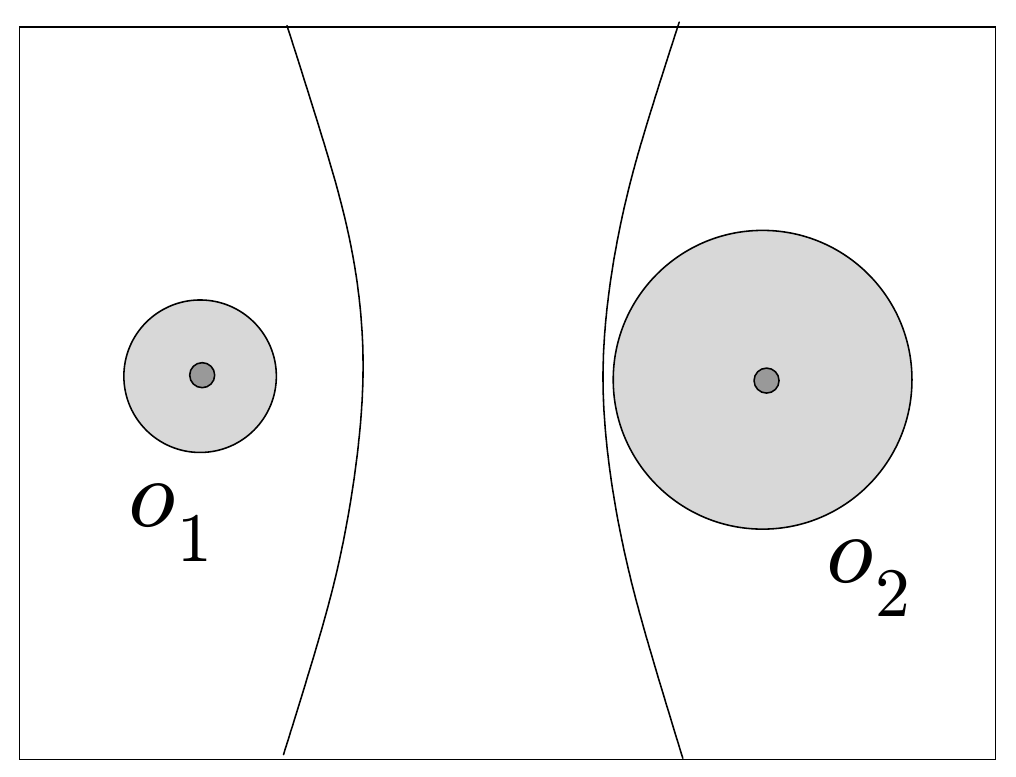}
    \caption{A guaranteed Voronoi diagram}
    \label{fig:guranteed-vd}
\end{figure}

Since $r_i$ and $r_j$ are constants, the points $p$ that satisfy
the above equation lie on the hyperbola (with foci $c_i$ and
$c_j$) arm closest to $o_i$. Figure~\ref{fig:guranteed-vd} shows
an example of this Voronoi diagram for uncertain objects $o_1$ and
$o_2$. The figure also shows the neutral region (the region
between two hyperbolic arms) for which the NN cannot be defined by
using this Voronoi diagram. Since this Voronoi diagram divides the
space based on only the distances (i.e., \emph{mindist} and
\emph{maxdist} of objects), there may not exist any partition of
the space when there is no point such that \emph{mindist} of an
object is equal to \emph{maxdist} of the other object, i.e., when
the regions of objects overlap or too close to each other.

In this approach, a Voronoi cell $V_i$ only contains those points
in the data space that have $o_i$  as the nearest object with
probability one. Thus, this diagram is called a guaranteed Voronoi
diagram for a given set of uncertain objects. However, in our
application domain, an uncertain database can contain objects with
overlapping ranges or objects with close proximity (or densely
populated)~\cite{Cheng03.sigmod,Cheng04.tkde,beskales08.vldb}.
Hence a PNN query returns a set of objects (possibly more than
one) which have the possibilities of being the NN to the query
point. Having such a data distribution, the guaranteed Voronoi
diagram cannot divide the space at all, and as a result the
neutral regions cover most of the data space for which no nearest
object can be determined. However, for an efficient PMNN query
evaluation we need to continuously find the most probable nearest
object for each point of the query path. We propose a
Probabilistic Voronoi Diagram (PVD) that works for any
distribution of data objects.

Cheng et al.~\cite{Cheng.ICDE10} also propose a Voronoi diagram
for uncertain data, called Uncertain-Voronoi diagram (UV-diagram).
The UV-diagram partitions the space based on the distance metric
similar to the guaranteed Voronoi diagram~\cite{evans08.CCCG}. For
each uncertain object $o_i$, the UV-diagram defines a region (or
UV-cell) where $o_i$ has a non-zero probability of being the NN
for any point in this region. The main difference of the
UV-diagram from the guaranteed Voronoi diagram is that the
guaranteed Voronoi diagram concerns about finding the region for a
object where the object is \emph{guaranteed} to be the NN for any
point in this region, on the other hand UV-diagram concerns about
defining a region for an object where the object has a
\emph{chance} of being the NN for any point in this region. For
example, in Figure~\ref{fig:guranteed-vd}, all points that are
left side of the hyperbolic arm closest to $o_2$ have non-zero
probabilities of $o_1$ being the NN, and thus the region left to
this hyperbolic line (i.e., closest to $o_2$) defines the UV-cell
for object $o_1$. Similarly, the region right to the hyperbolic
line closest to $o_1$ defines the UV-cell for object $o_2$. Since
both UV-diagram and guaranteed Voronoi diagram are based on the
concept of similar distance metrics, the UV-diagram suffers from
similar limitations as of the guaranteed Voronoi diagram (as
discussed above) and is not suitable for our purpose. \eat{More
importantly, UV-diagram answers static PNN queries, where it finds
all objects that has a chance of being the NN for a given query
point, on the other hand, we are interested in finding the most
probable NN for a moving query point.}

\eat{Note that, in this paper we assume that the privacy can be
one of the many reasons of uncertainty of data objects, and focus
on PMNN query processing on uncertain objects. The privacy-aware
continuous query processing paradigm can also be benefited from
this work, where the privacy of users associated with the data
objects are important.}

\eat{The important innovation of our approach is dividing the
space based on probability measure rather than using just the
distance metric.}

\section{Probabilistic Voronoi Diagram}
\label{sec:pvd}

A Probabilistic Voronoi Diagram (PVD) is defined as follows:

\begin{definition}
\label{def:pvd} (\textit{PVD}) Let $O$ be a set of uncertain
objects in a $d$-dimensional data space. The probabilistic Voronoi
diagram partitions the data space into a set of disjoint regions,
called Probabilistic Voronoi Cells (PVCs). The PVC of an object
$o_{i}\in O$ is a region or a set of non-contiguous region, denoted by $PVC(o_{i})$, such that
$p(o_{i},q)>p(o_{j},q)$ for any point $q\in PVC(o_{i})$ and for
any object $o_{j}\in O-\{o_i\}$, where $p(o_{i},q)$ and
$p(o_{j},q)$ are the probabilities of  $o_{i}$ and $o_{j}$ of
being the NNs to $q$.
\end{definition}

\eat{Note that, although neighboring PVCs are joint at their
bisectors, we treat them as disjoint regions. }

The basic idea of computing a PVD is to identify the PVCs of all objects. To find a PVC of an object, we need to find the boundaries of the PVC with all neighboring objects. The boundary line/curve that separates two neighboring PVCs is called the probabilistic bisector of two corresponding objects, as both objects have equal probabilities of being the NNs for any point on the boundary. Let $o_i$ and $o_j$ be two uncertain objects, $pb_{o_{i}o_{j}}$ be the probabilistic bisector of $o_{i}$ and $o_{j}$ that separates $PVC(o_i)$ and $PVC(o_j)$. Then, for any point $q\in pb_{o_{i}o_{j}}$, $p(o_{i},q)=p(o_{j},q)$, and for any point $q\in PVC(o_i)$, $p(o_{i},q)>p(o_{j},q)$, and for any point $q\in PVC(o_j)$, $p(o_{i},q)<p(o_{j},q)$.

\eat{the probabilistic bisector for all
pairs of uncertain objects. The probabilistic bisector between two
objects is defined as follows:

The basic idea of
computing a PVD is to find the probabilistic bisector for all
pairs of uncertain objects. The probabilistic bisector between two
objects is defined as follows:}
\eat{
\begin{definition}
\label{def:pb} (\textit{Probabilistic Bisector}) Let $o_{i}$ and
$o_{j}$ be two uncertain objects in a data space $D$, and the
probabilities of $o_{i}$ and $o_{j}$ of being the NN to a point
$q\in D$ are $p(o_{i},q)$ and $p(o_{j},q)$, respectively. The
probabilistic bisector $pb_{o_{i}o_{j}}$ of $o_{i}$ and $o_{j}$
partitions $D$ into two sub-spaces $D_{1}$ and $D_{2}$ such that
$p(o_{i},q)=p(o_{j},q)$ for any point $q\in pb_{o_{i}o_{j}}$,
$p(o_{i},q)>p(o_{j},q)$ for any point $q\in D_{1}$, and
$p(o_{i},q)<p(o_{j},q)$ for any point $q\in D_{2}$.
\end{definition}
}

A naive approach to compute the PVD requires the processing of PNN
queries by using Equation~\ref{eq:pnn} at every possible location
in the data space for determining the PVCs based on the calculated
probabilities. This approach is \emph{prohibitively} expensive in
terms of computational cost and thus impractical. In this paper,
we propose an efficient and practical solution for computing the
PVD for uncertain objects. Next, we show how to efficiently
compute PVDs, focusing on 1-dimensional (1D) and 2-dimensional
(2D) spaces.\eat{Note that, in a 1D data space, the probabilistic
bisector is a point that divides the space into two subspaces; on
the other hand, in a 2D data space, the probabilistic bisector is
a curve (or line) divides the space into two subspaces.} We briefly
discuss higher dimensional cases at the end of this section.

\subsection{Probabilistic Voronoi Diagram in a 1D Space}

Applications such as environmental
monitoring, feature extraction systems capture one dimensional
uncertain attributes, and store these values in a database. In
this section, we derive the PVD for 1D uncertain objects.

An uncertain 1D object $o_i$ can be represented as a range
$[l_i,u_i]$, where $l_i$ and $u_i$ are lower and upper bounds of
the range. Let $m_i$ and $n_i$ be the midpoint and the length of
the range $[l_i,u_i]$, i.e., $m_i=\frac{l_i+u_i}{2}$ and
$n_i=u_i-l_i$. The probabilistic bisector $pb_{o_{i}o_{j}}$ of two
1D objects $o_i$ and $o_j$ is a point $x$ within the range
$[min(l_i,l_j),max(u_i,u_j)]$ such that $p(o_i,x)=p(o_j,x)$, and
$p(o_i,x^\prime)>p(o_j,x^\prime)$ for any point $x^\prime<x$ and
$p(o_i,x^{\prime\prime})<p(o_j,x^{\prime\prime})$ for any point
$x^{\prime\prime}>x$. Since only the equality condition is not
sufficient, other two conditions must also hold. In our proof for
lemmas, we will show that a probabilistic bisector needs to
satisfy all three conditions.

For example, Figure~\ref{fig:line_lemma}(b) shows two uncertain
objects $o_1$ and $o_2$, and their probabilistic bisector
$pb_{o_{1}o_{2}}$ as a point $x$. In this example, the lengths of
range for $o_1$ and $o_2$ are $n_{1}=8$ and $n_{2}=4$,
respectively, and the minimum distances from $x$ to $o_1$ and
$o_2$ are $d_{1}=1$ and $d_{2}=3$, respectively. Then based on
Equation~\ref{eq:pnn}, we can compute the probabilities of $o_1$
and $o_2$ of being the NN to $x$ as follows:
\begin{align*}
p(o_{1},x)&= \frac{2}{8}\cdot\frac{4}{4}+\frac{1}{8}\cdot\frac{3}{4}+\frac{1}{8}\cdot\frac{2}{4}+\frac{1}{8}\cdot\frac{1}{4}=\frac{14}{32},
\end{align*}
and
\begin{align*}
p(o_{2},x) &= \frac{1}{4}\cdot\frac{5}{8}+\frac{1}{4}\cdot\frac{4}{8}+\frac{1}{4}\cdot\frac{3}{8}+\frac{1}{4}\cdot\frac{2}{8}=\frac{14}{32}.
\end{align*}

A naive approach for finding the $pb_{o_{i}o_{j}}$ requires the
computation of probabilities (using Equation~\ref{eq:pnn}) of
$o_i$ and $o_j$ for every position within the range
$[min(l_i,l_j),max(u_i,u_j)]$. To avoid high computational
overhead of this naive approach, in our method we show that for
two equi-range objects (i.e., $n_i=n_j$), we can always directly
compute the probabilistic bisector (see Lemma~\ref{lemma:one}) by
using the upper and lower bounds of two candidate objects.
Similarly, we also show that for two non-equi-range objects, where
$n_i\neq n_j$, we can directly compute the probabilistic bisector
for certain scenarios shown in
Lemmas~\ref{lemma:two}-\ref{lemma:three}, and for the remaining
scenarios of non-equi-range objects we exploit these lemmas to
find probabilistic bisectors at reduced computational cost.

\begin{figure}[htbp]
    \centering
        \includegraphics[width=1.8in]{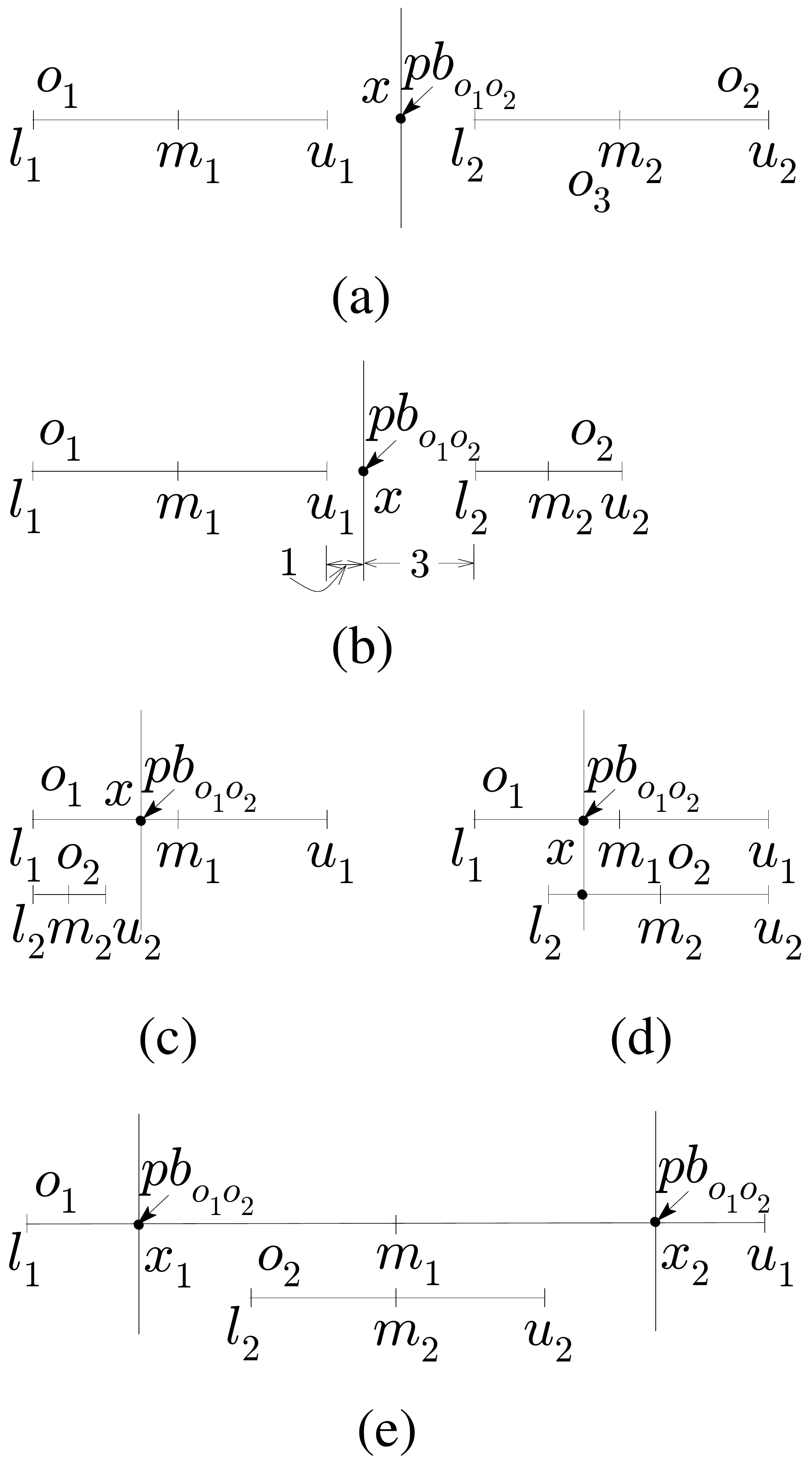}
    \caption{Scenarios of lemmas}
    \label{fig:line_lemma}
\end{figure}

Next, we present the lemmas for 1D objects. Lemma~\ref{lemma:one}
gives the probabilistic bisector of two equi-range objects,
overlapping and non-overlapping. Figure~\ref{fig:line_lemma}(a)
is an example of a non-overlapping case. (Note that if $l_i=l_j$
and $u_i=u_j$, then two objects $o_i$ and $o_j$ are assumed to be
the same and no probabilistic bisector exists between them.)

\begin{lemma}
\label{lemma:one} Let $o_i$ and $o_j$ be two objects where $m_{i}\neq m_{j}$. If
$n_i=n_j$, then the probabilistic bisector $pb_{o_{i}o_{j}}$ of
$o_i$ and $o_j$ is the bisector of $m_{i}$ and $m_{j}$.
\end{lemma}

\begin{proof}
Let $o_i$ and $o_j$ be two equi-range objects, i.e., $n_i=n_j$. Let $x$ be the bisector of two midpoints $m_i$ and $m_j$, i.e., $x=\frac{m_{i}+m_{j}}{2}$.

Then, by using Equation~\ref{eq:pnn}, we can calculate the probability of $o_i$ being the NN to $x$ as follows.
\begin{align*}
p(o_{i},x) &=
\sum_{s=1}^{n_{i}-1}\frac{1}{n_{i}}\frac{n_j-s}{n_{j}}.
\end{align*}
Similarly, we can calculate the probability of $o_j$ being the NN to $x$, as follows.
\begin{align*}
p(o_{j},x) &=
\sum_{s=1}^{n_{j}-1}\frac{1}{n_{j}}\frac{n_i-s}{n_{i}}.
\end{align*}

If we put $n_i=n_j$ in the above two equations, we have $p(o_{i},x)=p(o_{j},x)$. Thus, the probabilities of $o_i$ and $o_j$ of being the NN from the point $x$ are equal.

Now, let $x^{\prime}=\frac{m_{i}+m_{j}}{2}-\epsilon$ be a
point on the left side of $x$. Then we can calculate the probability of $o_i$ of being the NN to $x^{\prime}$
\begin{align*}
p(o_{i},x^{\prime}) &=
2\epsilon\frac{n_j}{n_{i}n_{j}}+\sum_{s=1}^{n_{i}-2\epsilon}\frac{1}{n_{i}}\frac{n_j-s}{n_{j}}.
\end{align*}
Similarly, we can calculate the probability of $o_j$ being the NN to $x^{\prime}$, as follows.
\begin{align*}
p(o_{j},x^{\prime}) &=
\sum_{s=2\epsilon+1}^{n_{i}-1}\frac{1}{n_{j}}\frac{n_i-s}{n_{i}}.
\end{align*}

Now, if we put  $n_i=n_j$ in the above two equations, then we have $p(o_{i},x^{\prime})>p(o_{j},x^{\prime})$ at $x^{\prime}$. Similarly we can prove that
$p(o_{i},x^{\prime\prime})<p(o_{j},x^{\prime\prime})$ for a point
$x^{\prime\prime}$ on the right side of $x$.

Thus, we can conclude that $x$ is the probabilistic bisector of $o_i$ and $o_j$, i.e., $pb_{o_{i}o_{j}}=x$.
\end{proof}

The following lemma shows how to compute the probabilistic
bisector of two non-equi-range objects that are non-overlapping
(see Figure~\ref{fig:line_lemma}(b)).

\begin{lemma}
\label{lemma:two} Let $o_i$ and $o_j$ be two non-overlapping
objects, where $n_i\neq n_j$. If there are no other objects within
the range $[min(l_i,l_j),max(u_i,u_j)]$, then the probabilistic
bisector $pb_{o_{i}o_{j}}$ of $o_i$ and $o_j$ is the bisector of
$m_{i}$ and $m_{j}$.
\end{lemma}

\begin{proof}
Let $n_i>n_j$, and $x$ be the bisector of two midpoints $m_i$ and $m_j$ of objects $o_i$ and $o_j$, respectively, i.e., $x=\frac{m_{i}+m_{j}}{2}$, and the minimum
distances from $x$ to $o_i$ and $o_j$ are $d_i$ and $d_j$,
respectively.

Then, by using Equation~\ref{eq:pnn}, we can calculate the probability of $o_i$ being the NN to $x$ as follows.
\begin{align*}
p(o_{i},x) &= (d_{j}-d_{i})\frac{1}{n_{i}}\frac{n_j}{n_{j}}+\sum_{s=1}^{n_{j}-1}\frac{1}{n_{i}}\frac{n_j-s}{n_{j}} \\
&= (d_{j}-d_{i})\frac{n_j}{n_{i}n_{j}}+\frac{n_{j}(n_{j}-1)}{2n_{i}n_{j}}.
\end{align*}
Similarly, we can calculate the probability of $o_j$ being the NN to $x$ as follows.
\begin{align*}
p(o_{j},x) &=
\sum_{s=1}^{n_{j}}\frac{1}{n_{j}}\frac{n_{i}-(d_{j}-d_{i}+s)}{n_{i}}.
\end{align*}
Since, we have $d_j-d_i=\frac{n_i-n_j}{2}$, i.e.,
$n_i=2(d_j-d_i)+n_j$. By replacing $n_i$ in the numerator of
$p(o_{j},x)$, we can have the following,
\begin{align*}
p(o_{j},x) &= \sum_{s=1}^{n_{j}}\frac{1}{n_{j}}\frac{2(d_j-d_i)+n_j-(d_j-d_i+s)}{n_{i}} \\
&= (d_{j}-d_{i})\frac{n_j}{n_{i}n_{j}}+\frac{n_{j}(n_{j}-1)}{2n_{i}n_{j}}.
\end{align*}

Since $p(o_{i},x)=p(o_{j},x)$, we have $pb_{o_{i}o_{j}}=x$.

On the other hand, let $x^{\prime}=\frac{m_{i}+m_{j}}{2}-\epsilon$ be a
point on the left side of the probabilistic bisector.

Then, by using Equation~\ref{eq:pnn}, we can calculate the probability of $o_i$ being the NN to $x^{\prime}$ as follows.

\begin{align*}
p(o_{i},x^{\prime}) &=
(d_{j}-d_{i}+2\epsilon)\frac{n_j}{n_{i}n_{j}}+\frac{n_{j}(n_{j}-1)}{2n_{i}n_{j}}.
\end{align*}

Similarly, we can calculate the probability of $o_j$ being the NN to $x^{\prime}$, as follows.

\begin{align*}
p(o_{j},x^{\prime}) &=
(d_{j}-d_{i}-2\epsilon)\frac{n_j}{n_{i}n_{j}}+\frac{n_{j}(n_{j}-1)}{2n_{i}n_{j}}.
\end{align*}

So, we can say $p(o_{i},x^{\prime})>p(o_{j},x^{\prime})$ for a point
$x^{\prime}$ on the left side of $pb_{o_{i}o_{j}}$. Similarly we can prove that
$p(o_{i},x^{\prime\prime})<p(o_{j},x^{\prime\prime})$ for a point
$x^{\prime\prime}$ on the right side of $pb_{o_{i}o_{j}}$.
\end{proof}

For two non-equi-range objects that are overlapping, the following
lemma directly computes the probabilistic bisector for the
scenarios where lower, upper, or mid-point values of two candidate
objects are same (see Figure~\ref{fig:line_lemma}(c), (d), and
(e)).

\begin{lemma}
\label{lemma:three} Let $o_i$ and $o_j$ be two overlapping
objects, where $n_i\neq n_j$, $l_{i}\leq l_{j}\leq u_{j}\leq
u_{i}$, and there are no other objects within the range
$[min(l_i,l_j),max(u_i,u_j)]$.
\begin{enumerate}
    \item If $l_{i}=l_{j}$, then the probabilistic
bisector $pb_{o_io_j}$ of $o_i$ and $o_j$ is the bisector of
$m_{i}$ and $u_{j}$.
    \item If $u_{i}=u_{j}$, then
the probabilistic bisector $pb_{o_io_j}$ of $o_i$ and $o_j$ is the
bisector of $m_{i}$ and $l_{j}$.
    \item If $m_{i}=m_{j}$, then the
probabilistic bisectors $pb_{o_io_j}$ of $o_i$ and $o_j$ are the
bisectors of $l_{i}$ and $l_{j}$, and  $u_{i}$ and $u_{j}$.
\end{enumerate}
\end{lemma}

\begin{proof}
Let $n_i>n_j$, $l_i=l_j$, $x=\frac{m_{i}+l_{j}}{2}$, and $d$ be
the distance from $x$ to both $m_i$ and $l_j$.

Then, by using Equation~\ref{eq:pnn}, we can calculate the probability of $o_i$ being the NN to $x$ as follows.
\begin{align*}
p(o_{i},x) &=
\sum_{s=1}^{d}\frac{2}{n_{i}}\frac{n_j}{n_{j}}+\sum_{s=1}^{n_j-1}\frac{2}{n_{i}}\frac{n_j-s}{n_{j}}\\
 &= \frac{2d}{n_{i}}+\frac{n_j(n_j-1)}{n_{i}n_{j}}.
\end{align*}
Similarly, we can calculate the probability of $o_j$ being the NN to $x$ as follows.
\begin{align*}
p(o_{j},x) &=
\sum_{s=1}^{n_j}\frac{1}{n_{j}}\frac{n_i-(2d+2s)}{n_{i}}.
\end{align*}
However, $\frac{n_i}{2}-n_j=2d$, that is $n_i=4d+2n_j$. By
replacing $n_i$ in the numerator and simplifying the term, we can
have the following,
$p(o_{j},x)=\frac{2d}{n_{i}}+\frac{n_j(n_j-1)}{n_{i}n_{j}}$. Since
$p(o_{i},x)=p(o_{j},x)$, $pb_{o_{i}o_{j}}=x$. Similar to
Lemma~\ref{lemma:two}, we can prove that
$p(o_{i},x^{\prime})>p(o_{j},x^{\prime})$ for any point
$x^{\prime}$ on the left, and
$p(o_{i},x^{\prime\prime})<p(o_{j},x^{\prime\prime})$ for any
point $x^{\prime\prime}$ on the right side of $pb_{o_{i}o_{j}}$.

Similarly, we can prove the case for $u_{i}=u_{j}$.

Let $m_i=m_j$, $x_{1}=\frac{l_{i}+l_{j}}{2}$, and $d$ be the
distance from $x_{1}$ to both $l_i$ and $l_j$.

Then, by using Equation~\ref{eq:pnn}, we can calculate the probability of $o_i$ being the NN to $x_{1}$ as follows.

\begin{align*}
p(o_{i},x_{1}) &=
\sum_{s=1}^{d}\frac{2}{n_{i}}\frac{n_j}{n_{j}}+\sum_{s=1}^{n_j-1}\frac{1}{n_{i}}\frac{n_j-s}{n_{j}}\\
&= \frac{2d}{n_{i}}+\frac{n_j(n_j-1)}{2n_{i}n_{j}}.
\end{align*}
Similarly, we can calculate the probability of $o_j$ being the NN to $x_{1}$ as follows.
\begin{align*}
p(o_{j},x_{1})=\sum_{s=1}^{n_j}\frac{1}{n_{j}}\frac{n_i-(2d+s)}{n_{i}}.
\end{align*}
However, $\frac{n_i}{2}-\frac{n_j}{2}=2d$, that is $n_i=4d+n_j$.
By replacing $n_i$ in the numerator and simplifying the term, we
can have the following,
$p(o_{j},x_{1})=\frac{2d}{n_{i}}+\frac{n_j(n_j-1)}{2n_{i}n_{j}}$.
Since $p(o_{i},x^\prime)=p(o_{j},x^\prime)$, we have
$pb_{o_{i}o_{j}}=x_{1}$. Similar to Lemma~\ref{lemma:two}, we can
prove that $p(o_{i},x^{\prime})>p(o_{j},x^{\prime})$ for any point
$x^{\prime}$ on the left, and
$p(o_{i},x^{\prime\prime})<p(o_{j},x^{\prime\prime})$ for any
point $x^{\prime\prime}$ on the right side of $pb_{o_{i}o_{j}}$.

Similarly, we can prove that the other probabilistic bisector
exists at $x_{2}=\frac{u_{i}+u_{j}}{2}$, as the case is symmetric to that of $x_1$.

Note that, since $n_i>n_j$ and $m_i=m_j$, $o_i$ completely contains $o_j$. Thus the probability of $o_j$ is higher than that of $o_i$ around the mid-point ($m_i$), and the probability of $o_i$ is higher than that of $o_j$ towards the boundary points ($l_i$ and $u_i$). Therefore in this case, we have two probabilistic bisectors between $o_i$ and $o_j$.
\end{proof}

Figures~\ref{fig:line_lemma}(c-e) show an example of three cases
as described in Lemma~\ref{lemma:three}.
Figure~\ref{fig:line_lemma}(c) shows the first case for objects
$o_1$ and $o_2$, where $l_1=l_2$ and
$pb_{o_1o_2}=\frac{m_1+u_2}{2}$. Similarly, Figure~\ref
{fig:line_lemma}(d) shows an example of the second case for
objects $o_1$ and $o_2$, where $u_1=u_2$ and
$pb_{o_1o_2}=\frac{m_1+l_2}{2}$. Finally, Figure~\ref
{fig:line_lemma}(e) shows an example of the third case for objects
$o_1$ and $o_2$, where $m_1=m_2$, and $x_1=\frac{l_1+l_2}{2}$ and
$x_2=\frac{u_1+u_2}{2}$ are two probabilistic bisectors. In such a
case, two probabilistic bisectors, $x_1$ and $x_2$, divide the
space into three subspaces. That means, the Voronoi cell of object
$o_1$ comprises of two disjoint subspaces. In Figure~\ref{fig:line_lemma}(e), the
subspace left to $x_1$ and the subspace right to $x_2$ form the
Voronoi cell of $o_1$, and the subspace bounded by $x_1$ and $x_2$
forms the Voronoi cell of $o_2$.

Apart from the above mentioned scenarios, the remaining scenarios
of two overlapping non-equi-range objects are shown in
Figure~\ref{fig:line_lemmano}, where it is not possible to compute
the probabilistic bisector directly by using lower and upper
bounds of two candidate objects. In these scenarios,
Lemma~\ref{lemma:three} can be used for choosing a point, called
the initial probabilistic bisector, which approximates the actual
probabilistic bisector and thereby reducing the computational
overhead. Figure~\ref{fig:line_lemmano} (a), (b), (c) show three
scenarios, where three cases of Lemma~\ref{lemma:three} (1), (2),
(3), are used to compute the initial probabilistic bisector,
respectively, for our algorithm. We will see (in Algorithm~\ref{algo:Prob-Bisectors-OneDim}) how to use our lemmas to find the probabilistic bisectors for these scenarios.

\begin{figure}[htbp]
    \centering
        \includegraphics[width=2.0in]{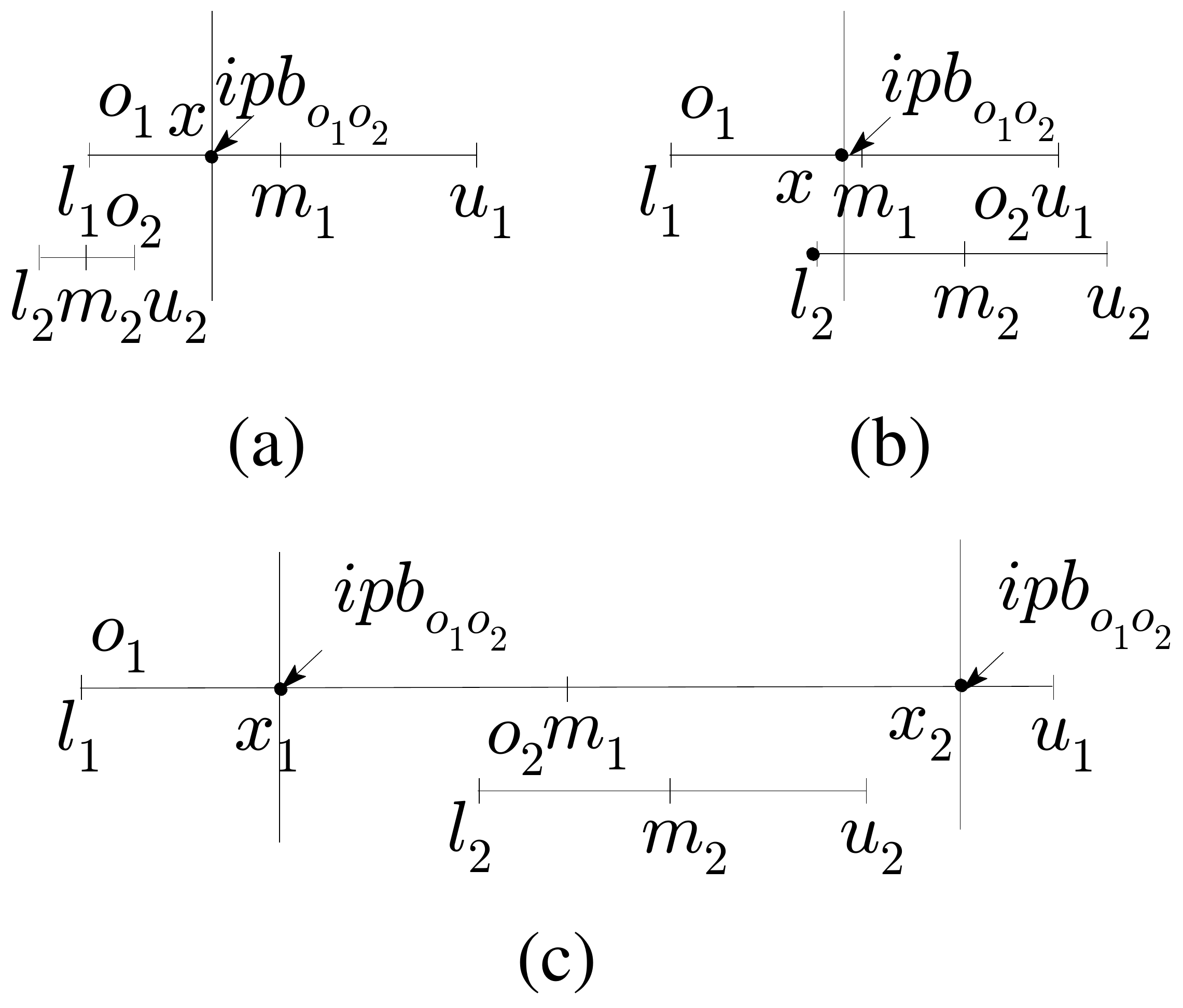}
    \caption{Remaining scenarios}
    \label{fig:line_lemmano}
    \vspace{0mm}
\end{figure}

So far we have assumed that no other objects exist within the
ranges of two candidate objects. However, the probabilities of two
candidate objects may change in the presence of any other objects
within their ranges (as shown in Equation~\ref{eq:pnn}). Only the probabilistic bisector of two equi-range objects remains
the same in the presence of any other object within their ranges.

Let $o_k$ be the third object that overlaps with the range
$[min(l_i,l_j),max(u_i,u_j)]$ for the case in Figure~\ref{fig:line_lemma}(a). Then, using Equation~\ref{eq:pnn}, we can calculate the NN probability of object $o_i$ from $x$ as follows.

\begin{align*}
p(o_{i},x) &=
\sum_{s=1}^{n_{i}-1}\frac{1}{n_{i}}\frac{n_j-s}{n_{j}}\frac{n_k-s}{n_{k}}.
\end{align*}
Similarly, we can calculate the NN probability of object $o_j$ from $x$ as follows.
\begin{align*}
p(o_{j},x) &=
\sum_{s=1}^{n_{j}-1}\frac{1}{n_{j}}\frac{n_i-s}{n_{i}}\frac{n_k-s}{n_{k}}.
\end{align*}

Since $n_i=n_j$, we have $p(o_{i},x)=p(o_{j},x)$ and
$pb_{o_{i}o_{j}}=x$. Therefore, the probabilistic bisector
$pb_{o_{i}o_{j}}$ does not change with the presence a third
object.

Therefore, for scenarios, except for the case when two candidate
objects are equi-range, when any other object exists within the
ranges two candidate objects, we again use one of the
Lemmas~\ref{lemma:one}-\ref{lemma:three} to compute the initial
probabilistic bisector, and then find the actual probabilistic
bisector. For example, if two non-equi-range candidate objects do
not overlap each other (see Figure~\ref{fig:line_lemma}(b)) and a
third object exists, which is not shown in figure, within the
range of these two candidate objects, then we use
Lemma~\ref{lemma:two} to find the initial probabilistic bisector.
Similarly, we choose the corresponding lemmas for other scenarios
to compute initial probabilistic bisectors. Then we use these
computed initial probabilistic bisectors to find actual
probabilistic bisectors.

The position of a probabilistic bisector depends on the relative
positions and the uncertainty regions of two candidate objects. We
have shown that for some scenarios the probabilistic bisectors can
be directly computed using the proposed lemmas. In some other
scenarios, there is no straightforward way to compute
probabilistic bisectors. For this latter case, the initial
probabilistic bisector of two candidate objects is chosen based on
the actual probabilistic bisector of the scenario that can be
directly computed and has the most similarity (relative positions
of candidate objects) with two candidate objects. This ensures
that the initial probabilistic bisector is essentially close to
the actual probabilistic bisector.

\setlength{\algomargin}{2em} \dontprintsemicolon
\begin{algorithm}[htbp]
\begin{small}
\caption{ProbBisector1D($o_i,o_j,O$)}
\label{algo:Prob-Bisectors-OneDim}
    $pb_{o_{i}o_{j}}\leftarrow\emptyset$\;

 \uIf{$o_i$ and $o_j$ satisfy one of the Lemmas~\ref{lemma:one}-~\ref{lemma:three}}
    {
         $pb_{o_{i}o_{j}} \leftarrow BisectorBasedOnLemmas(o_{i},o_{j},O)$\;
    }
    \Else
    {
    (any pair of objects that does not satisfy Lemmas~\ref{lemma:one}-~\ref{lemma:three})

        \uIf{$o_i$ and $o_j$ do not overlap}
        {
             $ipb\leftarrow Bisector(m_{i},m_{j})$\;
             $pb_{o_{i}o_{j}}\leftarrow FindProbBisector1D(o_i,o_j,ipb,O)$\;
        }
        \Else
        {
        (three possible cases for overlapping pairs of objects (Figure~\ref{fig:line_lemmano}))

            (assume $l_i\leq l_j$ (the other case $l_i\geq l_j$ is
            symmetric))\;
            \uIf{$l_{i}<l_{j}$ and $u_{j}<u_{i}$}
             {
                $ipb_{1}\leftarrow Bisector(l_{i},l_{j})$\;
                $ipb_{2}\leftarrow Bisector(u_{i},u_{j})$\;

                $pb_{o_{i}o_{j}}\leftarrow FindProbBisector1D(o_i,o_j,ipb_{1},O)\cup FindProbBisector1D(o_i,o_j,ipb_{2},O)$\;
             }
            \ElseIf{$l_{j}-l_{i}<u_{j}-u_{i}$}
            {
                $ipb\leftarrow Bisector(m_{j},u_{i})$\;
                $pb_{o_{i}o_{j}}\leftarrow FindProbBisector1D(o_i,o_j,ipb,O)$\;
             }
             \Else{ 
                $ipb\leftarrow Bisector(m_{i},l_{j})$\;
                $pb_{o_{i}o_{j}}\leftarrow FindProbBisector1D(o_i,o_j,ipb,O)$\;
             }
        }
    }
    \Return $pb_{o_{i}o_{j}}$;
\end{small}
\end{algorithm}

\noindent\emph{\textbf{Algorithms: }}\par Based on the above
lemmas, Algorithm~\ref{algo:Prob-Bisectors-OneDim} summarizes the
steps of computing the probabilistic bisector $pb_{o_{i}o_{j}}$
for any two objects $o_i$ and $o_j$, where $O$ is a given set of
objects and $o_i,o_j \in O$. If $o_i$ and $o_j$ satisfy any of
Lemmas~\ref{lemma:one}-\ref{lemma:three} the algorithm directly
computes $pb_{o_{i}o_{j}}$
(Lines~\ref{algo:Prob-Bisectors-OneDim}.2-~\ref{algo:Prob-Bisectors-OneDim}.3).
Otherwise, if any other object exists within the range of two
candidate non-equi-range objects $o_i$ and $o_j$, or two candidate
non-equi-range objects fall in any of the scenarios shown in
Figure~\ref{fig:line_lemmano}. The algorithm first computes an
initial probabilistic bisector $ipb$ using our lemmas, where the
given scenario has the most similarity in terms of relative
positions of candidate objects to the corresponding lemma. Then,
the algorithm uses the function $FindProbBisector1D$ to find
$pb_{o_{i}o_{j}}$ by using $ipb$ as a base.

\eat{ When two non-equi-range objects are non-overlapping and any
other object exists within their ranges, the algorithm computes
$ipb$ using Lemma~\ref{lemma:two}
(Line~\ref{algo:Prob-Bisectors-OneDim}.6). We choose
Lemma~\ref{lemma:two} for computing $ipb$ because,
Lemma~\ref{lemma:two} can directly find the probabilistic bisector
for any two non-equi-range non-overlapping objects when no other
object presents within the range these two candidate objects, and
if a third objects exist within the range of these two objects,
the probabilistic bisector slightly deviates from the case when no
third object exists. Similarly, we choose $ipb$ for other possible
cases using lemmas. If the ranges of two candidate objects overlap
each other, they satisfy one of the three conditions shown in
Lines~\ref{algo:Prob-Bisectors-OneDim}.10,~\ref{algo:Prob-Bisectors-OneDim}.14,
and~\ref{algo:Prob-Bisectors-OneDim}.17. When one object
completely contains the other object, the algorithm uses
Lemma~\ref{lemma:three}(3) to compute $ipb$
(Lines~\ref{algo:Prob-Bisectors-OneDim}.11-~\ref{algo:Prob-Bisectors-OneDim}.12).
When the lower bounds of two candidate objects are closer to each
other than that of the upper bounds, the algorithm computes $ipb$
using Lemma~\ref{lemma:three}(1)
(Line~\ref{algo:Prob-Bisectors-OneDim}.15). If the above two
conditions fail (e.g., Figure~\ref{fig:line_lemmano}(b)), the
algorithm uses Lemma~\ref{lemma:three}(2) to compute $ipb$
(Line~\ref{algo:Prob-Bisectors-OneDim}.18). }

After computing the $ipb$ the
algorithm calls a function $FindProbBisector1D$ to find the probabilistic bisector
$pb_{o_{i}o_{j}}$
(Lines~\ref{algo:Prob-Bisectors-OneDim}.8,~\ref{algo:Prob-Bisectors-OneDim}.15,~\ref{algo:Prob-Bisectors-OneDim}.18,
and~\ref{algo:Prob-Bisectors-OneDim}.21).

The function $FindProbBisector1D$ computes
$pb_{o_{i}o_{j}}$ by refining $ipb$. If the probabilities of $o_i$
and $o_j$ of being the NN from $ipb$ are equal, then the algorithm
returns $ipb$ as the probabilistic bisector. Otherwise, the
algorithm decides in which direction from $ipb$ it should continue
the search for $pb_{o_{i}o_{j}}$. Let $x=ipb$. We also assume that $o_i$ is left to $o_j$. If $p(o_i,x)$ is
smaller than $p(o_j,x)$, then $pb_{o_{i}o_{j}}$ is to the left of
$x$ and within the range $[min(l_i,l_j),x]$, otherwise
$pb_{o_{i}o_{j}}$ is to the right of $x$ and within the range
$[x,max(l_i,l_j)]$. Since using lemmas, we choose $ipb$ as close
as possible to $pb_{o_{i}o_{j}}$, in most of the cases the
probabilistic bisector is found very close to the position of
$ipb$. Thus, as an alternative to directly running a binary search
within the range, one can perform a step-wise search first, by
increasing (or decreasing) the value of $x$ \eat{(in every step by
$2^i$, where $i=1$ unit at step 1, $i=2$ units at step 2, $i=4$
units at step 3, and so on) }until the probability ranking of two
objects swaps. Since the precision of probability measures affects
the performance of the above search, we assume that the two
probability measures are equal when the difference between them is
smaller than a threshold. The value of the threshold can be found
experimentally given an application domain.

\setlength{\algomargin}{2em} \dontprintsemicolon
\begin{algorithm}[htbp]
\begin{small}
 \caption{ProbVoronoi1D($O$)}
 \label{algo:Prob-Voronoi-OneDim}
    $PVD\leftarrow\emptyset$\;
    $PBL\leftarrow\emptyset$\;
    $SortObjects(O)$\;
    \For{each $o_{i}\in O$}
    {
        $o_{j}\leftarrow getNext(O)$
        $pb_{o_{i}o_{j}}\leftarrow ProbBisector1D(o_{i},o_{j},O)$
        $PBL\leftarrow PBL \cup pb_{o_{i}o_{j}}$\;
        $N\leftarrow getCandidateObjects(O,pb_{o_{i}o_{j}})$\;
        \For{each $o_{k}\in N$}
         {
            $pb_{o_{i}o_{k}}\leftarrow ProbBisector1D(o_{i},o_{k},O)$
            $PBL\leftarrow PBL \cup pb_{o_{i}o_{k}}$\;

         }

    }
    $SPBL\leftarrow\ SortProbBisectors(PBL)$\;

    $o^{\prime} \leftarrow initialMostProbableObject()$\;

    \While{$SPBL$ is not empty}
    {
        $pb_{o_{i}o_{j}} \leftarrow popNextPB(SPBL)$\;
        $left \leftarrow LeftSideObject()$\;
        $right \leftarrow RightSideObject()$\;
        \uIf{$PVD$ is empty OR $o^{\prime} = left$}
            {
                $PVD \leftarrow PVD \cup ProbBisector(pb_{o_{i}o_{j}},o_{i},o_{j})$\;
                $o^{\prime} \leftarrow right$ \;

            }
            \Else
            {
                $Discard(pb_{o_{i}o_{j}})$\;
            }
    }

   \Return $PVD$;
\end{small}
\end{algorithm}

Finally, Algorithm~\ref{algo:Prob-Voronoi-OneDim} shows the steps
for computing a PVD for a set of 1D uncertain objects $O$. In 1D
data space, the PVD contains a list of bisectors that divides the
total data space into a set of Voronoi cells or 1D ranges. The
basic idea of Algorithm~\ref{algo:Prob-Voronoi-OneDim} is that,
once we have the probabilistic bisectors of all pairs of objects
in a sorted list, a sequential scan of the list can find the
candidate probabilistic bisectors that comprise the probabilistic
Voronoi diagram in 1D space.

To avoid computing probabilistic bisectors for all pairs of
objects $o_i,o_j\in O$, we use the following heuristic:

\begin{heuristic}
\label{heuristic:one} Let $o_i$ be an object in the ordered
(in ascending order of $l_i$) list of objects $O$, and $o_j$ be
the next object right to $o_i$ in $O$. Let $x=pb_{o_{i}o_{j}}$,
and $d=dist(x,l_i)$. Let $o_k$ be an object in $O$. If
$dist(x,l_{k})>d$, then the probabilistic bisector
$pb_{o_{i}o_{k}}$ of $o_i$ and $o_k$ is $x^{\prime}$, and
$x^{\prime}$ is to the right of $x$, i.e., $x^{\prime}>x$;
therefore $pb_{o_{i}o_{k}}$ does not need to be computed.
\end{heuristic}

Algorithm~\ref{algo:Prob-Voronoi-OneDim} runs as follows. First,
the algorithm sorts all objects in ascending order of their lower
bounds (Line~\ref{algo:Prob-Voronoi-OneDim}.3). Second, for each
object $o_i$, it computes probabilistic bisectors of $o_i$ with
the next object $o_j\in O$ and with a set $N$ of objects returned
by the function $getCandidateObjects$ based on
Heuristic~\ref{heuristic:one}
(Lines~\ref{algo:Prob-Voronoi-OneDim}.4-\ref{algo:Prob-Voronoi-OneDim}.8).
$PBL$ maintains the list all computed probabilistic bisectors.
Third, the algorithm sorts the list $PBL$ in ascending order of
the position of probabilistic bisectors and assigns the sorted
list to $SPBL$ (Line~\ref{algo:Prob-Voronoi-OneDim}.9). Finally,
from $SPBL$, the algorithm selects probabilistic bisectors that
contribute to the PVD
(Lines~\ref{algo:Prob-Voronoi-OneDim}.10-\ref{algo:Prob-Voronoi-OneDim}.19).
For this final step, the algorithm first finds the most probable
NN $o^\prime$ with respect to the starting position of the data
space. Then for each $pb_{o_{i}o_{j}}\in SPBL$, the algorithm
decides whether $pb_{o_{i}o_{j}}$ is a candidate for the $PVD$
(Lines~\ref{algo:Prob-Voronoi-OneDim}.11-\ref{algo:Prob-Voronoi-OneDim}.19).
We assume that $o_i$ is the left side object and $o_j$ is the
right side object of the probabilistic bisector. If
$o^\prime=o_i$, then $pb_{o_{i}o_{j}}$ is included in the $PVD$,
and $o^\prime$ is updated with the most probable object on the
right region of $pb_{o_{i}o_{j}}$
(Line~\ref{algo:Prob-Voronoi-OneDim}.17). Otherwise,
$pb_{o_{i}o_{j}}$ is discarded
(Line~\ref{algo:Prob-Voronoi-OneDim}.19). This process continues
until $SPBL$ becomes empty, and the algorithm finally returns
$PVD$.

The proof of correctness and the complexity of this algorithm are
provided as follows.

\emph{Correctness}: Let $SPBL$ be the list of probabilistic
bisectors in ascending order of their positions. Let $o^{\prime}$
be the most probable NN with respect to the starting point $l$ of
the 1D data space. Let $pb_{o_{i}o_{j}}$ be the next probabilistic
bisector fetched from $SPBL$. Now we can have the following two
cases: (i) Case 1: $o^{\prime}=o_{i}$. The probability $p_i$ of
$o_i$ being the nearest is the highest for all points starting
from $l$ to $pb_{o_{i}o_{j}}$ and the probability $p_j$ of $o_j$
being the nearest is the highest for points on the right side of
$pb_{o_{i}o_{j}}$ until the next valid probabilistic bisector is
found. Hence, $pb_{o_{i}o_{j}}$ is a valid probabilistic bisector
and is added to the PVD. Then the algorithm updates $o^{\prime}$
by $o_j$ since $o_j$ will be the most probable on the right of
$pb_{o_{i}o_{j}}$ and will be on the left region of the next valid
probabilistic bisector. (ii) Case 2: $o^{\prime}\neq o_{i}$. Let
us assume that $p_i>p^{\prime}$ at $pb_{o_{i}o_{j}}$. We already
know that $p^{\prime}>p_i$ at the starting point $l$. So there
should be some point within the range [$l,pb_{o_{i}o_{j}}$] where
$p^{\prime}=p_i$, which is the position of the probabilistic
bisector of $o^{\prime}$ and $o_i$. Since no such bisector is
found within this range, $p_i>p^{\prime}$ is not true at
$pb_{o_{i}o_{j}}$. Thus, $p^{\prime}$ is the highest even at
$pb_{o_{i}o_{j}}$, and will remain the highest until it fetches
another $pb_{o_{i^\prime}o_{j^\prime}}$ from $SPBL$, where
$o^{\prime}=o_{i^\prime}$.  The above process continues until the
algorithm reaches the end of the data space.

\eat{This is because according the definition~\ref{def:pb}, $p_{i}
= p_{j}$ at $pb_{o_{i}o_{j}}$, $p_{i}>p_{j}$ for points on the
left side of $pb_{o_{i}o_{j}}$, and $p_{i}<p_{j} for points on the
right side of $pb_{o_{i}o_{j}}$}

\emph{Complexity}: The complexity of
Algorithm~\ref{algo:Prob-Voronoi-OneDim} can be determined as
follows. Let $C_b$ be the cost of computing the probability of an
object being the NN of a query point, and $C_{pb}$ be the cost of
finding the probabilistic bisector of two objects. The complexity
of Algorithm~\ref{algo:Prob-Voronoi-OneDim} is dominated by the
complexity of executing the
Lines~\ref{algo:Prob-Voronoi-OneDim}.4-\ref{algo:Prob-Voronoi-OneDim}.8,
which is $O(nNC_{pb})$, where $n$ is the total number of objects,
and $N$ is the expected number of probabilistic bisectors that
need to be computed for each object in $O$. For real data sets, $N$ is found to be a small value since each object has a small number of surrounding objects (in the worst case it can be $n-1$). The cost of
$C_{pb}=O(C_{b}\log_{2}D)$, where $D$ is the expected distance
between our initial probabilistic bisector $ipb$ and the actual
probabilistic bisector. This is because, the cost of finding a
probabilistic bisector is $O(1)$ for the cases when our algorithm
can directly compute the probabilistic bisector, and for other
cases our algorithm first finds $ipb$ by $O(1)$ and then searches
for the actual probabilistic bisector using $FindProbBisector1D$ by
$O(\log D)$.

\subsection{Probabilistic Voronoi Diagram in a 2D Space}

In location-based applications, locations of objects such as a
passenger and a building, in a 2D space can be uncertain due to
the imprecision of data capturing devices or the privacy concerns of users. In these
applications, the location of an object $o_i$ can be represented
as a circular region $R_i=(c_i,r_i)$, where $c_i$ is the center
and $r_i$ is the radius of the region, and  the actual location of
$o_i$ can be anywhere in $R_i$. The area of $o_i$ is expressed as
$A_i=\pi r_{i}^2$. In this section, we derive the PVD for 2D
uncertain objects.

Similar to the 1D case, a naive approach to find the probabilistic
bisector $pb_{o_{i}o_{j}}$ of $o_i$ and $o_j$ requires an
exhaustive computation of probabilities using
Equation~\ref{eq:pnn} for every position in a large area. In our
approach, we first show that we can directly compute
$pb_{o_{i}o_{j}}$ as the bisector $bs_{c_ic_j}$ of $c_i$ and $c_j$
when two candidate objects are equi-range (i.e., $r_i=r_j$). Next,
we show that for two non-equi-range objects (i.e., $r_i\neq r_j$),
depending on radii and relative positions of objects
$pb_{o_{i}o_{j}}$ slightly shifts from $bs_{c_ic_j}$. In this
case, we use $bs_{c_ic_j}$ to choose a line, called the initial
probabilistic bisector, to approximate the actual probabilistic
bisector $pb_{o_{i}o_{j}}$. Although for simplicity of
presentation, we will use examples where two candidate objects are
non-overlapping, Lemmas~\ref{lemma:five}-\ref{lemma:eight} also
hold for overlapping objects.

For two equi-range uncertain circular objects $o_i$ and $o_j$, we
have the following lemma:

\begin{lemma}
\label{lemma:five} Let $o_i$ and $o_j$ be two circular uncertain
objects with uncertain regions $(c_i,r_i)$ and $(c_j,r_j)$, respectively. If $r_{i}=r_{j}$, then the probabilistic bisector
$pb_{o_{i}o_{j}}$ of $o_i$ and $o_j$ is the bisector $bs_{c_ic_j}$
of $c_{i}$ and $c_{j}$.
\end{lemma}

\begin{proof}
Let $x$ be any point on $bs_{c_ic_j}$, and $d=mindist(x,o_i)$(or
$mindist(x,o_j)$). Let there be no other objects within the
circular range centered at $x$ with radius $d+2r_i$. Suppose
circles centered at $x$ with radii $d+1$ to $d+2r_i$ partition
$o_i$ into $2r_i$ sub-regions
$o_{i_{1}},o_{i_{2}},...,o_{i_{2r_{i}}}$, such that
$\sum_{s=1}^{2r_i}\frac{o_{i_{s}}}{A_i}=1$. Similarly, $o_j$ is
divided into $2r_i$ sub-regions
$o_{j_{1}},o_{j_{2}},...,o_{j_{2r_{i}}}$, where
$\sum_{s=1}^{2r_i}\frac{o_{j_{s}}}{A_j}=1$. By using Equation~\ref{eq:pnn}, we can calculate the
probability of $o_i$ being the nearest from $x$, as follows.
\begin{align*}
p(o_i,x) &=
\sum_{s=d+1}^{2r_i+d}\frac{o_{i_{s-d}}}{A_i}(1-\sum_{u=d+1}^{s}\frac{o_{j_{u-d}}}{A_j}).
\end{align*}
Similarly, we can calculate the
probability of $o_j$ being the nearest from $x$, as follows.
\begin{align*}
p(o_j,x) &=
\sum_{s=d+1}^{2r_i+d}\frac{o_{j_{s-d}}}{A_j}(1-\sum_{u=d+1}^{s}\frac{o_{i_{u-d}}}{A_i}).
\end{align*}
Since, $r_i=r_j$ and $o_{i_{s}}=o_{j_{s}}$ for all $1\leq s\leq
2r_{i}$, we have $p(o_i,x)=p(o_j,x)$.
\eat{ Let $o_k$ be the third
object that overlaps with the circular boundary centered at $x$
with radius $d+2r_{i}$. Let $d^\prime$ be the $mindist(x,o_k)$.
Suppose circles with radii $d^\prime+1$ to $d^\prime+2r_k$
partition $o_k$ into $2r_k$ sub-regions
$o_{k_{1}},o_{k_{2}},...,o_{k_{2r_{k}}}$, where
$\sum_{s=1}^{2r_k}\frac{o_{k_{s}}}{Area(R_k)}=1$. We can
recalculate
\begin{align*}
&p(o_i,x)=\\&\sum_{s=d+1}^{2r_i+d}\frac{o_{i_{s-d}}}{A_i}(1-\sum_{u=d+1}^{s}\frac{o_{j_{u-d}}}{A_j})(1-\sum_{v=d^\prime+1}^{s}\frac{o_{k_{v-d^\prime}}}{A_k}),
\end{align*}
and
\begin{align*}
&p(o_j,x)=\\&\sum_{s=d+1}^{2r_1+d}\frac{o_{j_{s-d}}}{A_j}(1-\sum_{u=d+1}^{s}\frac{o_{i_{u-d}}}{A_i})(1-\sum_{v=d^\prime+1}^{s}\frac{o_{k_{v-d^\prime}}}{A_k}).
\end{align*}
Since $o_k$ contributes to both $p(o_i,x)$ and $p(o_j,x)$ in the
same proportion, we have $p(o_i,x)=p(o_j,x)$.}
\end{proof}

The probabilistic bisector $pb_{o_1o_2}$ of two equi-range objects
$o_1$ and $o_2$ is shown in Figure~\ref{fig:circle_lemma1}.

\begin{figure}[htbp]
    \centering
        \includegraphics[width=1.5in]{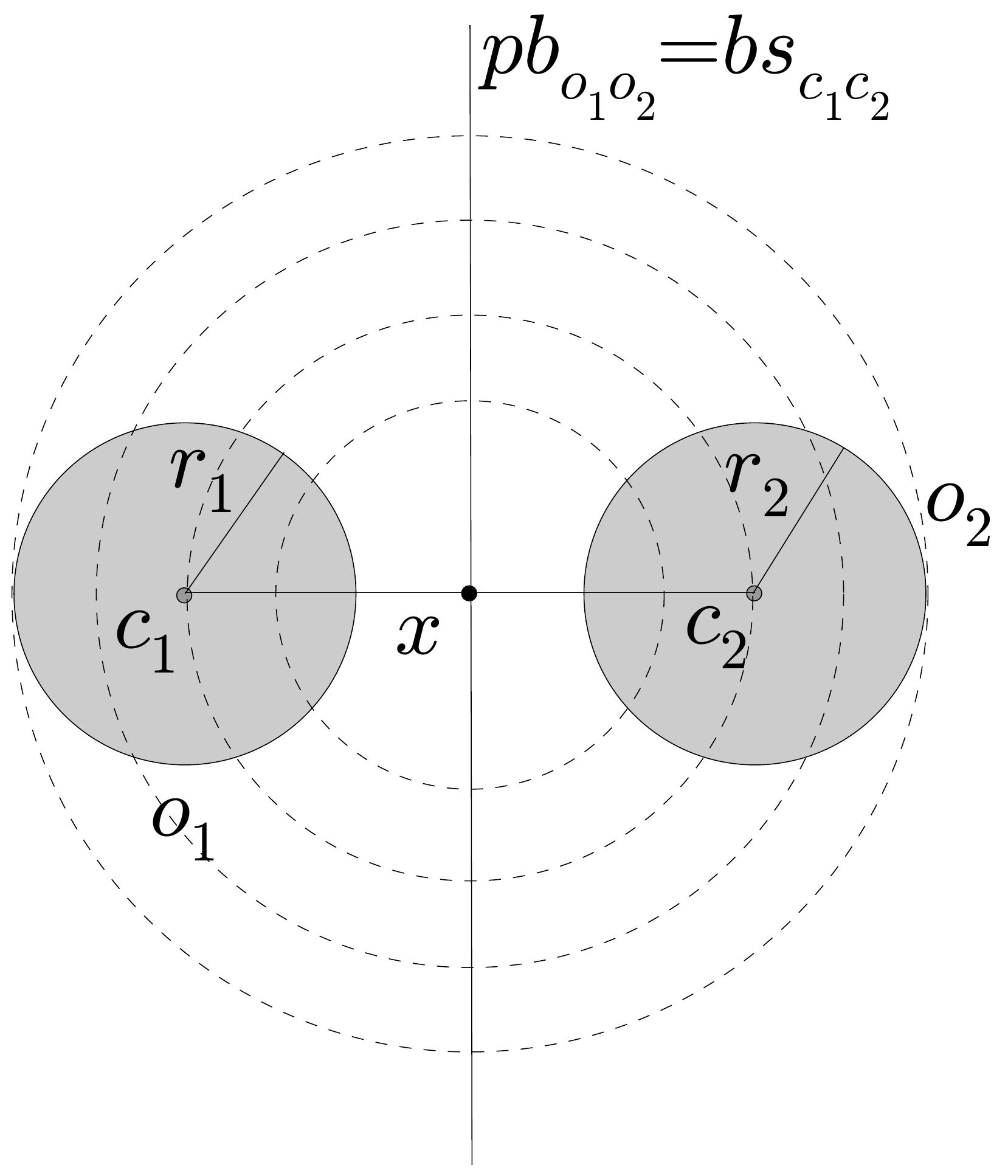}
    \caption{The probabilistic bisector of objects $o_1$ and $o_2$, where $r_{1}=r_{2}$}
    \label{fig:circle_lemma1}
\end{figure}

Lemmas~\ref{lemma:six} and~\ref{lemma:seven} show how the
probabilistic bisector of two non-equi-range objects $o_i$ and
$o_j$ is related to the bisector of $c_i$ and $c_j$
(Figure~\ref{fig:circle_lemma3} and ~\ref{fig:circle_lemma2}).

\begin{figure}[htbp]
    \centering
        \includegraphics[width=2in]{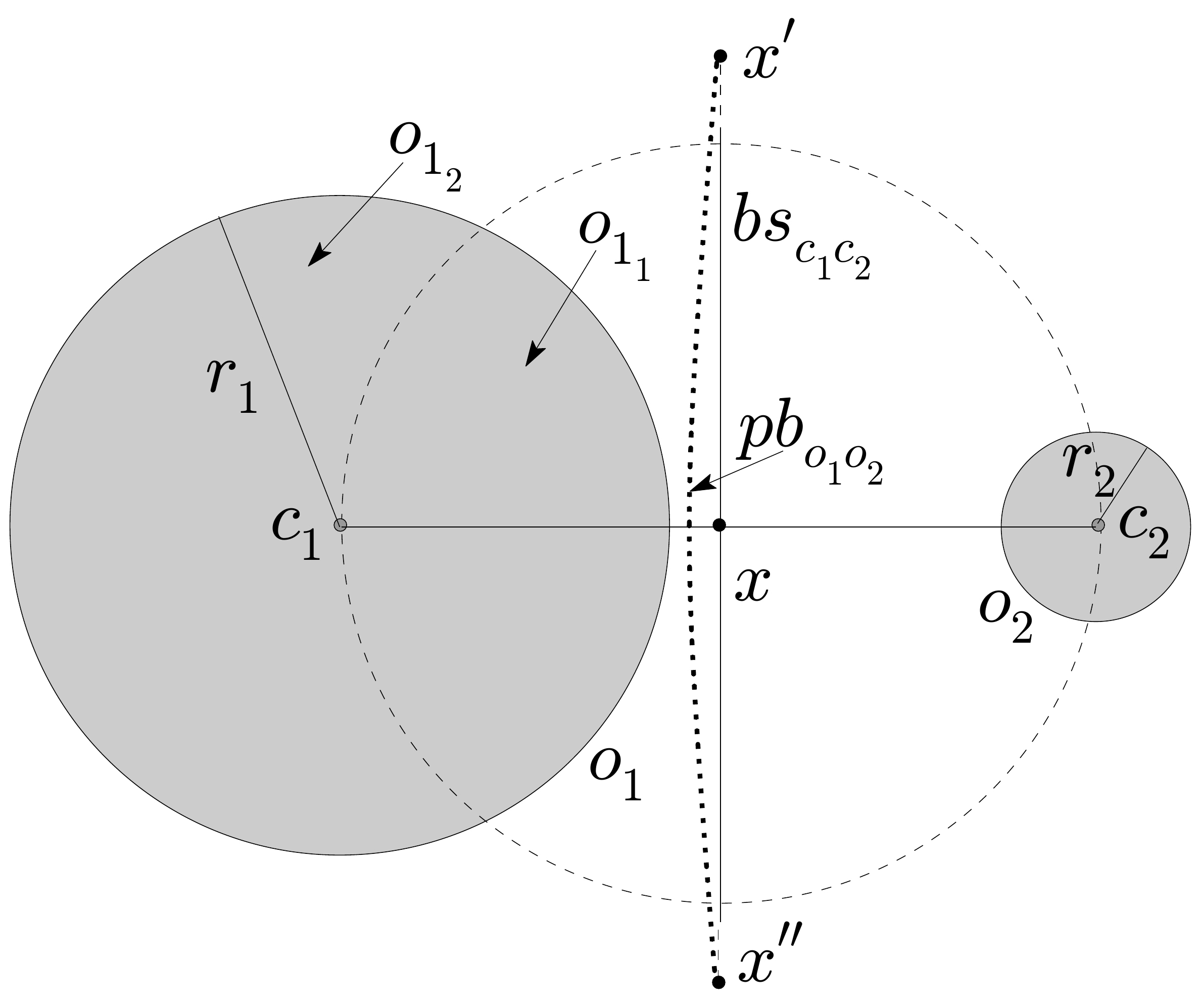}
    \caption[The probabilistic bisector of objects $o_1$ and $o_2$, where $r_{1}>r_{2}$]{The probabilistic bisector of objects $o_1$ and $o_2$, where $r_{1}>r_{2}$. The curve, $pb_{{o_1}{o_2}}$, is the probabilistic bisector
between $o_1$ and $o_2$, i.e., $p(o_1,x)=p(o_2,x)$, for any point
$x \in \{pb_{{o_1}{o_2}}\}$}
    \label{fig:circle_lemma3}
\end{figure}

Next, we will show in Lemma~\ref{lemma:six} that the shape of
$pb_{o_{i}o_{j}}$ for two non-equi-range circular objects $o_i$
and $o_j$ is a curve, and the distance of this curve from
$bs_{c_ic_j}$ is maximum on the line $\overline{c_{i}c_{j}}$.
Figure~\ref{fig:circle_lemma3} shows the bisector $bs_{c_1c_2}$
and the probabilistic bisector $pb_{o_{1}o_{2}}$ for $o_1$ and
$o_2$.

\begin{lemma}
\label{lemma:six} Let $o_i$ and $o_j$ be two objects with non-equi-range
uncertain circular regions $(c_i,r_i)$ and $(c_j,r_j)$, respectively, and $bs_{c_ic_j}$ be the bisector of
$c_i$ and $c_j$. Then the maximum distance between $bs_{c_ic_j}$
and $pb_{o_{i}o_{j}}$ occurs on the line $\overline{c_{i}c_{j}}$.
This distance gradually decreases as we move towards positive or
negative infinity along the bisector $bs_{c_ic_j}$.
\end{lemma}

\begin{proof}
Let $x=\frac{c_i+c_j}{2}$ be the intersection point of
$bs_{c_ic_j}$ and $\overline{c_{i}c_{j}}$. Suppose a circle
centered at $x$ with radius $\frac{dist(c_i,c_j)}{2}$ divides
$o_{i}$ into $o_{i_{1}}$ and $o_{i_{2}}$, where
$\frac{o_{i_{1}}}{A_i}+\frac{o_{i_{2}}}{A_i}=1$, and $o_{j}$ into
$o_{j_{1}}$ and $o_{j_{2}}$, where
$\frac{o_{j_{1}}}{A_j}+\frac{o_{j_{2}}}{A_j}=1$. According to
curvature properties of circles, since $r_{i}>r_{j}$, we have
$\frac{o_{i_{1}}}{A_i}<\frac{o_{j_{1}}}{A_j}$ (in
Figure~\ref{fig:circle_lemma3},
$\frac{o_{1_1}}{A_1}<\frac{o_{2_1}}{A_2}$), which intuitively
means, $o_j$ is a more probable NN than $o_i$ to $x$, i.e.,
$p(o_j,x)>p(o_i,x)$. Thus, $x$ needs to be shifted to a point
towards $c_i$ (along the line $\overline{xc_{i}}$), such that the
probabilities of $o_i$ and $o_j$ being the NNs to the new point
become equal.

Suppose a point $x^\prime$ is on $bs_{c_ic_j}$ at the positive
infinity. If a circle centered at $x^\prime$ goes through the
centers of both objects $o_i$ and $o_j$, then the curvature of the
portion of the circle that falls inside an object ($o_i$ or $o_j$)
will become a straight line. This is because, in this case we
consider a small portion of the curve of an infinitely large
circle. This circle divides both objects $o_i$ and $o_j$ into two
equal parts $o_{i_{1}}=o_{i_{2}}$ and $o_{j_{1}}=o_{j_{2}}$,
respectively. Thus, the probabilities of $o_i$ and $o_j$ being the
NNs will approach to being equal at positive infinity, i.e.,
$p(o_j,x^\prime)\approx p(o_i,x^\prime)$, for a large values of
$dist(x^{\prime},x)$. Similarly, we can show the case for a point
$x^{\prime\prime}$ at the negative infinity on $bs_{c_ic_j}$ (see
Figure~\ref{fig:circle_lemma3}).
\end{proof}

Next, we show in Lemma~\ref{lemma:seven} that $pb_{o_{i}o_{j}}$
shifts from  $bs_{c_{i}c_{j}}$ towards the object with larger
radius, and the distance of $pb_{o_{i}o_{j}}$ from
$bs_{c_{i}c_{j}}$ widens with the increase of the ratio of two
radii (i.e., $r_i$ and $r_j$). Figure~\ref{fig:circle_lemma2}
shows an example of this case.

\begin{lemma}
\label{lemma:seven} Let $o_i$ and $o_j$ be two objects with non-equi-range
uncertain circular regions $(c_i,r_i)$ and $(c_j,r_j)$, respectively, and $x=\frac{c_i+c_j}{2}$ be the
midpoint of the line segment $\overline{c_{i}c_{j}}$. If
$r_{i}>r_{j}$, then the probabilistic bisector $pb_{o_{i}o_{j}}$
meets $\overline{c_{i}c_{j}}$ at point $x^{\prime}$, where
$x^{\prime}$ lies between $x$ and $c_i$. If the circular range of
$o_i$ increases such that $r_{i}^{\prime}>r_{i}$, then the new
probabilistic bisector $pb^{\prime}_{o_{i}o_{j}}$ meets
$\overline{c_{i}c_{j}}$ at point $x^{\prime\prime}$ , where
$x^{\prime\prime}$ lies between $x$ and $c_{i}$, and
$dist(x,x^\prime)<dist(x,x^{\prime\prime})$.
\end{lemma}

\begin{proof}
Suppose a circle centered at $x$ with radius
$\frac{dist(c_i,c_j)}{2}$ divides $o_{i}$ into $o_{i_{1}}$ and
$o_{i_{2}}$, where
$\frac{o_{i_{1}}}{A_i}+\frac{o_{i_{2}}}{A_i}=1$, and $o_{j}$ into
$o_{j_{1}}$ and $o_{j_{2}}$, where
$\frac{o_{j_{1}}}{A_j}+\frac{o_{j_{2}}}{A_j}=1$. According to
curvature properties of circles, since $r_{i}>r_{j}$, we have
$\frac{o_{i_{1}}}{A_i}<\frac{o_{j_{1}}}{A_j}$ (in
Figure~\ref{fig:circle_lemma2},
$\frac{o_{1_1}}{A_1}<\frac{o_{2_1}}{A_2}$), which intuitively
means, $o_j$ is a more probable NN than $o_i$ to $x$, i.e.,
$p(o_j,x)>p(o_i,x)$. Thus, $x$ needs to be shifted to a point
$x^\prime$ towards $c_i$, such that the probabilities of $o_i$ and
$o_j$ being the NN to $x^\prime$ become equal. Let
$o_{i}^{\prime}$ be an object, such that $r_{i}^{\prime}>r_{i}$
and $c_{i}^{\prime}=c_{i}$. Then the circle centered at $x$ with
radius $\frac{dist(c_i,c_j)}{2}$ divides $o^{\prime}_{i}$ into
$o^{\prime}_{i_{1}}$ and $o^{\prime}_{i_{2}}$, where
$\frac{o^{\prime}_{i_{1}}}{A_{i}^{\prime}}+\frac{o^{\prime}_{i_{2}}}{A_{i}^{\prime}}=1$.
Now, we have
$\frac{o^{\prime}_{i_{1}}}{A_{i}^{\prime}}<\frac{o_{i_{1}}}{A_i}<\frac{o_{j_{1}}}{A_j}$.
Thus, $x$ needs to be shifted to a point $x^{\prime\prime}$ more
towards $c_i$, i.e., $dist(x,x^\prime)<dist(x,x^{\prime\prime})$,
such that the probabilities of $o^{\prime}_{i}$ and $o_j$ being
the NNs become equal at $x^{\prime\prime}$.\eat{ Therefore, we
have $dist(x,x^\prime)<dist(x,x^{\prime\prime})$.}
\end{proof}

\begin{figure}[htbp]
    \centering
        \includegraphics[width=2in]{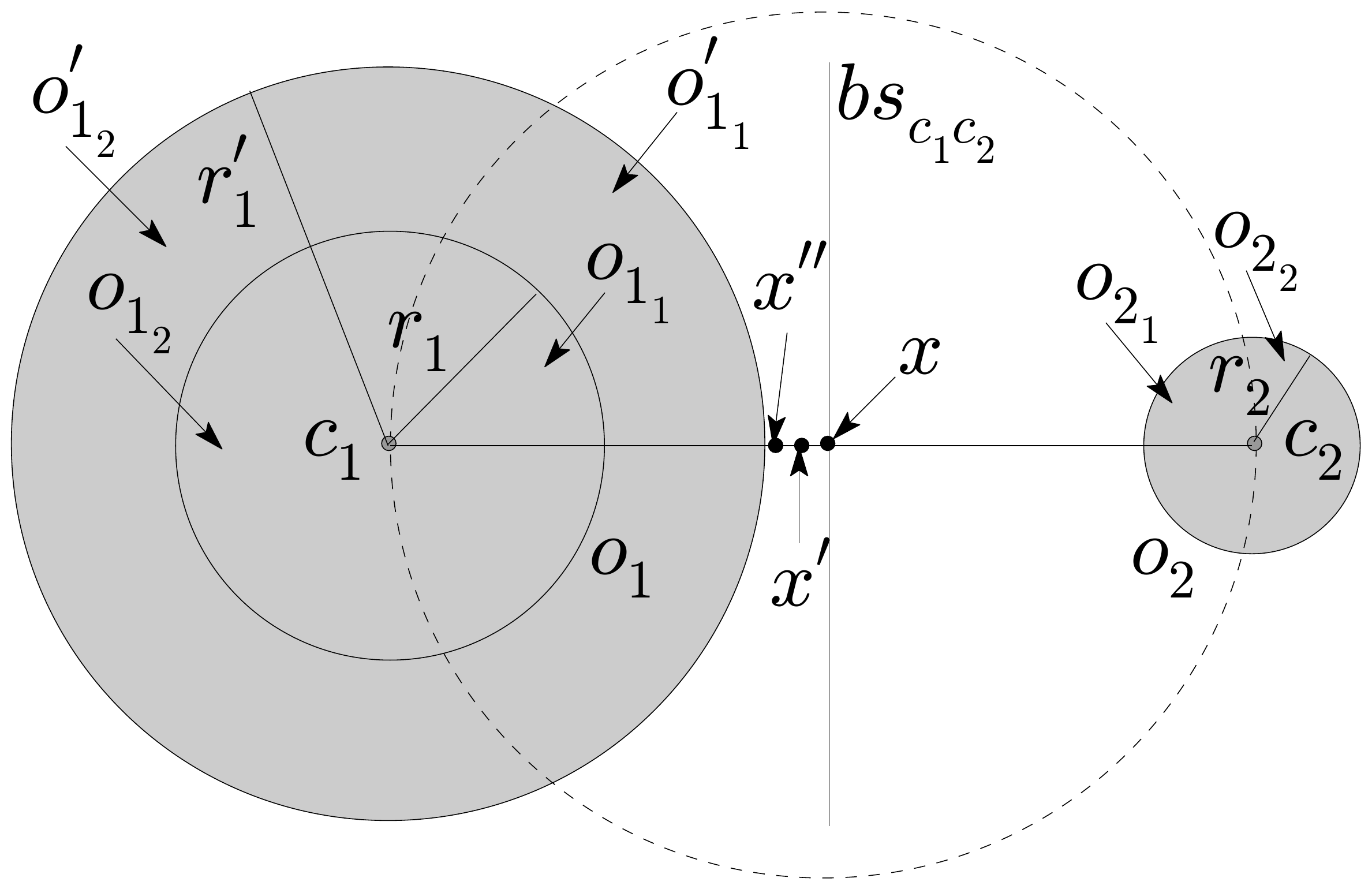}
    \caption{Influence of objects' sizes on the probabilistic bisector}
    \label{fig:circle_lemma2}
\end{figure}

The next lemma shows the influence of a third object on the
probabilistic bisector of two non-equi-range objects. (Note that
the probabilistic bisector of two equi-range objects does not
change with the influence of any other object (see
Lemma~\ref{lemma:five})). Figure~\ref{fig:circle_influence} shows
an example, where object $o_3$ influences the probabilistic
bisector of objects $o_1$ and $o_2$. In this figure, the dotted
circle centered at $s_1$ with radius $dist(s_1,c_1)+r_1$ encloses
one candidate object $o_1$, but only touches the third object
$o_3$. Thus, the probability of $o_3$ being the NN to $s_1$ is
zero. However, for any point between $s_1$ and $s_2$, $o_3$ has a
non-zero probability of being the NN of that point, and thus $o_3$
influences $pb_{o_1o_2}$.

\begin{lemma}
\label{lemma:eight} Let $o_i$ and $o_j$ be two objects with non-equi-range
uncertain circular regions $(c_i,r_i)$ and $(c_j,r_j)$, respectively, where $r_{i}<r_{j}$, and $bs_{c_ic_j}$ be the bisector of $c_i$ and
$c_j$. An object $o_k$ influences the probabilistic bisector
$pb_{o_{i}o_{j}}$ for the part of the segment $[s_{1},s_{2}]$ on
the line $bs_{c_ic_j}$, where
$dist(s,c_{i})+r_i>dist(s,c_{k})-r_k$ for $s\in bs_{c_ic_j}$.
\end{lemma}

\begin{proof}
Since $r_{i}<r_{j}$, we have $maxdist(s,o_i)<maxdist(s,o_j)$.
Thus, if the minimum distance $mindist(s,o_k)$ of an object $o_k$
from $s$ is greater than the maximum distance $maxdist(s,o_j)$ of
$o_j$ from $s$, i.e., $dist(s,c_{k})-r_k>dist(s,c_{i})+r_i$, the
object $o_k$ cannot be the NN to the point $s$, otherwise $o_k$ has the possibility of being
the NN to $s$ and hence $o_k$ influences $pb_{o_io_j}$.
\end{proof}

It is noted when the centers of two non-equal objects coincide
each other, the probability of the smaller object dominates the
probability of the larger object.\eat{ We can give similar proof as
shown in Lemma~\ref{lemma:seven}, where we can show that the
probability of a smaller circular object to a query point will be
higher than that of a larger circular object.} Therefore, in those
cases, we only consider the object with a smaller radius, and the
other object is discarded. Also, if two objects are equal and
their centers coincide each other, no probabilistic bisector
exists between them, thus any one of these two objects is
considered for computing the PVD.

\begin{figure}[htbp]
    \centering
        \includegraphics[width=1.5in]{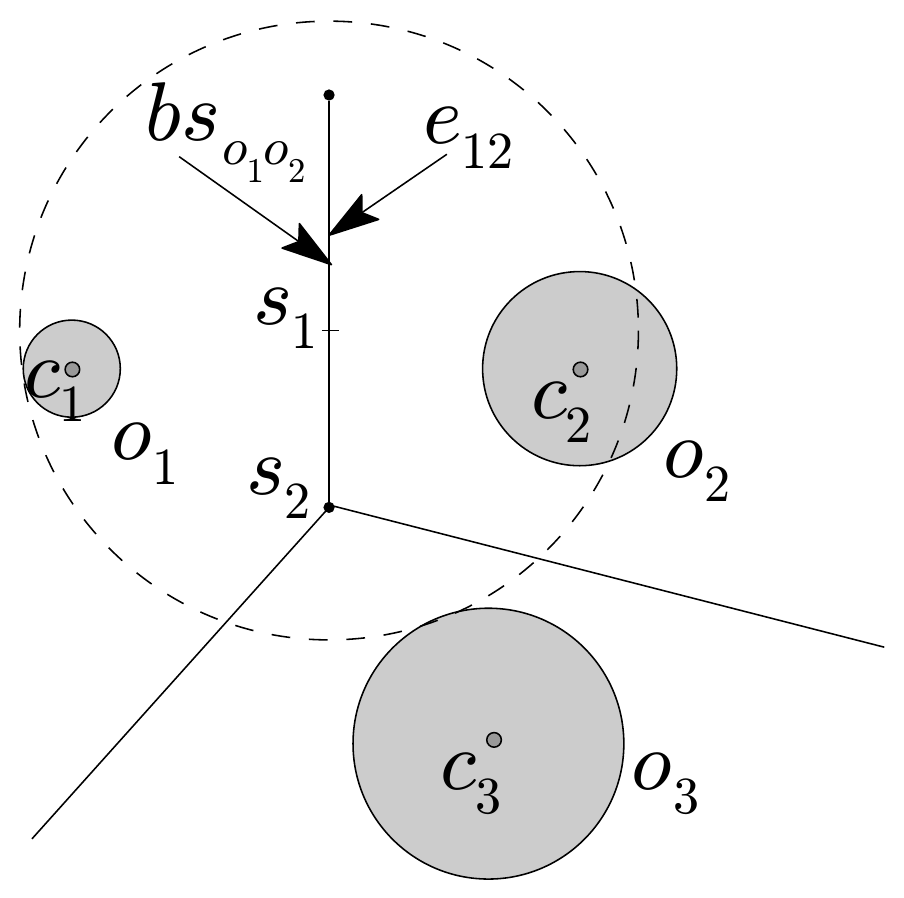}
    \caption{Influence of object $o_3$ on the probabilistic bisector of $o_1$ and $o_2$}
    \label{fig:circle_influence}
\end{figure}

%
\noindent\emph{\textbf{Algorithms: }}\par Based on the above
lemmas, we propose algorithms to find the probabilistic bisector
of any two uncertain 2D objects. We have shown in
Lemma~\ref{lemma:five} that the probabilistic bisector of two
circular uncertain objects is a straight line when the radii of
two objects are equal. On the other hand,
Lemma~\ref{lemma:six}-Lemma~\ref{lemma:seven} show that the
probabilistic bisector is a curve when the radii of two objects
are non-equal. However, to avoid the computational and maintenance
costs, we maintain a bounding box (i.e., quadrilateral) that
encloses the actual probabilistic bisector of two objects. Hence,
we name the probabilistic bisector of two circular objects as the
\emph{Probabilistic Bisector Region} (PBR). For example, the bounding box that encloses the curve in Figure~\ref{fig:circle_lemma3} is the PBR for two objects $o_1$ and $o_2$. In our algorithm, we
first create an ordinary Voronoi diagram by using the centers of
all uncertain objects. Then, from each Voronoi edge $e_{ij}$
(i.e., $bs_{c_ic_j}$) of two objects $o_i$ and $o_j$, we compute
the PBR that encloses $pb_{o_{i}o_{j}}$.

Algorithm~\ref{algo:Prob-Bisector-TwoDim} computes the
probabilistic bisector of two equi-range objects according to
Lemma~\ref{lemma:five} (Line~\ref{algo:Prob-Bisector-TwoDim}2).
Otherwise, it calls the function $FindProbBisector2D$ to determine
$pb_{o_{i}o_{j}}$ for two non-equi-range objects $o_i$ and $o_j$.

\setlength{\algomargin}{2em} \dontprintsemicolon
\begin{algorithm}[htbp]
\begin{small}
\caption{ProbBisector2D($o_i,o_j,e_{ij},O$)}
\label{algo:Prob-Bisector-TwoDim}
    \uIf{$r_{i}=r_{j}$}
    {
             $pb_{o_{i}o_{j}} \leftarrow
             ProbBisector(o_i,o_j,e_{ij})$\;
    }
    \Else
    {
        $pb_{o_{i}o_{j}} \leftarrow FindProbBisector2D(o_{i},o_{j},e_{ij},O)$\;
    }
    \Return $pb_{o_{i}o_{j}}$;
    \end{small}
\end{algorithm}

The function $FindProbBisector2D$ (see
Algorithm~\ref{algo:calc-prob-bisect2D}) takes two non-equi-range
objects $o_i$, $o_j$, the bisector $e_{ij}$ (i.e., a Voronoi edge)
of $c_i$ and $c_j$, and the set of objects $O$ as input, and
returns $pbr$ for $pb_{o_{i}o_{j}}$. The algorithm finds lower
($lval$) and upper ($uval$) bounds representing the required
deviations of the probabilistic bisector from the bisector of
$c_i$ and $c_j$, such that the PBR can be computed by drawing two
lines parallel to $e_{ij}$ at $lval$ and $uval$, respectively.
Algorithm~\ref{algo:calc-prob-bisect2D} first initializes $ipb$
with the intersection point of $e_{ij}$ and $\overline{c_ic_j}$
(Line~\ref{algo:calc-prob-bisect2D}.1). Then, the function
$InitPBRBound$ computes initial lower ($lval$) and upper ($uval$)
bounds of $pbr$ (Line~\ref{algo:calc-prob-bisect2D}.2). This
function first determines a point $x^\prime$ on the line
$\overline{c_ic_j}$ where $p(o_i,x^\prime)\approx p(o_j,x^\prime)$
(We use a similar search technique as described for the 1D
space).\setlength{\algomargin}{2em} \dontprintsemicolon
\begin{algorithm}[htbp]
\begin{small}
\caption{FindProbBisector2D($o_i,o_j,e_{ij},O$)}
\label{algo:calc-prob-bisect2D}
    $ipb\leftarrow Intersect(e_{ij},\overline{c_ic_j})$\;
    $InitPBRBound(lval,hval,ipb)$\;
    $IL\leftarrow FindInfluencedPart(o_i,o_j,e_{ij})$\;
   \For{each $ls \in IL$}
    {
        $UpdatePBRBound(lval,hval,ls)$\;
    }
    $pbr\leftarrow [lval,hval]$\;
    \Return $pbr$;
    \end{small}
\end{algorithm}If $x^\prime$ is to the left of $e_{ij}$, then $lval$ and
$hval$ are set to $x^\prime$ and $x$, respectively. On the other
hand, if $x^\prime$ is to the right of $e_{ij}$, then $lval$ and
$hval$ are set to $x$ and $x^\prime$, respectively. After that,
the function $FindInfluencePart$ finds a list $IL$ that contains
different segments of the bisector $e_{ij}$, where other objects
influence $pb_{o_io_j}$ (see Lemma~\ref{lemma:eight}). The
function returns $IL$ as an empty list when no other object
influences the probabilistic bisector. In that case the current
$lval$ and $hval$ defines $pbr$. In Lemma~\ref{lemma:seven}, we
have seen that the maximum distance of $pb_{o_io_j}$ from the
bisector of $c_i$ and $c_j$ is on the line $\overline{c_ic_j}$.
Thus, the initially computed $pbr$ encloses the curve of
$pb_{o_io_j}$. On the other hand, if $IL$ is not empty, then for
each line segment $ls\in IL$, the function $UpdatePBRBound$ is
called to update $lval$ and $hval$ based on the influence of other
objects. As $lval$ and $hval$ represent the deviation of
$pb_{o_io_j}$ from $e_{ij}$, we need compute the deviations for
each line segment $ls$, and then take the minimum of all $lval$s
and the maximum of all $hval$s to compute the $pbr$. To avoid a
brute-force approach of computing $lval$ and $hval$ for every
point of an $ls\in IL$, we compute $lval$ and $hval$ for two
extreme points and the mid-point of $ls$. Finally, the algorithm
returns $pbr$ for $pb_{o_io_j}$.

\setlength{\algomargin}{2em} \dontprintsemicolon
\begin{algorithm}[htbp]
\begin{small}
\caption{ProbVoronoi2D(O)}
\label{algo:Prob-Voronoi-TwoDim}
    $PVD\leftarrow\emptyset$\;
    $VD\leftarrow VornoiDiagramOfCentroids(O)$\;

    \For{each Voronoi edge $e_{ij} \in VD$, where $o_i,o_j \in O$}
    {
            $PVD \leftarrow PVD \cup ProbBisector2D(o_i,o_j,e_{ij},O)$\;
    }
    \Return $PVD$;
    \end{small}
\end{algorithm}

Algorithm~\ref{algo:Prob-Voronoi-TwoDim} shows the steps of
$ProbVoronoi2D$ that computes $PVD$ for a given set $O$ of 2D
objects. In Line~\ref{algo:Prob-Voronoi-TwoDim}.2, the algorithm
first creates a Voronoi diagram $VD$ for all centers $c_i$ of
objects $o_i\in O$ using~\cite{fortune87.algo}. Then for each
Voronoi edge $e_{ij}$ between two objects $o_i$ and $o_j$, the
algorithm calls the function $ProbBisector2D$ to compute the
probabilistic bisector as PBR between two candidate objects, and
finally it returns the PVD for the given set $O$ of objects.

\begin{figure}[htbp]
    \centering
        \includegraphics[width=1.5in]{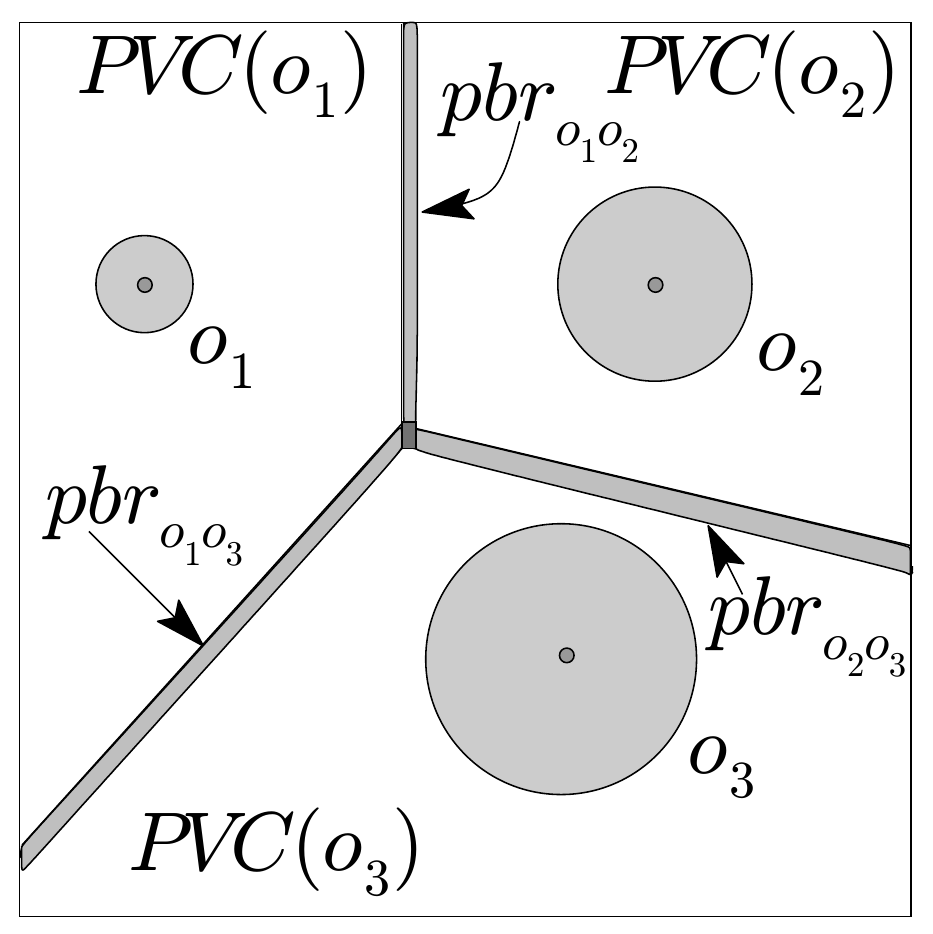}
    \caption{The PVD of three objects $o_1$, $o_2$, and $o_3$}
    \label{fig:circle_pvd1}
\end{figure}

Figure~\ref{fig:circle_pvd1} shows the PVD for objects
$o_1$, $o_2$, and $o_3$. In this figure, $PVC(o_1)$, $PVC(o_2)$, and $PVC(o_3)$ represent the PVCs for objects $o_1$, $o_2$, and $o_3$, respectively. The boundaries between PVCs, i.e., PBRs of objects, $pbr_{o_{1}o_{2}}$, $pbr_{o_{2}o_{3}}$, and $pbr_{o_{1}o_{3}}$, are shown using grey bounded regions. For any point inside a PVC, the corresponding object is guaranteed to be the most probable NN. On the other hand, for any point inside a PBR, any of the two objects that share the PBR can be the most probable NNs. If more than two PBRs intersect each other in a region, any object associated with these
PBRs can be the most probable NN to a query point in that region.
Figure~\ref{fig:circle_pvd1} shows a dark grey region where
$pbr_{o_1o_2}$, $pbr_{o_1o_3}$, and $pbr_{o_2o_3}$ meets.

\eat{When the query point is at $q^\prime$, PVD-PMNN
returns $o_3$ as the most probable NN as $q^\prime$ falls into
PVC($o_3$). When the query point moves to $q^{\prime\prime}$, the
algorithm returns $o_2$ as the answer. If more than two PBRs
intersect each other in a region, any object associated with these
PBRs can be the most probable NN to a query point in that region.
Figure~\ref{fig:circle-pvd}(a) shows a dark grey region where
$pbr_{o_1o_2}$, $pbr_{o_1o_3}$, and $pbr_{o_2o_3}$ meets.}

\emph{Complexity}: The complexity of
Algorithm~\ref{algo:Prob-Voronoi-TwoDim} can be estimated as
follows. The complexity of creating a Voronoi diagram (Line
\ref{algo:Prob-Voronoi-TwoDim}.2) is $O(n\log
n)$~\cite{fortune87.algo}, where $n$ is the number of objects. The
complexity of finding probabilistic bisectors
(Lines~\ref{algo:Prob-Voronoi-TwoDim}.3-\ref{algo:Prob-Voronoi-TwoDim}.4)
is $O(n_{e}C_{pb})$, where $n_{e}$ is the number of Voronoi edges
and $C_{pb}$ is the expected cost of computing the probabilistic
bisector between two circular objects. For real data sets, $n_e$
is expected to be a small integer since an object has only a small
number of surrounding objects. The total complexity of the
algorithm is $O(n\log n)$+$O(n_{e}C_{pb})$. $C_{pb}$ can be
estimated as follows. Let $C_b$ be the cost of computing the
probability of an object being the NN of a query point, $D$ be the
expected distance between the initial probabilistic bisector $ipb$
and the actual probabilistic bisector, and $L$ be the expected
number of points in the bisector that needs to be considered to
find upper and lower bounds of the probabilistic bisector. Then we
have $C_{pb}=O(LC_{b}\log_{2}D)$. This is because, the cost of
finding a probabilistic bisector is $O(1)$ for the cases when our
algorithm can directly compute the probabilistic bisector, and for
other cases our algorithm first finds $ipb$ by $O(1)$ and then
search for the actual probabilistic bisector by using
Algorithm~\ref{algo:calc-prob-bisect2D} by $O(L\log D)$. Note
that, for both 1D and 2D, the run-time behavior of our algorithm
is dominated by those cases for which there is no closed form for
a given probabilistic bisector, i.e., the algorithm needs to
search for the bisector by using the initial probabilistic
bisector.

\subsection{Discussion}
\label{subsec:discus}

\noindent\emph{PVD for Other Distributions:} In this paper, we
assume the uniform distributions for the pdf of uncertain objects
to illustrate the concept of the PVD. However, the pdf that
describes the distribution of an object inside the uncertainty
region can follow arbitrary distributions, e.g., Gaussian. The
concept of PVD can be extended for any arbitrary distribution. For
example, for an object with Gaussian pdf having a circular
uncertain region, the probability of the object of being around
the center of the circular region is higher than that of the
boundary region of the circle. For such distributions, a
straightforward approach to compute the probabilistic bisector
between any pair of objects is as follows. First, we can use the
bisector of the centroids of two candidate objects as the initial
probabilistic bisector. Then we can refine the initial
probabilistic bisector to find the actual probabilistic bisector.
Finding suitable initial probabilistic bisectors for efficient
computation of probabilistic bisectors (e.g., lemmas for different
cases for 1D and 2D data sets similar to the uniform pdf) for an
arbitrary distribution is the scope of future investigation.

\noindent\emph{PVD for Higher Dimensions:} We can compute the PVD for higher
dimensional spaces, similar to 1D and 2D spaces. For example, in a 3D space, an uncertain
object can be represented as a sphere instead of a circle in 2D.
Then, the probabilistic bisector of two equal size spheres will be
a plane bisecting the centers of two spheres. Using this as a
base, similar to 2D, we can compute the PVD for 3D objects. We
omit a detailed discussion on PVDs in spaces of more than 2
dimensions.

\noindent\emph{Higher order PVDs}: In this paper, we focus on the
first order PVD. By using this PVD, we can find the NN for a given
query point. Thus, the PVD can be used for continuously reporting
1-NN for a moving query point. To generalize the concept for $k$NN
queries, we need to develop the $k$-order PVD. The basic idea
would be to find the probabilistic bisectors among size-$k$
subsets of objects. The detailed investigation of higher order
PVDs is a topic of future study.

\noindent\emph{Handling Updates}: To handle updates on the data
objects, like traditional Voronoi diagrams, a straightforward
approach is to recompute the entire PVD. There are
algorithms~\cite{de94:incv,mosta03:voronoiupdate} to incrementally
update a traditional Voronoi diagram. Similar ideas can be applied
to the PVD to derive incremental update algorithms. We will defer
such incremental update algorithms for future work.

It is noted that, to avoid an expensive computation of the PVD for
the whole data set and to cope with updates for the data objects,
we propose an alternative approach based on the concept of local
PVD (see Section~5.5.2). In this approach, only a subset of
objects that fall within a specified range of the current position
of the query is retrieved from the server and then the local PVD
is created for these retrieved objects to answer PMNN queries. If
there is any update inside the specified range, the process needs
to be repeated. Since, this approach works only with the
surrounding objects of a query, updates from objects that are
outside the range do not affect the performance of the system.

\eat{To handle updates on the data objects, like traditional
Voronoi diagrams, a straightforward approach is to recompute the
entire PVD. There are algorithms~\cite{de94:incv} to incrementally
update a traditional Voronoi diagram. Similar ideas can be applied
to the PVD to derive incremental update algorithms. We will defer
such incremental update algorithms for future work.}

%

\section{Processing PMNN Queries}
\label{sec:pmnn}

In this section, based on the concept of PVD we propose two
techniques: a pre-computation approach and an incremental approach for
answering PMNN queries. In the pre-computation approach, we first
create the PVD for the whole data set and then index the PVCs for
answering PMNN queries. We name the
pre-computation based technique for processing PMNN queries as
P-PVD-PMNN. On the other hand, in the incremental
approach, we retrieve a set of surrounding objects with respect
the current query location and then create the local PVD for these
retrieved data set, and finally use this local PVD to answer PMNN
queries. We name this approach I-PVD-PMNN in this paper.
%


\subsection{Pre-computation Approach}
\label{subsec:precom}

In the pre-computation approach, we first create the PVD for all
objects in the database. After computing the PVD, we only need to
determine the current Probabilistic Voronoi Cell (PVC), where the
current query point is located. The query evaluation algorithm can
be summarized as follows.

Initially, the query issuer requests the most probable NN for the
current query position $q$. After receiving the PMNN request for
$q$, the server algorithm finds the current PVC to which the query
point falls into using a function $IdentifyPVC$ and updates $cpvc$
with the current PVC. The algorithm reports the corresponding
object $p$ as the most probable NN and the cell $cpvc$ to the
query issuer. Next time when $q$ is updated at the query issuer,
if $q$ falls inside $cpvc$, no request is made to the server as
the most probable NN has not been changed. Otherwise, the query
issuer again sends the PMNN request to the server to determine the
new PVC and the answer for the updated query position.

\eat{Then, we  organize the PVCs of the PVD in
a hierarchy using the $R$-tree~\cite{guttman:rteee}. Finally, we process PMNN queries
by using Algorithm~\ref{algo:Prob-Rank-Update}.  We name the
pre-computation based technique for processing PMNN queries as
P-PVD-PMNN.}

\begin{figure}[htbp]
    \centering
        \includegraphics[width=2.5in]{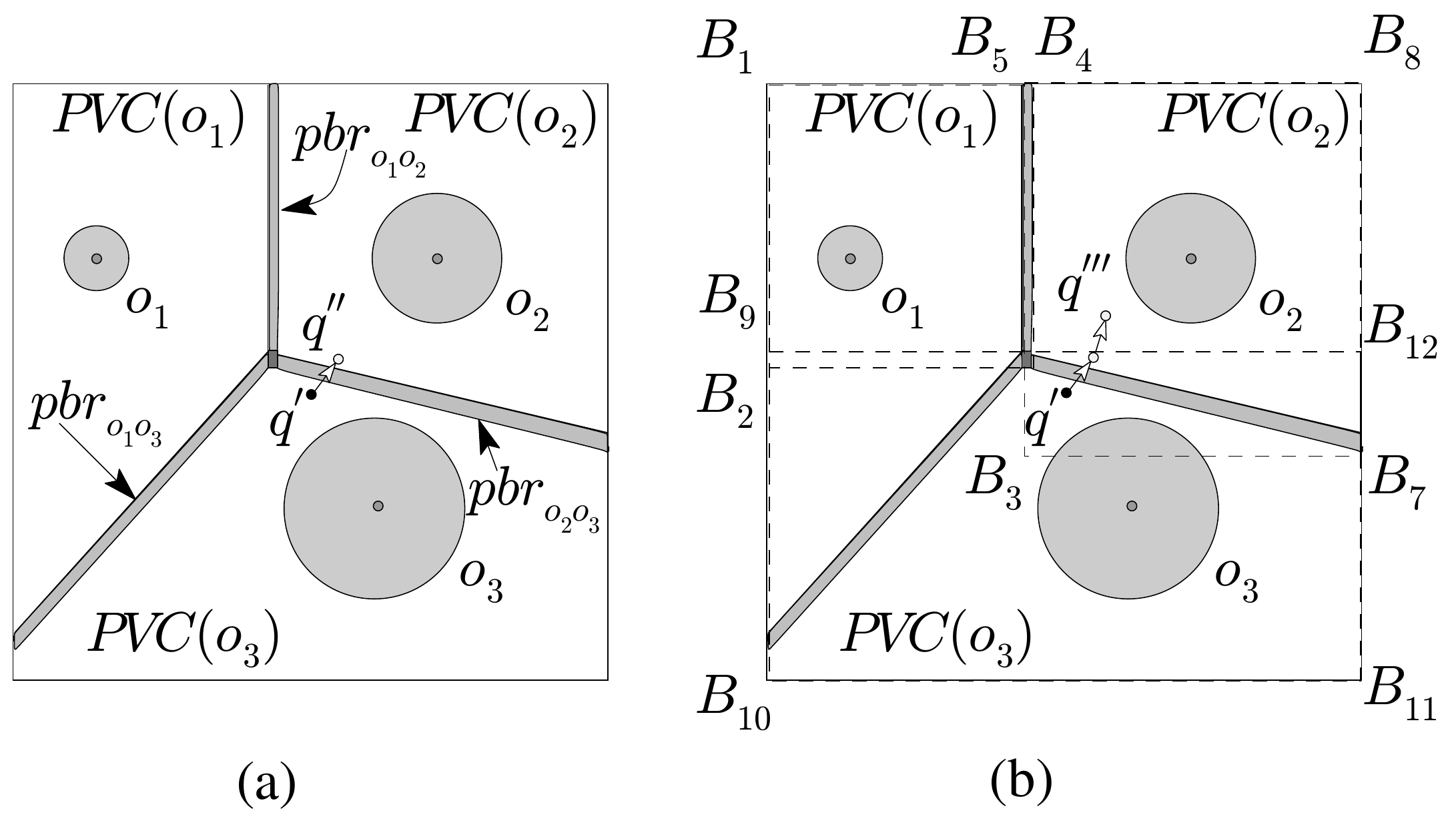}
    \caption{(a) The PVD, and (b) the MBRs of PVCs for objects $o_1$-$o_3$}
    \label{fig:circle-pvd}
\end{figure}

As the PVD in a 1D space contains a set of non-overlapping ranges
representing PVCs for objects, the algorithm returns a single object as
the most probable NN for any query point. On the other hand, in a
2D space, the boundary between two PVCs is a region (i.e., PBR) rather than a
line. When a query point falls inside a PBR, the algorithm can
possibly return both objects that share the PBR as the most probable NNs,
or preferably can decide the most probable NN by computing a top-1-PNN
query. (Since, for a realistic setting a PBR is small region compared to that of PVCs, our approach incurs much less computational overhead than that of the sampling based approach for processing a PMNN query.)

Figure~\ref{fig:circle-pvd}(a) shows that when the query point is at $q^\prime$, PVD-PMNN
returns $o_3$ as the most probable NN as $q^\prime$ falls into
PVC($o_3$). When the query point moves to $q^{\prime\prime}$, the
algorithm returns $o_2$ as the answer.

A naive approach of identifying the desired PVC (i.e.,
$IdentifyPVC$ function) requires an exhaustive search of all the
PVCs in a PVD, which is an expensive operation. Indexing Voronoi
diagrams~\cite{berchtold98:vdind,berchtold00:vdindex,samet:multidimension}
is an well-known approach for efficient nearest neighbor search in
high-dimensional spaces. Thus, for efficient search of the PVCs,
we index the PVCs of the PVD using an
$R^{*}$-tree~\cite{beckmann:rstar}, a variant of the
$R$-tree~\cite{guttman:rteee,samet:multidimension}. In a 1D space,
each PVC is represented as a 1D range and is indexed using a 1D
$R^{*}$-tree. Since there is no overlap among PVCs, a query point
always falls inside a single PVC, where the corresponding object
is the most probable NN to the query point. On the other hand, in
a 2D space, each PVC cell is enclosed using a Minimum Bounding
Rectangle (MBR), and is indexed using a 2D $R^{*}$-tree. Since the
MBRs representing PVCs overlap each other, when a query point
falls inside only a single MBR, the corresponding object is
confirmed to be the most probable NN to the query point. However,
when a query point falls inside the overlapping region of two or
more MBRs, the actual most probable NN can be identified by
checking the PVCs of all candidate MBRs.
Figure~\ref{fig:circle-pvd}(b) shows the MBRs $[B_1,B_2,B_3,B_4]$,
$[B_5,B_6,B_7,B_8]$, and $[B_9,B_{10},B_{11},B_{12}]$ for the PVCs
of objects $o_1$, $o_2$, and $o_3$, respectively. In this example,
the query point $q^\prime$ intersects both $[B_5,B_6,B_7,B_8]$ and
$[B_9,B_{10},B_{11},B_{12}]$, and the actual most probable NN
$o_3$ can be determined by checking the PVCs of $o_3$ and $o_2$;
on the other hand, the query point $q^{\prime\prime\prime}$ only
intersects a single MBR $[B_5,B_6,B_7,B_8]$, so the corresponding
object $o_2$ is the most probable NN to $q^{\prime\prime\prime}$.

Since the above approach only retrieves the current PVC of a
moving query point, it needs to access the PVD using the
$R^{*}$-tree as soon as the query leaves the current PVC. This may
incur more I/O costs than what can be achieved. To further reduce
I/O and improve the processing time, we use a buffer management
technique, where instead of only retrieving the PVC that contains
the given query point, we retrieve all PVCs whose MBRs intersect
with a given range, called a buffer window, for a given query
point. These PVCs are buffered and are used to answer subsequent
query points of a moving query. This process is repeated for a
PMNN query when the buffered cells cannot answer the query.

\eat{It is noted that, to avoid an expensive computation of the
PVD for the whole data set and to cope with updates for the data
objects, for a PMNN query the server only needs to retrieve a
sub-set of objects that falls within a specified range of the
current position of the query, and then creates a PVD for those
objects. This process needs to be repeated as soon as the most
probable NN for the moving query cannot be answered by the already
retrieved data. We will further investigate this approach in our
future work.}

Since the creation of the entire PVD is computationally expensive,
the pre-computation based approach is justified when the PVD can
be re-used which is the case for static data, or when the query
spans over the whole data space. To avoid expensive
pre-computation, next, we propose an incremental approach which is
preferable when the query is confined to a small region of the
data space or when there are frequent updates in the database.

\subsection{Incremental Approach}
 \label{subsec:psr}

\eat{In this section, we describe the evaluation of a PMNN query
in a client-server model. The moving user issues a query to the
server through the wireless network. For a PMNN, a user is needed
to be updated continuously about her most probable nearest
neighbors as the user moves in along her trajectory. Due to high
latency and low bandwidth of wireless networks, communication
costs dominate processing costs while evaluating a PMNN query in a
client-server model. To reduce the communication overheads, we
adopt the concept of known region that retrieves extra data from
the server~\cite{nutanong08:vstar}. }

\eat{In the incremental approach, we first retrieve a set of
surrounding objects by expanding the search space from the current
query location. After retrieving all the objects inside the search
space, also called known region, we develop a local PVD for those
objects. We use the local PVD and  the knowledge of the known
region to efficiently answer the PMNN query. Since, the total data
space is unknown in this case, we exploit the knowledge of the
known region to compute a probabilistic bound of a data object of
being the NN from the query point.}

In this section, we describe our incremental evaluation technique
for processing a PMNN query based on the concept of known region
and the local PVD. Next, we briefly discuss the concept of known
region, and then present the detailed algorithm of our incremental
approach.

\textbf{Known Region: }Intuitively, the \emph{known region} is an
explored data space where the position of all objects are known.
We define the known region as a circular region that bounds the
top-$k$ probable NNs with respect to the current query point (i.e., the center point of the
region). For a given point $q_s$, the server expands the search space
to incrementally access objects in the order of their $mindist$
from $q_s$ until it finds top $k$ probable nearest neighbors with
respect to $q_s$ (we use existing algorithm~\cite{beskales08.vldb} to find
top-$k$NN). Then the known region is
determined by a circular region centered at $q_s$ that encloses all
these $k$ objects. Figure~\ref{fig:circle_safe} shows the
known circular region $K(q_{s},r)$ using a dashed circle, where
$k=3$. Then the radius $r$ of this known area is
determined by
$max(maxdist(q_{s},o_{1}),maxdist(q_{s},o_{2}),maxdist(q_{s},o_{3}))$. In this example, top-3 most probable nearest neighbors are $o_{1}$,$o_{2}$, and $o_{3}$.

\eat{In addition to these $k$ objects, the server
also retrieves objects that intersect with this known region.

Let $q_s$ be the current query position for which the data
retrieval request is made. Figure~\ref{fig:circle_safe} shows the
known circular region $K(q_{s},r)$ using a dashed circle, where
$k=3$. Then the radius $r$ of this known area is determined by the circular explored space. In this example, the server retrieves top-3 probable nearest neighbors $o_{1}$,$o_{2}$, and $o_{3}$.}

\begin{figure}[htbp]
    \centering
        \includegraphics[width=1.5in]{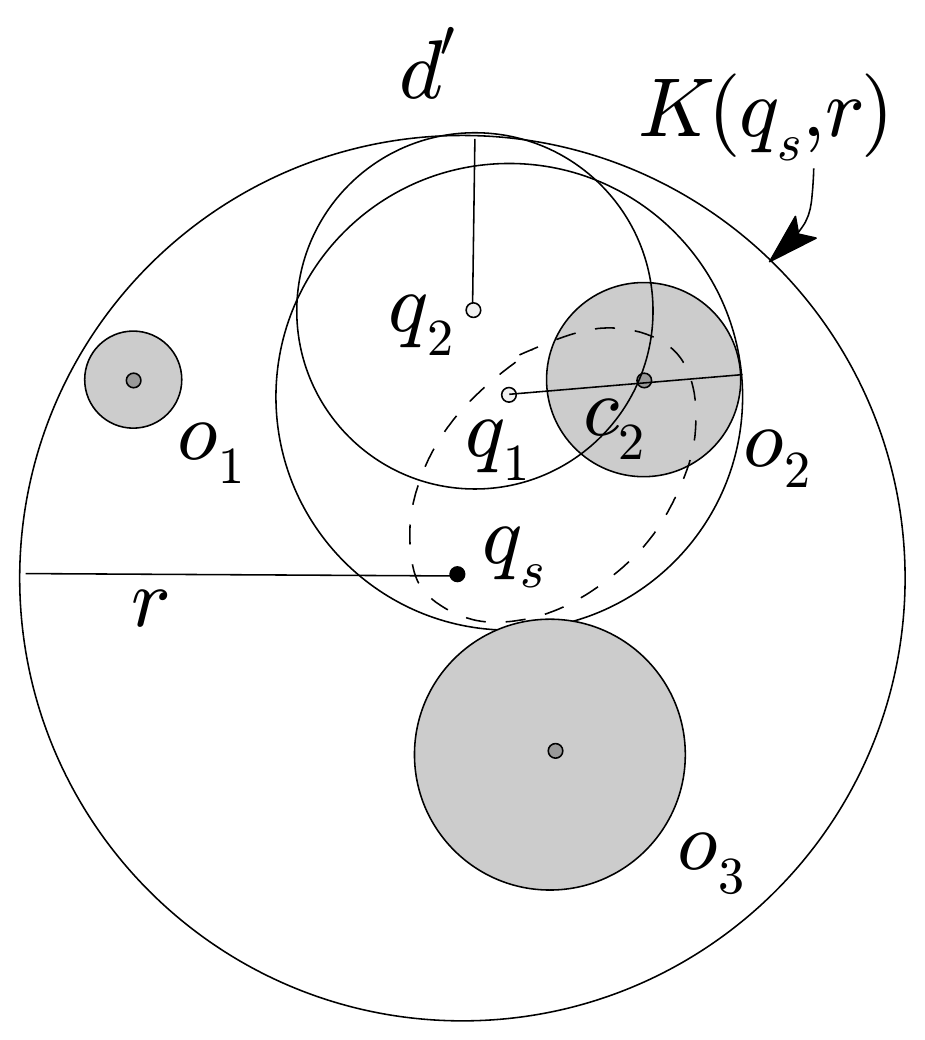}
    \caption{Known region and objects $o_1$, $o_2$, and $o_3$}
    \label{fig:circle_safe}
\end{figure}

\eat{Then the radius $r$ of this known area is determined by
$max(maxdist(q_{s},o_{1}),maxdist(q_{s},o_{2}),...,maxdist(q_{s},o_{5}))$.
In this example, the server retrieves top-3 probable nearest
neighbors $o_{1},o_{2}, and o_{3}$. In addition, the
server also retrieves all the objects that fall inside or
intersect with this known region $K(q_{s},r)$. In our example,
$o_{6}$ and $o_7$ are also retrieved as these objects also
intersect with $K(q_{s},r)$.}

\eat{Our algorithm first requests the server for objects and the
known region by providing the starting point of the trajectory as
a query point. This process needs to be repeated as soon as the
user's request for a PMNN cannot be satisfied by the already
retrieved data at the client. Next we explain how to evaluate a
PMNN at the client based on the retrieved objects within the known
region. }

The key idea of incremental approach is to consider only a sub-set
of objects surrounding the moving query point while evaluating a
PMNN query. For example, in a client-server paradigm, the client
first requests the server for objects and the known region by
providing the starting point of the moving query path as a query
point. Then the client locally creates a PVD based on the
retrieved objects, and uses the local PVD for answering the PMNN
query. This process needs to be repeated as soon as the user's
request for the PMNN query cannot be satisfied by the already
retrieved data at the client. Though this incremental approach
applies to both centralized and client-server paradigms, without
loss of generality, next we explain how to incrementally evaluate
a PMNN query in the client-server paradigm.

\textbf{Algorithm: }After retrieving a set of objects from the
server, the client locally computes a PVD for those objects. Then,
the client can use the local PVD to determine the most probable
nearest neighbor among the objects inside the known region.
However, since the client does not have any knowledge about
objects that are outside of the known region, the most probable
nearest neighbor based on the local PVD formed for objects inside
the known region, might not guarantee the most probable nearest
neighbor with respect to \emph{all} objects in the database. This
is because, a PVC of the local PVD determines the region where the
corresponding object is the most probable NN with respect to
objects inside the known region. However, certain locations of the
PVC can have other non-retrieved objects, which are outside the
known region, as the most probable NN. Thus, we need to determine
a region in the PVC for which the query result is guaranteed. That
is, all locations inside this guaranteed region will have the
corresponding object as the most probable NN. To define the
guaranteed region for an object, we have two conditions.

\eat{
Thus, in our approach, we classify PVC cells into
two categories: inner cell, and outer cell. A cell that is totally
contained within the known region is an inner cell. On the other
hand, a cell that has some portions outside the boundary of the
known region is an outer cell. If the current query position is
anywhere inside an inner PVC cell (not shown in the figure), then
the corresponding object is ensured of being the most probable
nearest neighbor among all the objects in the database. On the
other hand, if the query location is inside an outer PVC cell,
then it is only ensured that the corresponding object is the most
probable nearest among all the objects inside the known region.}

Let $q$ be a query point and $o_i$ be an object inside the known
region. Then, if the query point $q$ is inside a PVC cell of
object $o_i$ and the condition in the following equation (see
Equation~\ref{eq:100safe}) holds, then it is ensured that $o_i$ is
the most probable NN among all objects in the database.

\begin{equation}
\label{eq:100safe}
    maxdist(q,c_{i}) \leq r - dist(q,q_s).
\end{equation}

The condition in Equation~\ref{eq:100safe} ensures that no object
outside the known region can be the nearest neighbor for the given
query point. This is because, when a circle centered at $q$
completely contains an object, all objects outside this circle
will have zero probability of being the NN to $q$.

To formally define a region based on the above inequality, we re-arrange Equation~\ref{eq:100safe} as follows.

\begin{align*}
\label{eq:100safe1}
    dist(q,c_{i}) + r_i \leq r - dist(q,q_s)\\
    => dist(q,c_{i}) + dist(q,q_s) \leq r - r_i
\end{align*}

We can see that the boundary of the above formula forms an elliptic region in a 2D Euclidean space, where the two foci of the ellipse are $q_s$ and $c_i$. i.e.,  the sum of the distances from $q_s$ and $c_i$ to any point on the ellipse is $r-r_i$. Figure~\ref{fig:circle_safe} shows an example, where the elliptical region for object $o_2$ is shown using dashed border. Figure~\ref{fig:circle_safe} shows that when the query point is at
$q_{1}$, the object $o_2$ is confirmed to be the most probable nearest neighbor, as
$dist(q_{1},c_{2})+r_{2}<r-dist(q_{1},q_{s})$.

From the above discussion, we see that for an object $o_i$, the intersection region of the PVC and the elliptical region for $o_i$ forms a region where all points in this region has $o_i$ as the most probable NN. Figure~\ref{fig:circle_inc_pvd} shows the PVD and elliptical regions for objects $o_1$, $o_2$, and $o_3$, and a moving query path from $q^{\prime}$ to $q^{\prime\prime\prime}$. In this figure, since $q^{\prime}$ is inside the intersection region of $PVC(o_3)$ and elliptical region of $o_3$, thus $o_3$ is guaranteed to be the most probable NN for $q^{\prime}$ with respect to all objects in the database. Similarly, $o_2$ is the most probable NN when the query point moves to $q^{\prime\prime\prime}$.

If a query point is outside the intersection region of a PVC and the corresponding elliptical region, but falls inside the PVC, still there is a possibility that the object associated with this PVC is the most probable NN for the query point. For example, in Figure~\ref{fig:circle_safe}, when the query point is at $q_{2}$, then
the condition in Equation~\ref{eq:100safe} fails. For
this case, our algorithm relies on the lower bound of
the probability for the object $o_2$ of being the nearest
neighbor from the query point $q_{2}$. We define the second condition based on the lower bound probability of an object.

\begin{figure}[htbp]
    \centering
        \includegraphics[width=1.5in]{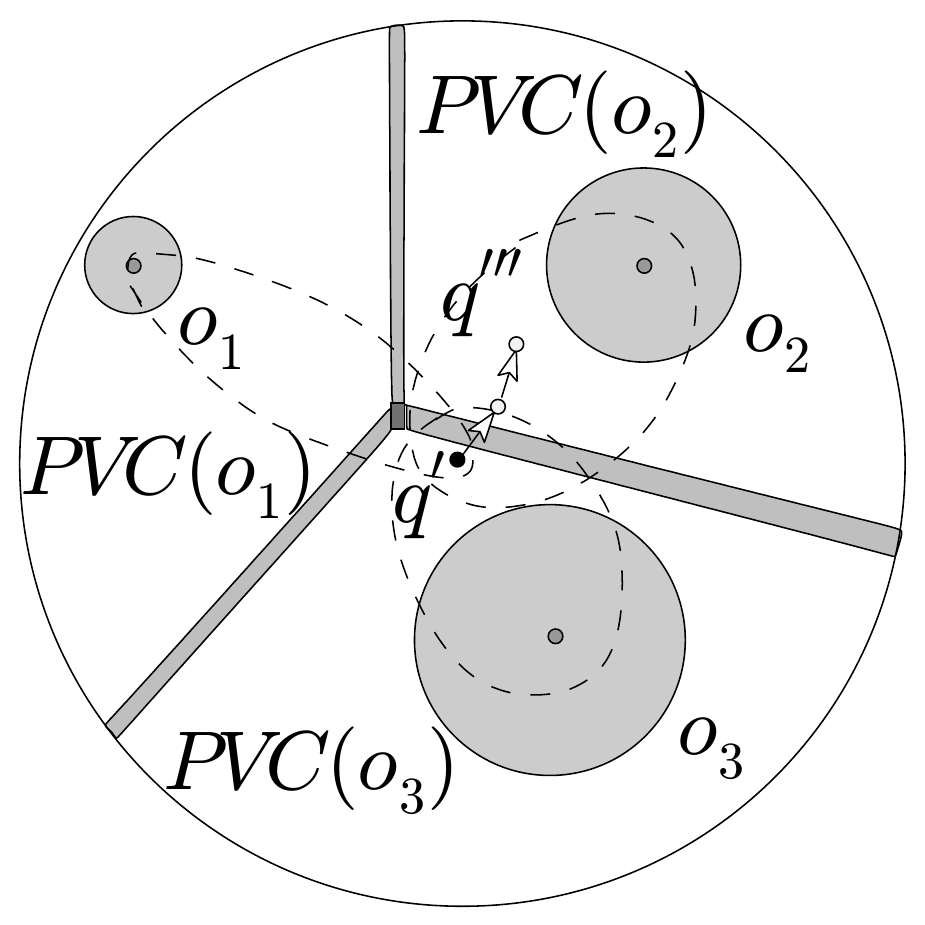}
    \caption{The incremental approach}
    \label{fig:circle_inc_pvd}
\end{figure}

We can compute the lower bound of probability,
$lp(o_i,q)$  for object $o_i$ of being the NN from
the query point $q$, by using pessimistic
assumption. For computing the lower bound
probability, we assume that a non-retrieved \emph{virtual} point
object is located at the minimum distance from the query point and
is just outside the boundary surface of the known region. For
example, in Figure~\ref{fig:circle_safe}, when the client is at
$q_{2}$, we assume that a point object exists at
$d^{\prime}$. Then, we estimate the probability of the
object $o_2$ being the NN to $q_{2}$, which gives us
the lower bound of the probability.

\eat{On the other hand, to
calculate the upper bound of the probability of an object, we use
an optimistic assumption i.e., we assume that there is no
non-retrieved object inside the circular region centered at the
query point $q$ with the radius $dist(q,maxdist(q,o_{i}))$. For
example, in Figure~\ref{fig:circle_safe}, for the current query
point $q^{\prime\prime}$, and the object $o_4$, we assume that
there is no object outside the known region that intersects or
fall inside the circle centered at $q^{\prime\prime}$ with a
radius $dist(q^{\prime\prime},c_4)+r_4$.}

By using the lower bound, the client can determine whether there
is a possibility of other non-retrieved objects being the most
probable NN from the current query location. If the probability of the virtual point object $o_v$,
$p(o_v,q)$ is less than the lower bound probability of the
candidate object $o_i$, $lp(o_i,q)$, then it is ensured that there
is no other object in the database that has higher probability for
being the NN of $q$ than that of $o_{i}$; otherwise there may
exist other object in the database with higher probabilities for
being the NN of $q$ than $o_i$. Thus, our second condition for the guaranteed region can be defined as follows:

\begin{equation}
\label{eq:100safe2}
   lp(o_i,q) \geq p(o_v,q).
\end{equation}

\eat{On the other hand, a user might be satisfied with the answer even
if an object is not guaranteed to be the most probable NN. In that
case, the user may want to know the maximum probability (upper
bound) of the object being the NN based upon the current knowledge
from the retrieved object within the known region. If the user's
required level of probability threshold, say $\beta$, is  smaller than the
estimated upper bound of object $o_i$, $up(o_i,q)$, the client
can avoid contacting the server as there is a possibility that $o_i$ is the most probable NN from $q$. This process
reduces the communication overheads for a PMNN query.}

Based on the above observations, we define a \emph{probabilistic safe
region} for an object $o_i$, as a region where $o_i$ is guaranteed
to be the most probable NN for every point
inside that region. Thus,
Equation~\ref{eq:100safe} and Equation~\ref{eq:100safe2} form the guaranteed region for an object $o_i$.

We use the above two conditions and the local PVD to incrementally
evaluate a PMNN query. The algorithm first retrieves a set of
surrounding objects for the given query point $q$, and creates a
PVD, named $lPVD$, for those objects. Then, the algorithm finds
the PVC  and the corresponding object $o_i$ as the most probable
nearest neighbor of the query point $q$ with respect to the
objects within the known region. If $q$ is inside a PVC
cell, the object $o_i$ is returned as the most probable nearest
neighbor if $q$ satisfies Equation~\ref{eq:100safe} or Equation~\ref{eq:100safe2}.
If none of the above condition holds, the
algorithm requests a new set of objects with respect to the
current query point $q$, and repeats the above process on newly
retrieved set of objects.

\noindent\emph{Discussion:}
Our pre-computation based approach computes the PVD for all objects in the database and then indexes the PVD using an $R$-tree to efficiently  process PMNN queries. Since, the pre-computation of the PVD for the entire data set is computationally expensive, the pre-computation based approach is justified when the PVD can be re-used for large number of queries, as the cost is amortized among queries(e.g.,~\cite{okabe00:voronoi,zhang:sq}).
Thus, the pre-computation based approach is appropriate for the following settings: the data set largely remains static, there are large number of queries in the system, and the query spans over the whole data space.

On the other hand, in our incremental approach, we retrieve a set of surrounding objects for the current query location, and then incrementally process PMNN queries based only on these retrieved set of objects. Only data close to the given query are accessed for query evaluation. As the evaluation of this approach depends on the location of the query, this approach is also called the query dependent approach, as opposed to the data dependent approach (e.g., pre-computation based approach) where the location of queries are not taken into account. This incremental approach is preferred for the cases when there are updates in the database or the query is confined to a small region in the data space. A comparative discussion between the pre-computation approach and the incremental approach for point data sets can be found in~\cite{nutanong08:vstar,zhang:sq}.

\section{Experimental Study}
\label{sec:exp} We compare our PVD based approaches for the PMNN
query (P-PVD-PMNN and I-PVD-PMNN) with a sampling based approach
(Naive-PMNN), which processes a PMNN query as a sequence of static
PNN queries at sampled locations. Though in Naive-PMNN we use the
most recent technique of static top-1-PNN
queries~\cite{beskales08.vldb}, any existing technique for static
PNN queries~\cite{Cheng03.sigmod,Cheng04.tkde} can be used. Note
that, by using the existing method in~\cite{Cheng.ICDE10}, for
each uncertain object $o_i$, we could only define a region (or
UV-cell) where $o_i$ has a non-zero probability of being the NN
for any point in this region. Thus, this method cannot be used to
determine whether an object has the highest probability of being
the NN to a query point. Therefore, we compare our approach with a
sampling based approach.

In our experiments, we measure the query processing time, the I/O
costs, and the communication costs as the number of communications
between a client and a server. Note that while the processing and
I/O costs are the performance measurement metric for both
centralized and client-server paradigms, the communication cost
only applies to the client-server paradigm. In this paper, we run
the experiments in the centralized paradigm, where the query
issuer and the processor reside in the same machine. Thus, we
measure the communication cost as the number of times the query
issuer communicates with the query processor while executing a
PMNN query.

\eat{
In this article, we focus on presenting experimental results for
2D data, with representative results for 1D data. Because 2D
data can be seen as superset of 1D data and the experimental
results of 1D data show similar results as of 2D data.}

\subsection{Experimental Setup}

We present experimental results for both 1D and 2D data sets.

For 2D data, we have used both synthetic and real data sets. We
normalize the data space into a span of $10,000\times10,000$
square units. We generated synthetic data sets with uniform (U)
and Zipfian (Z) distributions, representing a wide range of real
scenarios. For both uniform and Zipfian, we vary the data set size
from 5K to 25K. To introduce uncertainty in data objects, we
randomly choose the uncertainty range of an object between
$5\times5$ and $30\times30$ square units, and approximated the
selected range using a circle. For real data distributions, we use
the data sets from Los Angeles (L) with 12K geographical objects
described by ranges of longitudes and latitudes~\cite{tiger}. Note
that, in both uniform and Zipfian distributions, objects can
overlap each other. More importantly, in Zipfian distribution,
most of the objects are concentrated within a small region in the
space, thereby objects largely overlap with each other. Also, our
real datasets include objects with large and overlapping regions.
Thus, we do not present any sperate experimental results for
overlapping objects.

For 1D data, we have only used syntectic data sets. In this case, we generated
synthetic data sets with uniform (U) and Zipfian (Z) distributions
in the data space of 10,000 units. The uncertainty range of
an object is chosen as any random value between 5 and 30 units. We also vary the data set size from 100
to 500. These values are comparable to 2D data
set sizes and scenarios.

For query paths, we have generated two different types of query trajectories,
random (R) and directional (D), representing the query movement
paths covering a large number of real scenarios. The default length of a trajectory
is a fixed length of 1000 steps, and consecutive points are
connected with a straight line of a length of 5 units. For each
type of query path, we run the experiments for 20 different
trajectories starting at random positions in the data space, and
determine the average results. We present the processing time, I/O cost, and the communication cost for executing a complete trajectory (i.e., a PMNN query). In our experiments, since the
trajectory of a moving query path is unknown, we use the generated
trajectories as input, but do not provide these to the server in
advance.

We run the experiments on a desktop computer with Intel(R) Core(TM) 2 CPU 6600 at 2.40 GHz and 2 GB RAM.

\subsection{Performance Evaluation}
In this section, we evaluate our proposed techniques: pre-computation
approach (P-PVD-PMNN) and incremental approach (I-PVD-PMNN) in
Sections~\ref{subsubsec:ppvd} and~\ref{subsubsec:ipvd},
respectively.

It is well known that pre-computation based approach is
suitable for settings when the PVD can be re-used (e.g., static data sets) for large number of queries or the query span the whole data space, and on the other hand the
incremental or local approach is suitable for settings when the query is confined to
a small space and there are frequent updates in the database (e.g.,~\cite{nutanong08:vstar,okabe00:voronoi,zhang:sq}). Since two approaches aim at two different environmental settings
and also the parameters of these two techniques differ from each
other, we independently evaluate them and compare them with the
sampling based approach.

\subsubsection{Pre-computation Approach}
\label{subsubsec:ppvd} In the pre-computation approach, we first
create the PVD for the entire data set and  use an $R^{*}$-tree to
index the MBRs of PVCs. On the other hand, for Naive-PMNN we use
an $R^{*}$-tree to index uncertain objects. In both cases, we use
the page size of 1KB and the node capacity of 50 entries for the
$R^{*}$-tree.

\noindent\emph{\textbf{Experiments with 2D Data Sets: }}

We vary the following parameters in our experiments: the
length of a query trajectory, the data set size, and the size of the
buffer window that determines the number of PVCs retrieved
each time with respect to a query point.

\emph{Effect of the Length of a Query Trajectory}: In this set of
experiments, we vary the length of moving queries from 1000 to
5000 units of the data space. We run the experiments for data sets U(10K), Z(10K),
and L(12K). Since the real data set size is 12K, the data set sizes
for U and Z are both set to 10K. Figures~\ref{fig:vq} show the
processing time, I/O costs, and the number of communications
required for a PMNN query of different query trajectory length. Figures~\ref{fig:vq}(a)-(c) present the results for U data sets, where we can see that, for both P-PVD-PMNN and Naive-PMNN, the
processing time, I/O costs, and the number of communications increase with the increase of the length of the query trajectory, which is expected. Figures also show that our P-PVD-PMNN approach outperforms the Naive-PMNN by at least
an order of magnitude in all metrics. This is because, P-PVD-PMNN only needs to identify the current PVC
rather than computing top-1-PNN for every sampled location of the
moving query.

The results for both Z (see Figures~\ref{fig:vq}(d)-(f)) and L
(see Figures~\ref{fig:vq}(g)-(i)) data sets show similar trends
with U data set as described above.

\begin{figure*}[htbp]
 \begin{center}
    \begin{tabular}{cccc}
        \hspace{-5mm}
      \resizebox{40mm}{!}{\includegraphics{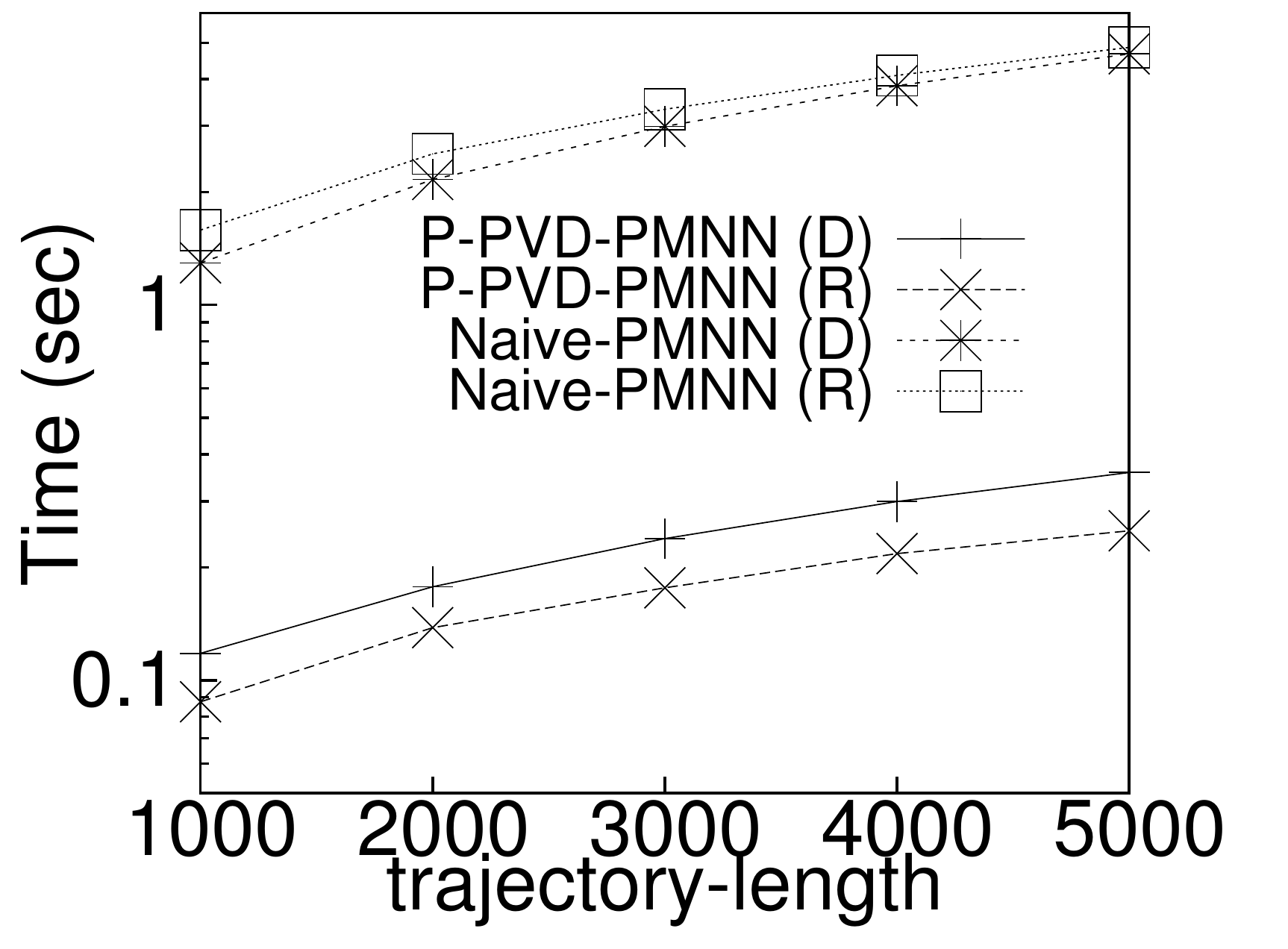}} &
        \hspace{-4mm}

        \resizebox{40mm}{!}{\includegraphics{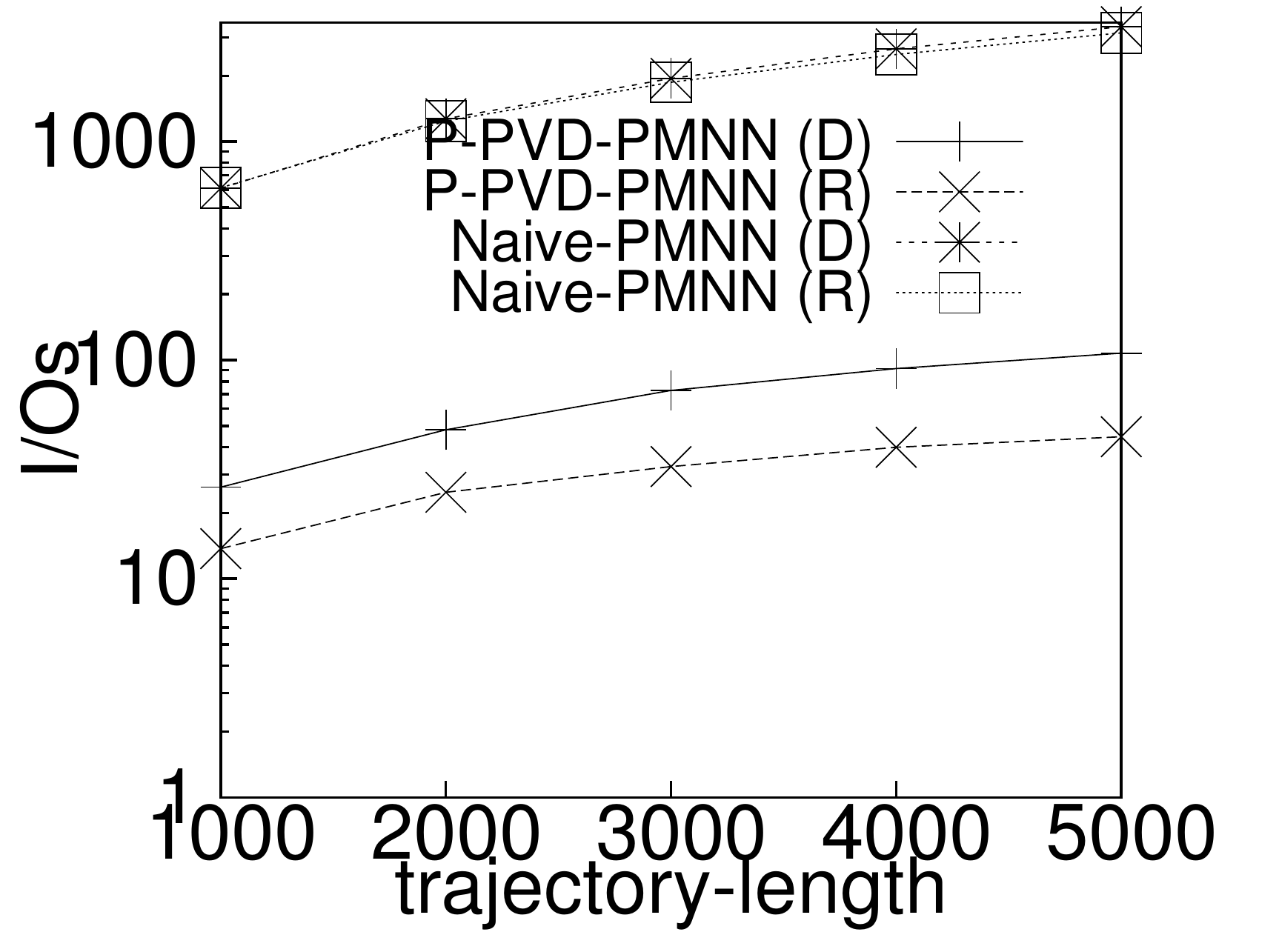}} &
         \hspace{-4mm}

        \resizebox{40mm}{!}{\includegraphics{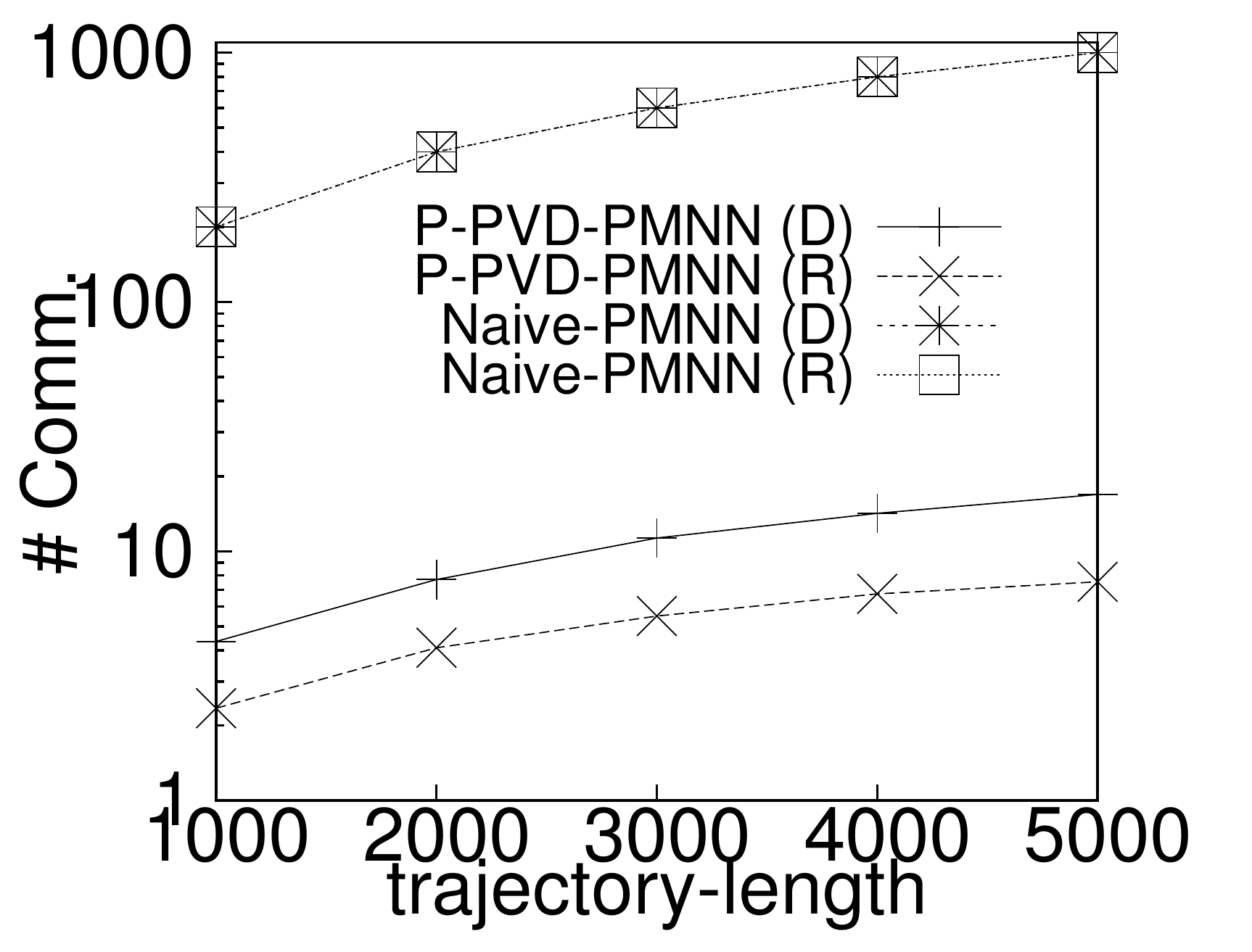}}\\
       \scriptsize{(a)\hspace{0mm}} & \scriptsize{(b)} & \scriptsize{(c)}\\
      \end{tabular}
    \begin{tabular}{cccc}
        \hspace{-5mm}
      \resizebox{40mm}{!}{\includegraphics{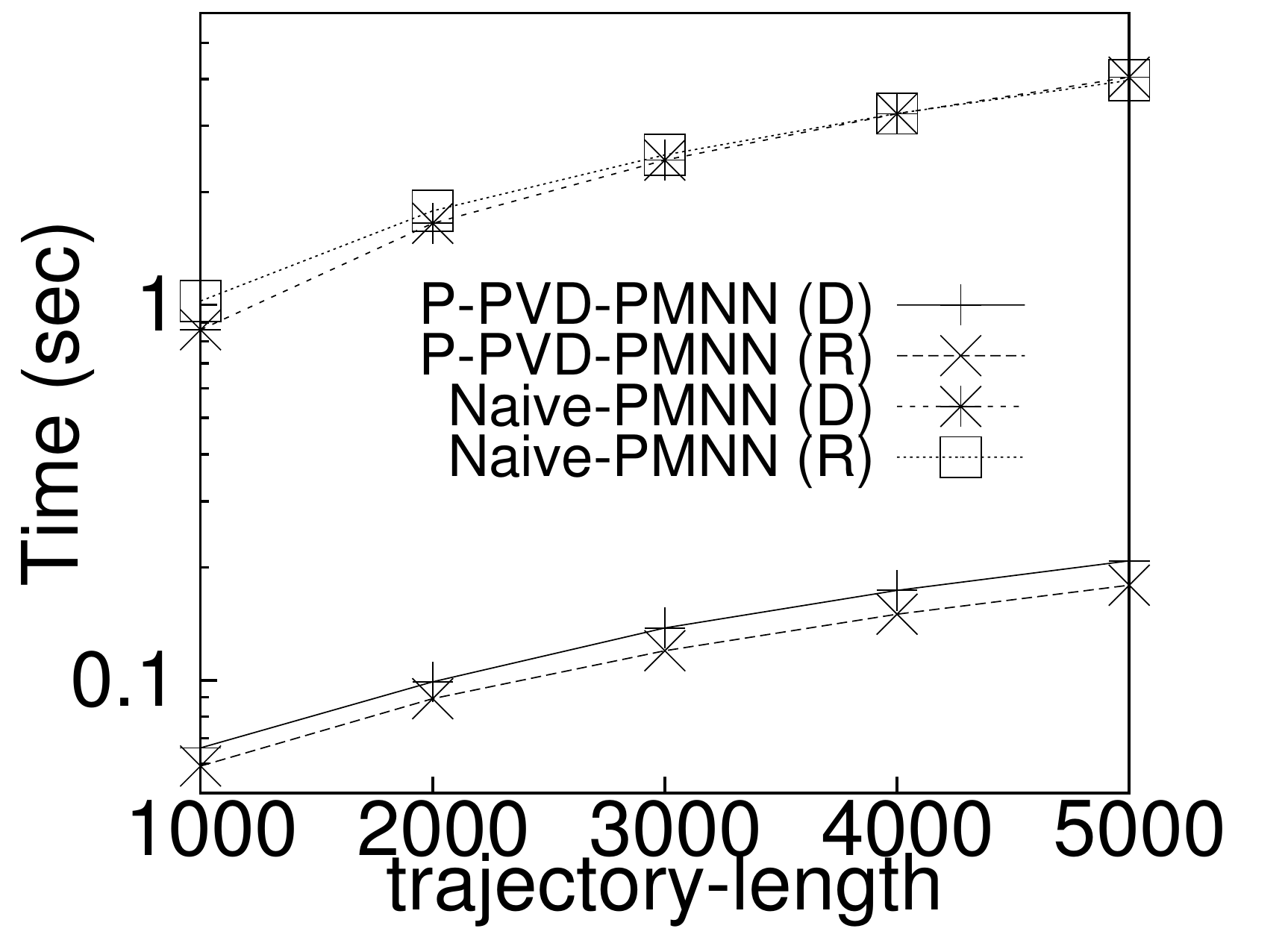}} &
        \hspace{-4mm}

        \resizebox{40mm}{!}{\includegraphics{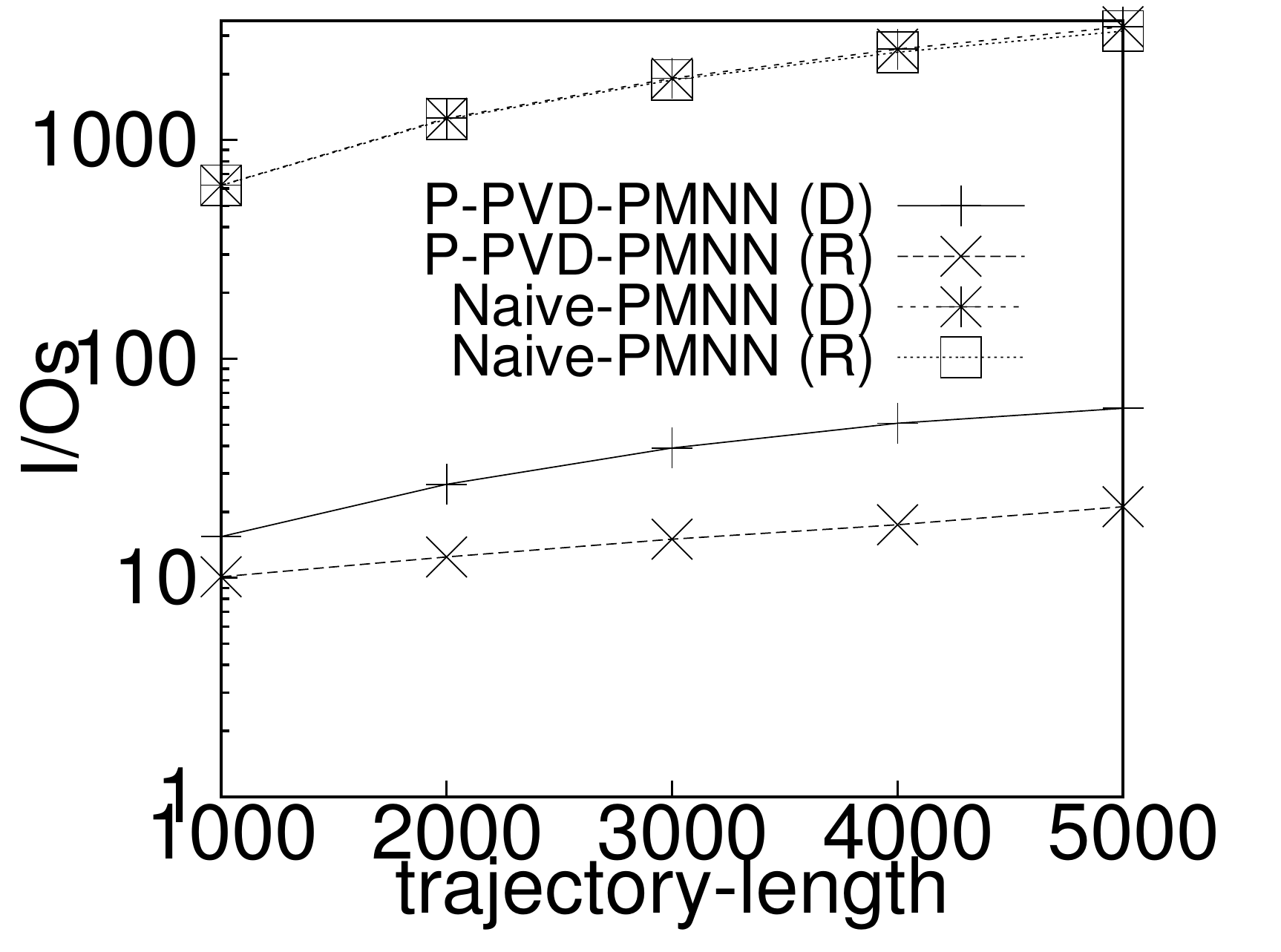}} &

        \hspace{-4mm}

        \resizebox{40mm}{!}{\includegraphics{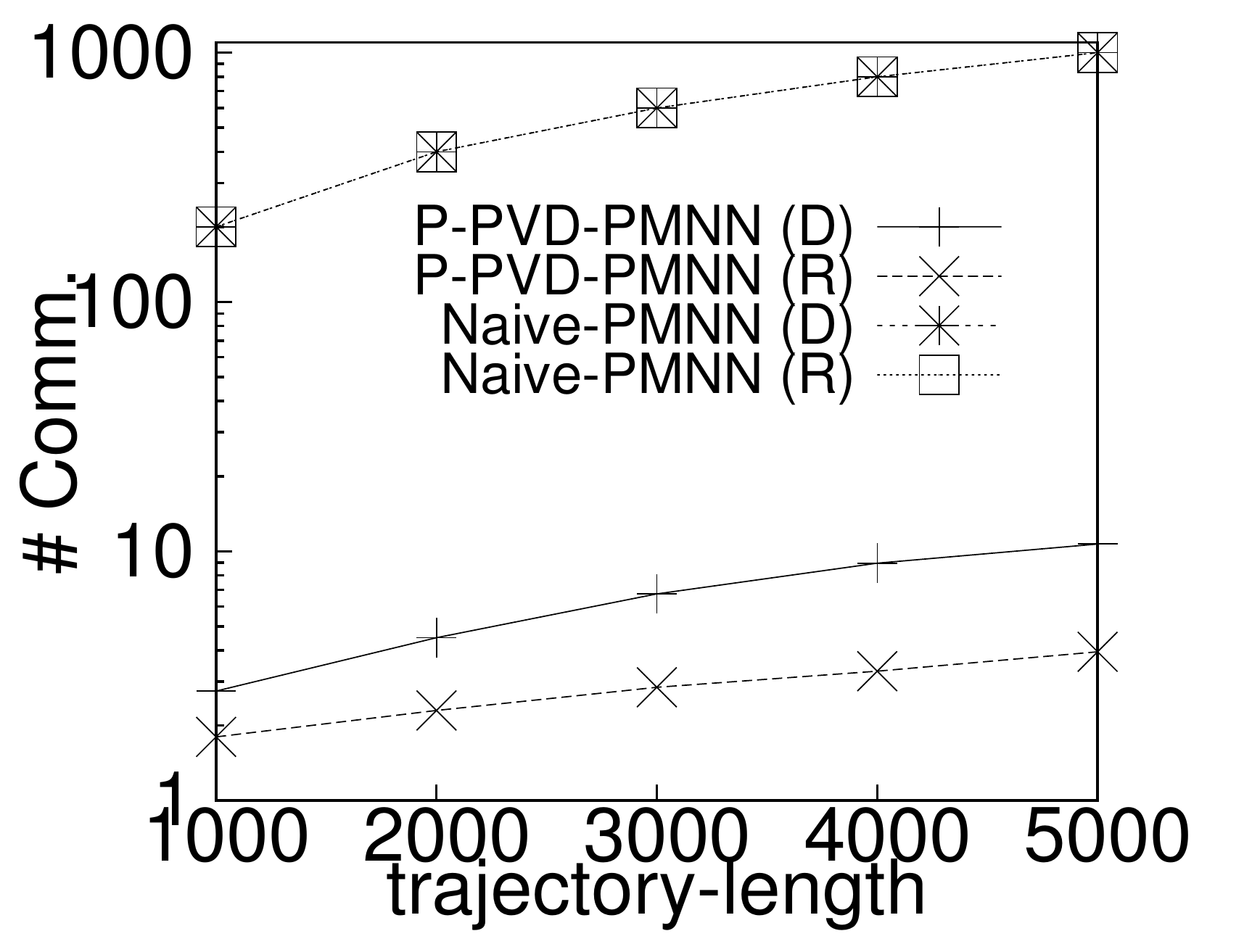}}\\
       \scriptsize{(d)\hspace{0mm}} & \scriptsize{(e)} & \scriptsize{(f)}\\
      \end{tabular}
      \begin{tabular}{cccc}
        \hspace{-5mm}
      \resizebox{40mm}{!}{\includegraphics{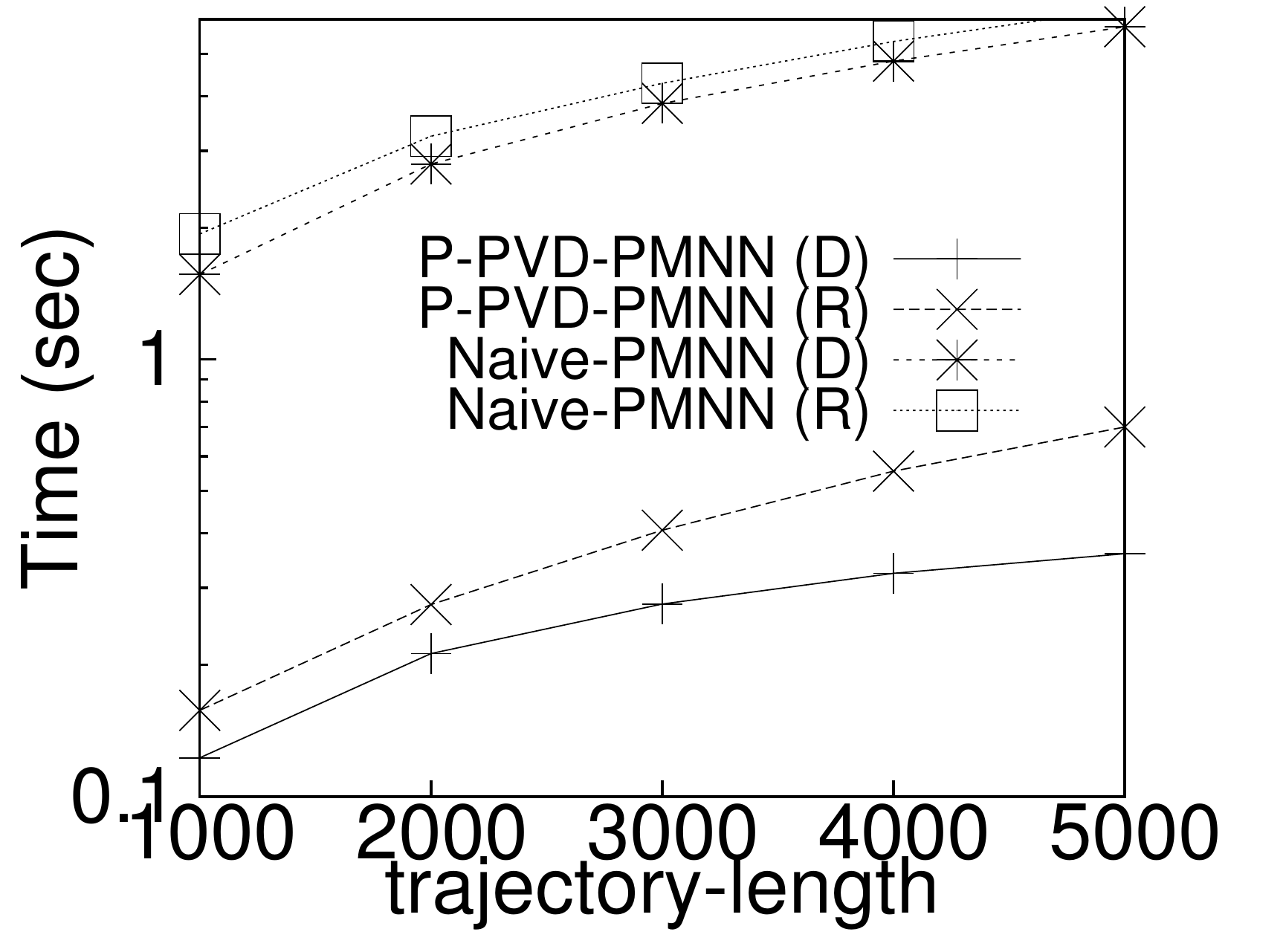}} &
        \hspace{-4mm}

        \resizebox{40mm}{!}{\includegraphics{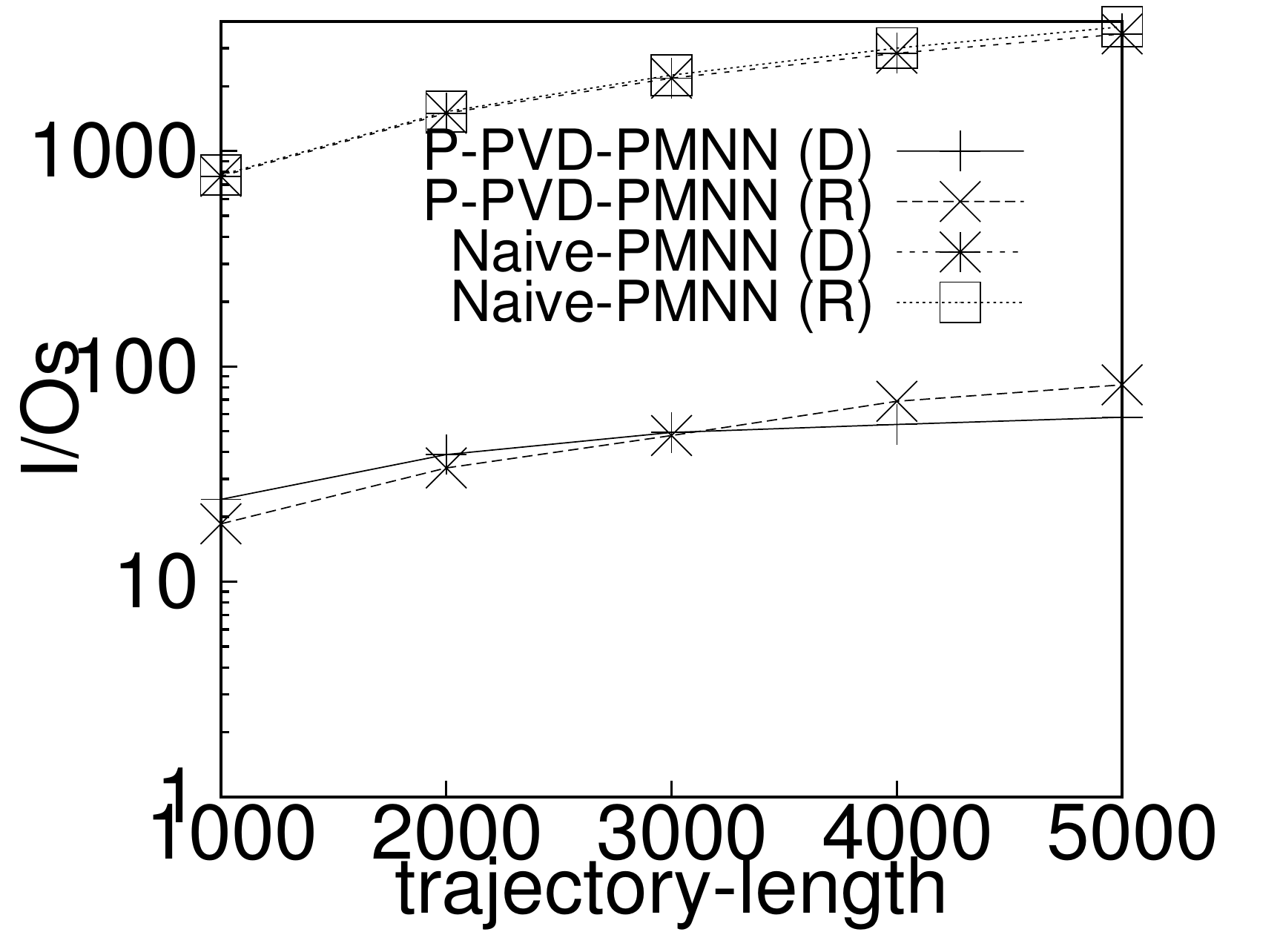}} &
         \hspace{-4mm}

        \resizebox{40mm}{!}{\includegraphics{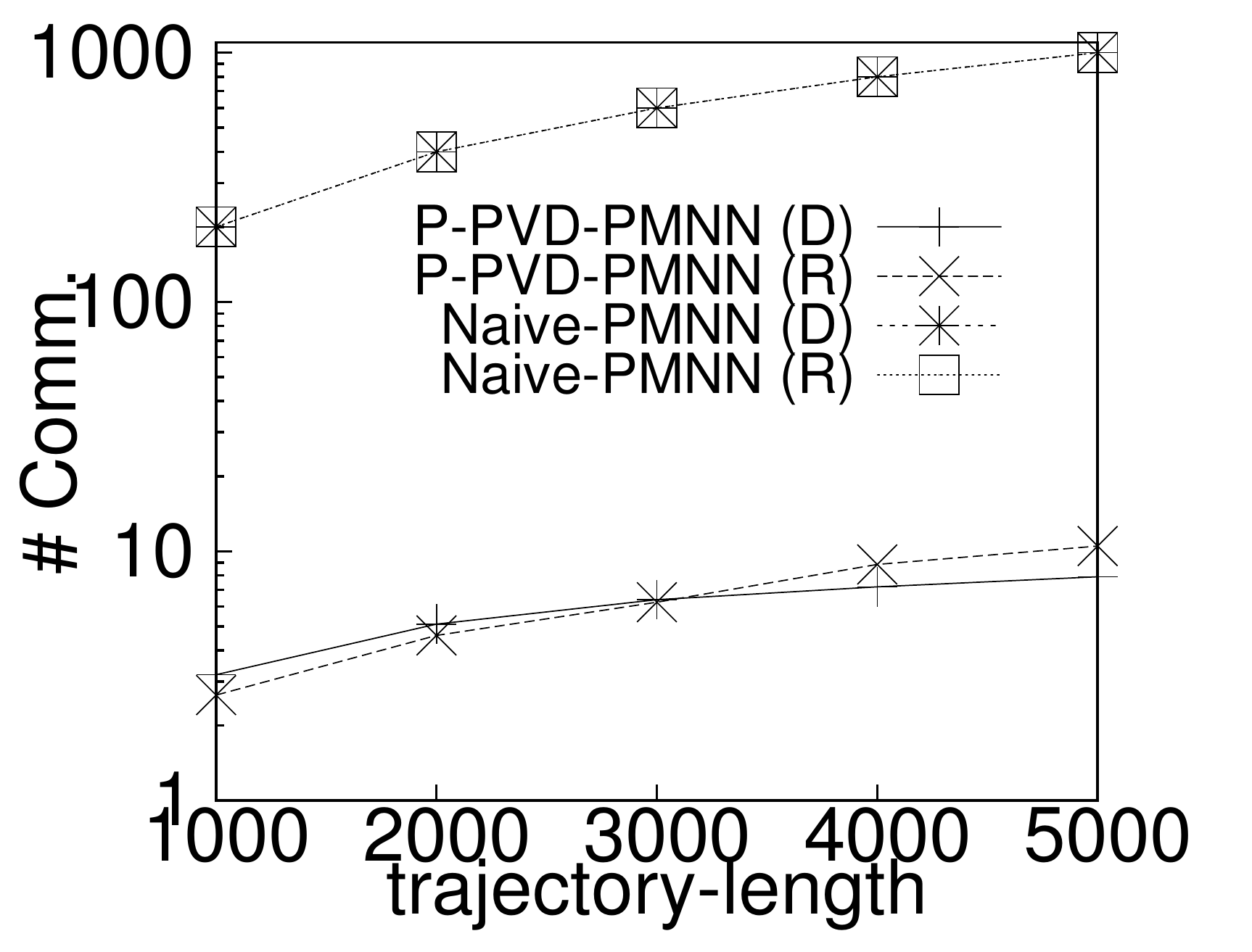}}\\
       \scriptsize{(g)\hspace{0mm}} & \scriptsize{(h)} & \scriptsize{(i)}\\
      \end{tabular}
    \caption{The effect of the query trajectory length in U (a-c), Z (d-f), and L (g-i)}
    \label{fig:vq}
  \end{center}
\end{figure*}

\begin{figure*}[htbp]
  \begin{center}
    \begin{tabular}{cccc}
        \hspace{-5mm}
      \resizebox{40mm}{!}{\includegraphics{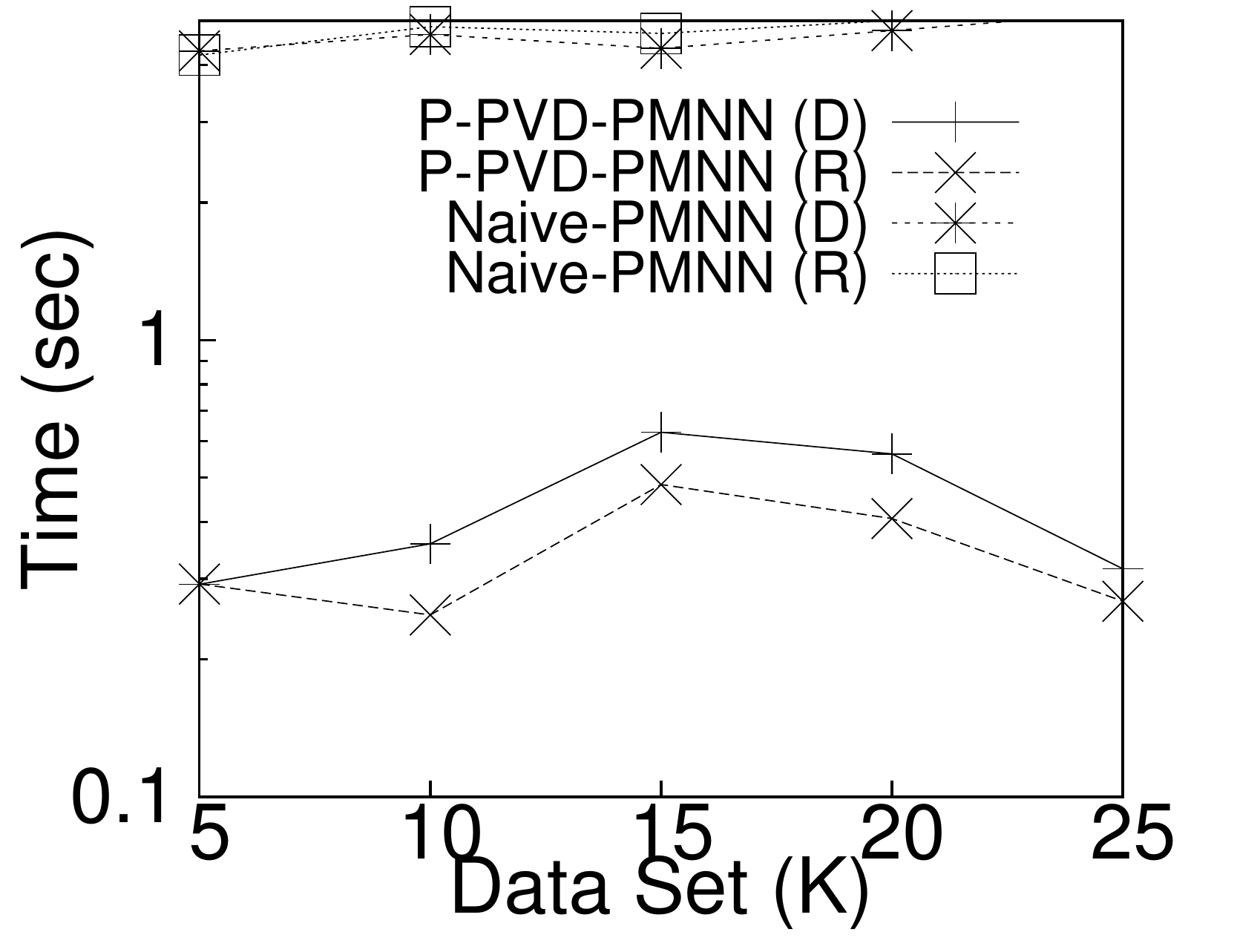}} &
        \hspace{-4mm}

        \resizebox{40mm}{!}{\includegraphics{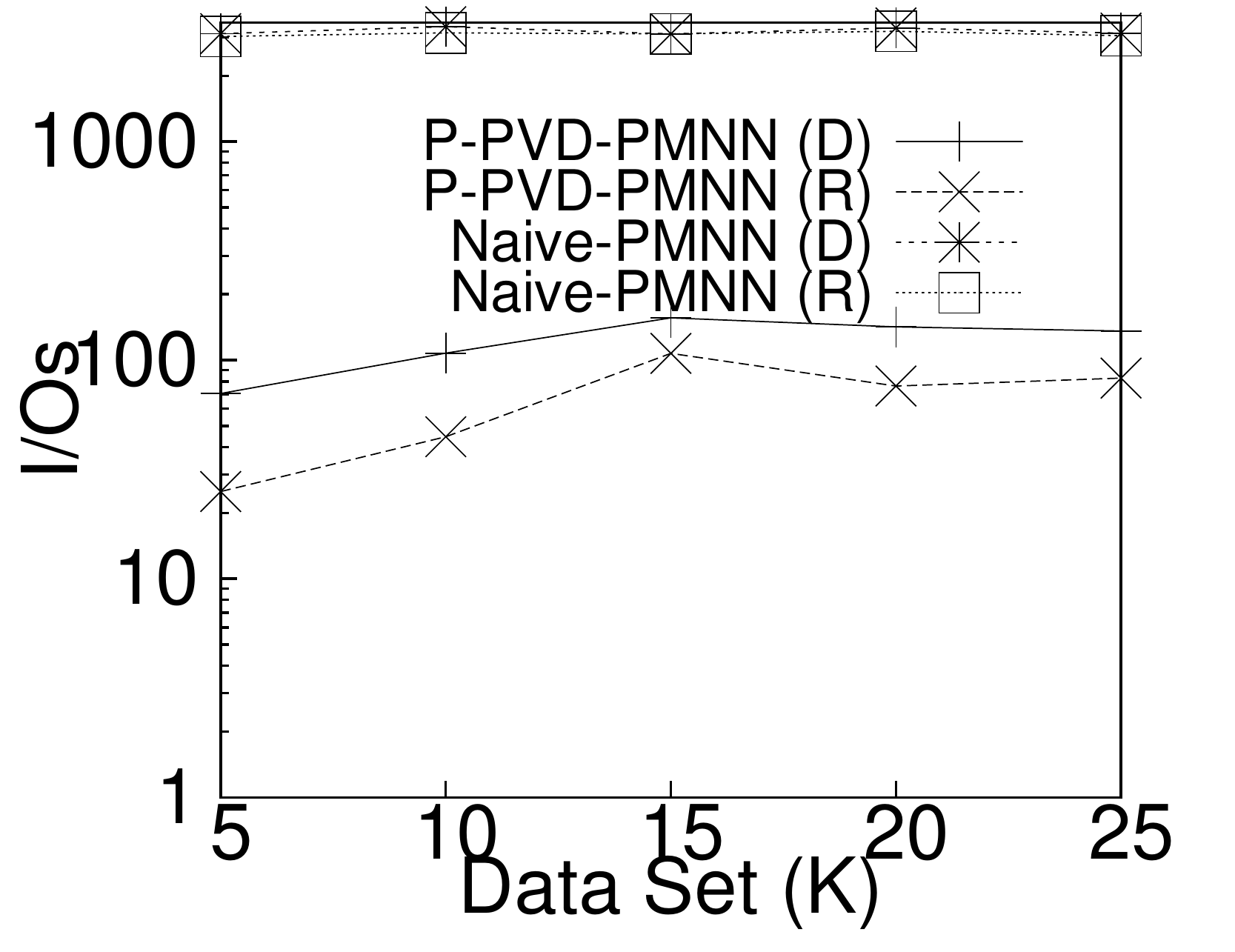}} &

        \hspace{-4mm}

        \resizebox{40mm}{!}{\includegraphics{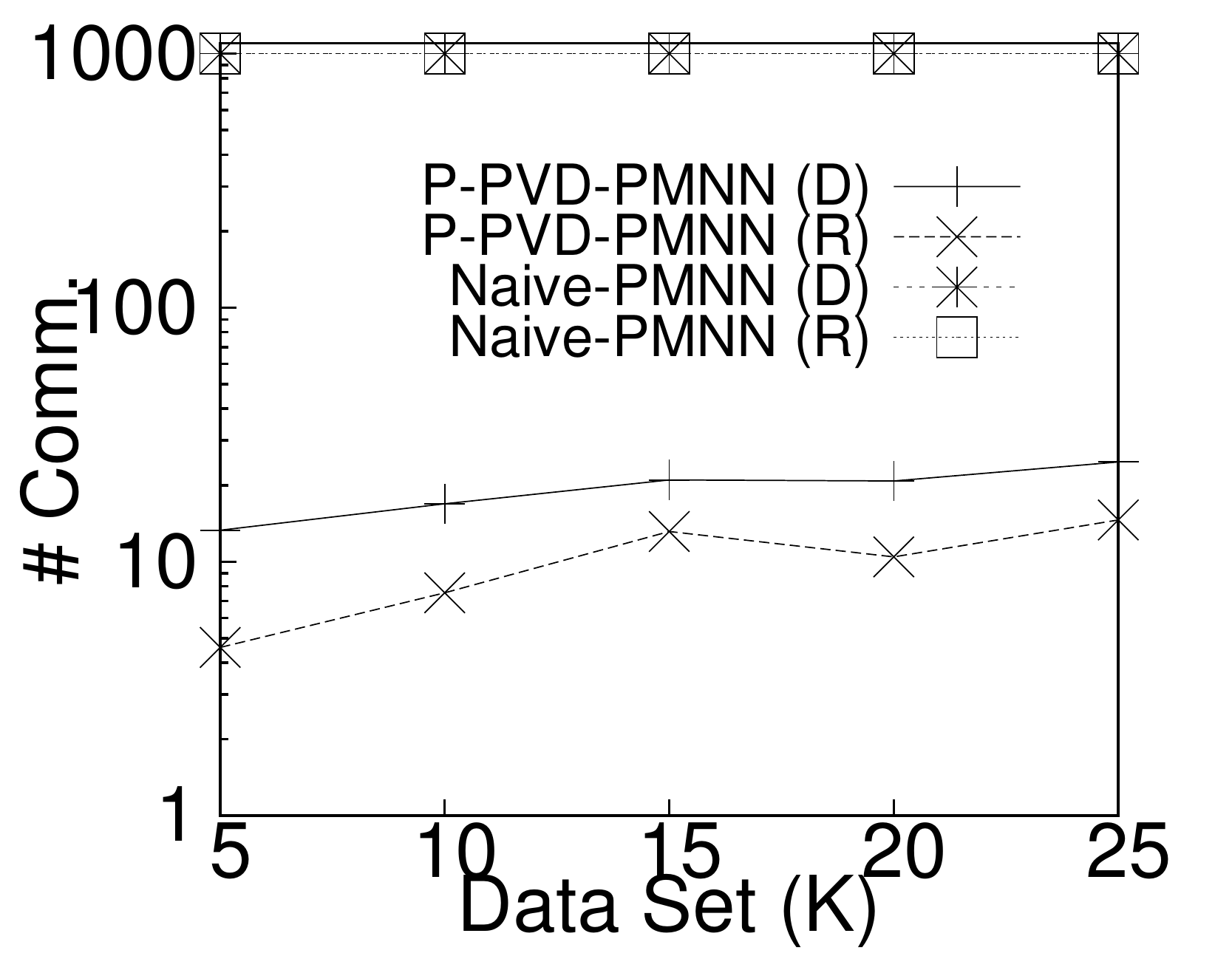}}\\
       \scriptsize{(a)\hspace{0mm}} & \scriptsize{(b)} & \scriptsize{(c)}\\
      \end{tabular}
      \begin{tabular}{cccc}
        \hspace{-5mm}
      \resizebox{40mm}{!}{\includegraphics{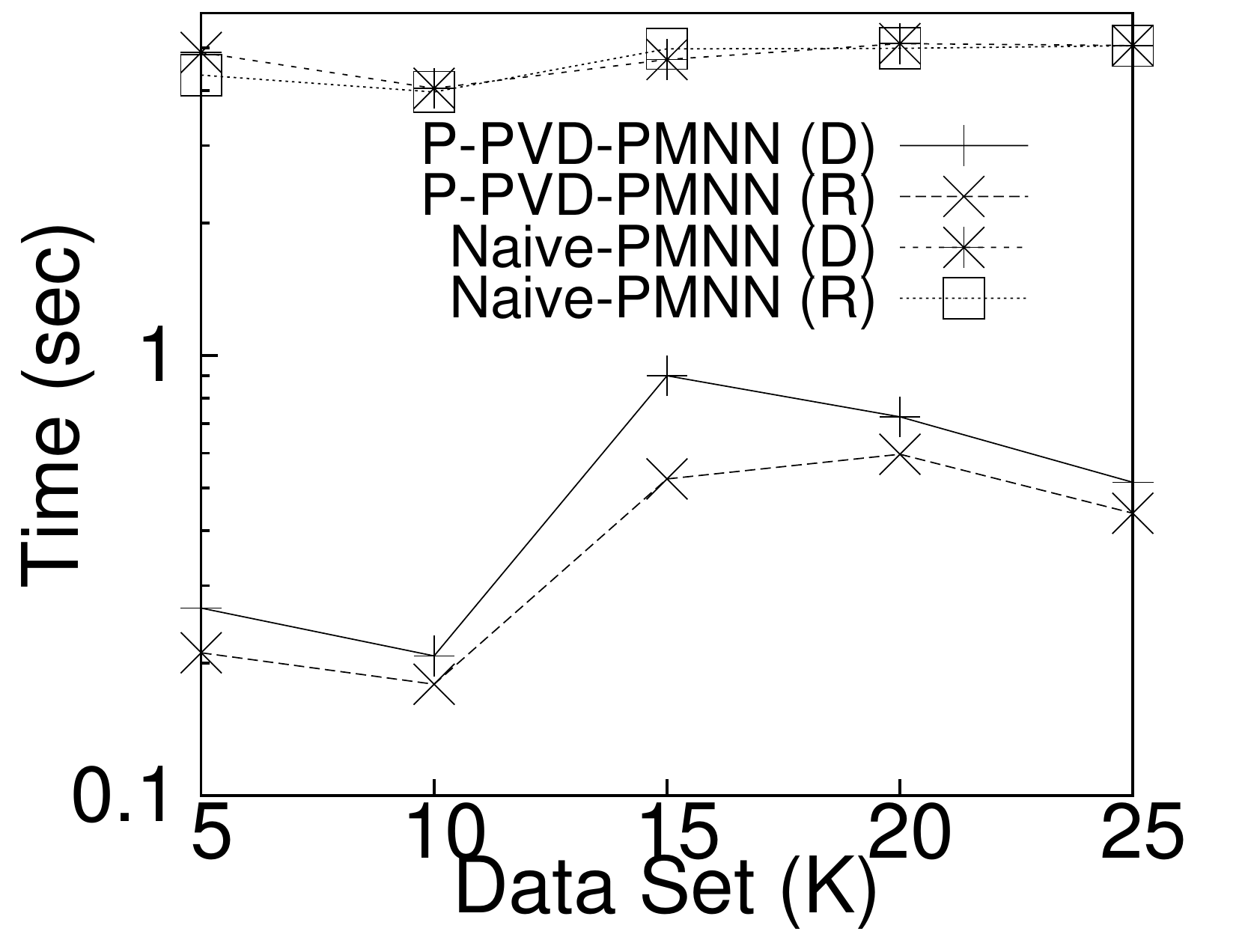}} &
        \hspace{-4mm}

        \resizebox{40mm}{!}{\includegraphics{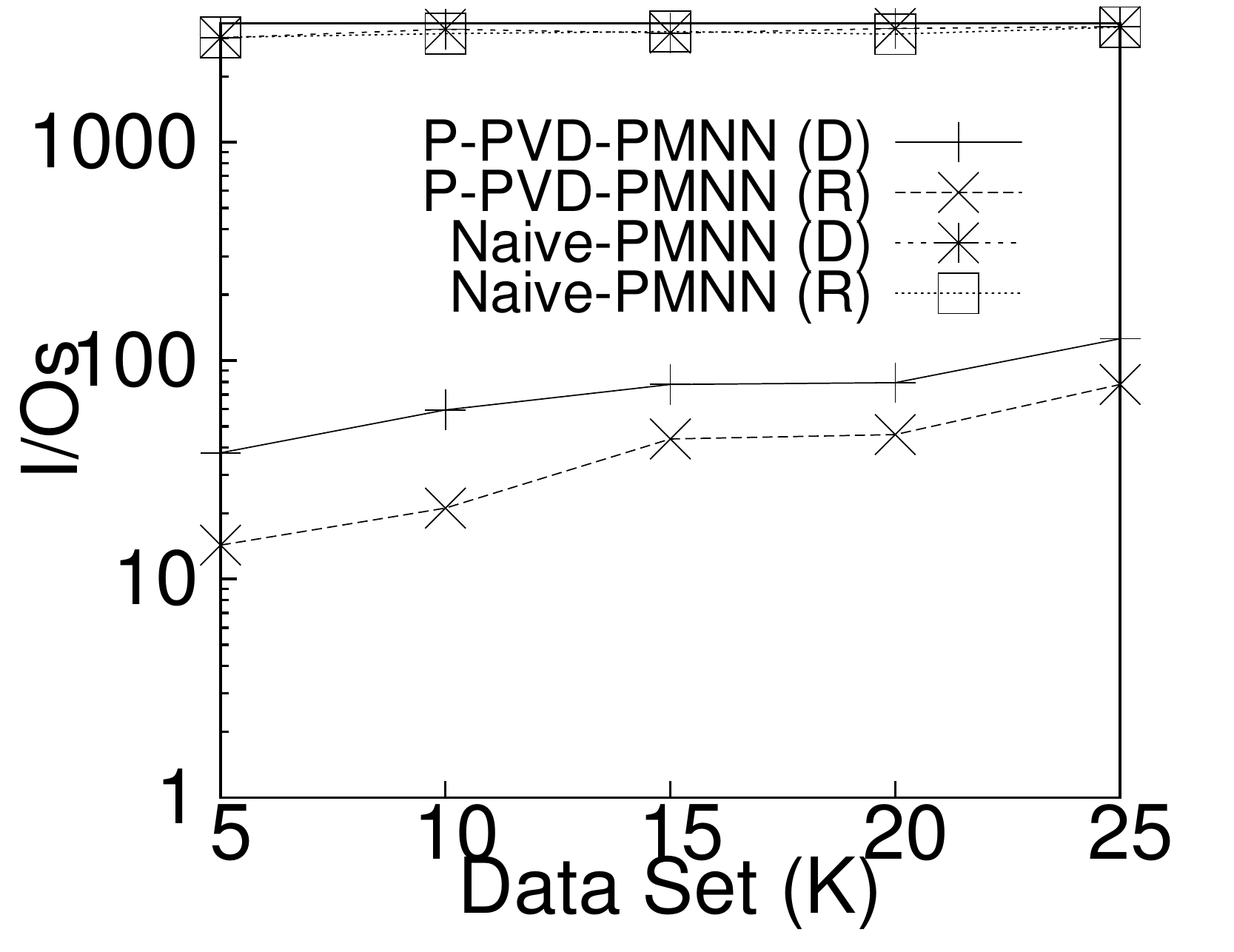}} &

        \hspace{-4mm}

        \resizebox{40mm}{!}{\includegraphics{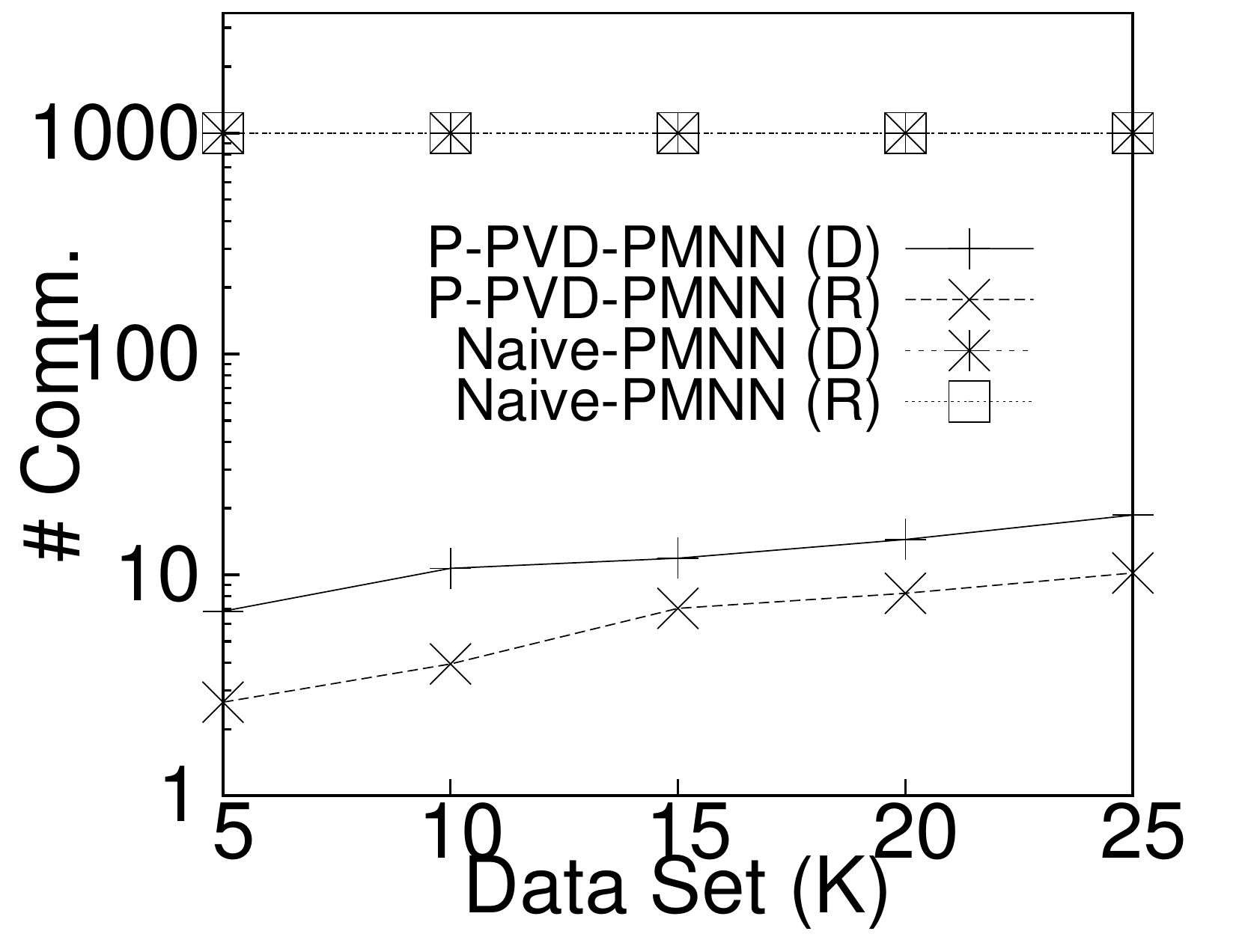}}\\
       \scriptsize{(d)\hspace{0mm}} & \scriptsize{(e)} & \scriptsize{(f)}\\
      \end{tabular}
    \caption{The effect of the data set size in U (a-c), Z (d-f)}
    \label{fig:vds}
  \end{center}
\end{figure*}

\emph{Effect of Data Set Size}: In this set of experiments, we
vary the data set size from 5K to 25K and compare the performance
of our P-PVD-PMNN with Naive-PMNN for both U (see
Figures~\ref{fig:vds}(a)-(c)) and Z (see
Figures~\ref{fig:vds}(d)-(f)) distributions. In these experiments,
we set the trajectory length to 5000 units.

Figures~\ref{fig:vds}(a)-(f) show that, in general for P-PVD-PMNN, the processing time and
I/O costs, and the number of communications increase with the increase of the data set size. The reason is as follows. For a larger data set, since the density of objects is high, we have smaller PVCs. Thus, for a larger data set, as the query point moves, it crosses the boundaries of PVCs more frequently than that of a smaller data set. This operation incurs extra computational overhead for a larger data set. On the other hand, for Naive-PMNN, the processing time, I/O costs, and the communication costs remain almost constant with the increase of the data set size. This is because, unless the $R^{*}$-tree has a new level due to the increase of the data set size, the processing costs for Naive-PMNN do not vary with increase of the data set size, which is the case in Figures~\ref{fig:vds}(a)-(f)).

Figures also show that our P-PVD-PMNN outperforms
Naive-PMNN by an order of magnitude in processing time, 2 orders
of magnitude in I/Os and number of communications for all data
sets. The results also show that P-PVD-PMNN performs similar for
both directional (D) and random (R) query movement paths.

\begin{figure*}[htbp]
  \begin{center}
    \begin{tabular}{cccc}
        \hspace{-5mm}
      \resizebox{40mm}{!}{\includegraphics{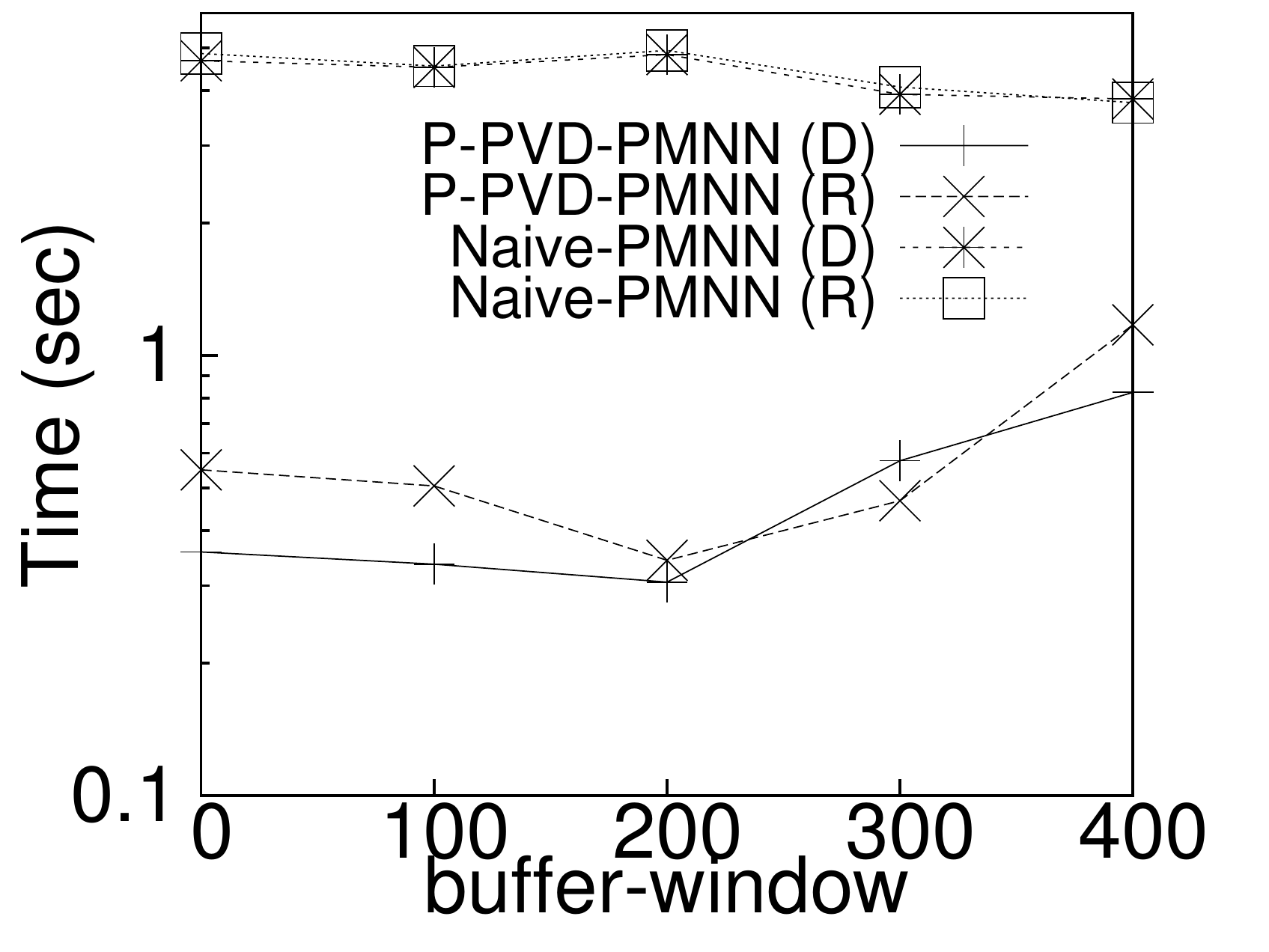}} &
        \hspace{-4mm}
        \resizebox{40mm}{!}{\includegraphics{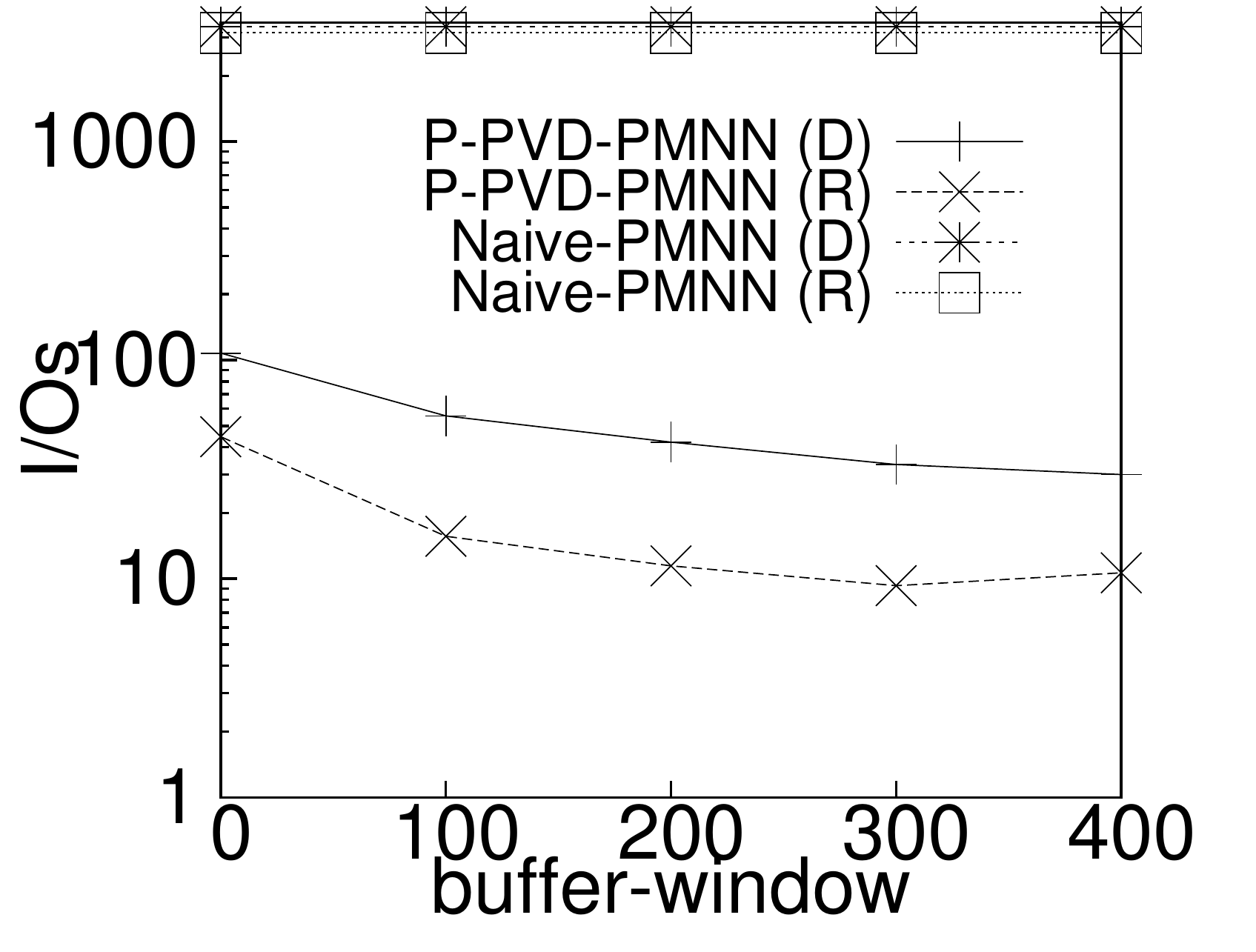}} &
         \hspace{-4mm}

        \resizebox{40mm}{!}{\includegraphics{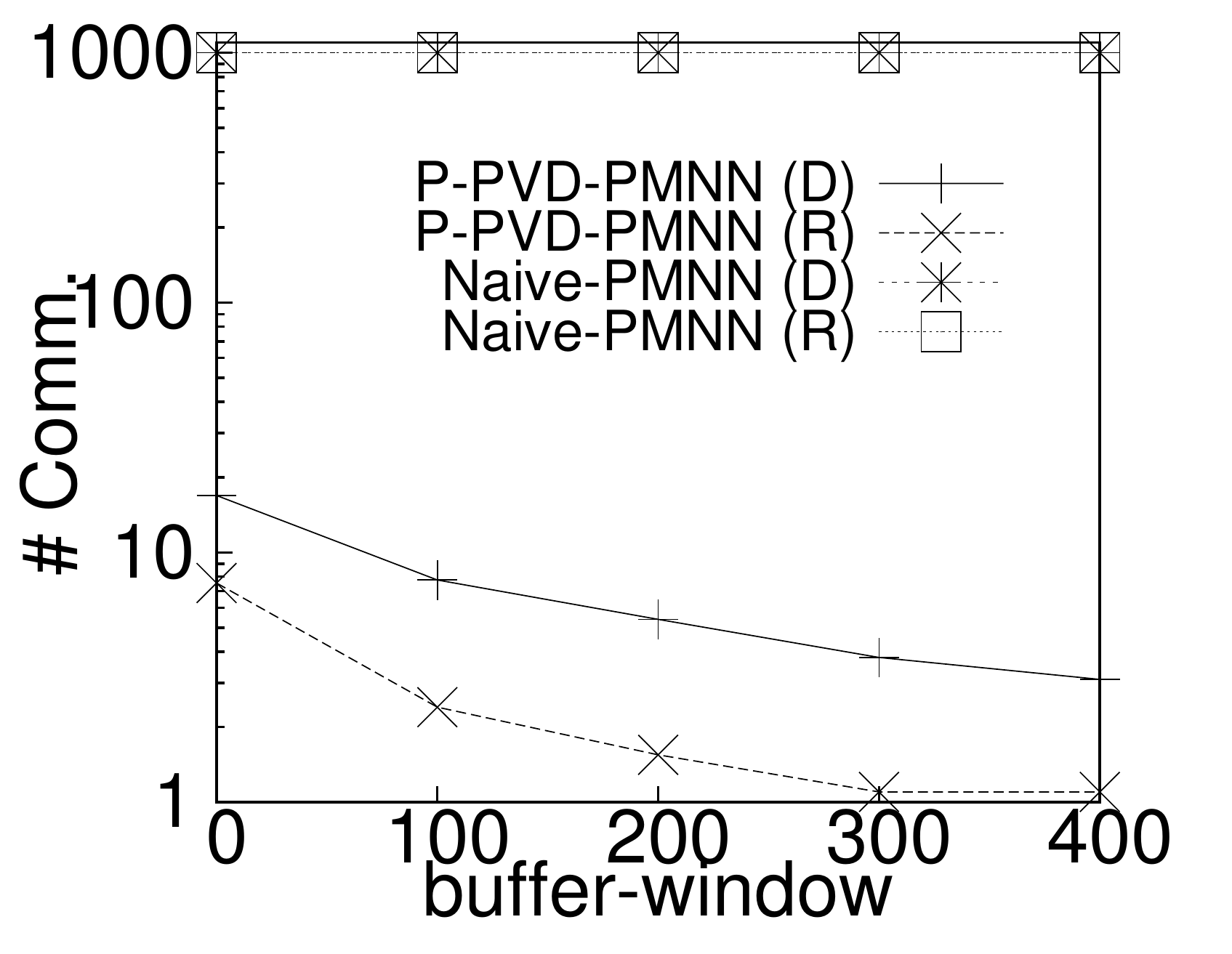}}\\
       \scriptsize{(a)\hspace{0mm}} & \scriptsize{(b)} & \scriptsize{(c)}\\
      \end{tabular}
      \begin{tabular}{cccc}
        \hspace{-5mm}
      \resizebox{40mm}{!}{\includegraphics{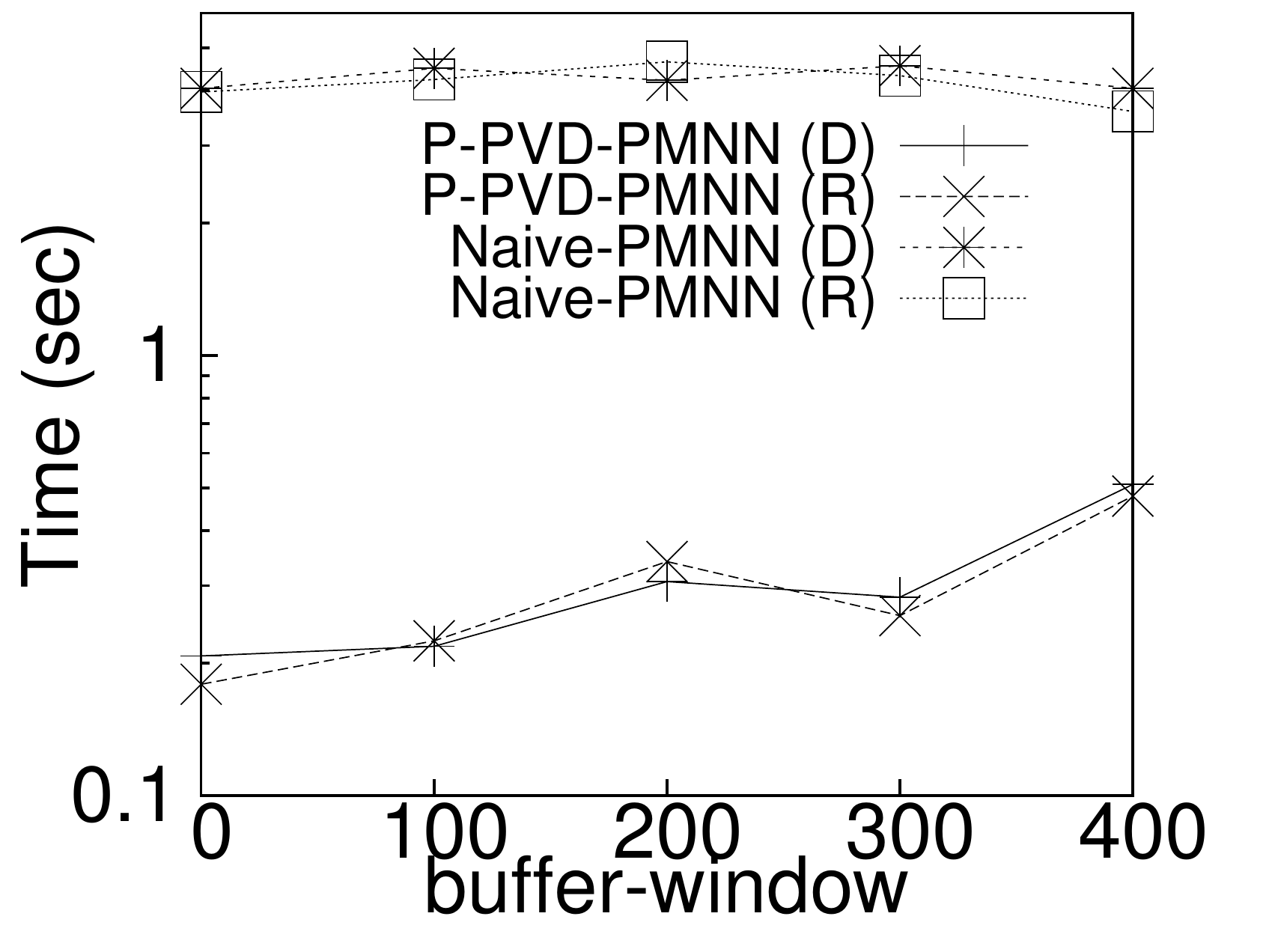}} &
        \hspace{-4mm}

        \resizebox{40mm}{!}{\includegraphics{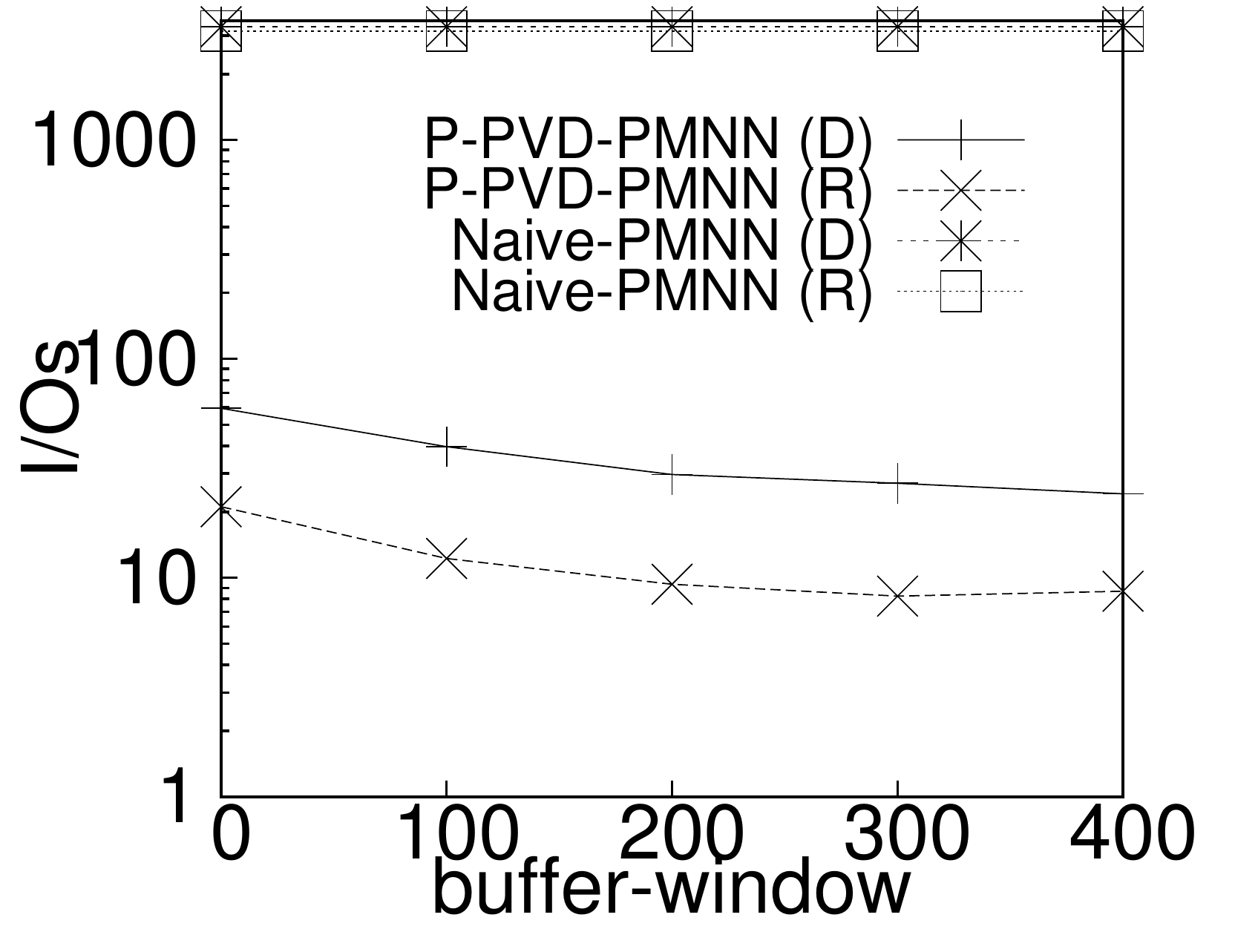}} &
        \hspace{-4mm}

        \resizebox{40mm}{!}{\includegraphics{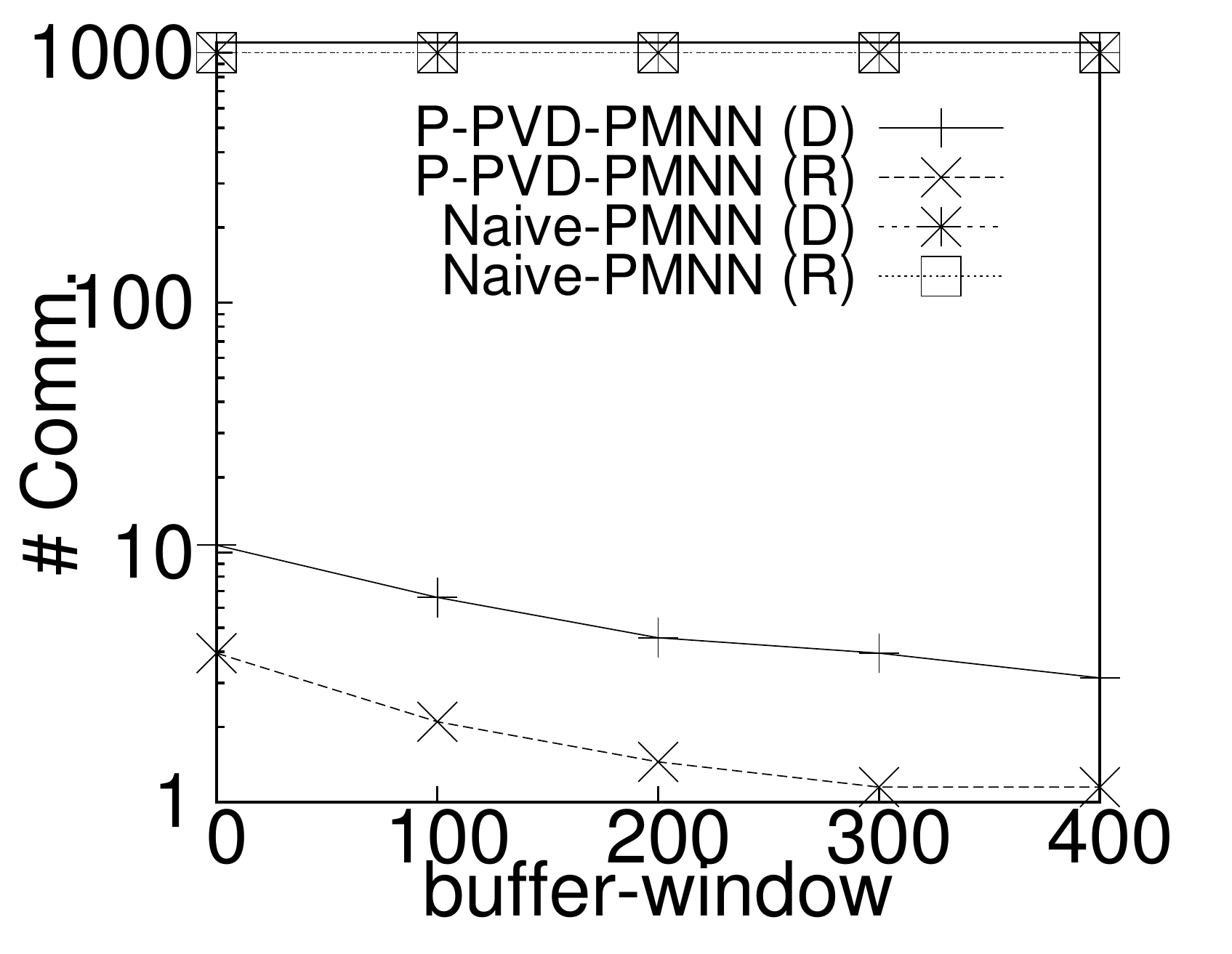}}\\
       \scriptsize{(d)\hspace{0mm}} & \scriptsize{(e)} & \scriptsize{(f)}\\
      \end{tabular}
      \begin{tabular}{cccc}
        \hspace{-5mm}
      \resizebox{40mm}{!}{\includegraphics{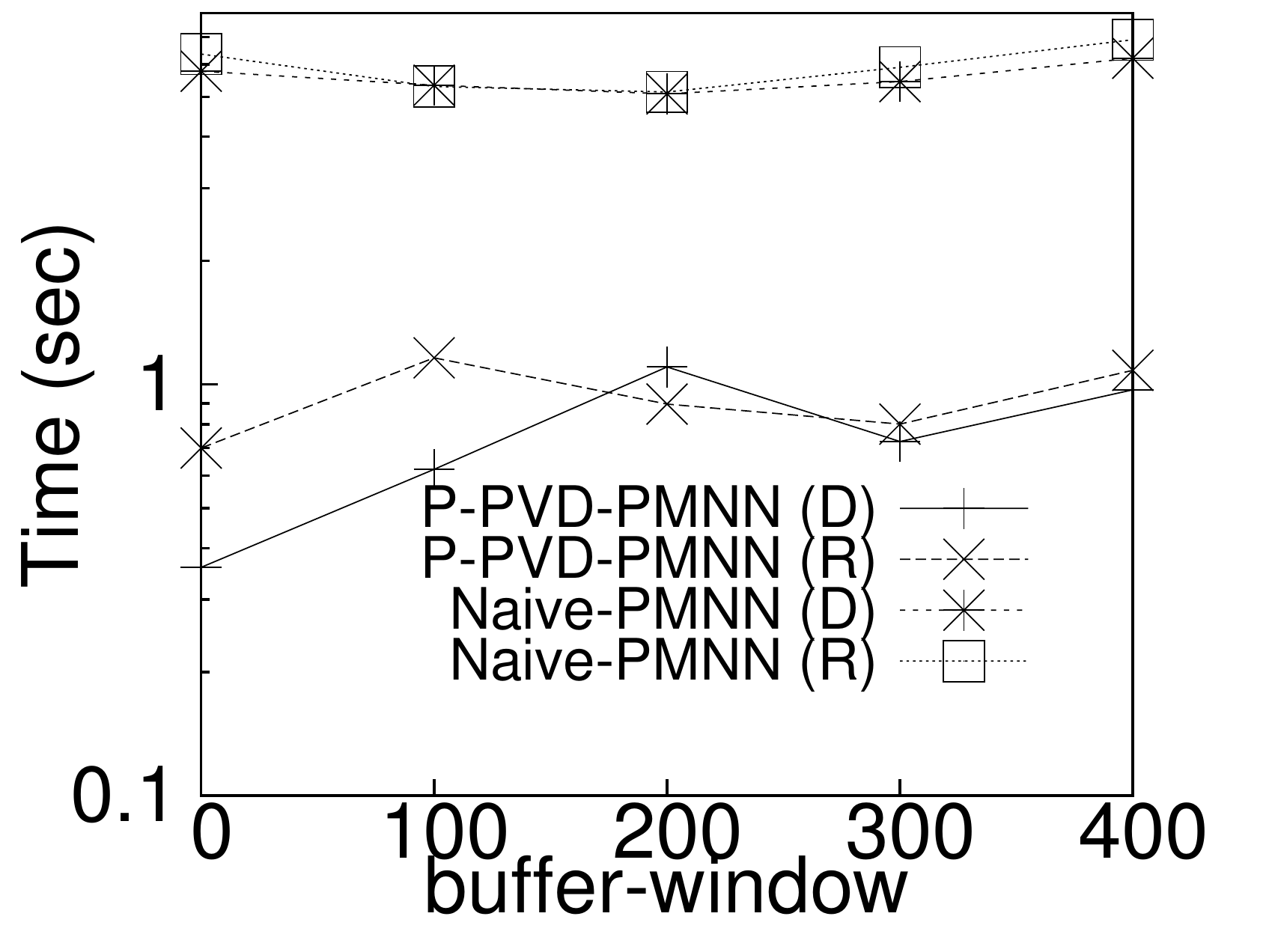}} &
        \hspace{-4mm}

        \resizebox{40mm}{!}{\includegraphics{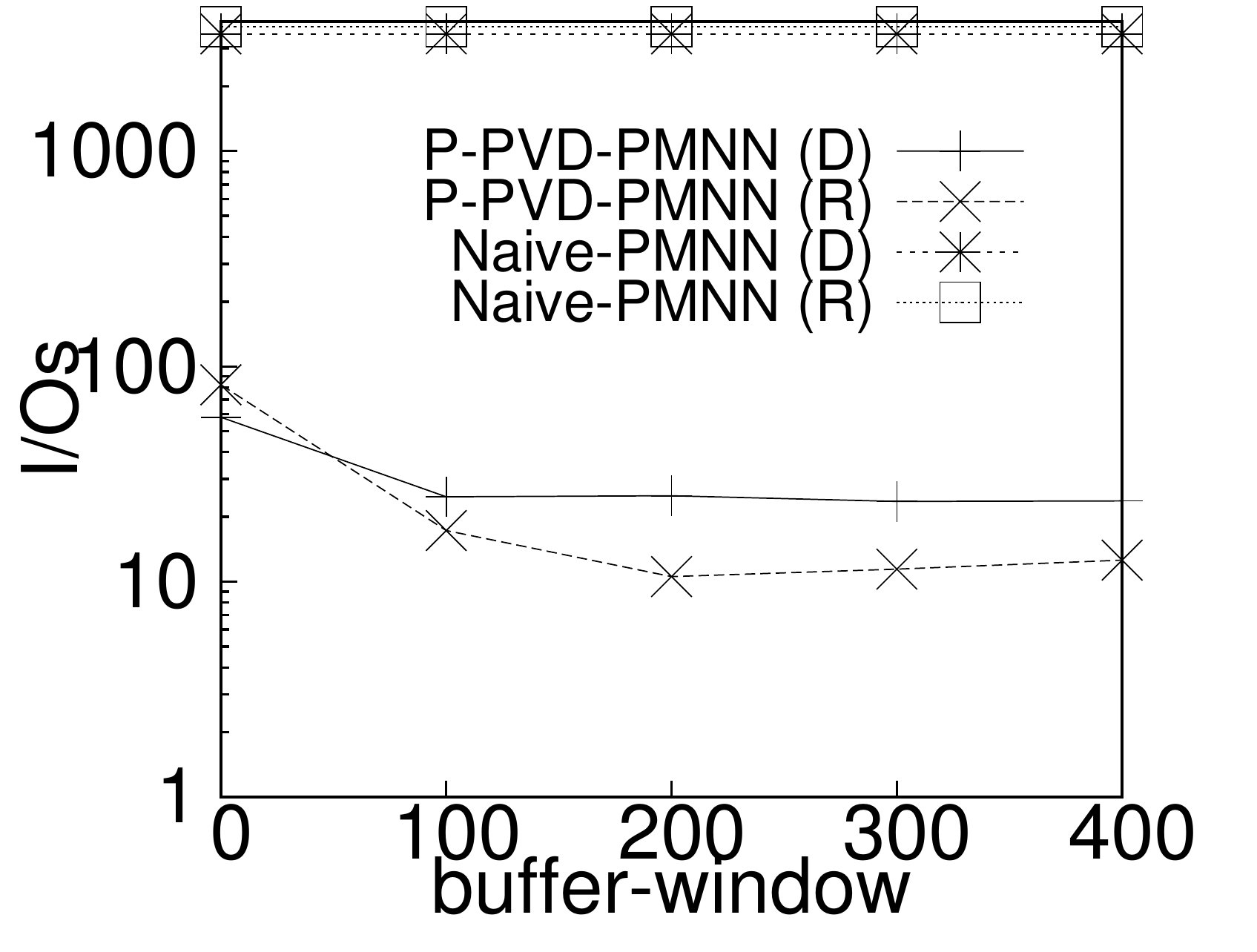}} &
         \hspace{-4mm}

        \resizebox{40mm}{!}{\includegraphics{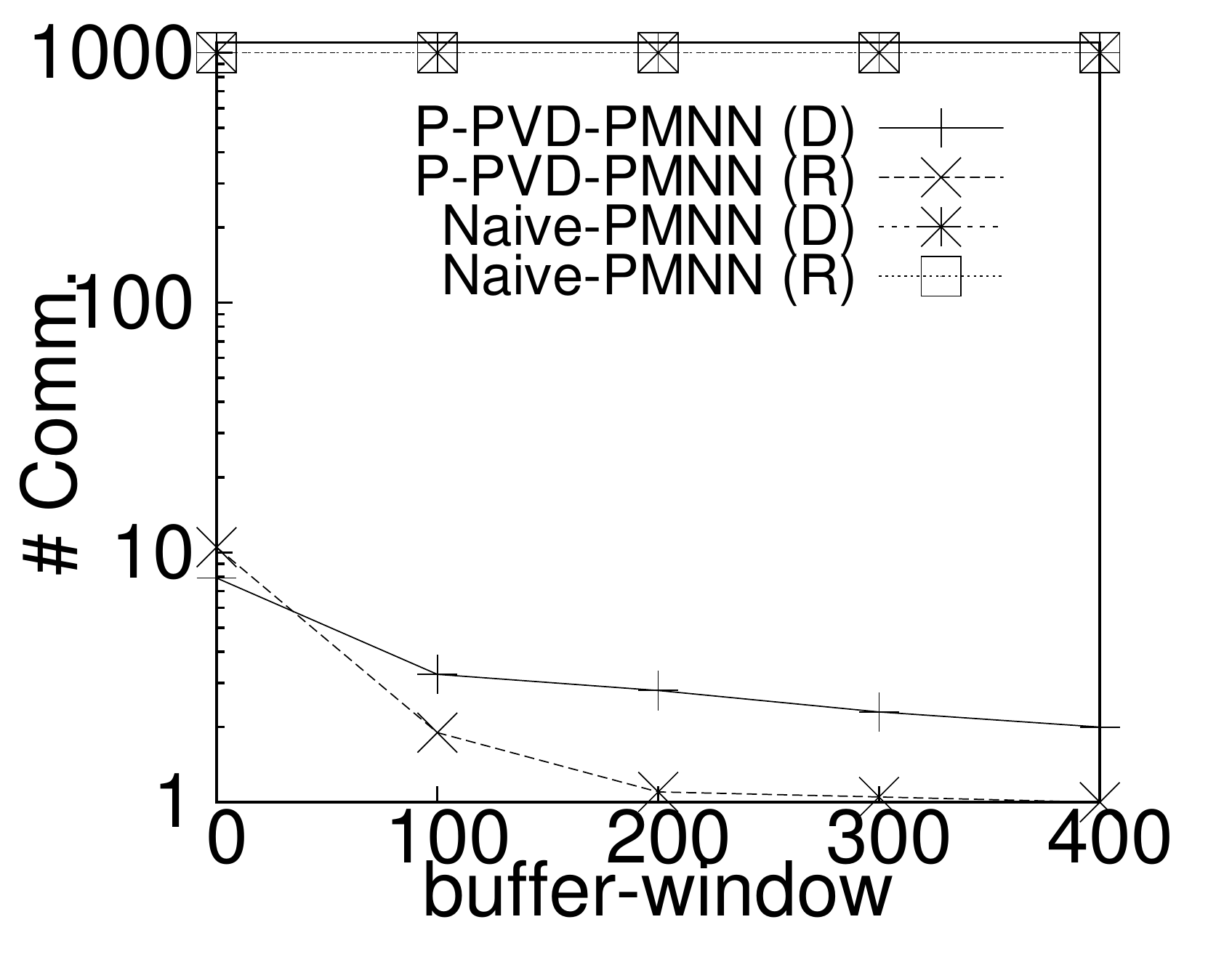}}\\
       \scriptsize{(g)\hspace{0mm}} & \scriptsize{(h)} & \scriptsize{(i)}\\
      \end{tabular}
    \caption{The effect of buffer window in U (a-c), Z (d-f), and L (g-i)}
    \label{fig:vr}
  \end{center}

\end{figure*}

\emph{Effect of Buffer Window}: In this set of experiments, we study the impact of introducing a buffer for processing a PMNN query. We vary the value of buffer window from 0 to 400 units of the data space,
and then run the experiments for data sets U(10K), Z(10K),
and L(12K). We set the trajectory length to 5000 units.

In these experiments, all PVCs whose MBRs intersect with a buffer
window centered at $q$ having the length and width of the buffer window are
retrieved from the $R^{*}$-tree and sent to the client. The client
stores these PVCs in its buffer. When buffer window is 0, the algorithm
only retrieves those PVCs whose MBRs contain the given query
point. On the other hand, when buffer window is 100, all PVCs whose MBRs
intersect with the buffer window centered at $q$ having the length and width of 100 units (i.e., the buffer window covers $100\times100$ square units in the data space) are retrieved. In this setting, we expect that the
I/O costs will be reduced for a larger value of buffer window, because the
server does not need to access the $R^{*}$-tree as long as these
buffered PVCs can serve the subsequent query points of a moving
query.

Figures~\ref{fig:vr}(a)-(c) show the processing time, the I/O
costs, and the number of communications, respectively, for varying
the size of the buffer window from 0 to 400 units for U data set. Figure~\ref{fig:vr}(a) shows
that for P-PVD-PMNN, in general the processing time increase with increase of buffer window. The reason is that for a very large buffer window, a large number of PVCs are buffered and the processing
time increases as the algorithm needs to check these PVCs for a
moving query.\eat{the processing time slightly reduces for
varying the buffer window from 0 to 200, and then increases for varying the buffer window
from 200 to 400. The reason is that for a very large buffer window
($>200$), a large number of PVCs are buffered and the processing
time increases as the algorithm needs to check these PVCs for a
moving query.} On the other hand, Figure~\ref{fig:vr}(b) shows
that for P-PVD-PMNN, I/O costs decrease with the increase of the buffer window.
This is because, for a larger value of buffer window the algorithm fetches
more PVCs at a time from the server, and thereby needs to access
the PVD using the $R^{*}$-tree reduced number of times. The figure
also shows that P-PVD-PMNN outperforms Naive-PMNN by an order of
magnitude in processing time and 2 orders of magnitude in I/O.
Figure~\ref{fig:vr}(c) shows that the number of communications for P-PVD-PMNN continuously
decreases with the increase of buffer window as the client fetches more
PVCs at a time from the server.  However, for Naive-PMNN, the client
communicates with the server for each sampled location of the
query, and thus the number of communications remain constant.

The results on Z (see Figures~\ref{fig:vr}(d)-(f)) and L (see
Figures~\ref{fig:vr}(g)-(i)) data sets show similar trends with U
data set described above.

\begin{figure*}[htbp]
  \begin{center}
    \begin{tabular}{cccc}
        \hspace{-5mm}
      \resizebox{40mm}{!}{\includegraphics{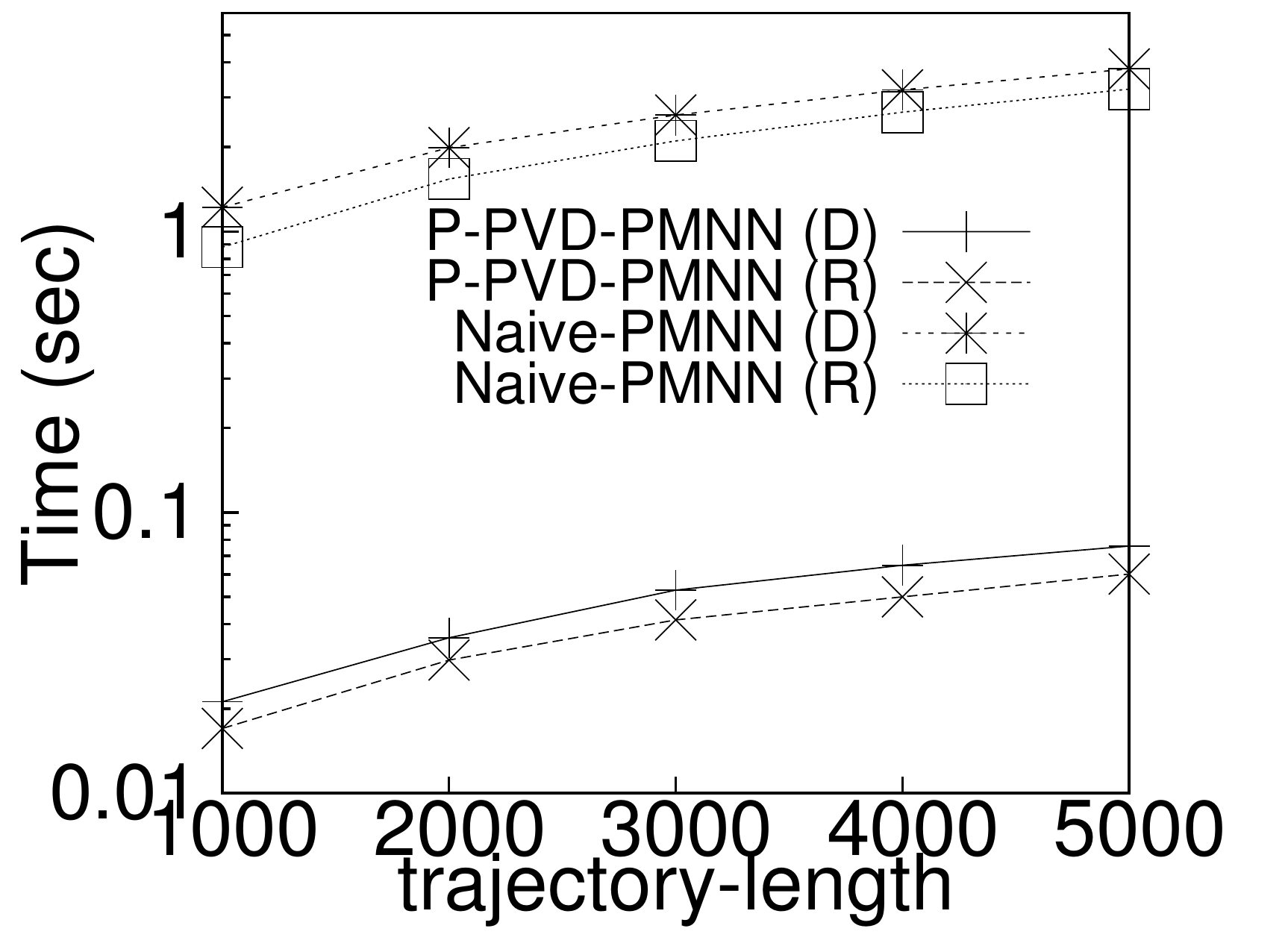}} &
        \hspace{-4mm}

        \resizebox{40mm}{!}{\includegraphics{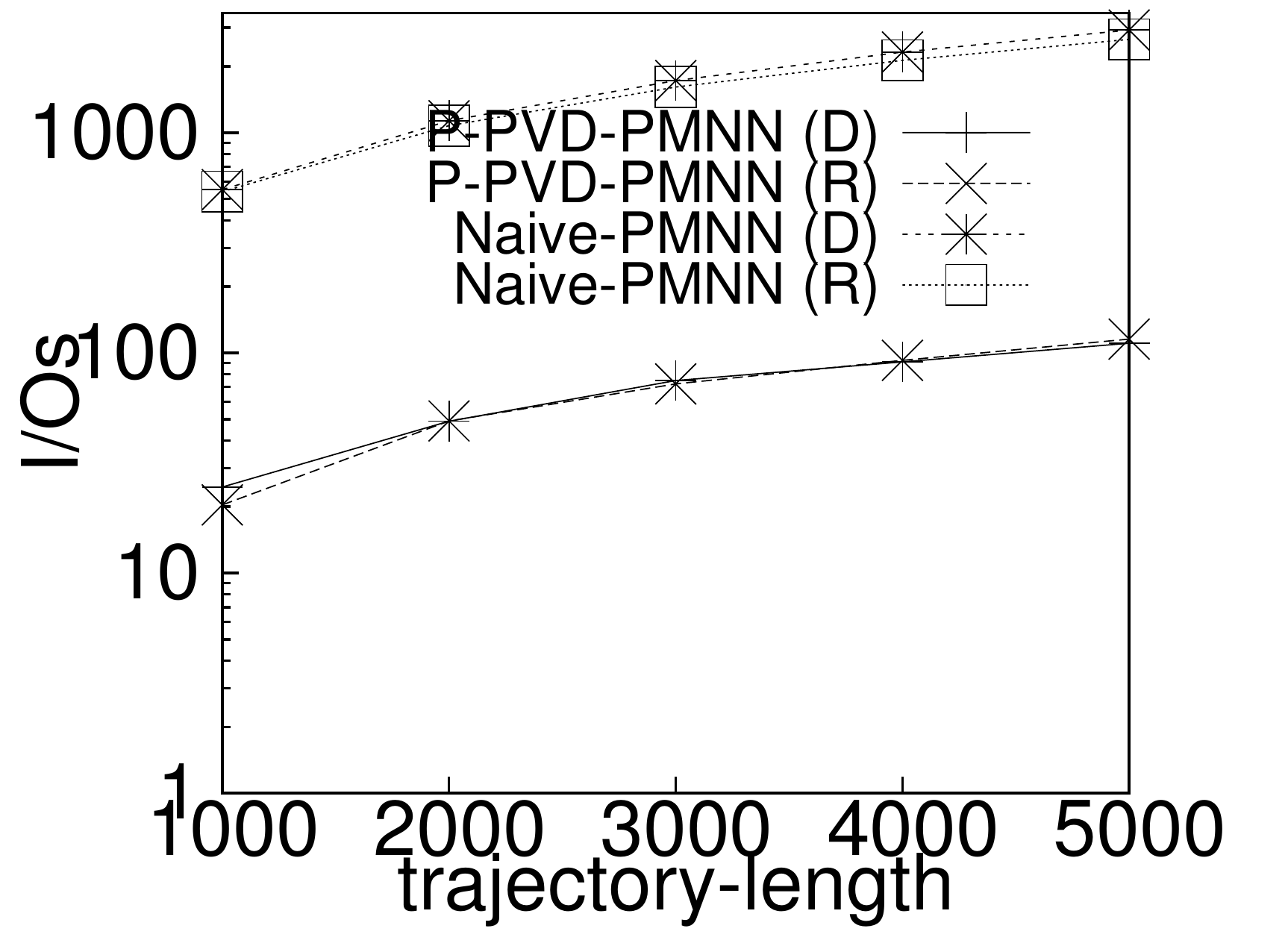}} &
         \hspace{-4mm}

        \resizebox{40mm}{!}{\includegraphics{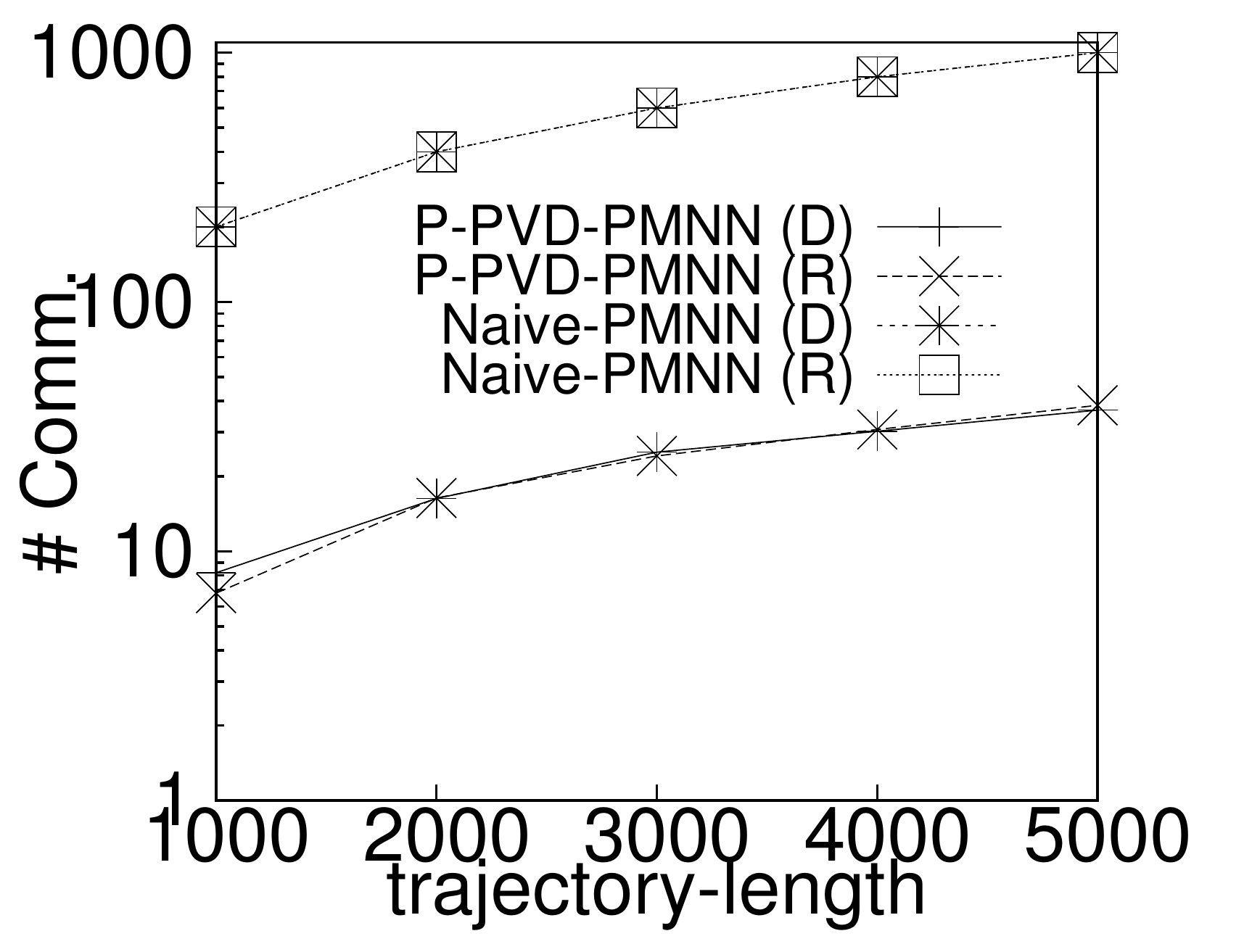}}\\
       \scriptsize{(a)\hspace{0mm}} & \scriptsize{(b)} & \scriptsize{(c)}\\
      \end{tabular}
    \begin{tabular}{cccc}
        \hspace{-5mm}
      \resizebox{40mm}{!}{\includegraphics{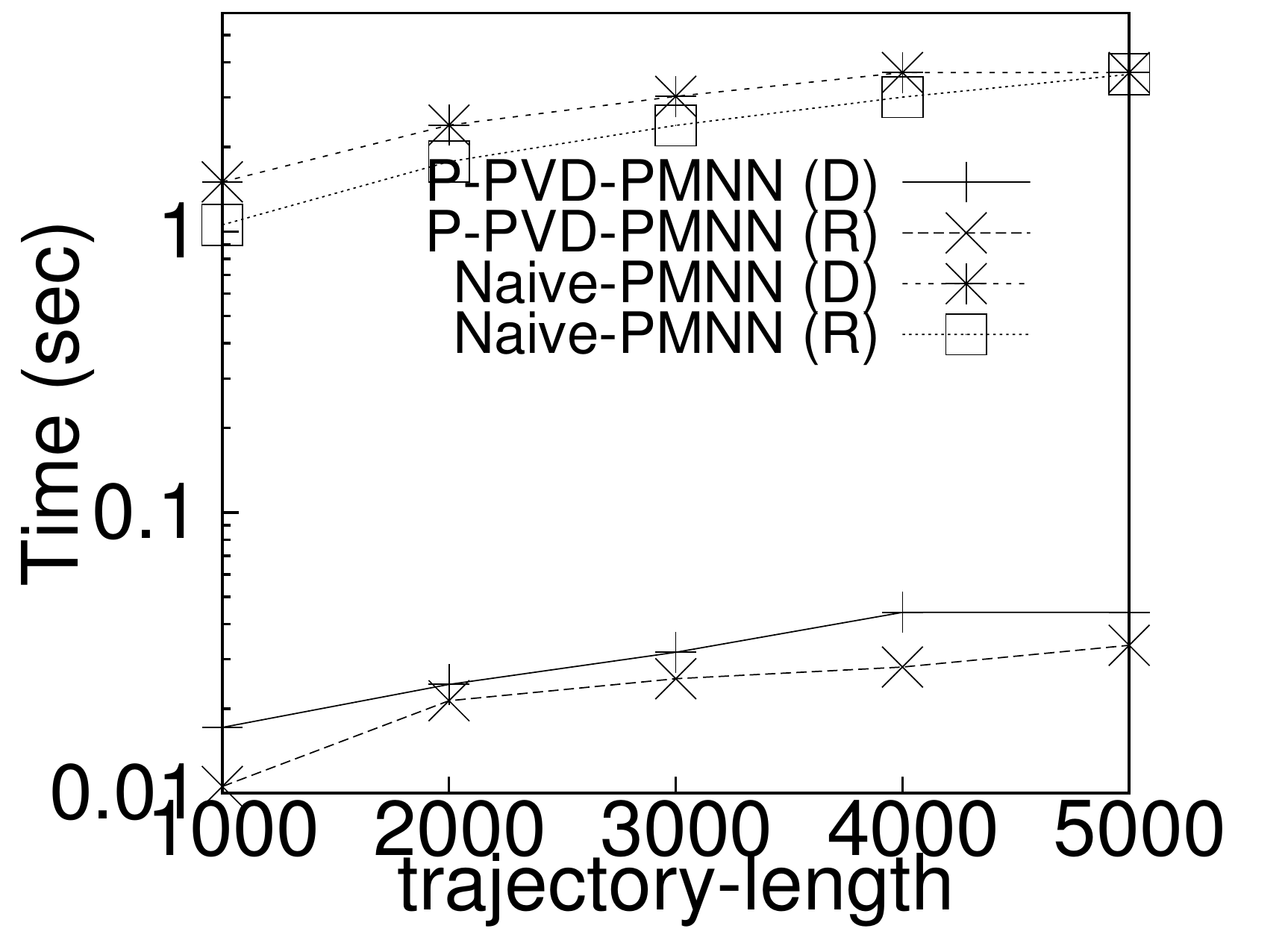}} &
        \hspace{-4mm}

        \resizebox{40mm}{!}{\includegraphics{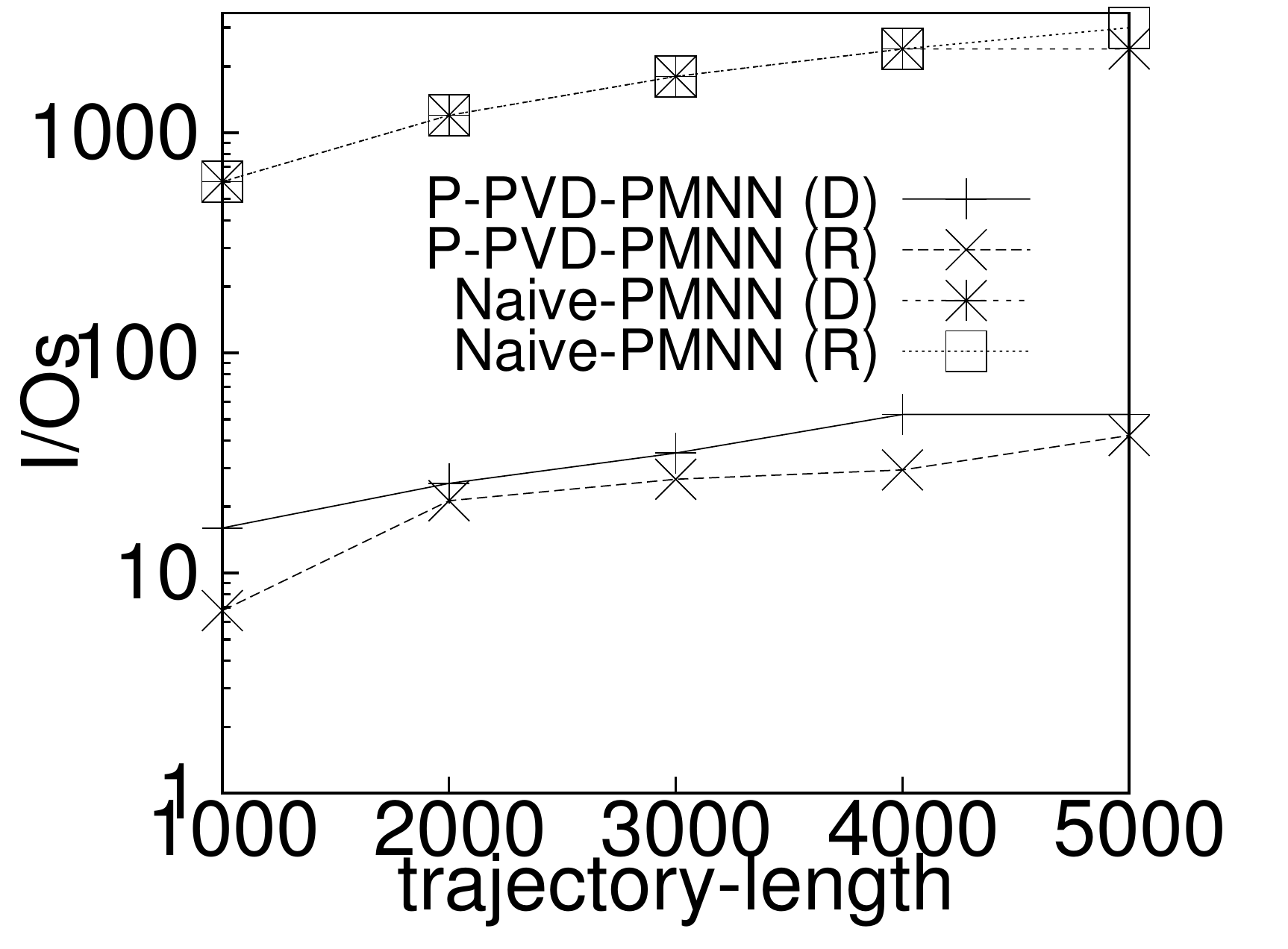}} &

        \hspace{-4mm}

        \resizebox{40mm}{!}{\includegraphics{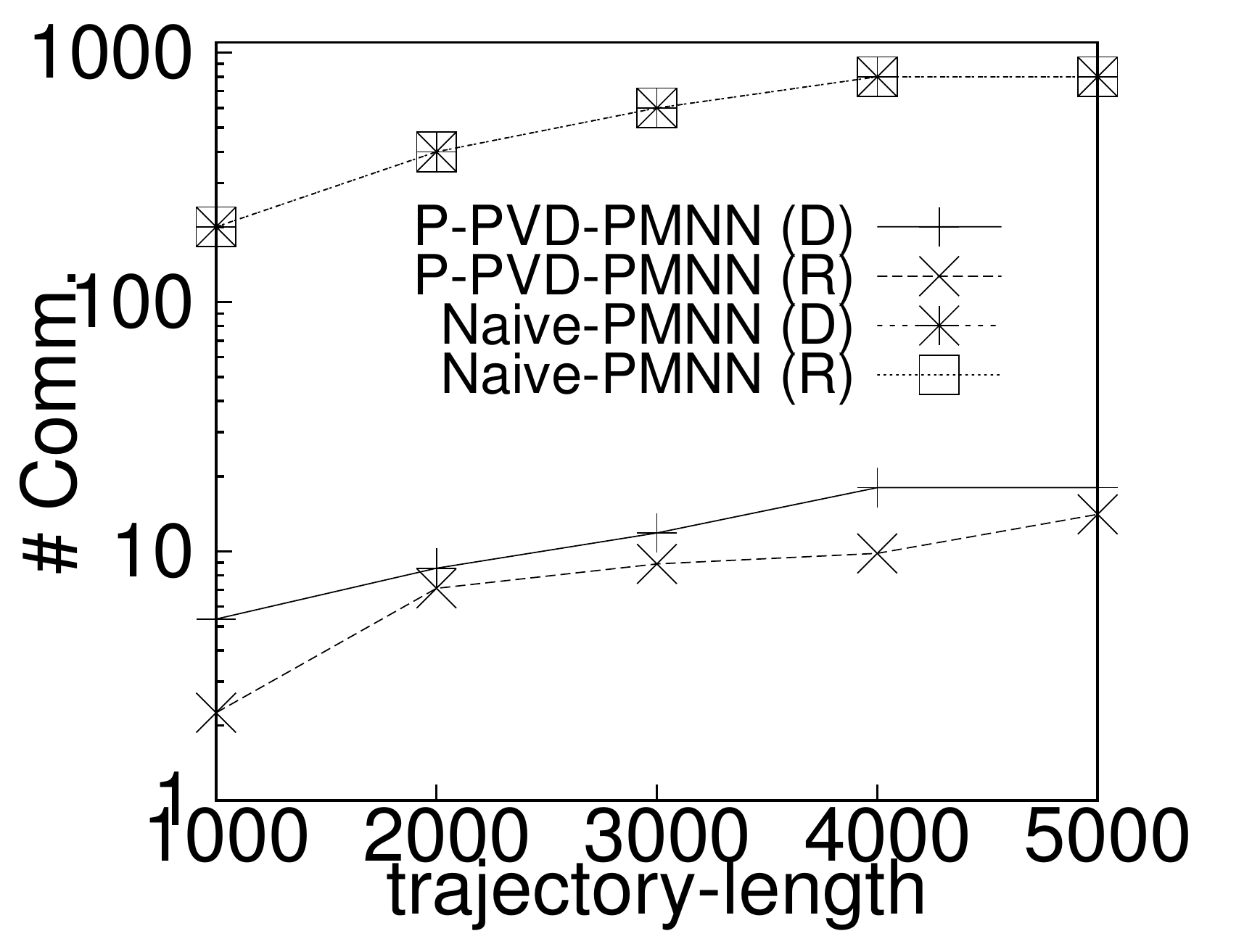}}\\
       \scriptsize{(d)\hspace{0mm}} & \scriptsize{(e)} & \scriptsize{(f)}\\
      \end{tabular}
    \caption{The effect of the query trajectory length in U (a-c), Z (d-f) for 1D data}
    \label{fig:vq1D}
  \end{center}
\end{figure*}

\begin{figure*}[htbp]
  \begin{center}
    \begin{tabular}{cccc}
        \hspace{-5mm}
      \resizebox{40mm}{!}{\includegraphics{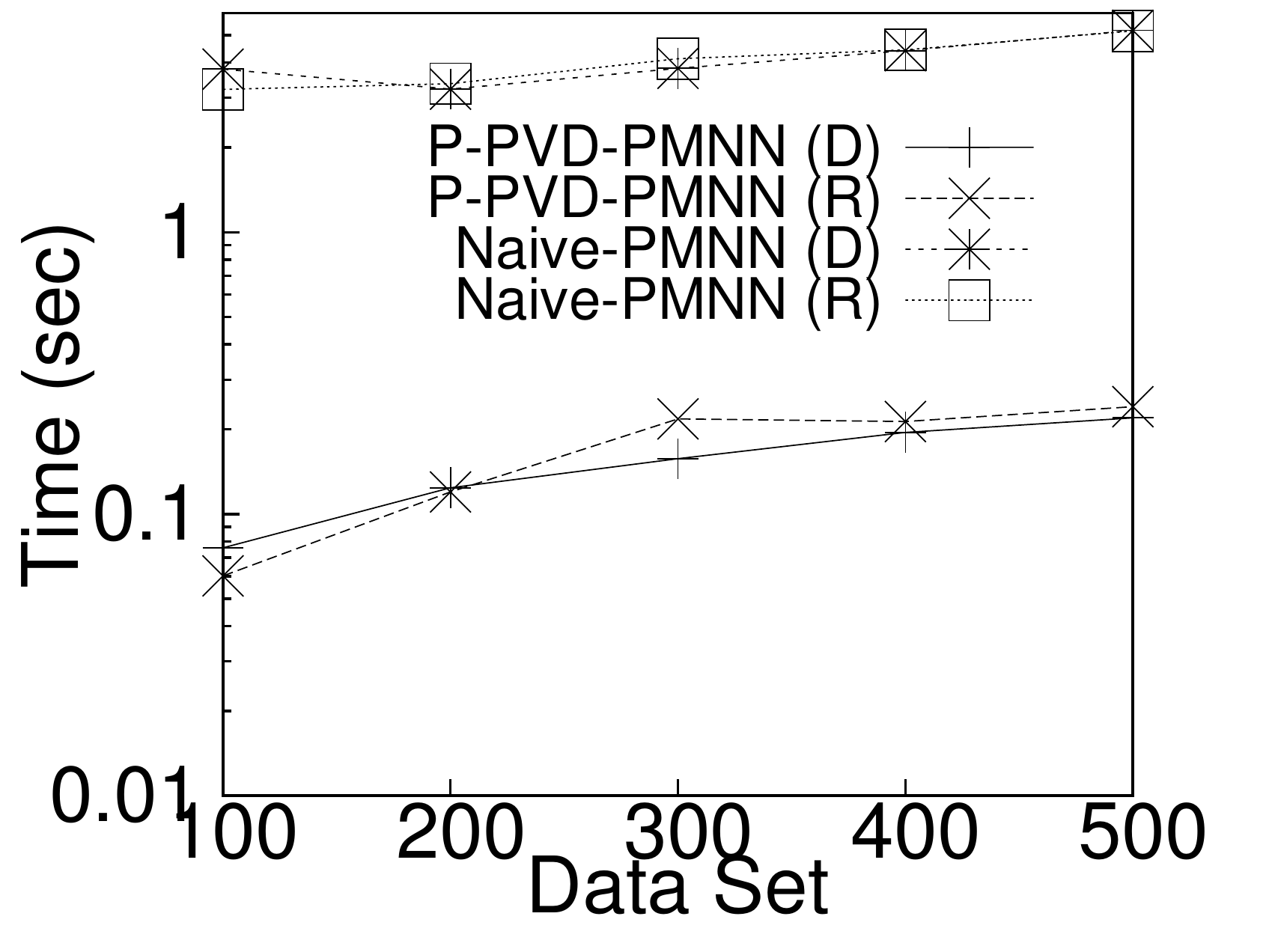}} &
        \hspace{-4mm}

        \resizebox{40mm}{!}{\includegraphics{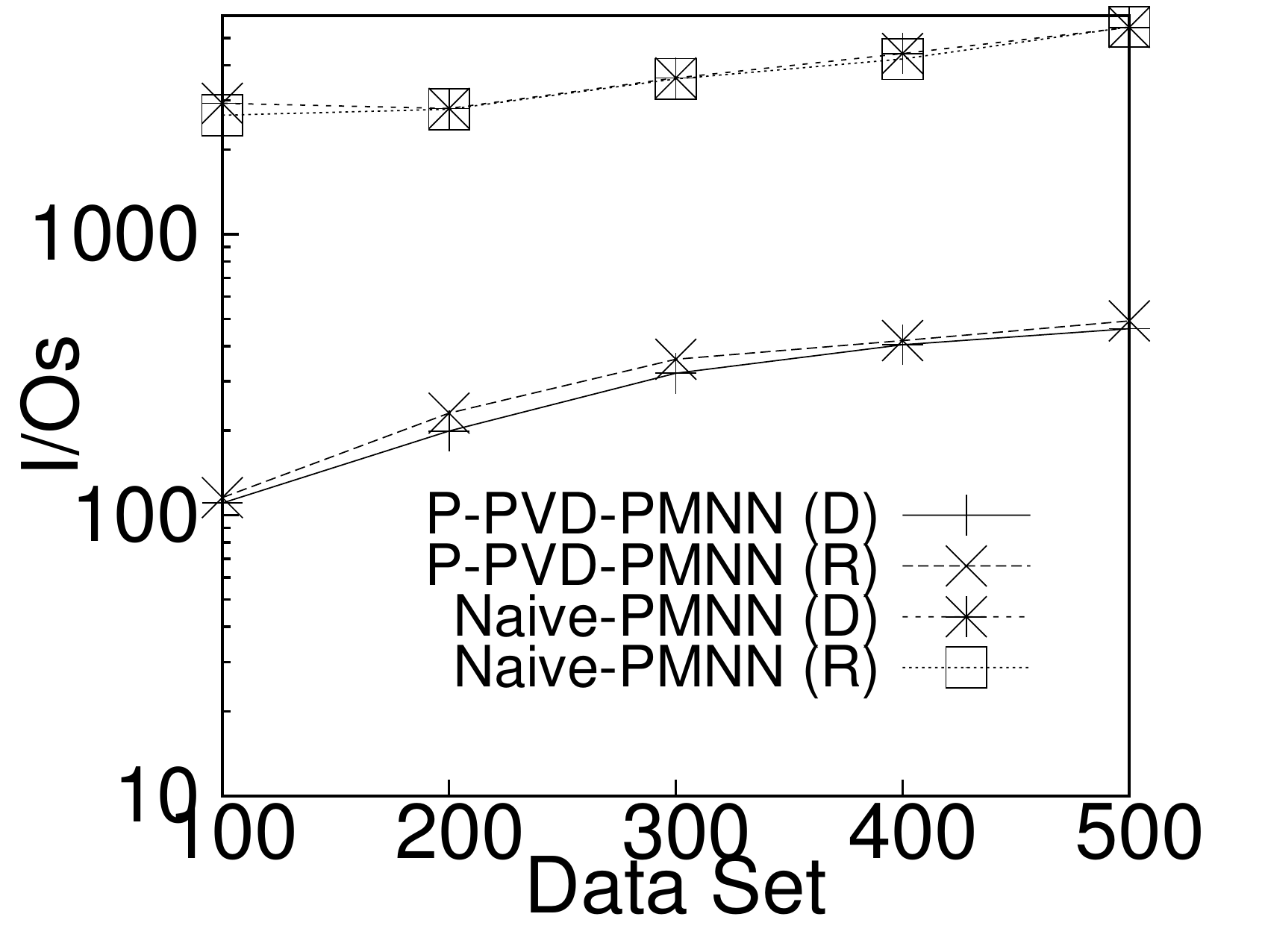}} &
         \hspace{-4mm}

        \resizebox{40mm}{!}{\includegraphics{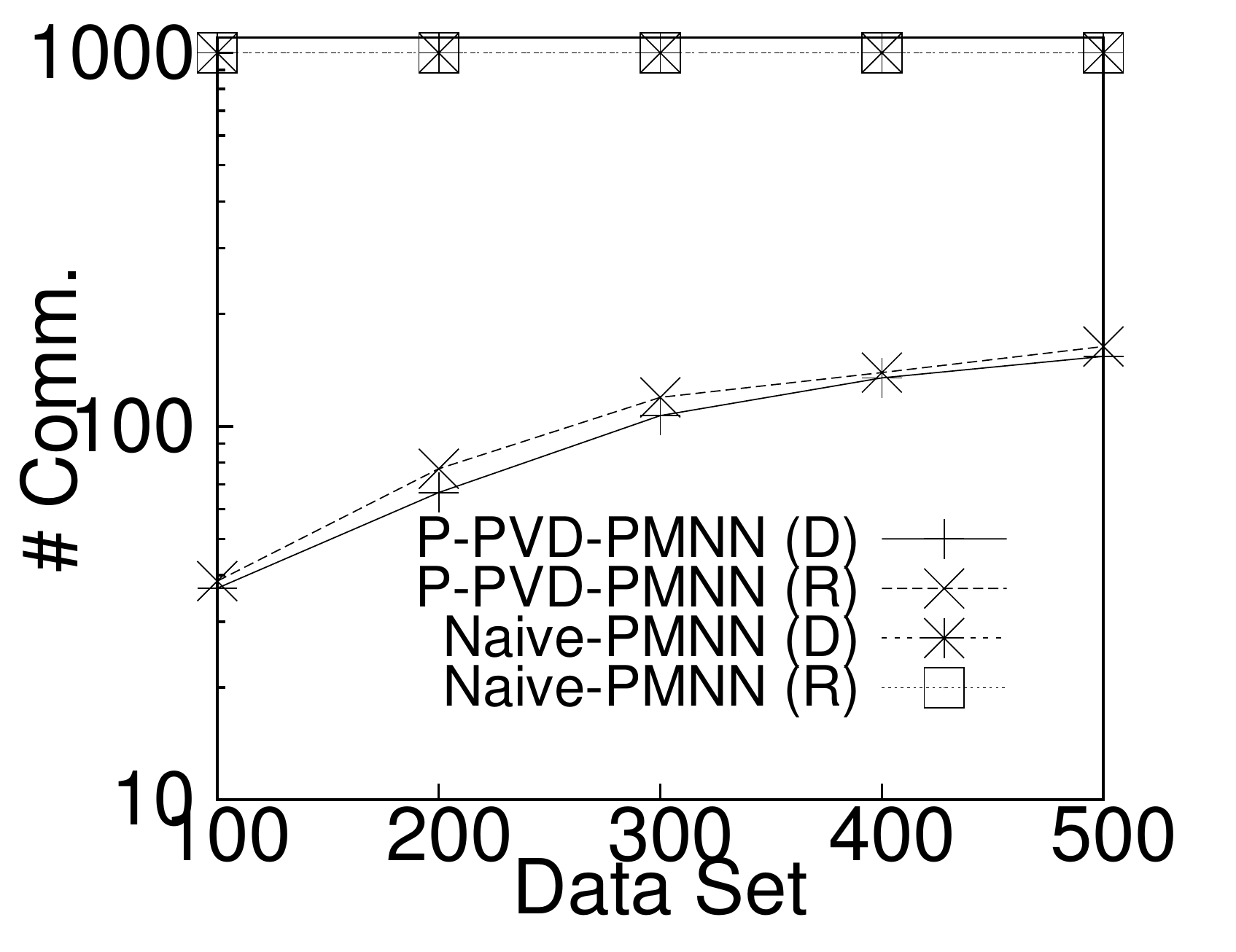}}\\
       \scriptsize{(a)\hspace{0mm}} & \scriptsize{(b)} & \scriptsize{(c)}\\
      \end{tabular}
      \begin{tabular}{cccc}
        \hspace{-5mm}
      \resizebox{40mm}{!}{\includegraphics{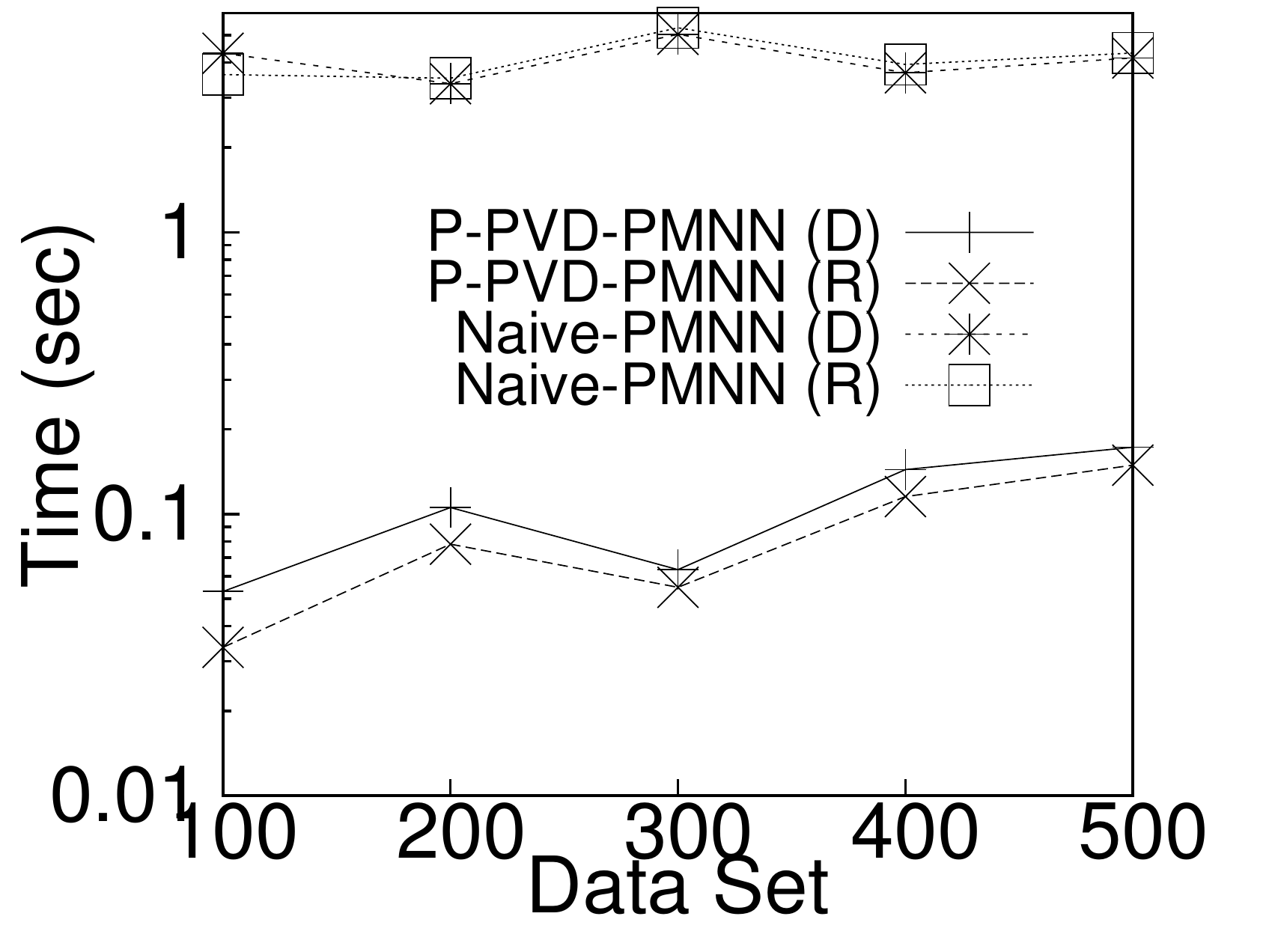}} &
        \hspace{-4mm}

        \resizebox{40mm}{!}{\includegraphics{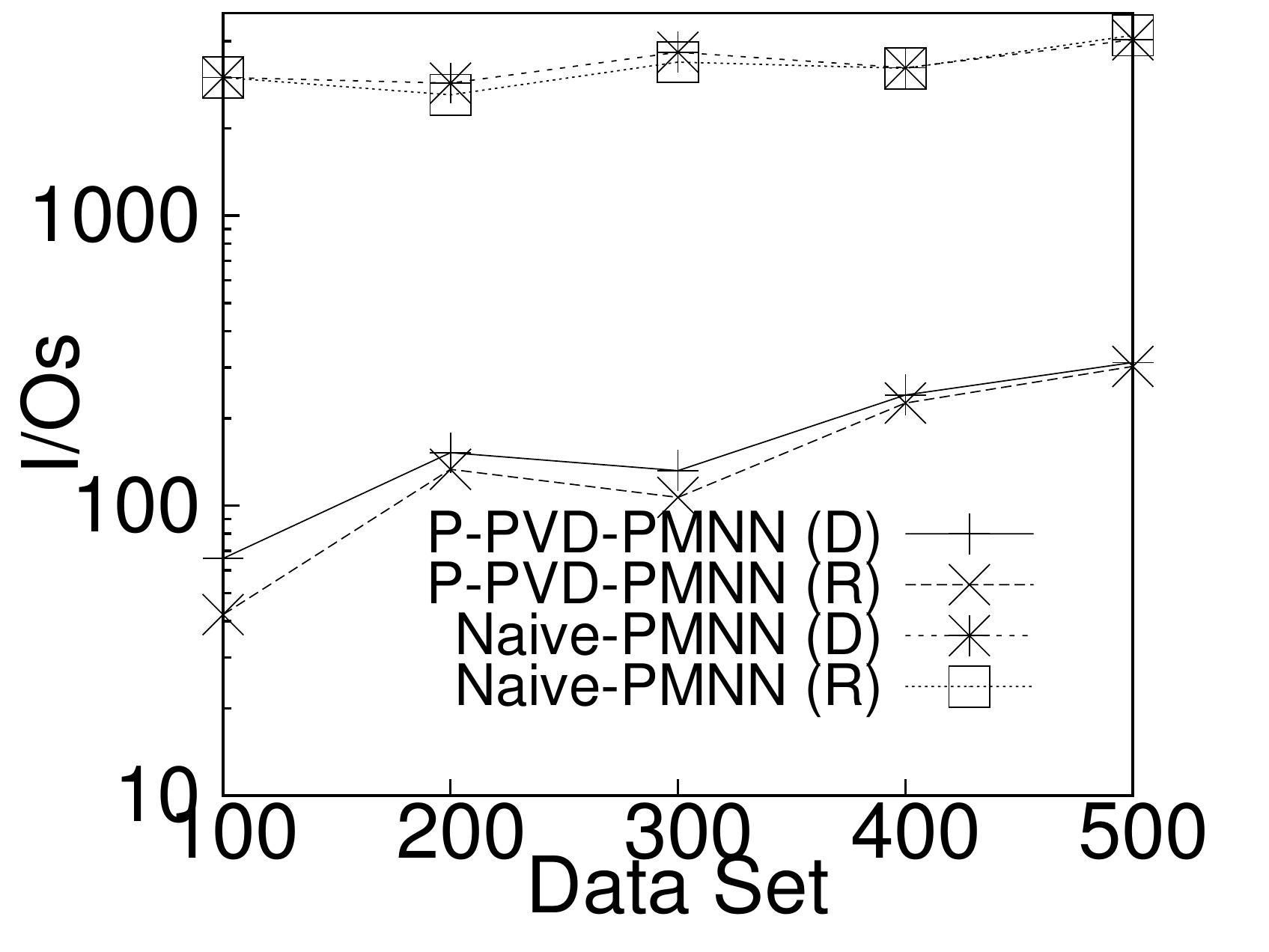}} &
        \hspace{-4mm}

        \resizebox{40mm}{!}{\includegraphics{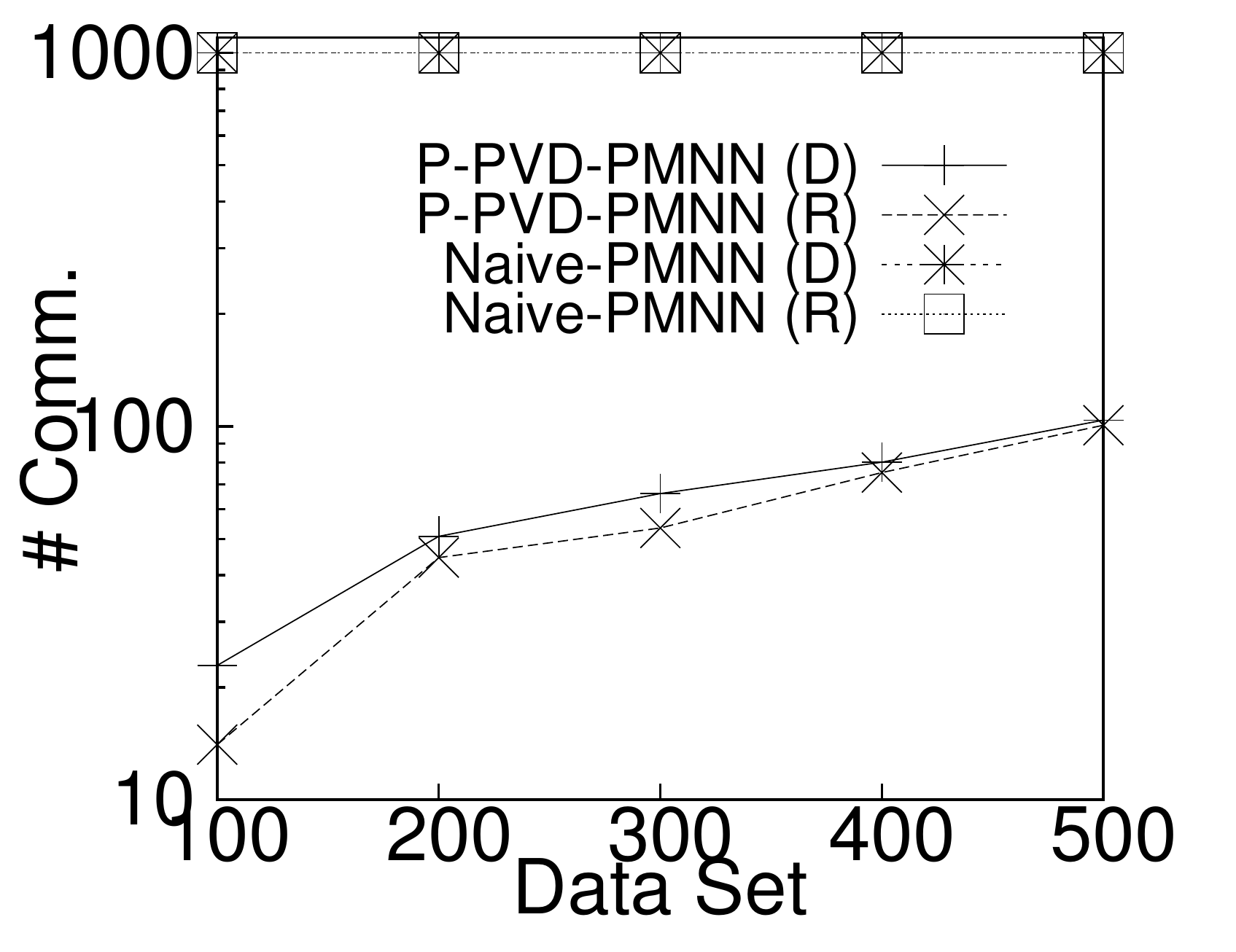}}\\
       \scriptsize{(d)\hspace{0mm}} & \scriptsize{(e)} & \scriptsize{(f)}\\
      \end{tabular}
    \caption{The effect of the data set size in U (a-c), Z (d-f) for 1D data}
    \label{fig:vds1D}
  \end{center}
\end{figure*}

\begin{figure*}[htbp]
  \begin{center}
    \begin{tabular}{cccc}
        \hspace{-5mm}
      \resizebox{40mm}{!}{\includegraphics{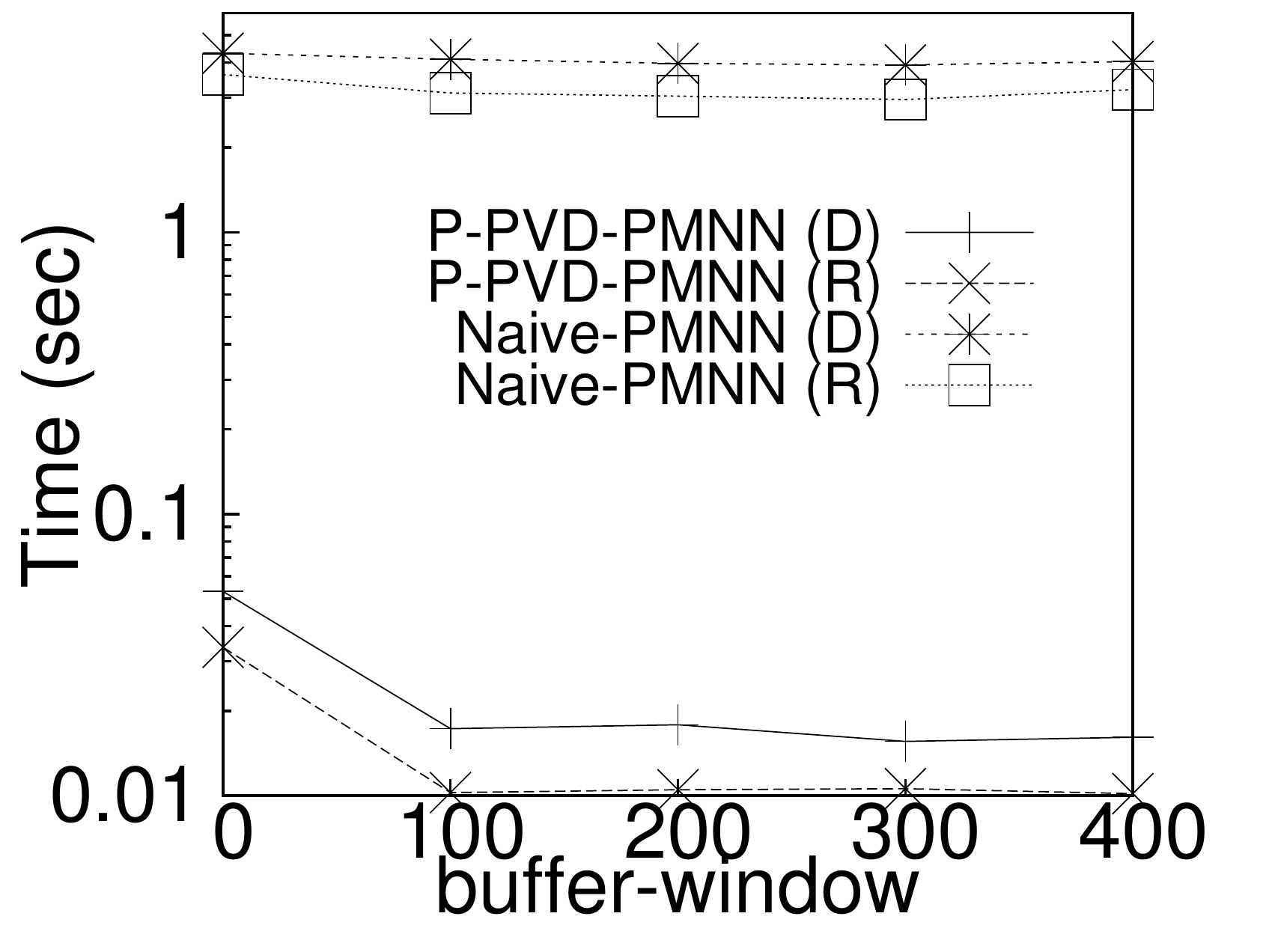}} &
        \hspace{-4mm}
        \resizebox{40mm}{!}{\includegraphics{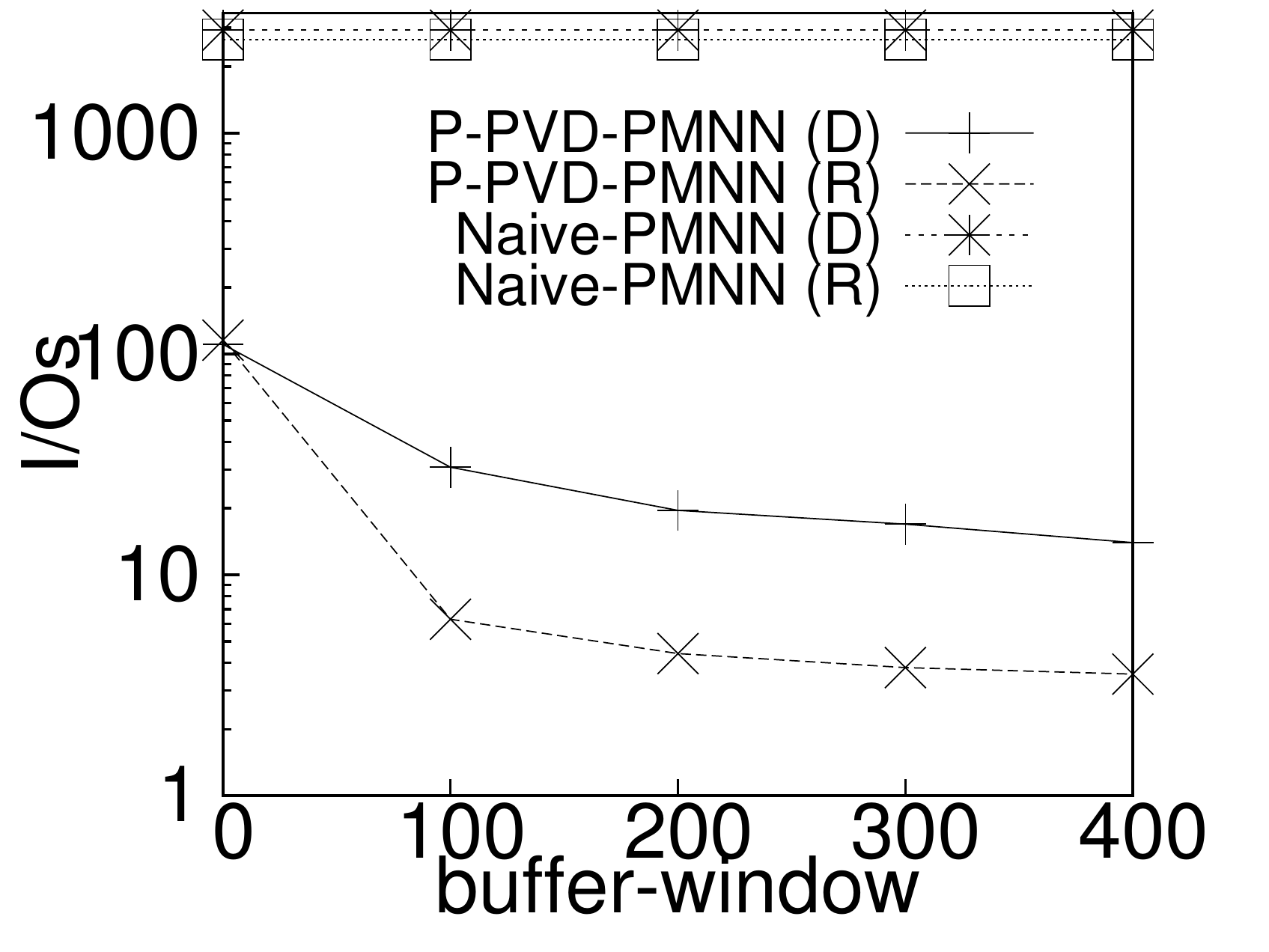}} &
         \hspace{-4mm}
        \resizebox{40mm}{!}{\includegraphics{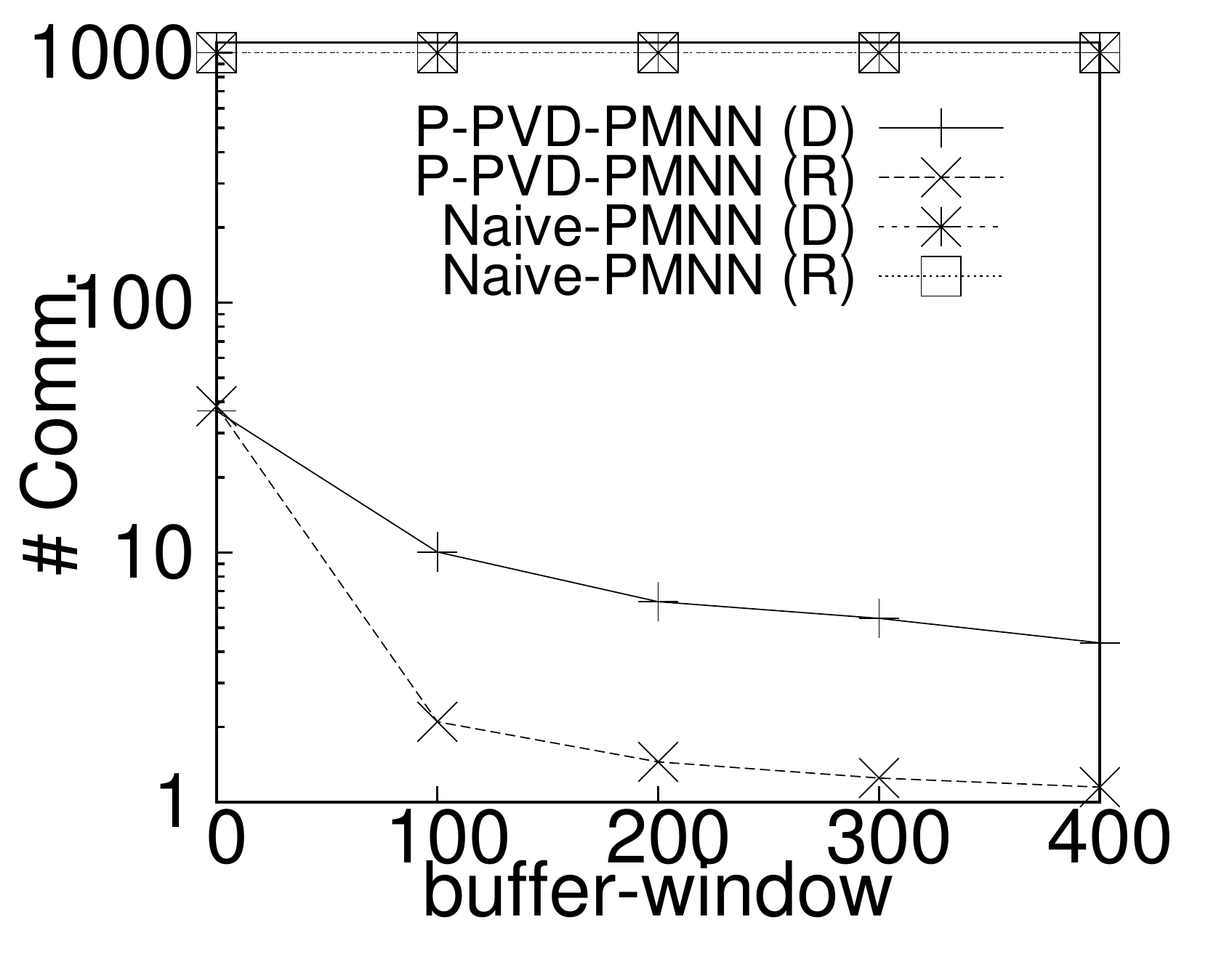}}\\
       \scriptsize{(a)\hspace{0mm}} & \scriptsize{(b)} & \scriptsize{(c)}\\
      \end{tabular}
      \begin{tabular}{cccc}
        \hspace{-5mm}
      \resizebox{40mm}{!}{\includegraphics{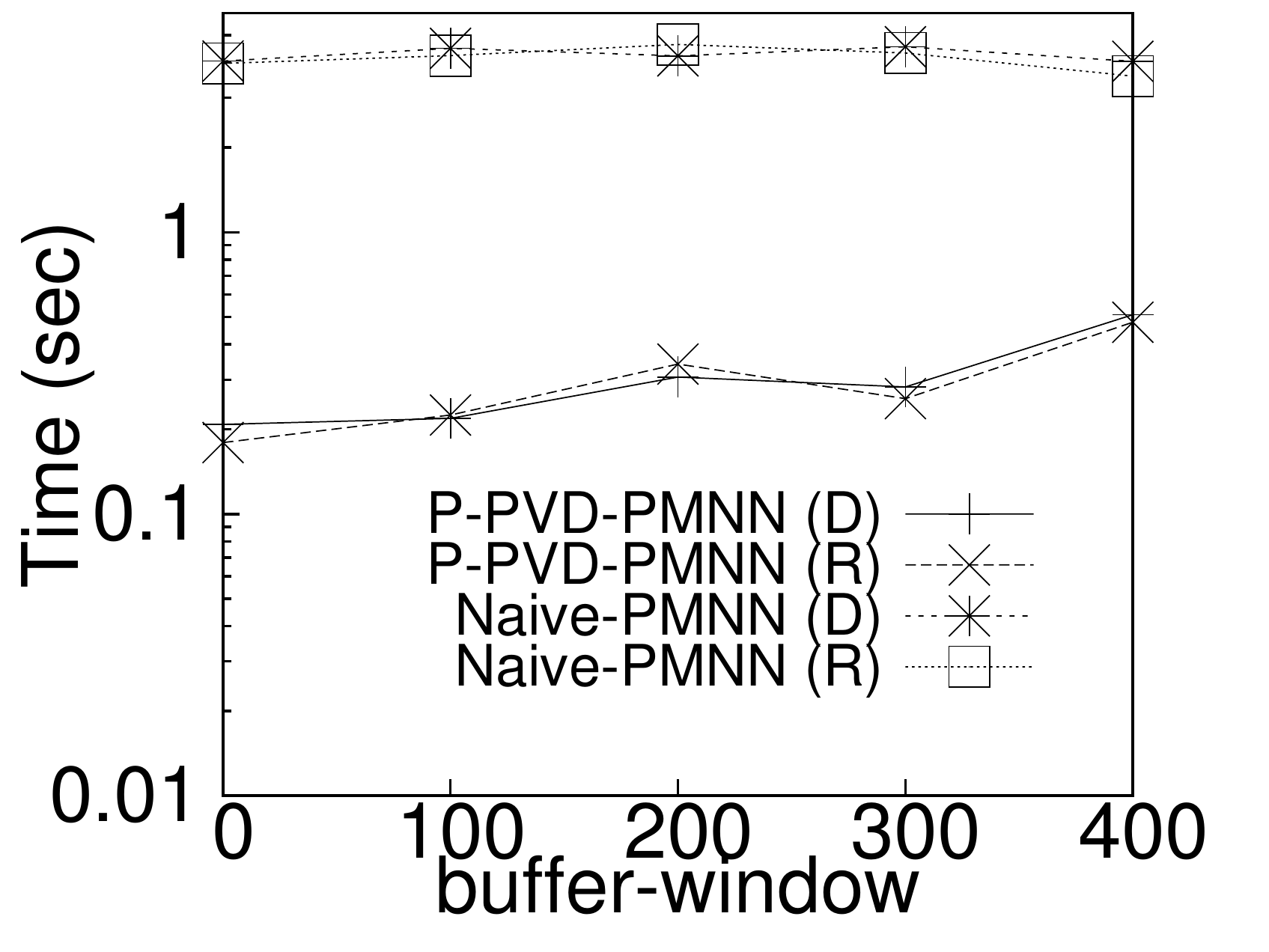}} &
        \hspace{-4mm}
        \resizebox{40mm}{!}{\includegraphics{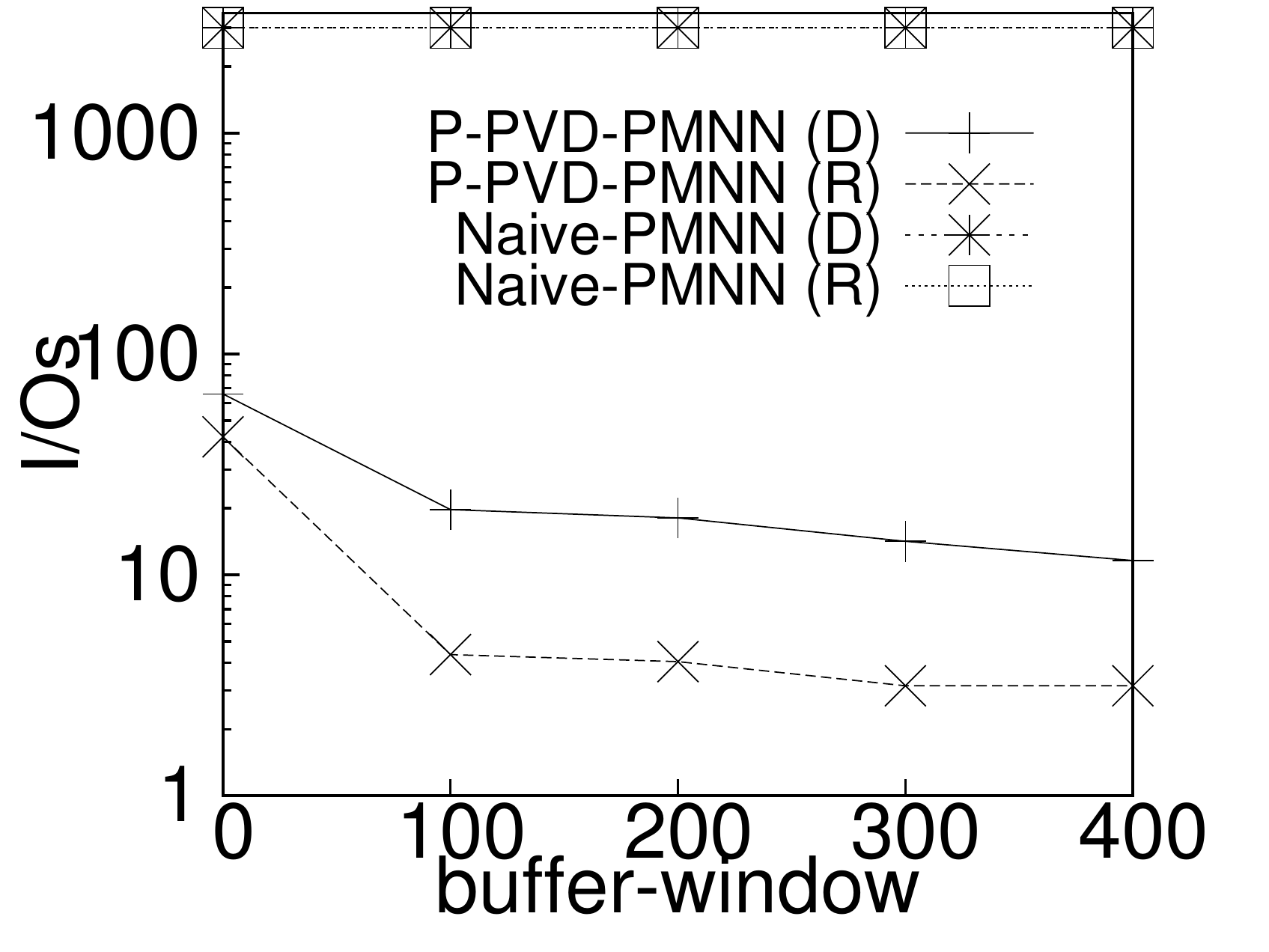}} &
          \hspace{-4mm}
        \resizebox{40mm}{!}{\includegraphics{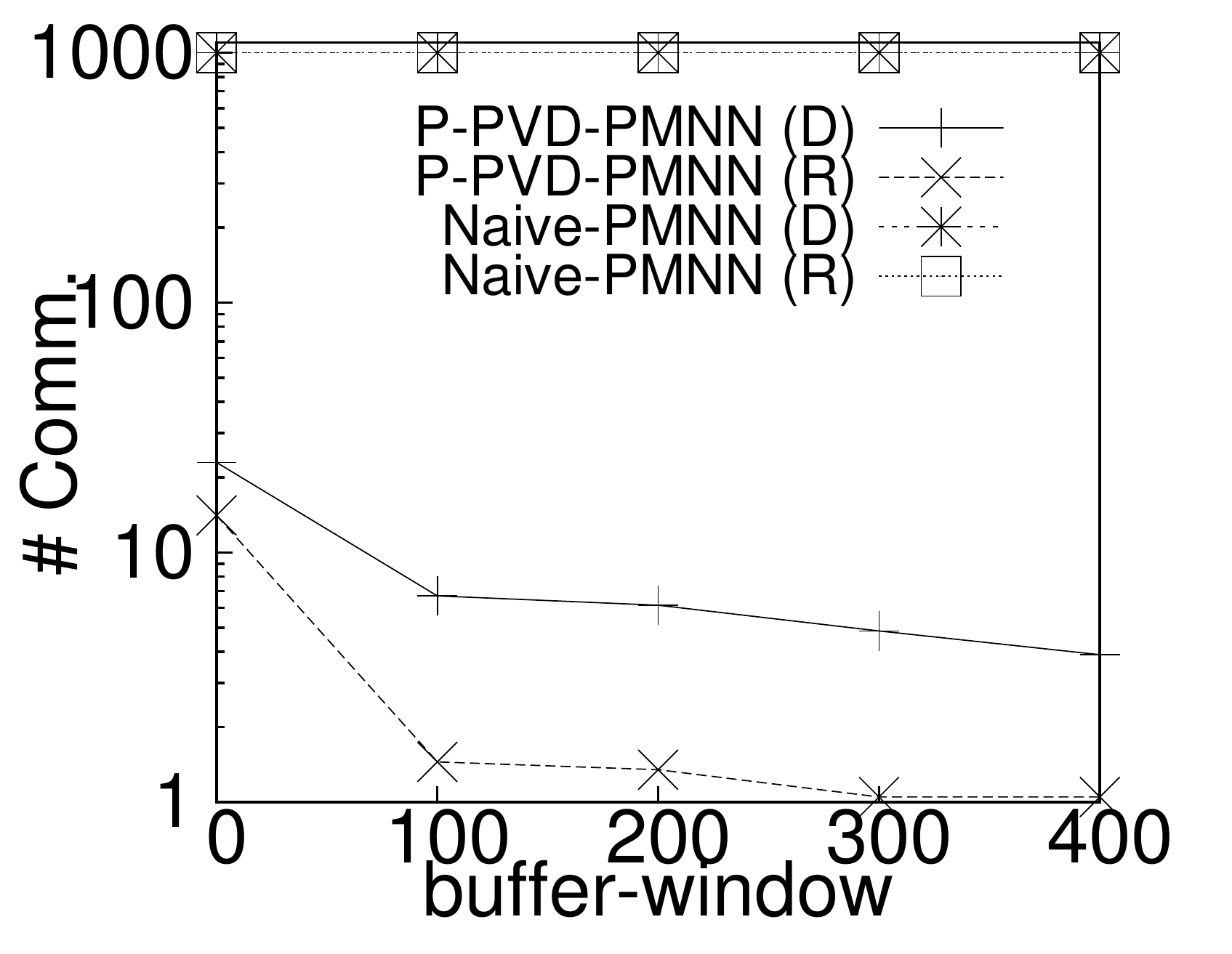}}\\
       \scriptsize{(c)\hspace{0mm}} & \scriptsize{(d)}\\
      \end{tabular}
    \caption{The effect of buffer window in U (a-c), Z (d-f) for 1D data}
    \label{fig:vr1D}
  \end{center}
\end{figure*}

\noindent\emph{\textbf{Experiments with 1D Data Sets:}}

For 1D data, we have run a similar set of experiments to 2D ones, where we vary the length of the query trajectory, the data set size, and the size of the buffer window.

\emph{Effect of the Length of a Query Trajectory}: In this set of
experiments, we vary the query trajectory length from 1000 to 5000
units while evaluating a PMNN query for 1D data sets. We run the
experiments for both U (see Figures~\ref{fig:vq1D}(a)-(c)) and Z
(see Figures~\ref{fig:vq1D}(d)-(f)) data sets. The data set size
is set to 100. We can see that, for both P-PVD-PMNN and
Naive-PMNN, the processing time, I/O costs, and number of
communications increase with the increase of the query trajectory
length for 1D sets, which is expected. Figures also show that our
P-PVD-PMNN outperforms the Naive-PMNN by at least an order of
magnitude in terms of processing time, I/Os, and communication
costs.

\emph{Effect of Data Set Size}: We also run the experiments with
varying data set size (see Figures~\ref{fig:vds1D}(a)-(c) for U
and Figures~\ref{fig:vds1D}(d)-(f) for Z data sets). In these
experiments, the trajectory length is set to 5000 units. Figures
show that, for P-PVD-PMNN, the processing time, I/O costs and
number of communications increase with the increase of data set
size for 1D sets. This is because, for a larger data set, we have
smaller PVCs and thereby a moving query needs to check higher
number of PBRs than that of a smaller data set.
Figures~\ref{fig:vds1D}(a)-(f) also show that our P-PVD-PMNN
outperforms Naive-PMNN by at least an order of magnitude in all
evaluation metrics. The results also show that P-PVD-PMNN performs
similar for both directional (D) and random (R) query movement
paths.

\emph{Effect of Buffer Window}: In this set of experiments, we
study the impact of introducing a buffer for processing a PMNN
query. We vary the value of buffer window from 0 to 400 units, and
then run the experiments for U (see Figures~\ref{fig:vr1D}(a)-(c))
and Z (see Figures~\ref{fig:vr1D}(d)-(f)) data sets. In these
experiments, we set the data set size to 100 and the trajectory
length to 5000 units. \eat{Like 2D data, the results show that the
processing time, I/O cost, and the number of communications
decrease with the increase of the buffer window. This is because,
for a large buffer window P-PVD-PMNN retrieves more PVCs at a time
than that of a smaller buffer window, and thereby reduces
communication, I/O and processing costs. However, for Zipfian data
the processing time slightly increases with the increase of the
buffer window, but the I/O cost always decreases with the increase
of the buffer window. Since the in-memory processing cost
increases with the increase of the buffer window, the total
processing costs can also increase with the increase of the buffer
window.}The experimental results show that P-PVD-PMNN outperforms
Naive-PMNN by 1-2 orders of magnitude for I/O and processing
costs, and 2-3 orders of magnitude in terms of communication
costs.

\subsubsection{Incremental Approach}
\label{subsubsec:ipvd}

In the incremental approach, we use an $R^{*}$-tree to index the
MBRs of uncertain objects, for both I-PVD-PMNN and Naive-PMNN.
In both cases, we use the page size of 1KB and the node capacity
of 50 entries for the $R^{*}$-tree.

\begin{figure*}[htbp]
  \begin{center}
    \begin{tabular}{cccc}
        \hspace{-5mm}
      \resizebox{40mm}{!}{\includegraphics{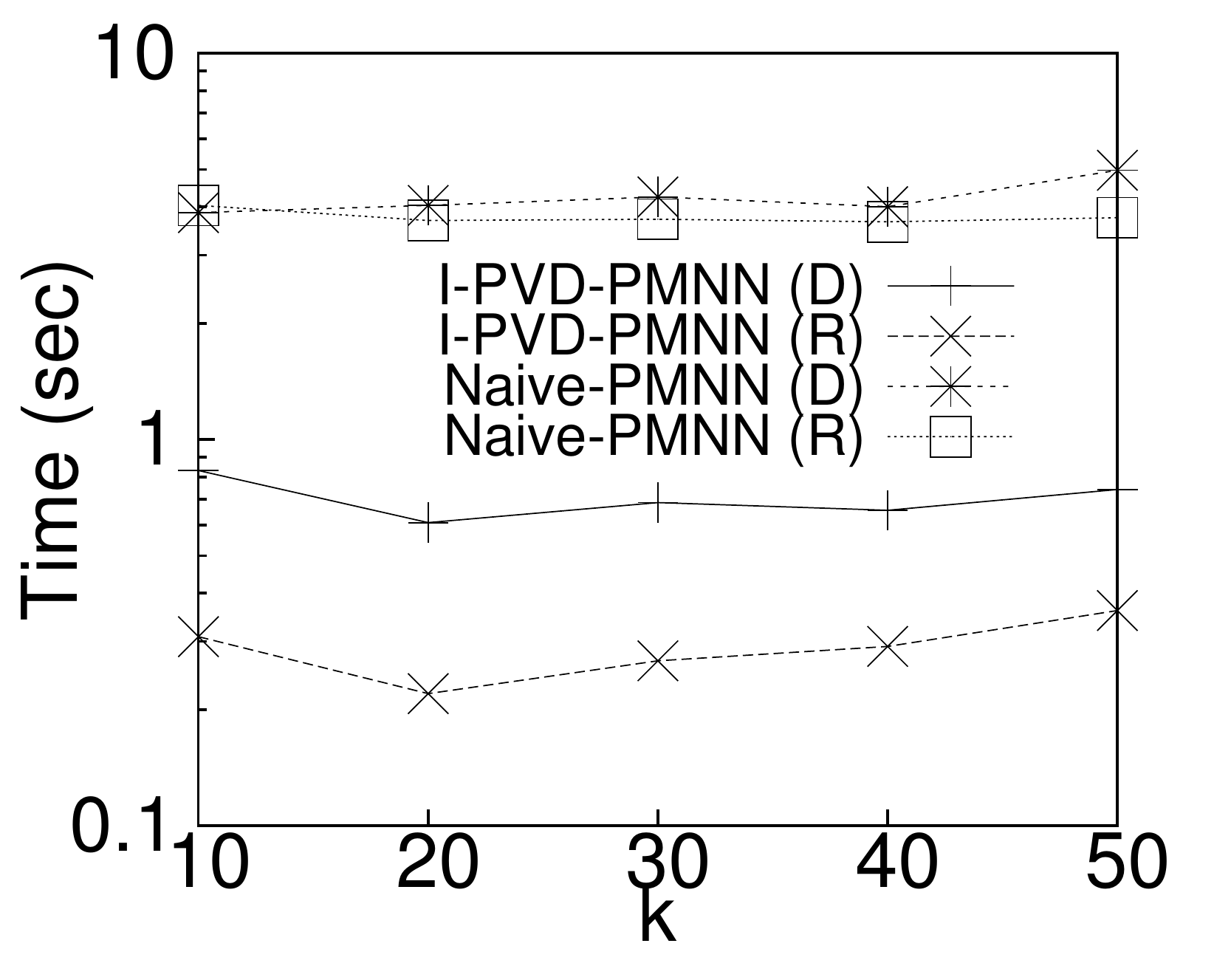}} &
        \hspace{-4mm}
        \resizebox{40mm}{!}{\includegraphics{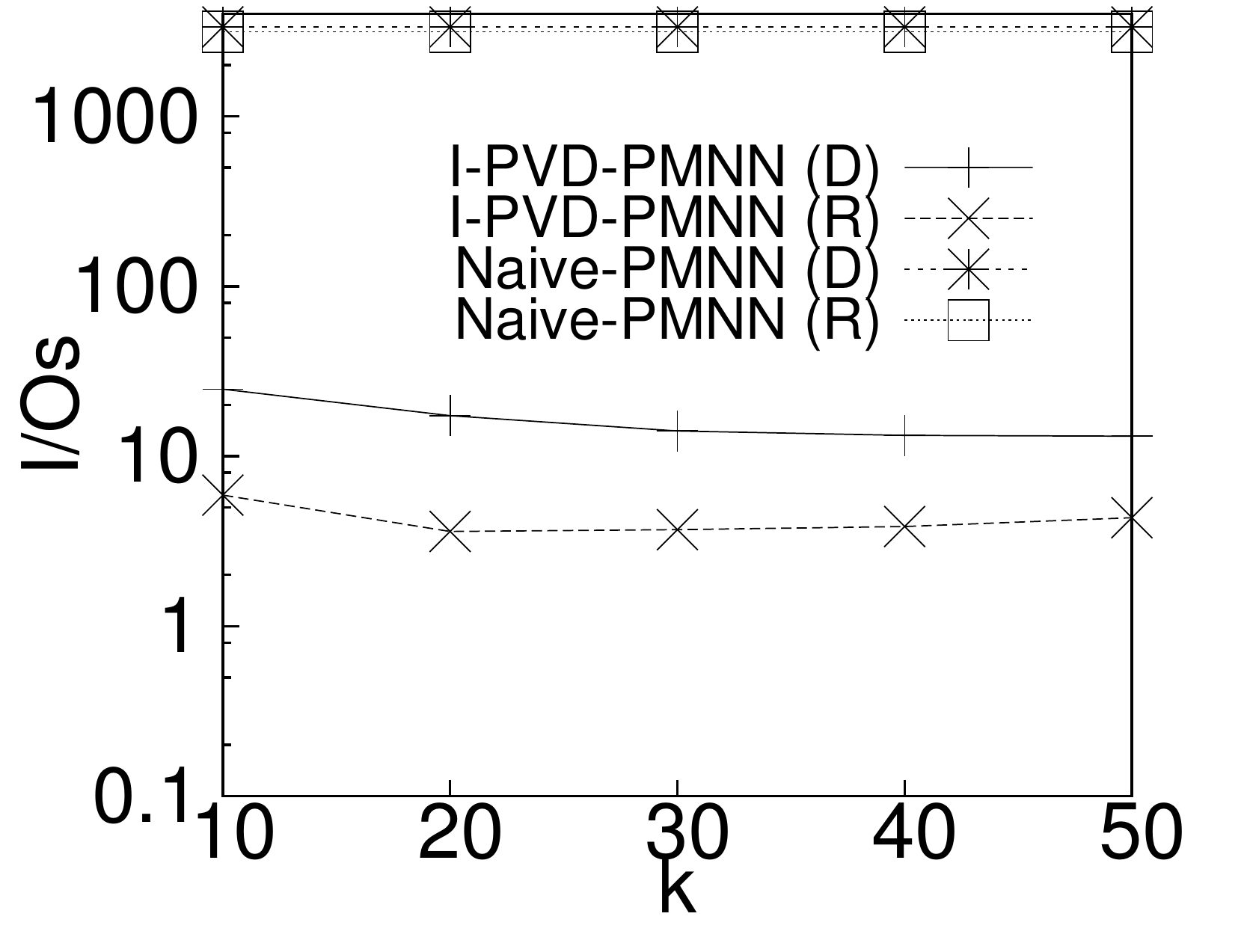}} &
         \hspace{-4mm}
        \resizebox{40mm}{!}{\includegraphics{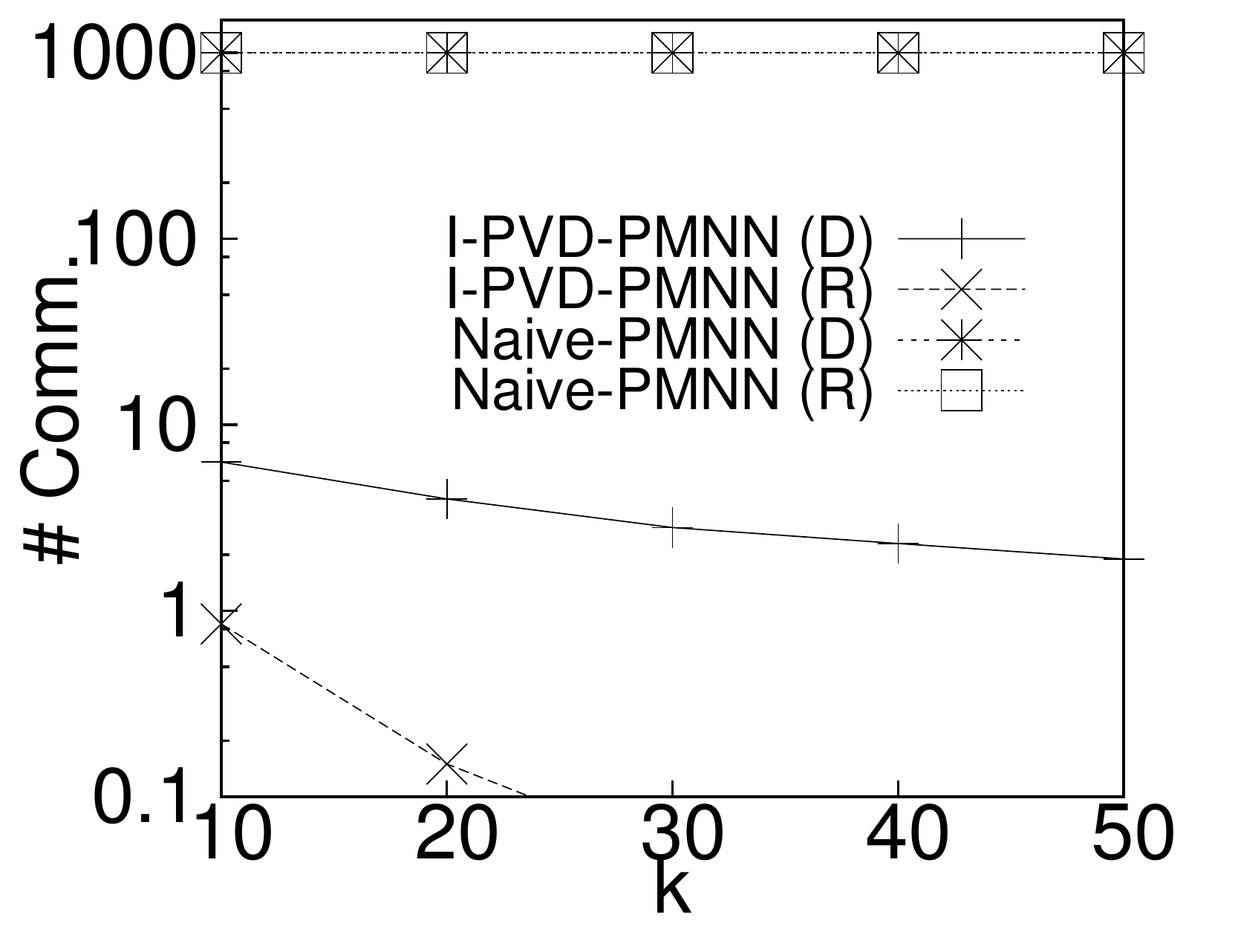}}\\
       \scriptsize{(a)\hspace{0mm}} & \scriptsize{(b)} & \scriptsize{(c)}\\
      \end{tabular}
      \begin{tabular}{cccc}
        \hspace{-5mm}
      \resizebox{40mm}{!}{\includegraphics{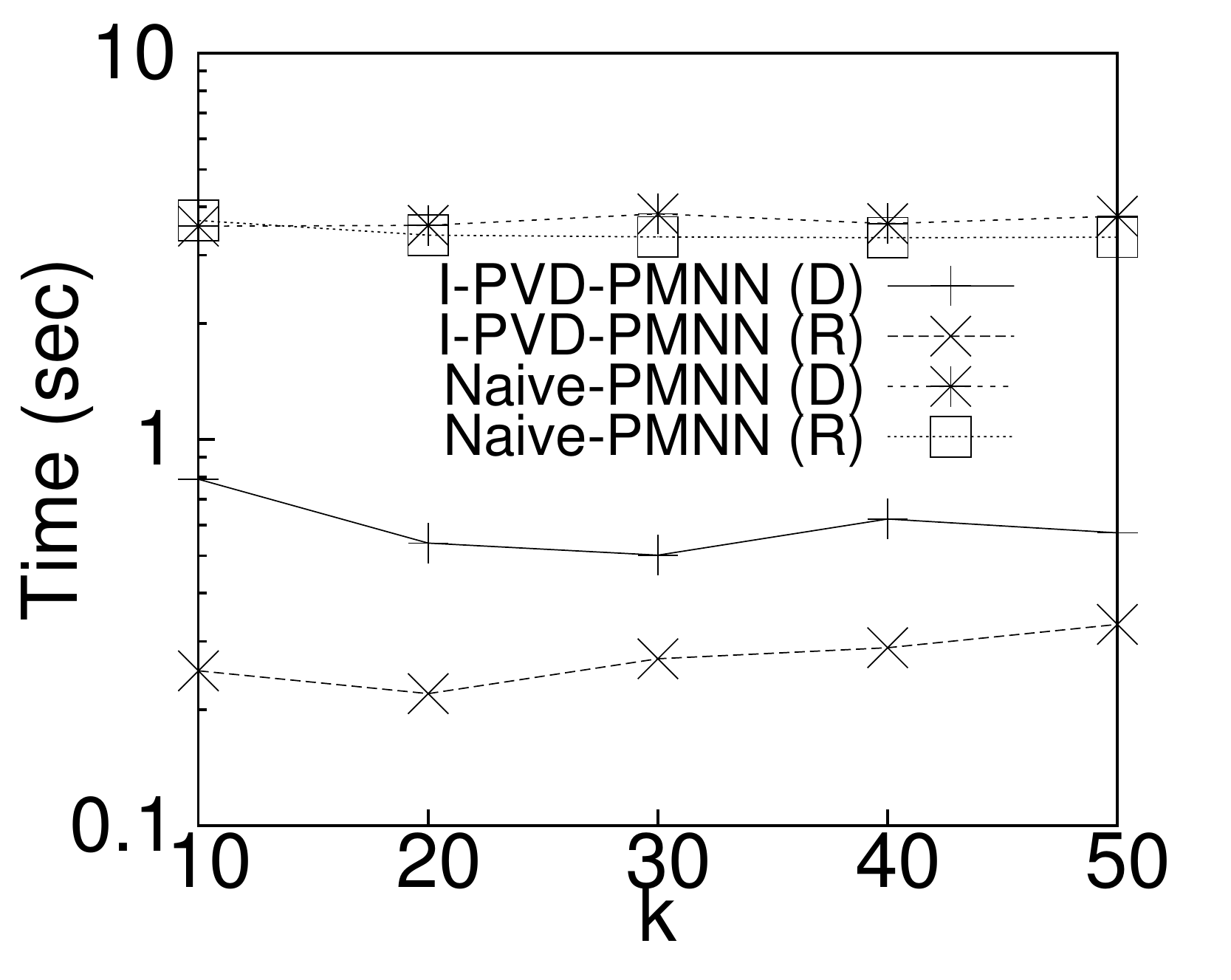}} &
        \hspace{-4mm}
        \resizebox{40mm}{!}{\includegraphics{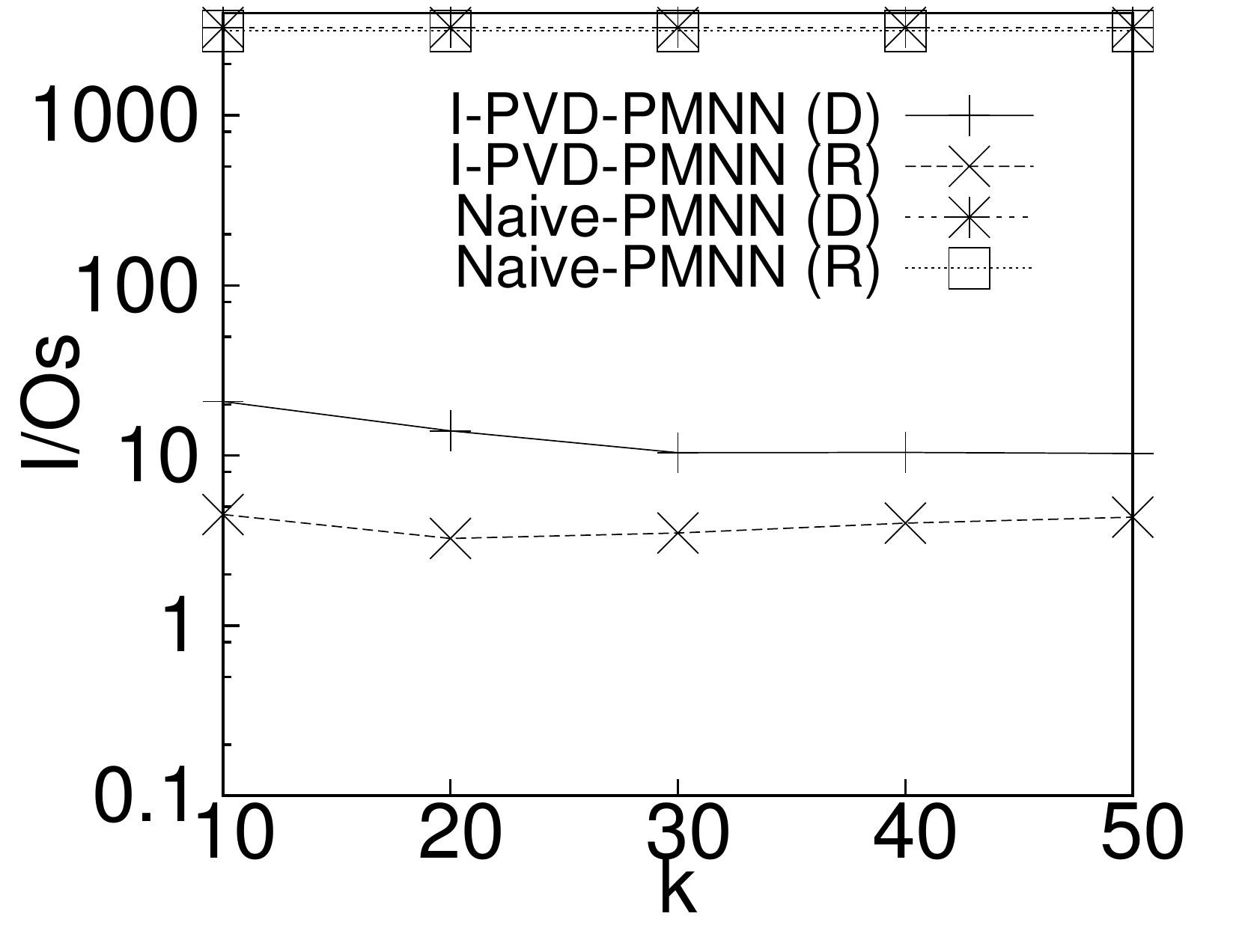}} &
          \hspace{-4mm}
        \resizebox{40mm}{!}{\includegraphics{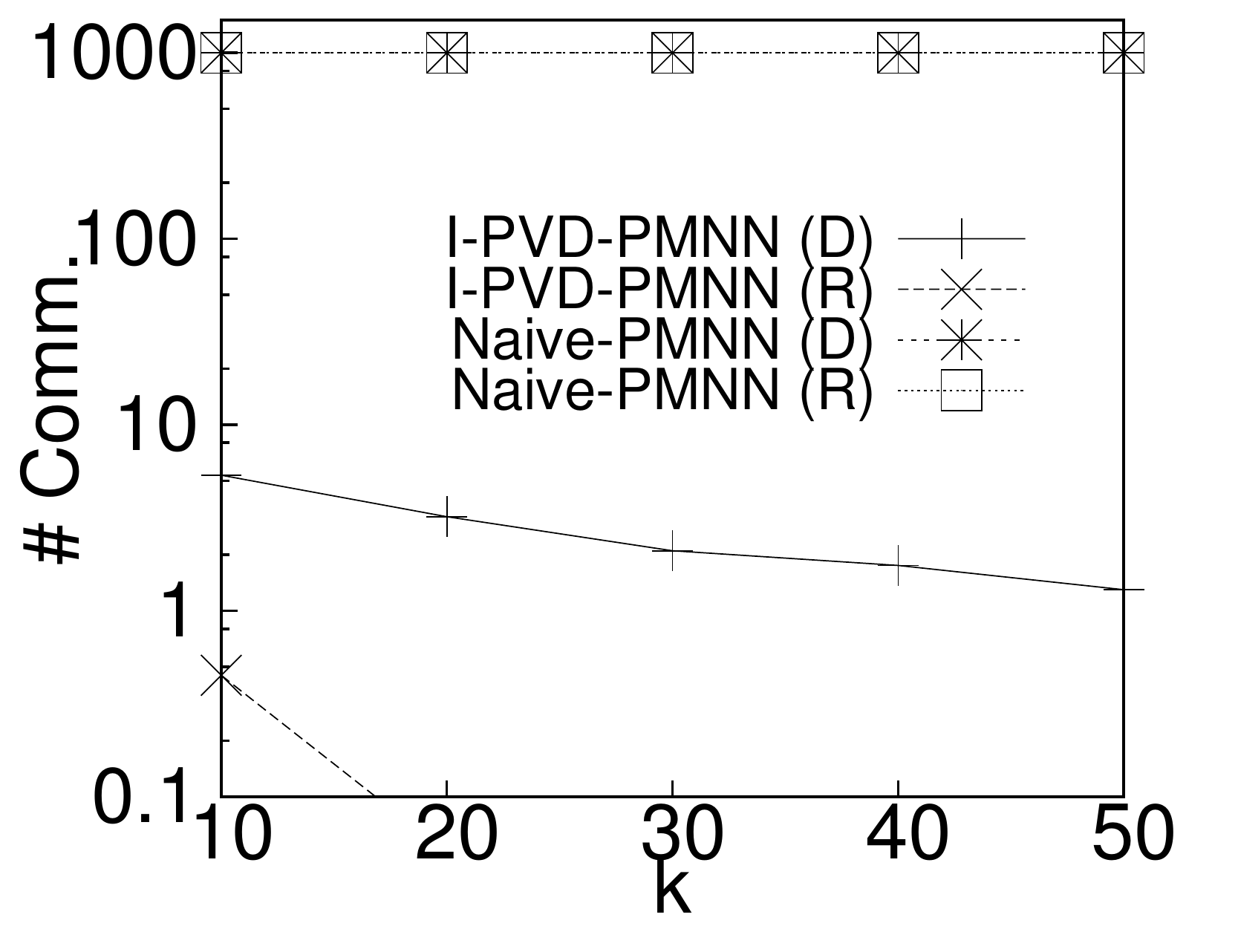}}\\
        \scriptsize{(d)\hspace{0mm}} & \scriptsize{(e)} & \scriptsize{(f)}\\
      \end{tabular}
          \begin{tabular}{cccc}
        \hspace{-5mm}
      \resizebox{40mm}{!}{\includegraphics{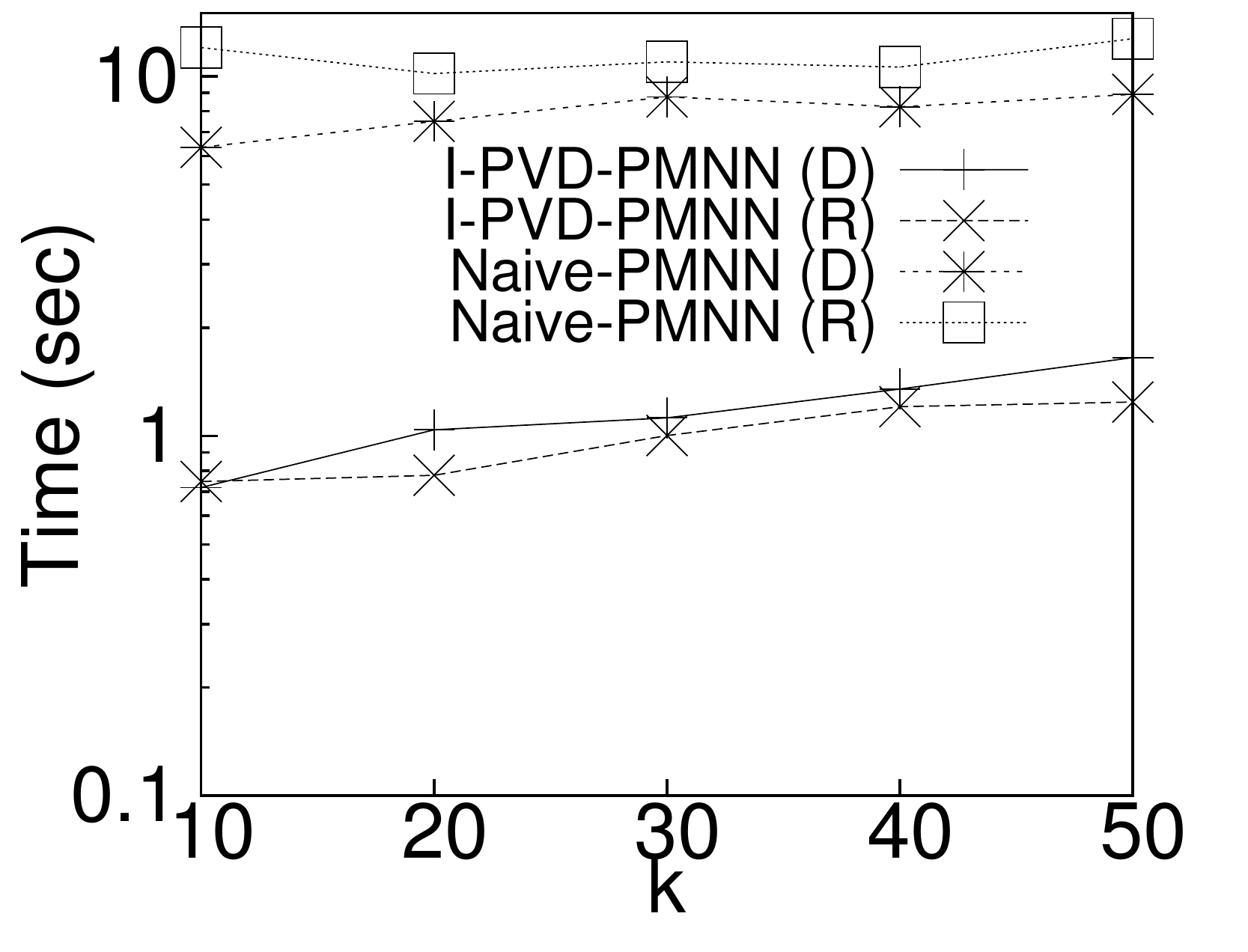}} &
        \hspace{-4mm}
        \resizebox{40mm}{!}{\includegraphics{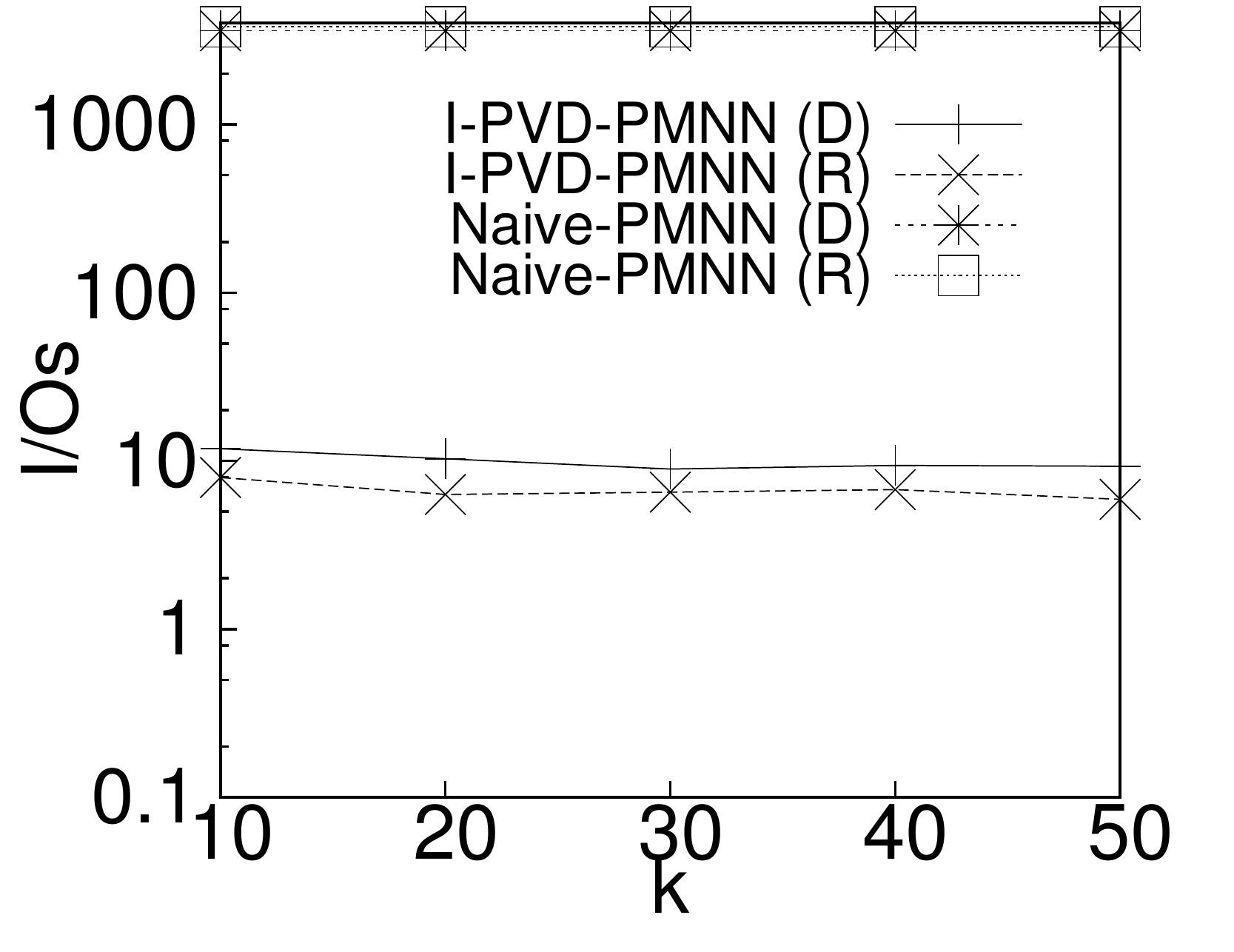}} &
          \hspace{-4mm}
        \resizebox{40mm}{!}{\includegraphics{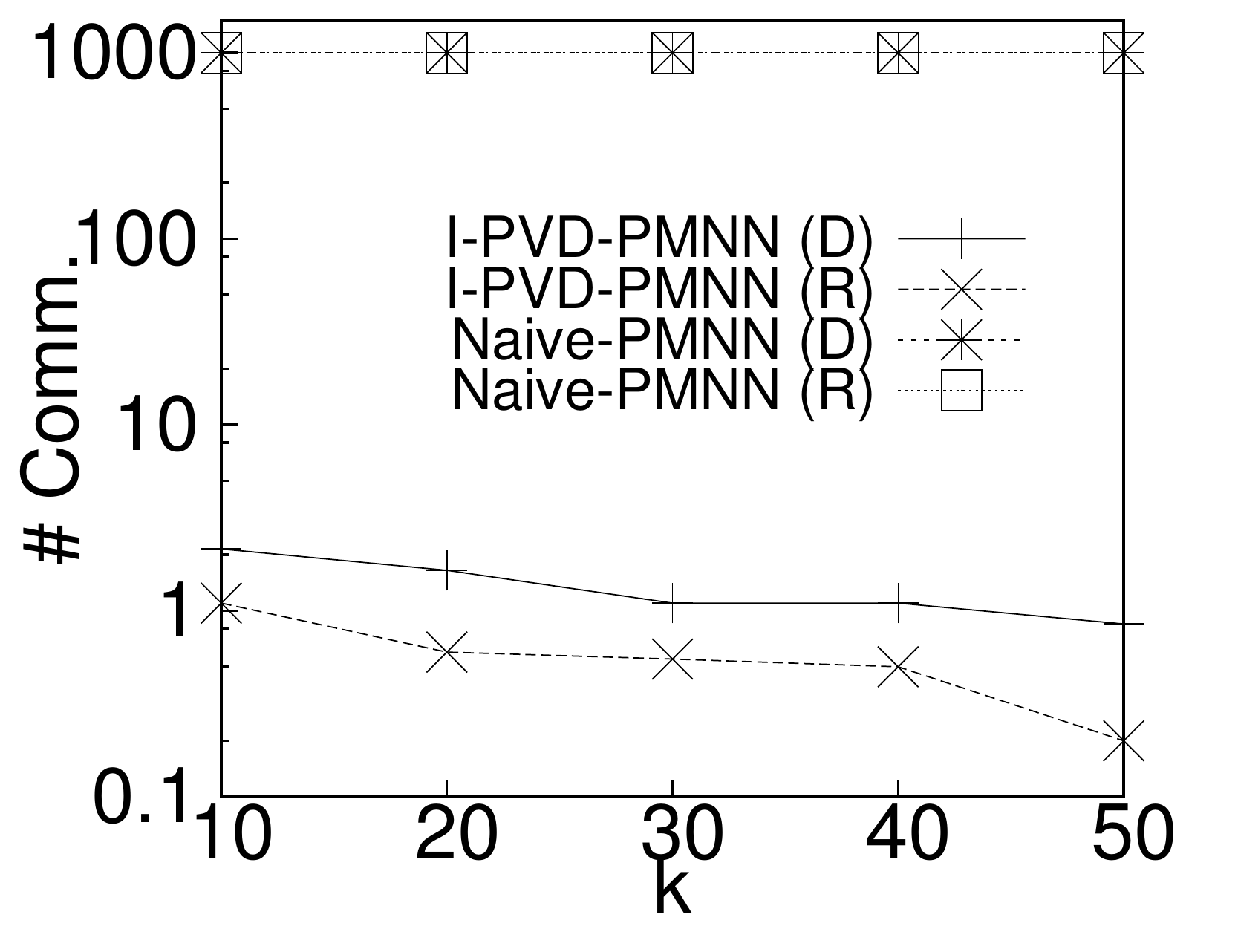}}\\
          \scriptsize{(g)\hspace{0mm}} & \scriptsize{(h)} & \scriptsize{(i)}\\
      \end{tabular}
    \caption{The effect of (\emph{k}) in U (a-c), Z (d-f), and L (g-i)}
    \label{fig:vk}
  \end{center}
\end{figure*}

\noindent\emph{\textbf{Experiments with 2D Data Sets:}}

We vary the following parameters in our experiments: the value of $k$ (i.e., the number of objects retrieved at each step), the data set size, and the length of the query trajectory, and compare the performance of I-PVD-PMNN with Naive-PMNN.

\emph{Effect of $k$: }In this set of experiments, we study the
impact of $k$ in the performance measure for processing a PMNN
query. We vary the value of $k$ from 10 to 50, and then run
the experiments for all available data sets (U, Z, and L). In these
experiments, for both U and Z, we have set the data set size to 10K. Figures~\ref{fig:vk}(a)-(c) show the processing time, the
I/O costs, and the number of communications, respectively, for
varying $k$ from 10 to 50 for U data set. Figure~\ref{fig:vk}(a)
shows that the processing time almost remains constant for varying
$k$. The processing time of I-PVD-PMNN is on average 6 times less for directional (D) query paths than that of Naive-PMNN, and on average 13 times less for random (R) query paths than that of Naive-PMNN.
On the other hand, Figures~\ref{fig:vk}(b)-(c) show that
I/O costs and the number of communications decrease with the
increase of $k$. This is because, for a larger value of $k$, the
client fetches more data at a time from the server, and thereby
needs to communicate less number of times with the server. Figures also show that our I-PVD-PMNN outperforms the Naive-PMNN by
2-3 orders of magnitude for both I/O and communication costs.

Figures~\ref{fig:vk}(d)-(f) and (g)-(i) show the performance
behaviors of Z and L data sets, respectively, which are
similar to U data set.

\begin{figure*}[htbp]
  \begin{center}
    \begin{tabular}{cccc}
        \hspace{-5mm}
      \resizebox{40mm}{!}{\includegraphics{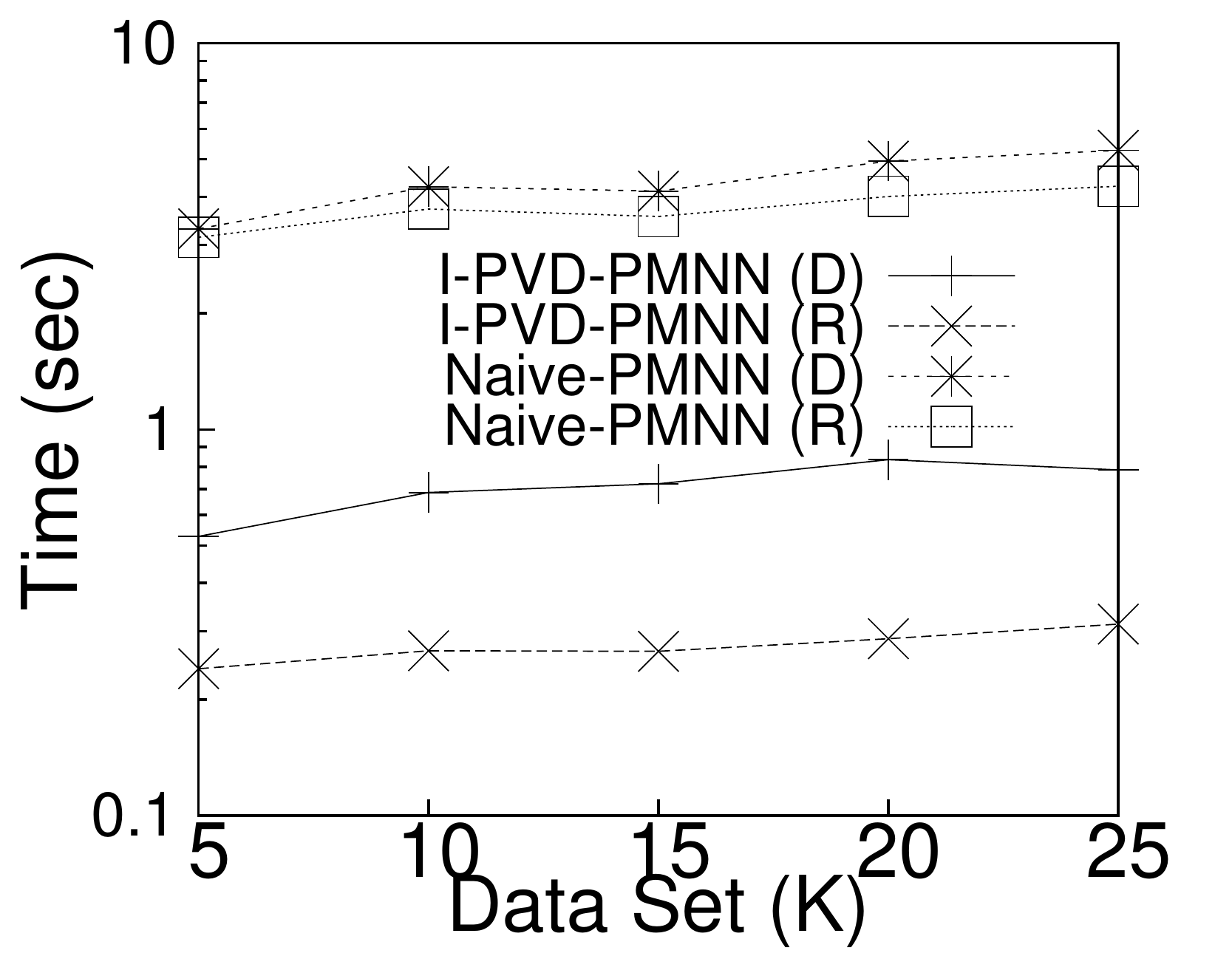}} &
        \hspace{-4mm}

        \resizebox{40mm}{!}{\includegraphics{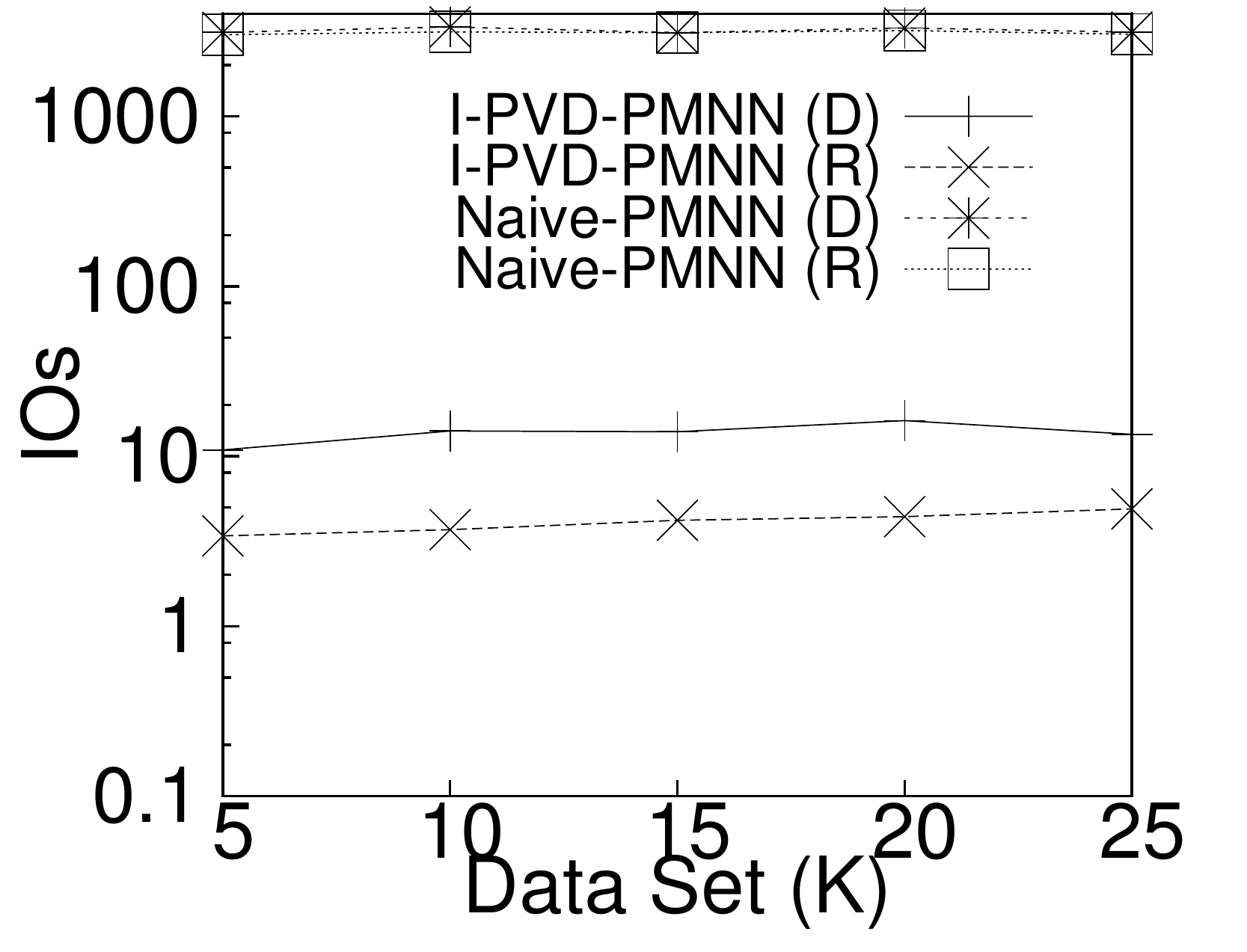}} &
         \hspace{-4mm}

        \resizebox{40mm}{!}{\includegraphics{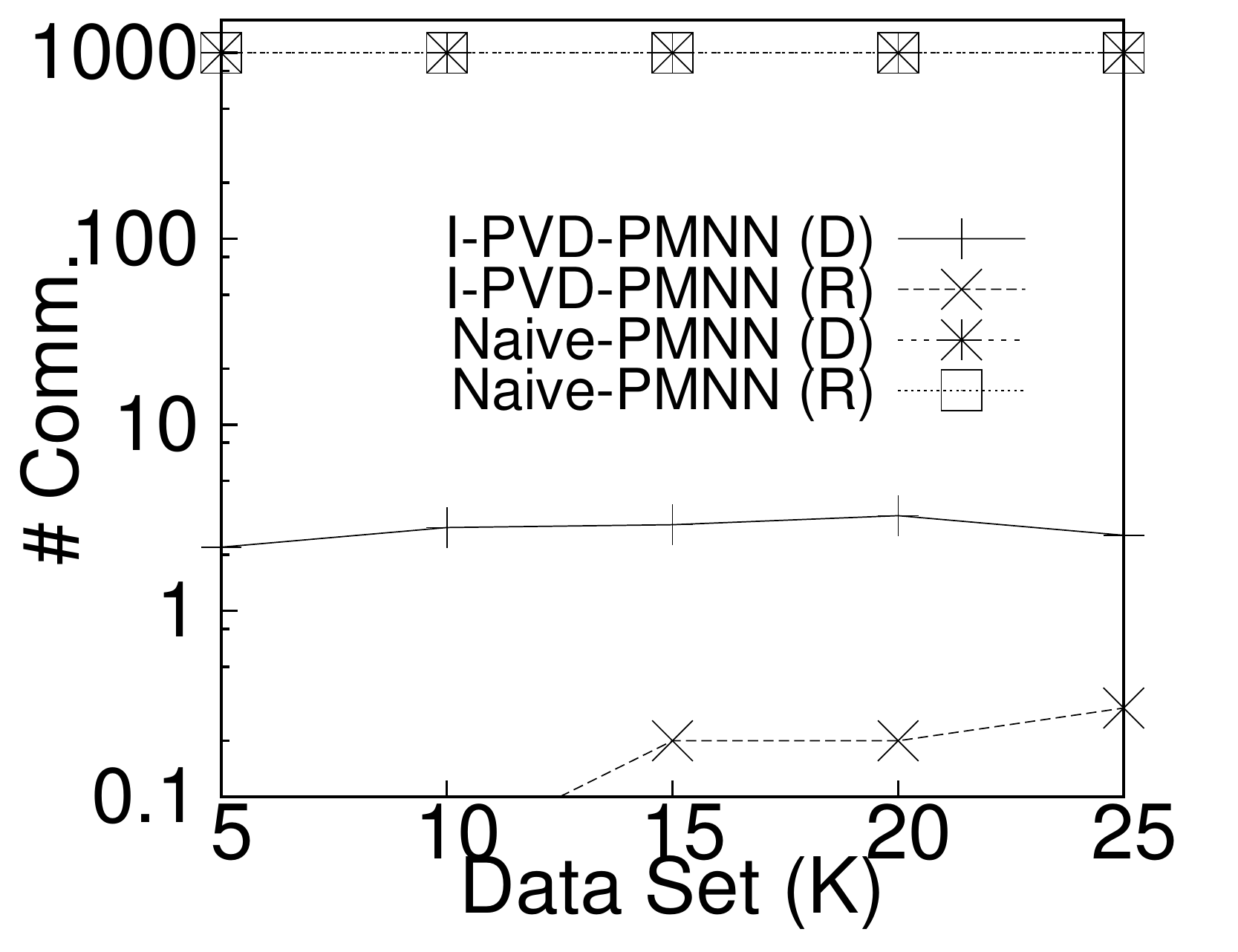}}\\
       \scriptsize{(a)\hspace{0mm}} & \scriptsize{(b)} & \scriptsize{(c)}\\
      \end{tabular}
      \begin{tabular}{cccc}
        \hspace{-5mm}
      \resizebox{40mm}{!}{\includegraphics{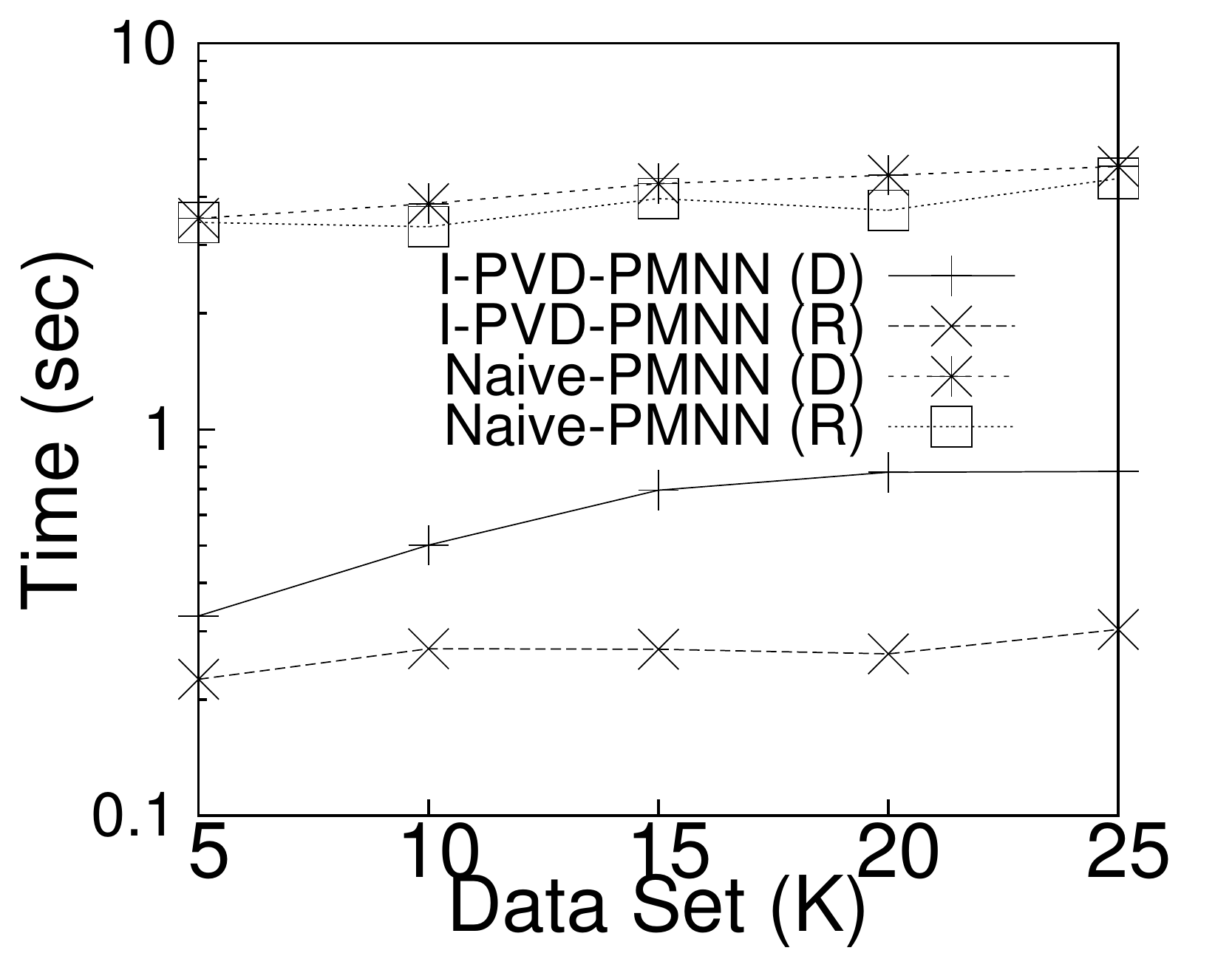}} &
        \hspace{-4mm}

        \resizebox{40mm}{!}{\includegraphics{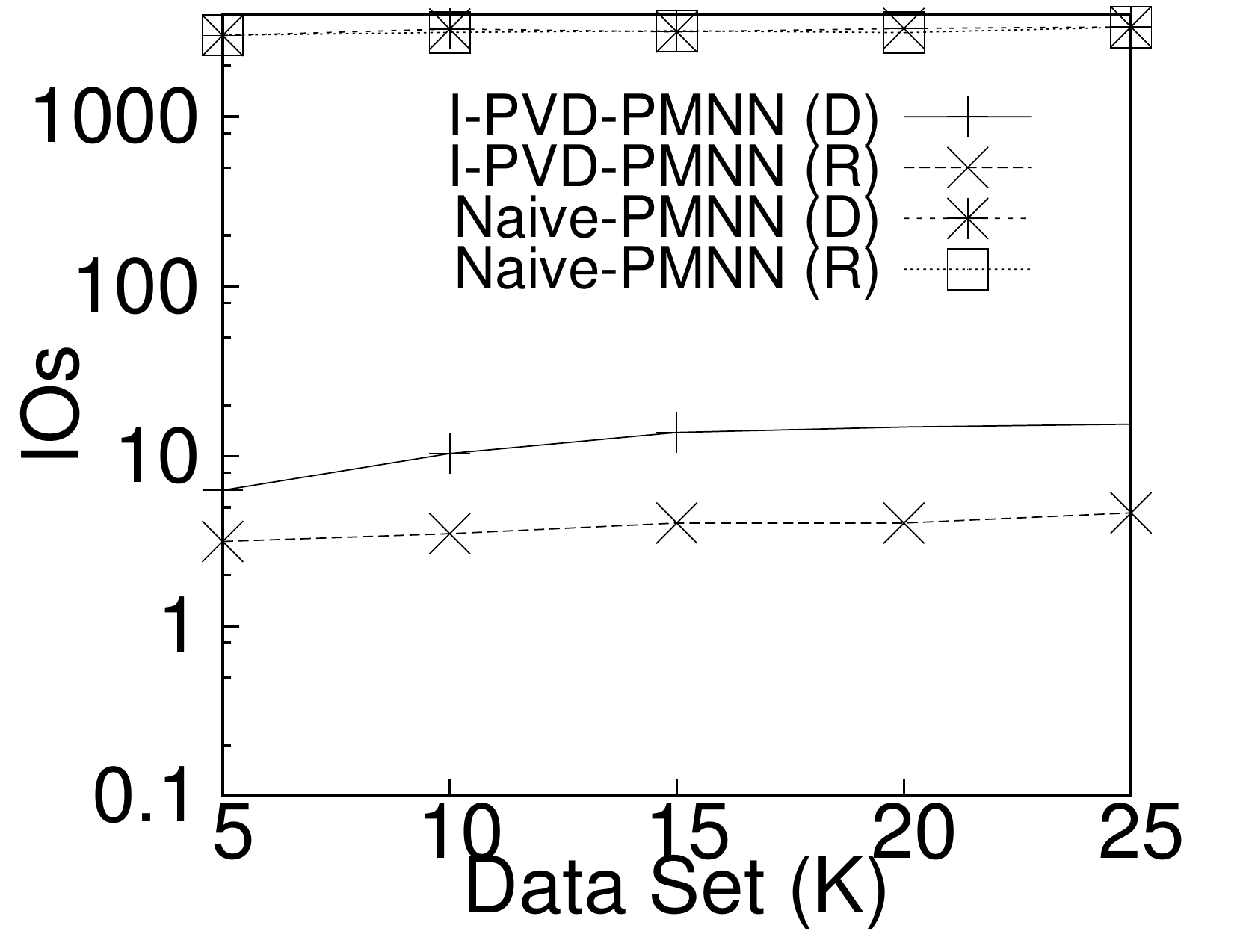}} &
        \hspace{-4mm}

        \resizebox{40mm}{!}{\includegraphics{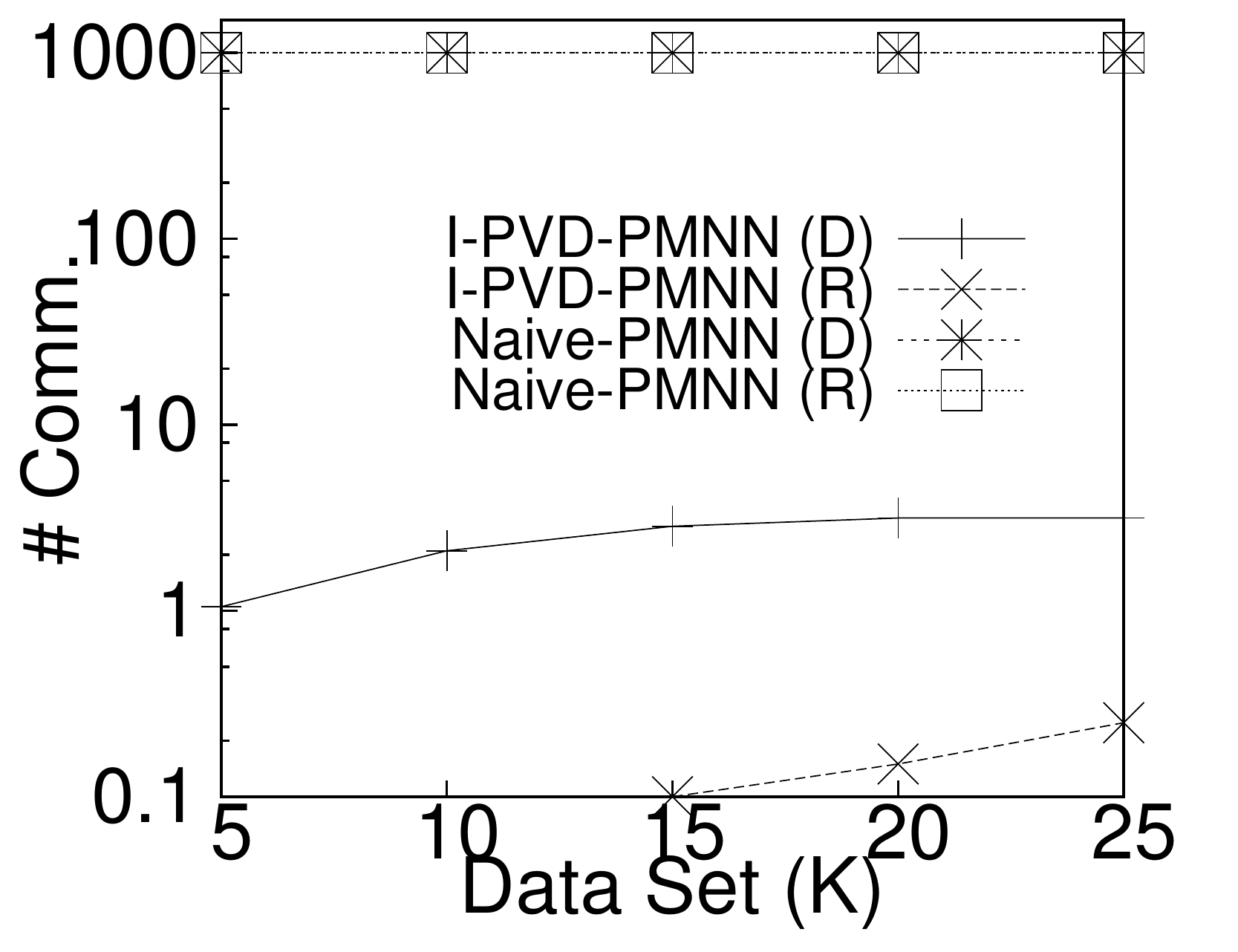}}\\
       \scriptsize{(d)\hspace{0mm}} & \scriptsize{(e)} & \scriptsize{(f)}\\
      \end{tabular}
    \caption{The effect of the data set size in U (a-c), Z (d-f)}
    \label{fig:vd}
  \end{center}
\end{figure*}

\begin{figure*}[htbp]
  \begin{center}
    \begin{tabular}{cccc}
        \hspace{-5mm}
      \resizebox{40mm}{!}{\includegraphics{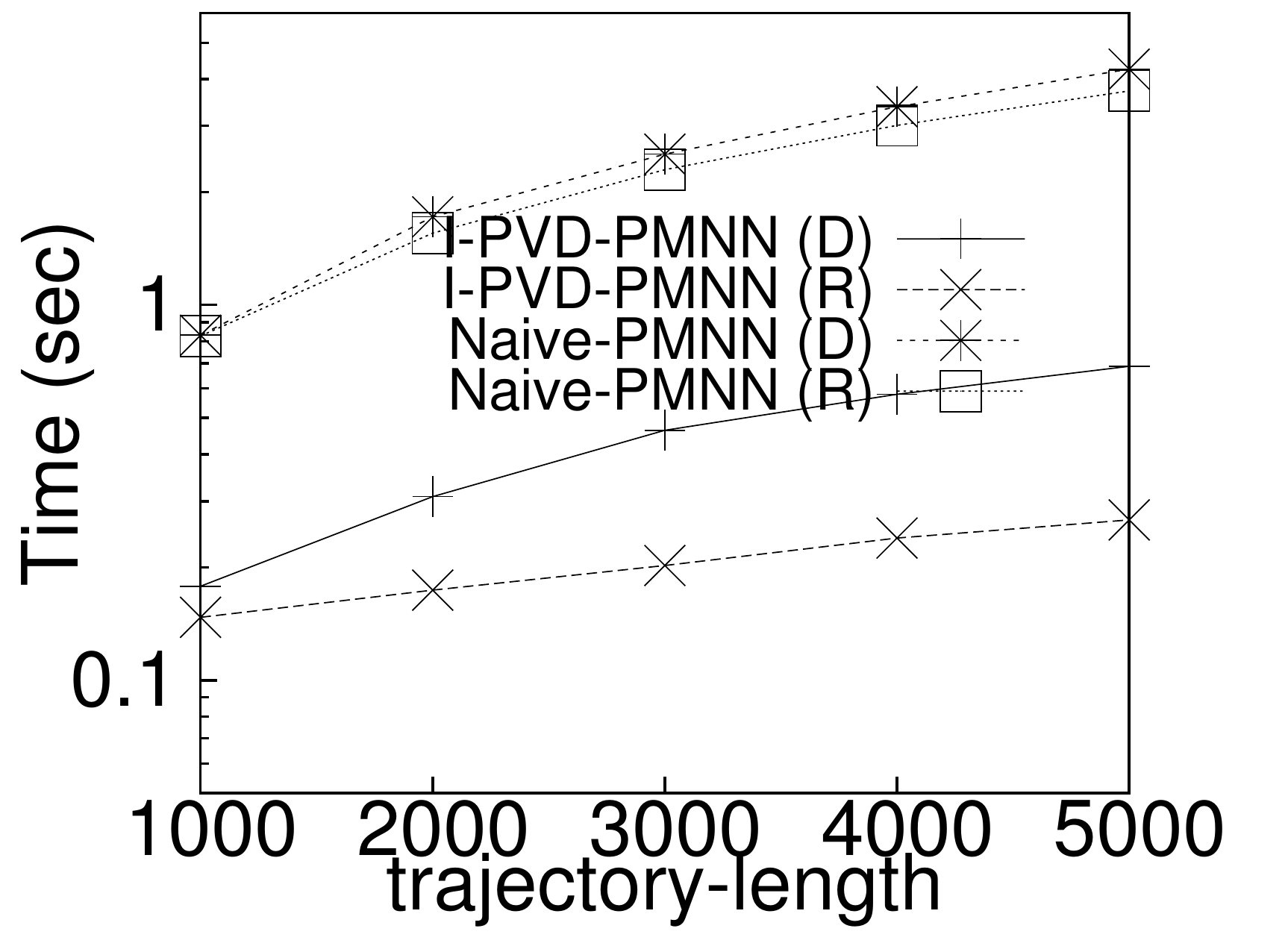}} &
        \hspace{-4mm}

        \resizebox{40mm}{!}{\includegraphics{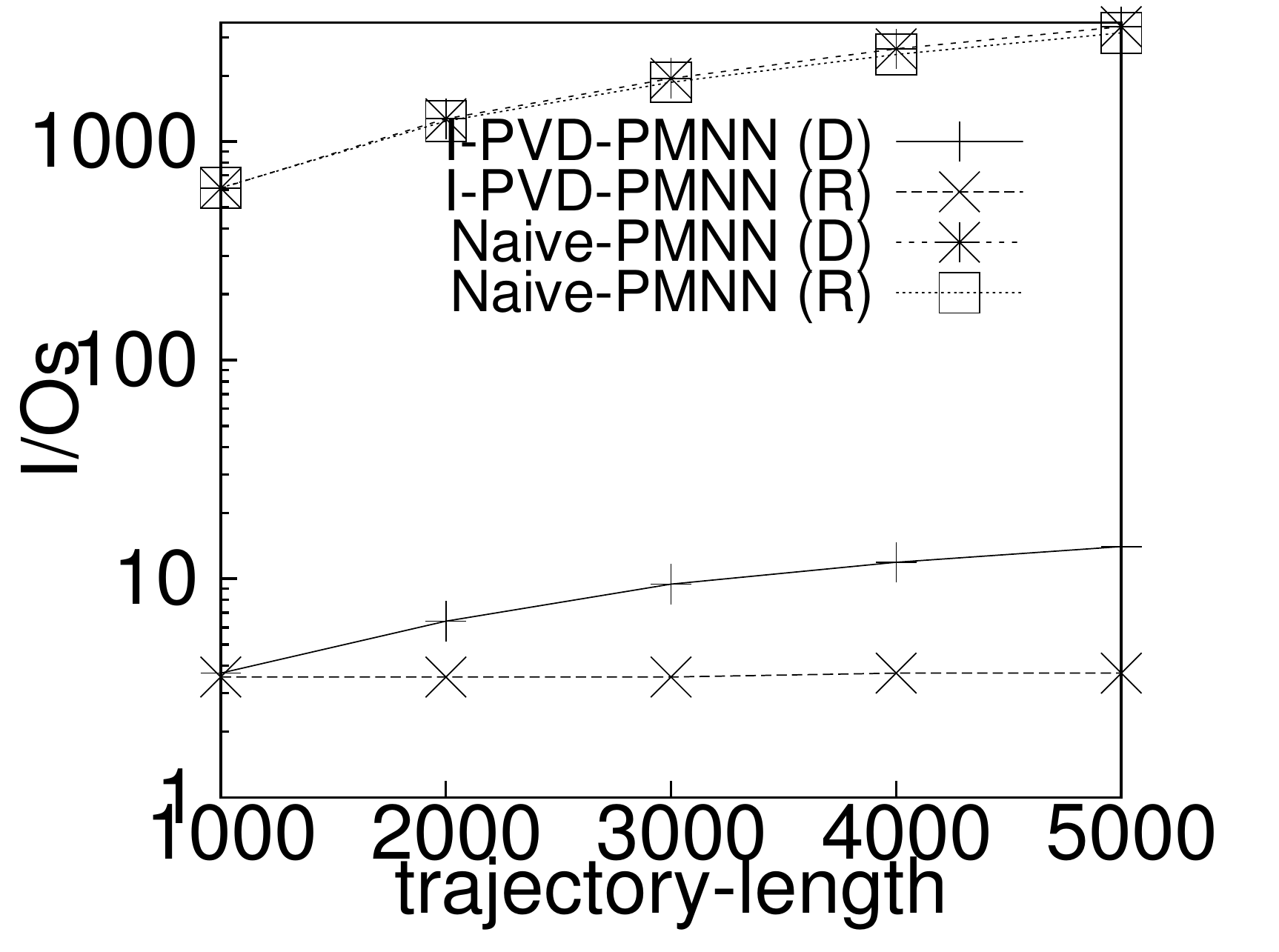}} &
         \hspace{-4mm}

        \resizebox{40mm}{!}{\includegraphics{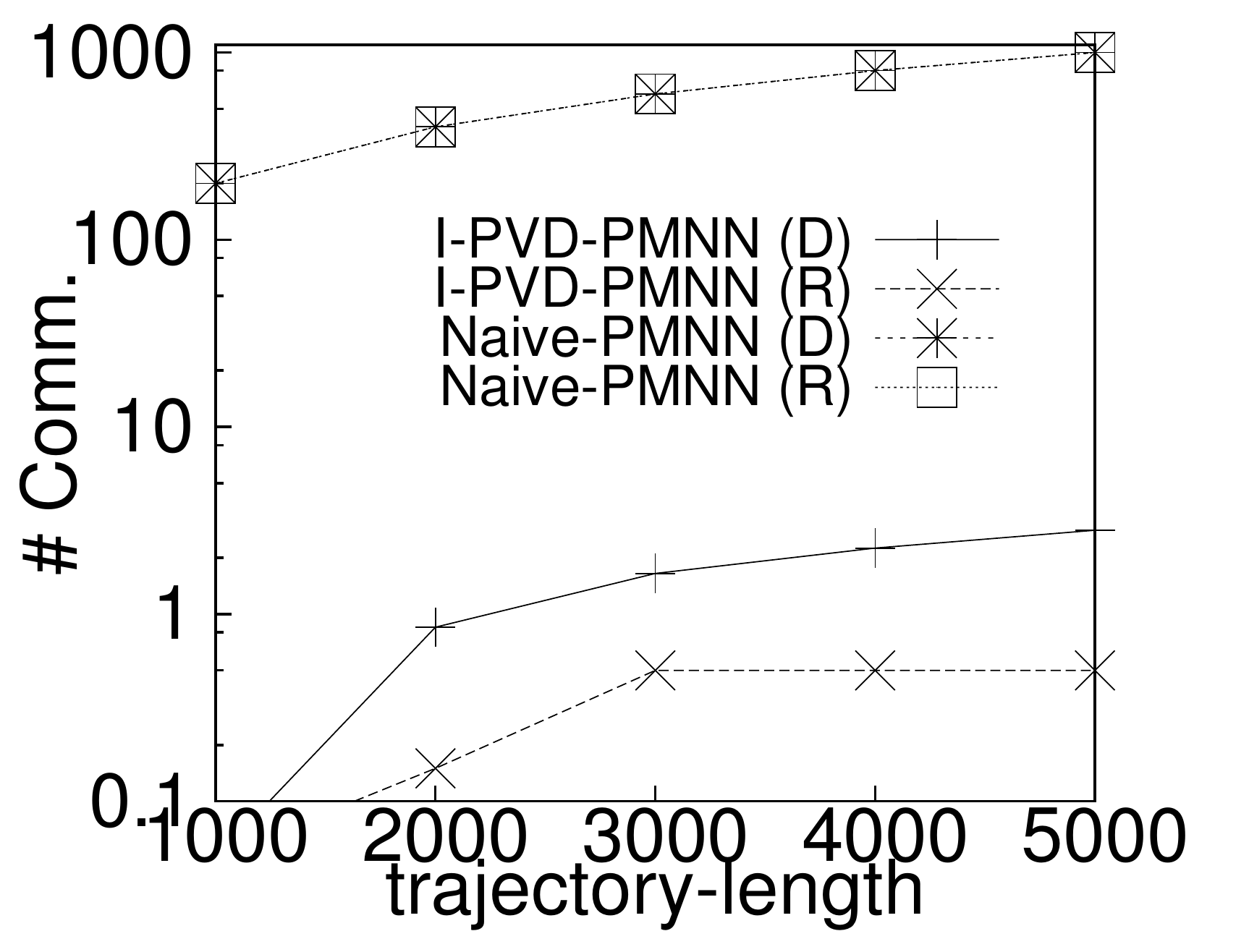}}\\
       \scriptsize{(a)\hspace{0mm}} & \scriptsize{(b)} & \scriptsize{(c)}\\
      \end{tabular}
    \begin{tabular}{cccc}
        \hspace{-5mm}
      \resizebox{40mm}{!}{\includegraphics{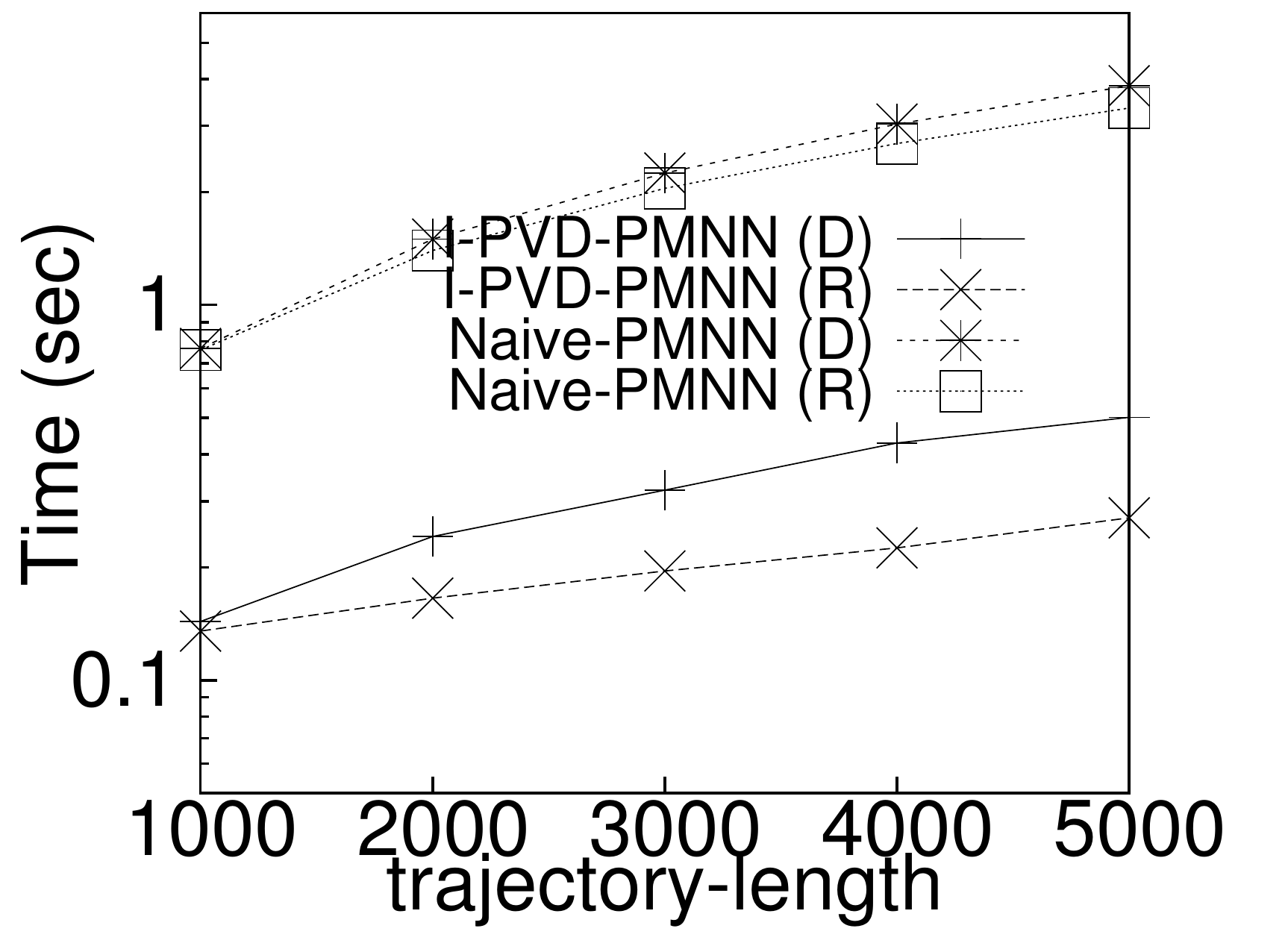}} &
        \hspace{-4mm}

        \resizebox{40mm}{!}{\includegraphics{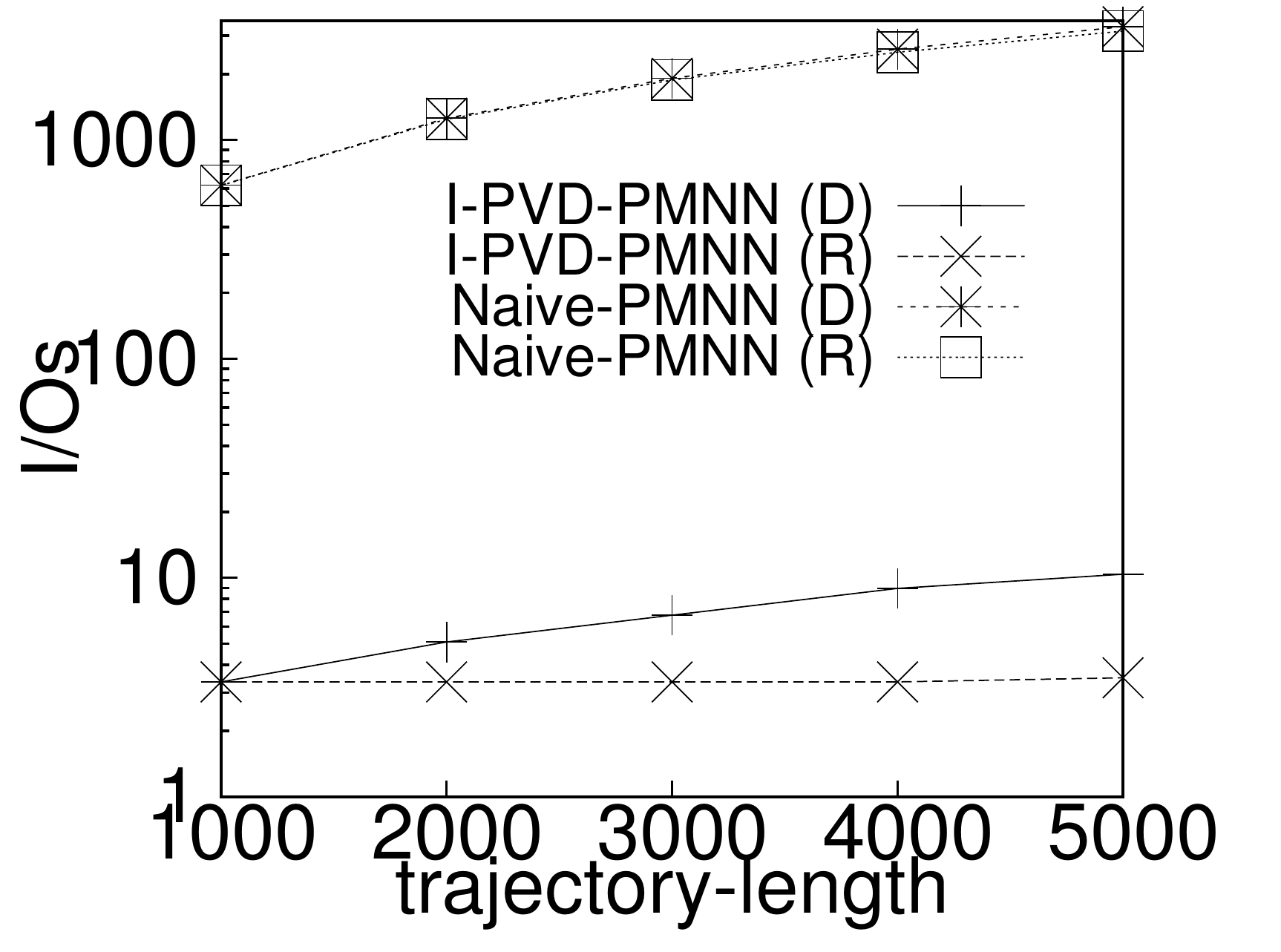}} &

        \hspace{-4mm}

        \resizebox{40mm}{!}{\includegraphics{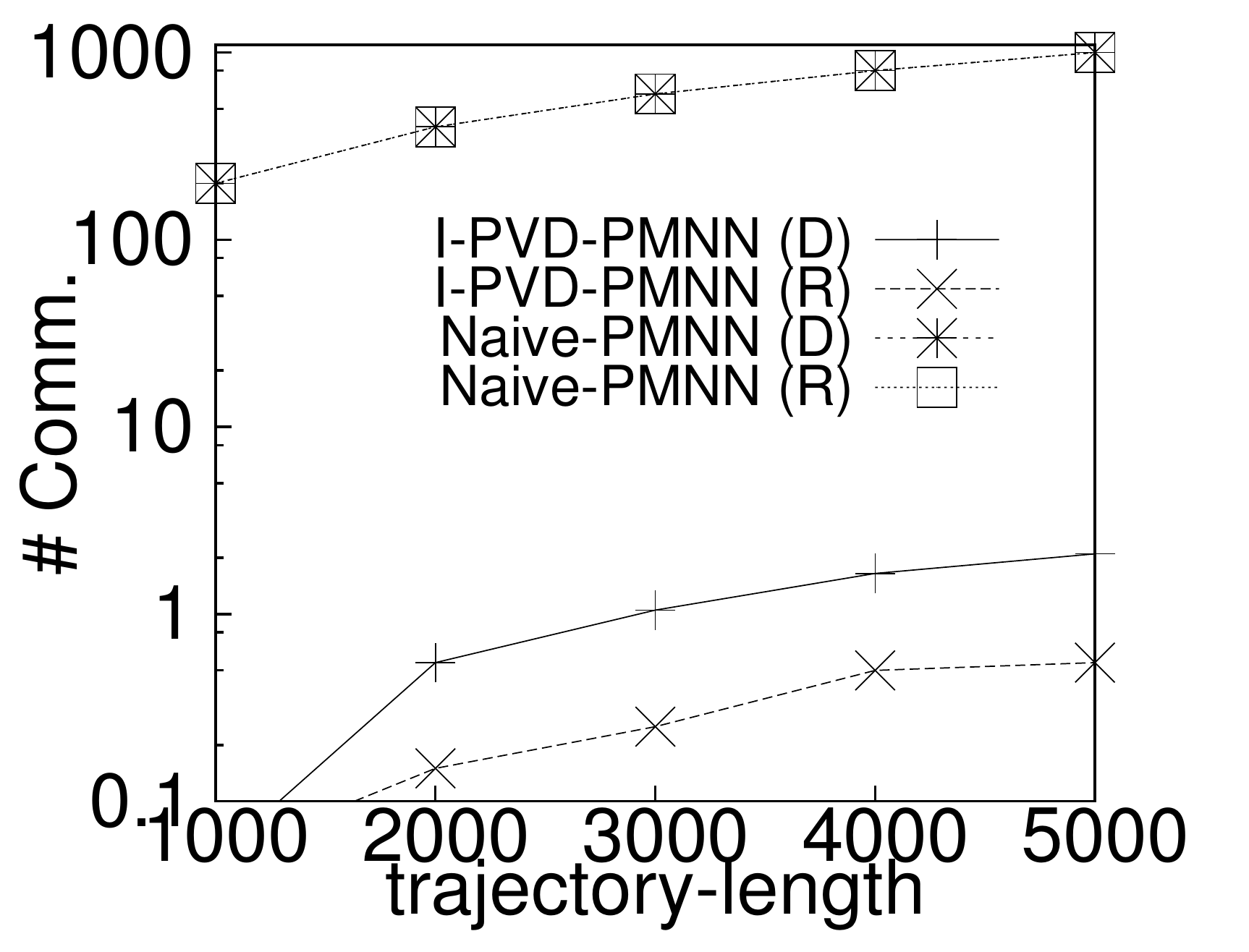}}\\
       \scriptsize{(d)\hspace{0mm}} & \scriptsize{(e)} & \scriptsize{(f)}\\
      \end{tabular}

    \caption{The effect of the query length in U (a-c), Z (d-f)}
    \label{fig:vqi}
  \end{center}

\end{figure*}

\emph{Effect of Data Set Size}: In this set of experiments,
we vary the data set size from 5K to 25K and compare the
performance of our approach I-PVD-PMNN with Naive-PMNN. We set the trajectory length to 5000 units. Also, in these
experiments, we have set the value of $k$ to 30. Figures~\ref{fig:vd} (a)-(c) and (d)-(f) show the
processing time, I/O costs, and the number of
communications for U and Z data sets, respectively.
Figures also show that our I-PVD-PMNN outperforms Naive-PMNN by
1-3 orders of magnitude for all data sets.

\emph{Effect of the Length of a Query Trajectory}: We vary the length of moving queries from 1000 to
5000 units of the data space. In these
experiments, for both U and Z, we have set the data set size to 10K. Also, in these
experiments, we have set the value of $k$ to 30. Figures~\ref{fig:vqi} show that the
processing time, I/O costs, and the number of communications
increase with the increase of the length of the query trajectory
for both U and Z data sets, which is
expected. The processing time of I-PVD-PMNN is on average 5 times less for directional (D) query path and is on average 10 times less for random (R) query paths compared to Naive-PMNN. Also I-PVD-PMNN outperforms Naive-PMNN by at least
an order of magnitude for both I/O and communication costs.

\begin{figure*}[htbp]
  \begin{center}
    \begin{tabular}{cccc}
        \hspace{-5mm}
      \resizebox{40mm}{!}{\includegraphics{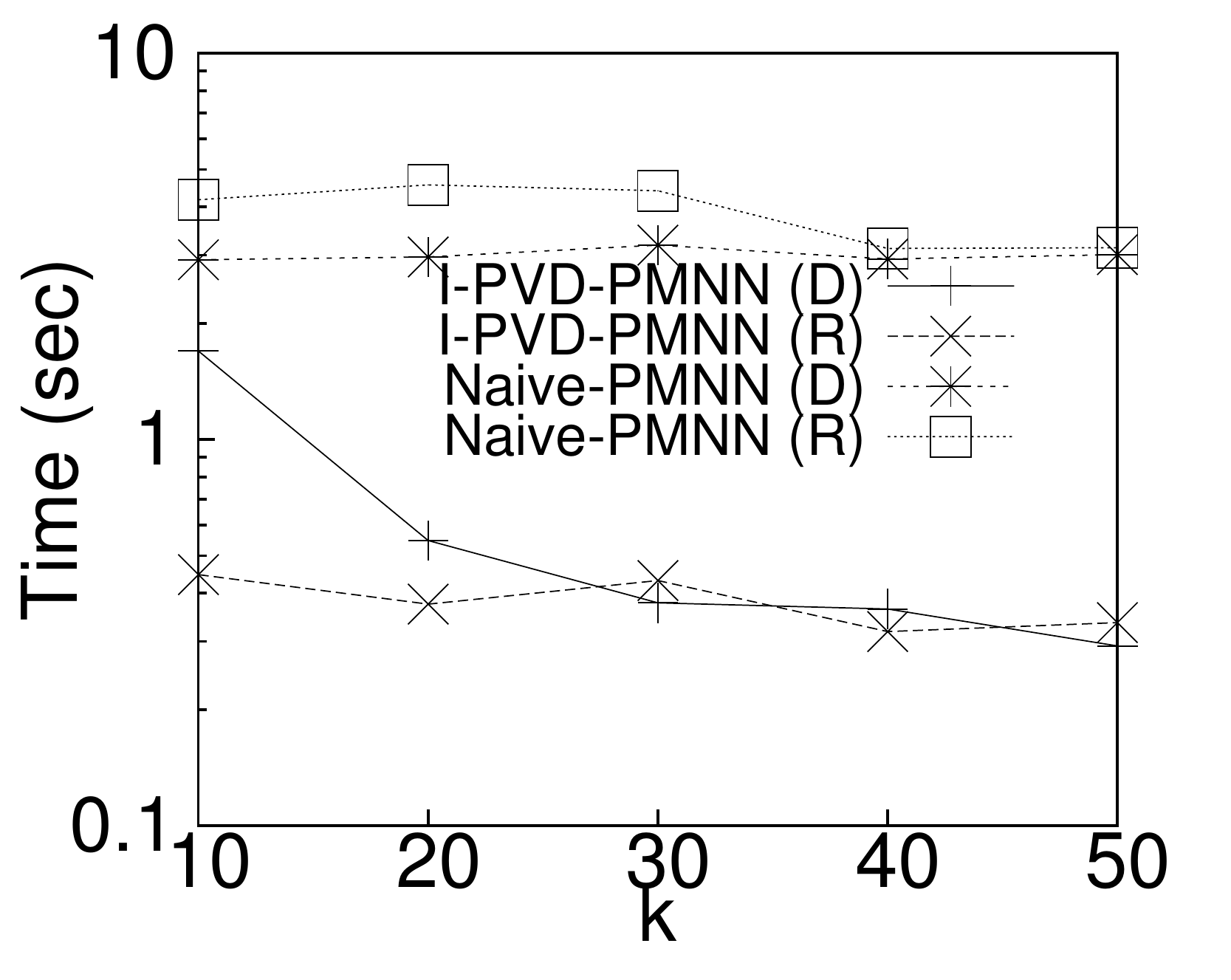}} &
        \hspace{-4mm}
        \resizebox{40mm}{!}{\includegraphics{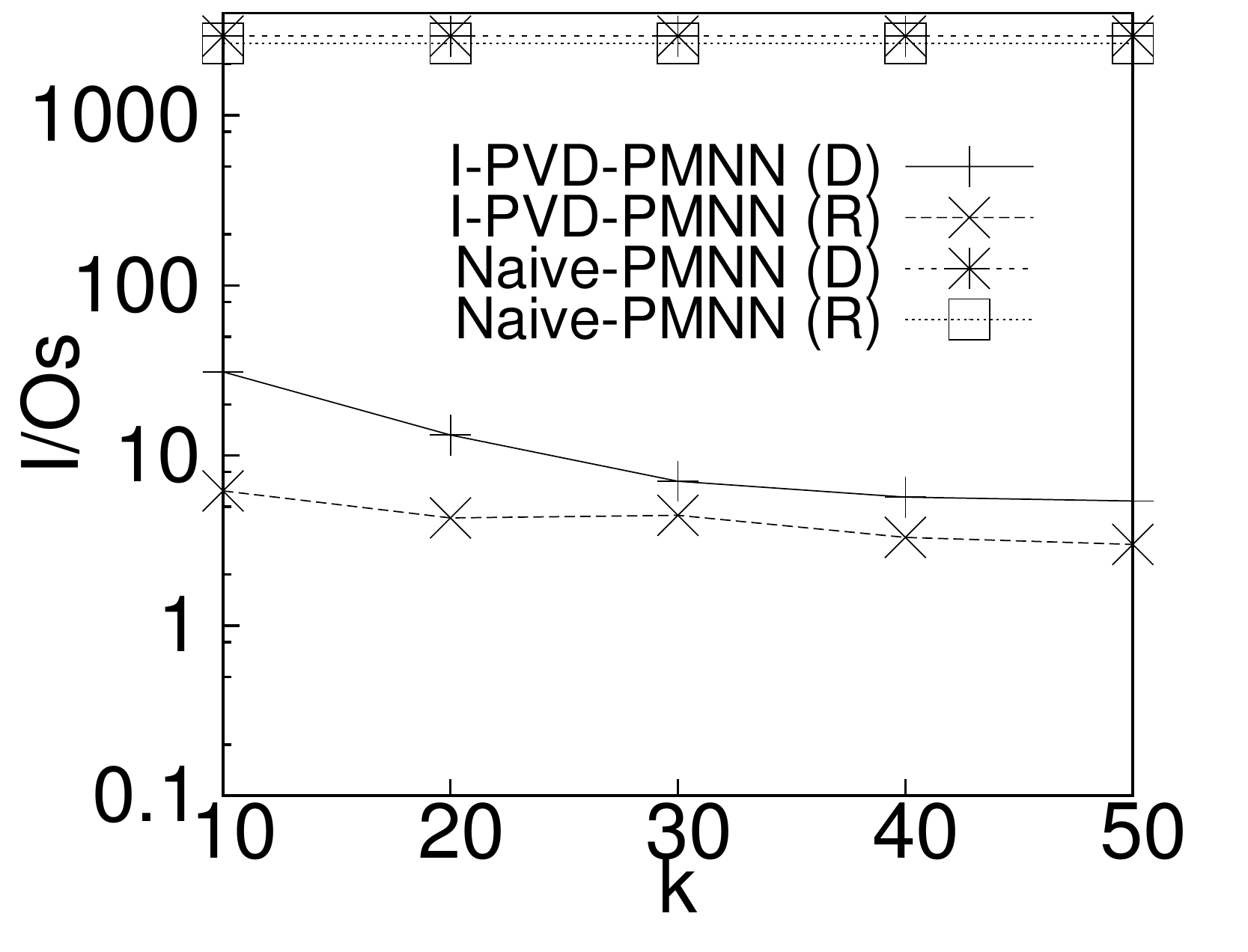}} &
         \hspace{-4mm}
        \resizebox{40mm}{!}{\includegraphics{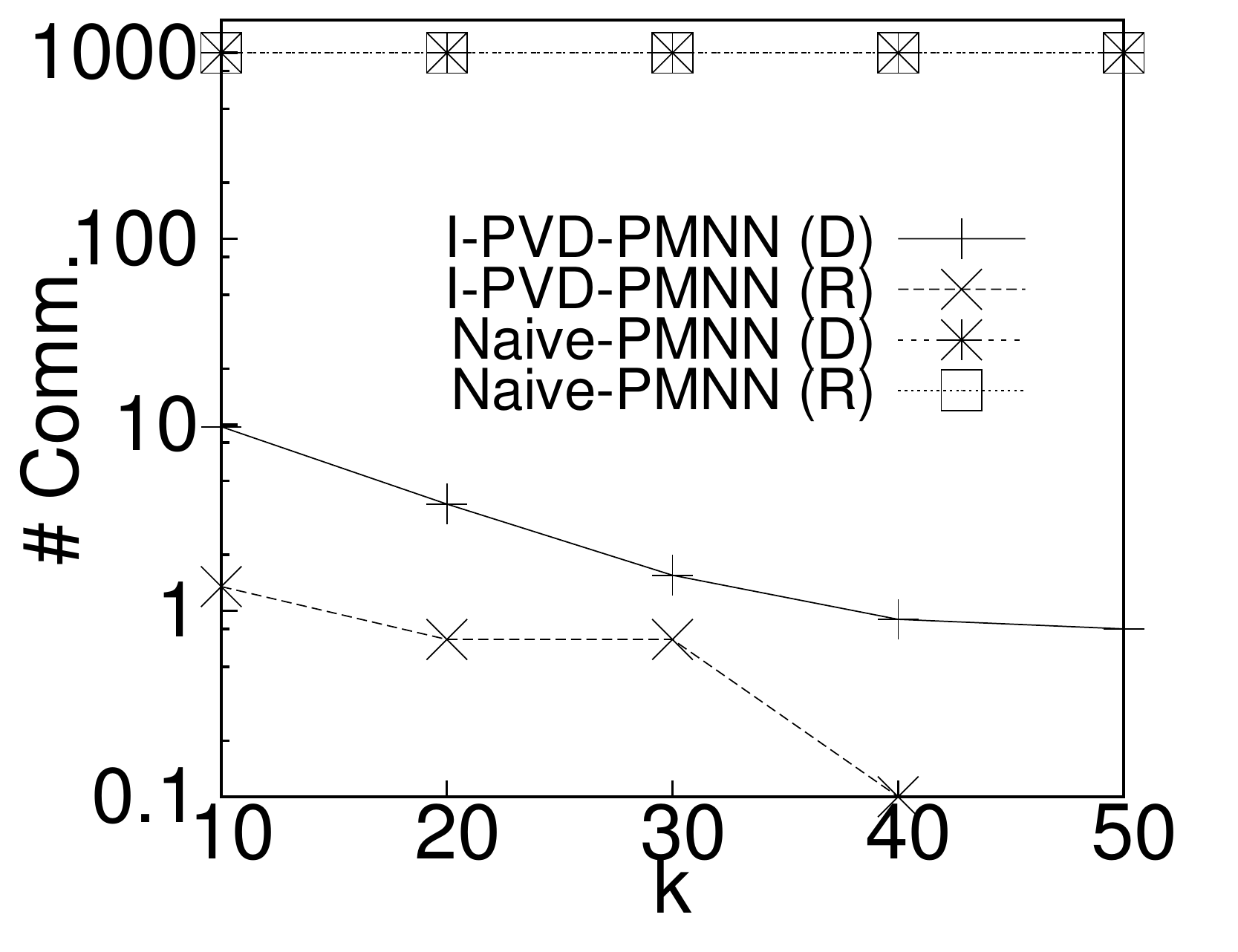}}\\
       \scriptsize{(a)\hspace{0mm}} & \scriptsize{(b)} & \scriptsize{(c)}\\
      \end{tabular}
      \begin{tabular}{cccc}
        \hspace{-5mm}
      \resizebox{40mm}{!}{\includegraphics{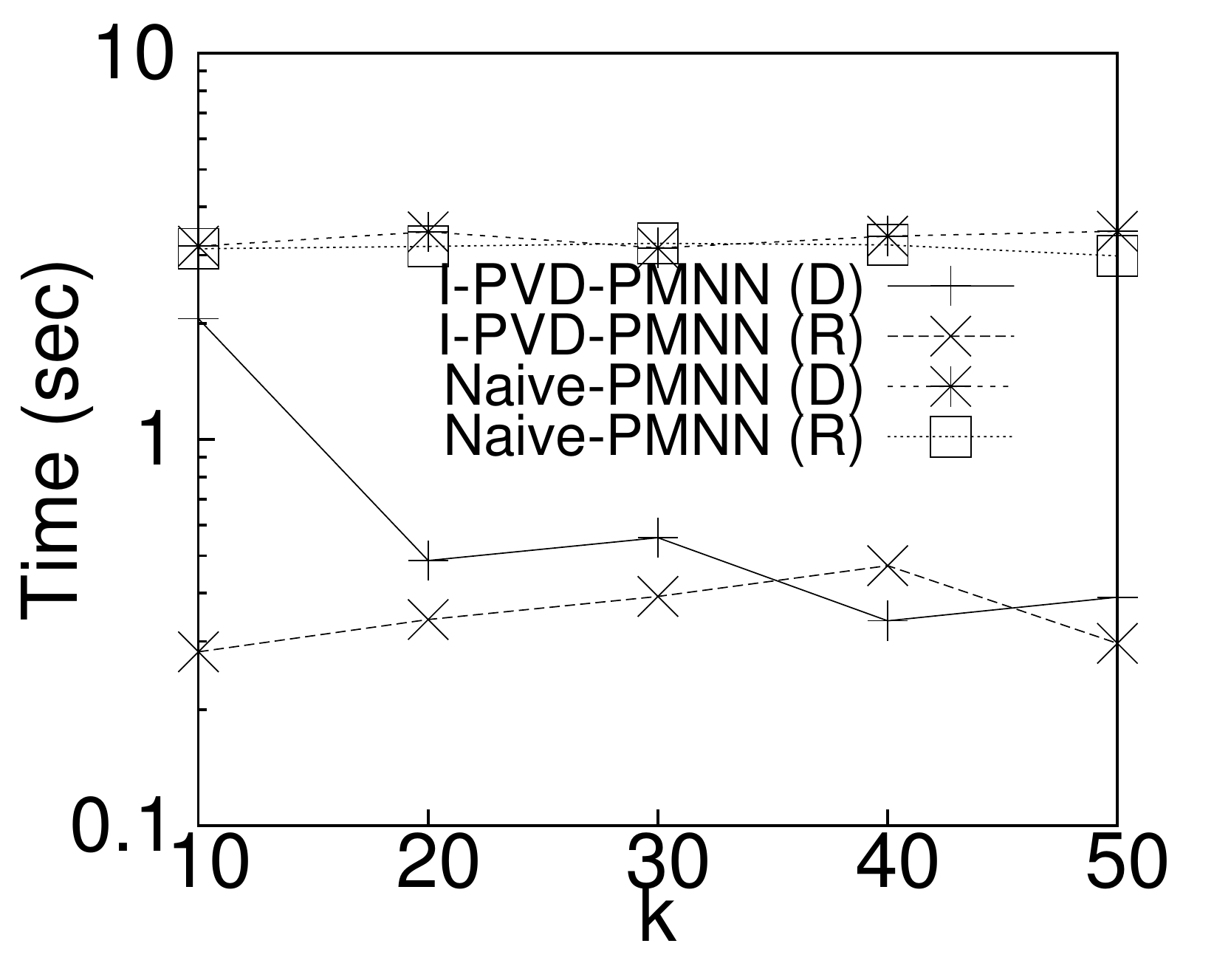}} &
        \hspace{-4mm}
        \resizebox{40mm}{!}{\includegraphics{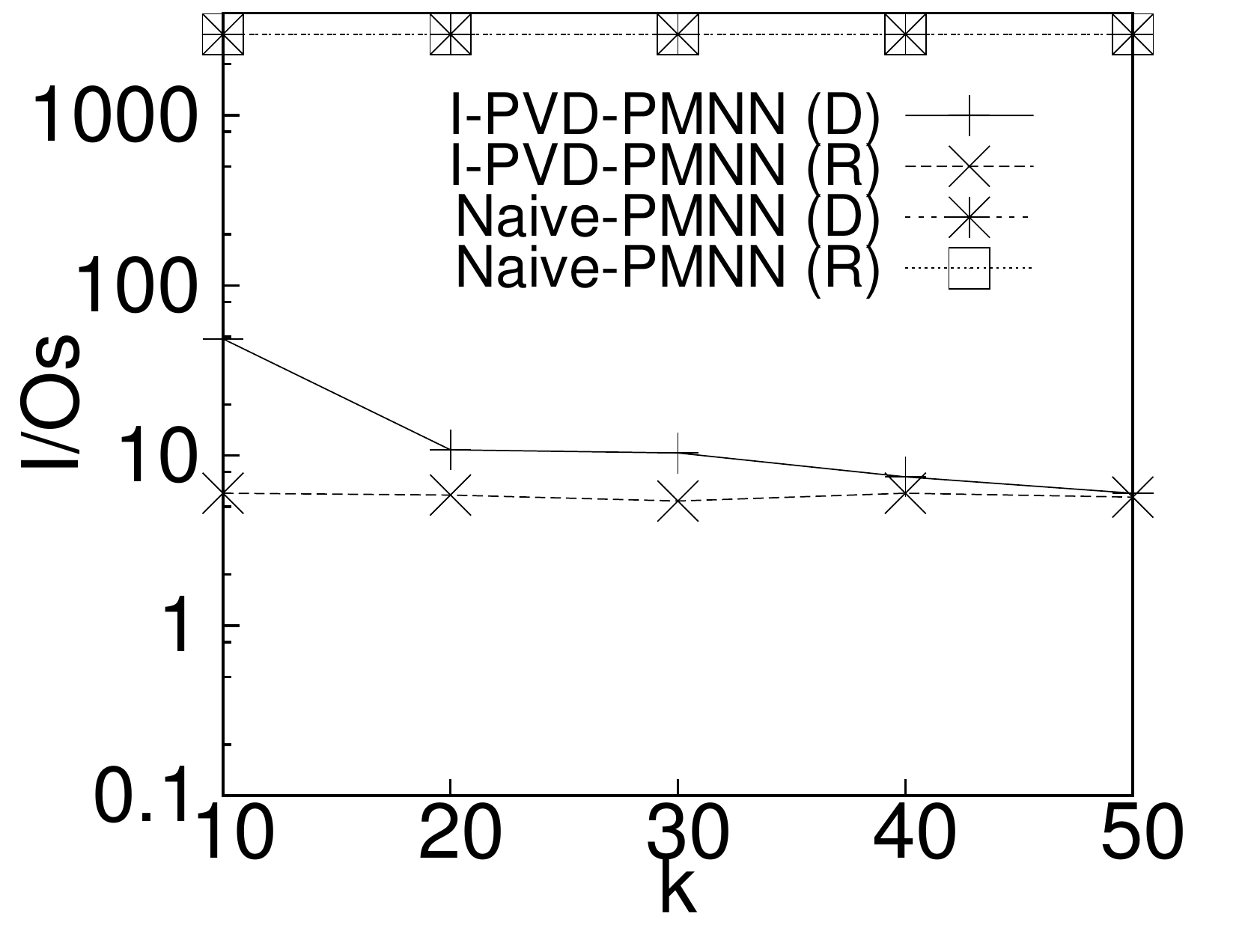}} &
          \hspace{-4mm}
        \resizebox{40mm}{!}{\includegraphics{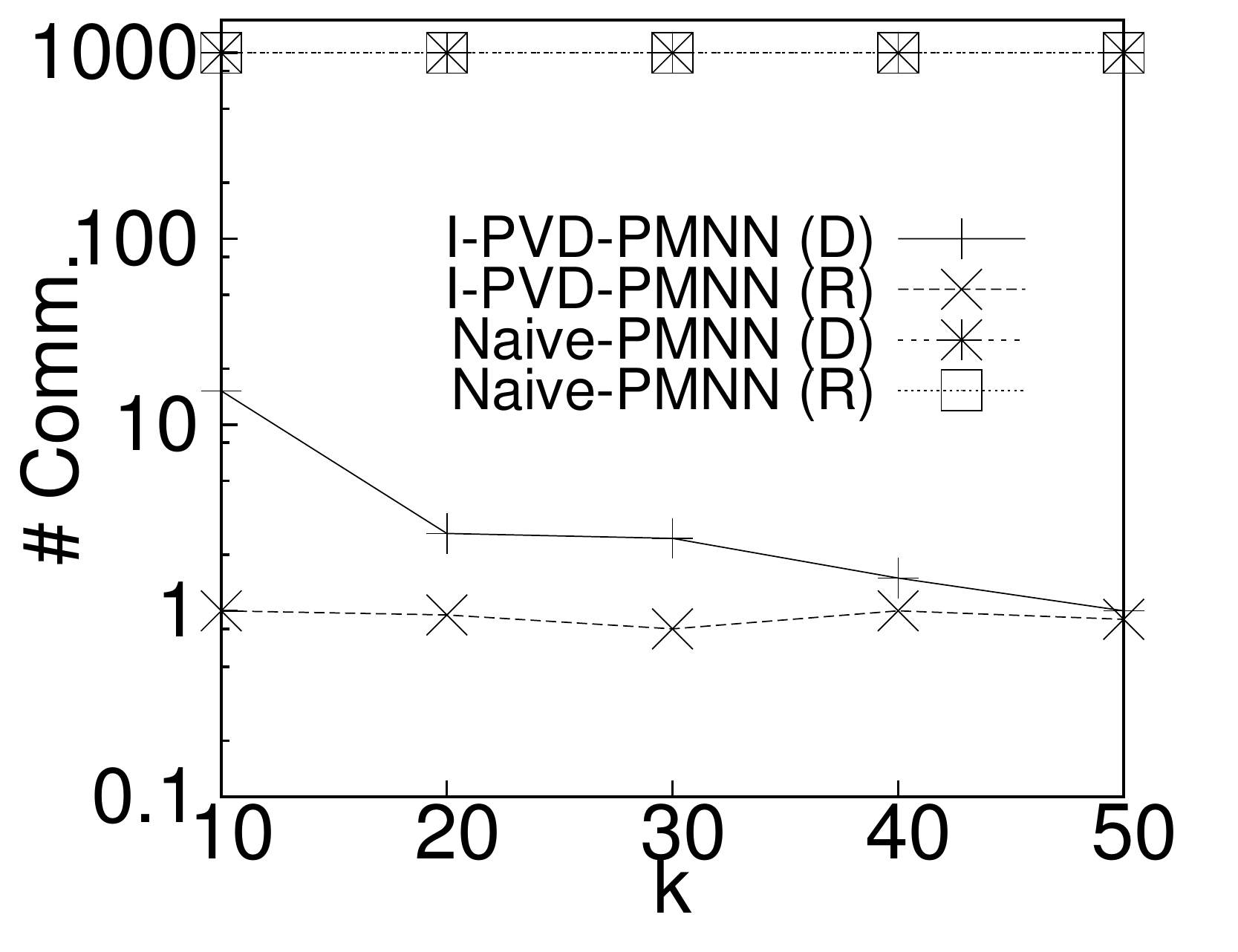}}\\
       \scriptsize{(d)\hspace{0mm}} & \scriptsize{(e)} & \scriptsize{(f)}\\
      \end{tabular}
    \caption{The effect of (\emph{k}) in U (a-c), Z (d-f) for 1D data}
    \label{fig:vk1D}
  \end{center}
\end{figure*}

\begin{figure*}[htbp]
  \begin{center}
    \begin{tabular}{cccc}
        \hspace{-5mm}
      \resizebox{40mm}{!}{\includegraphics{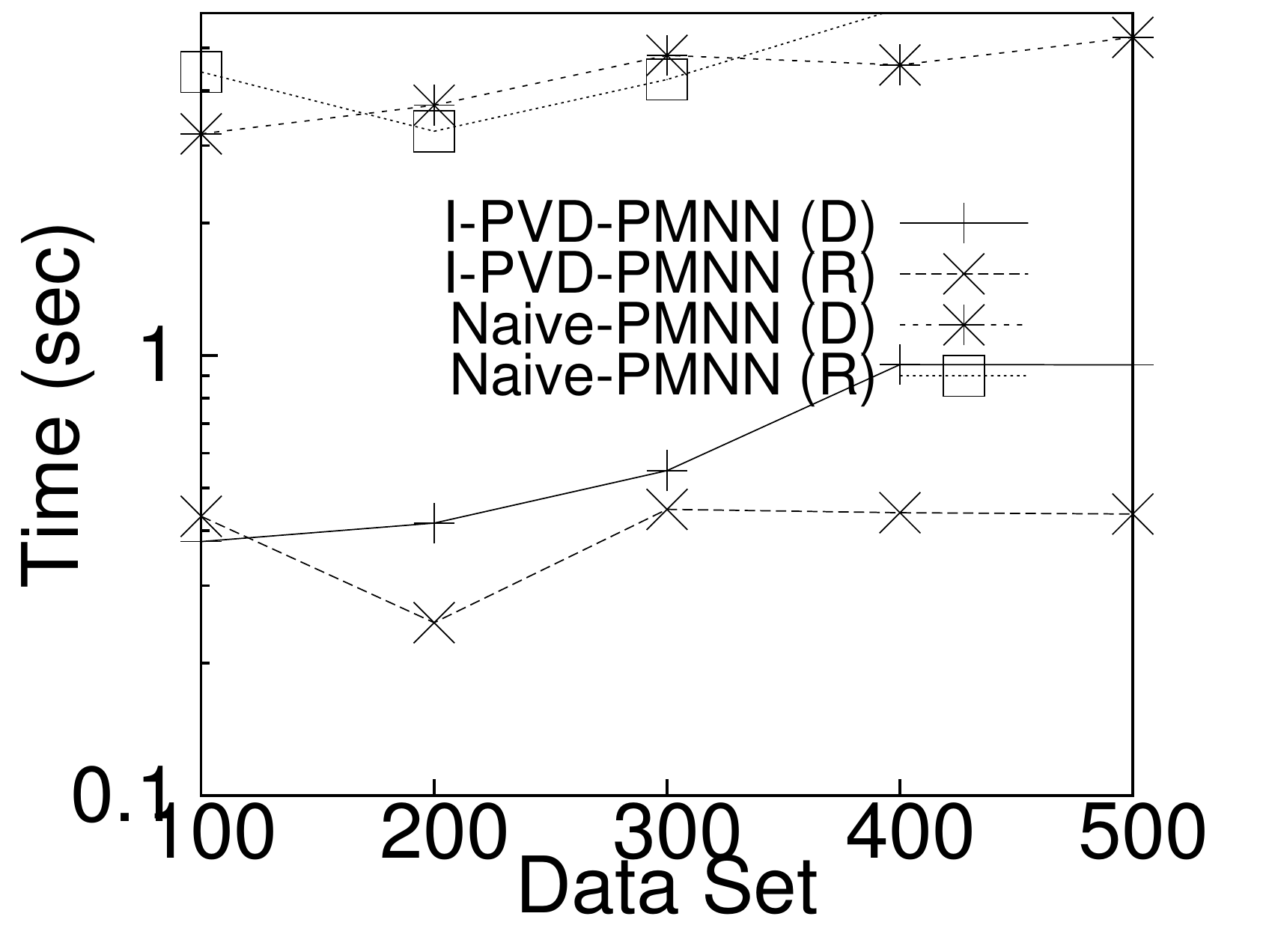}} &
        \hspace{-4mm}

        \resizebox{40mm}{!}{\includegraphics{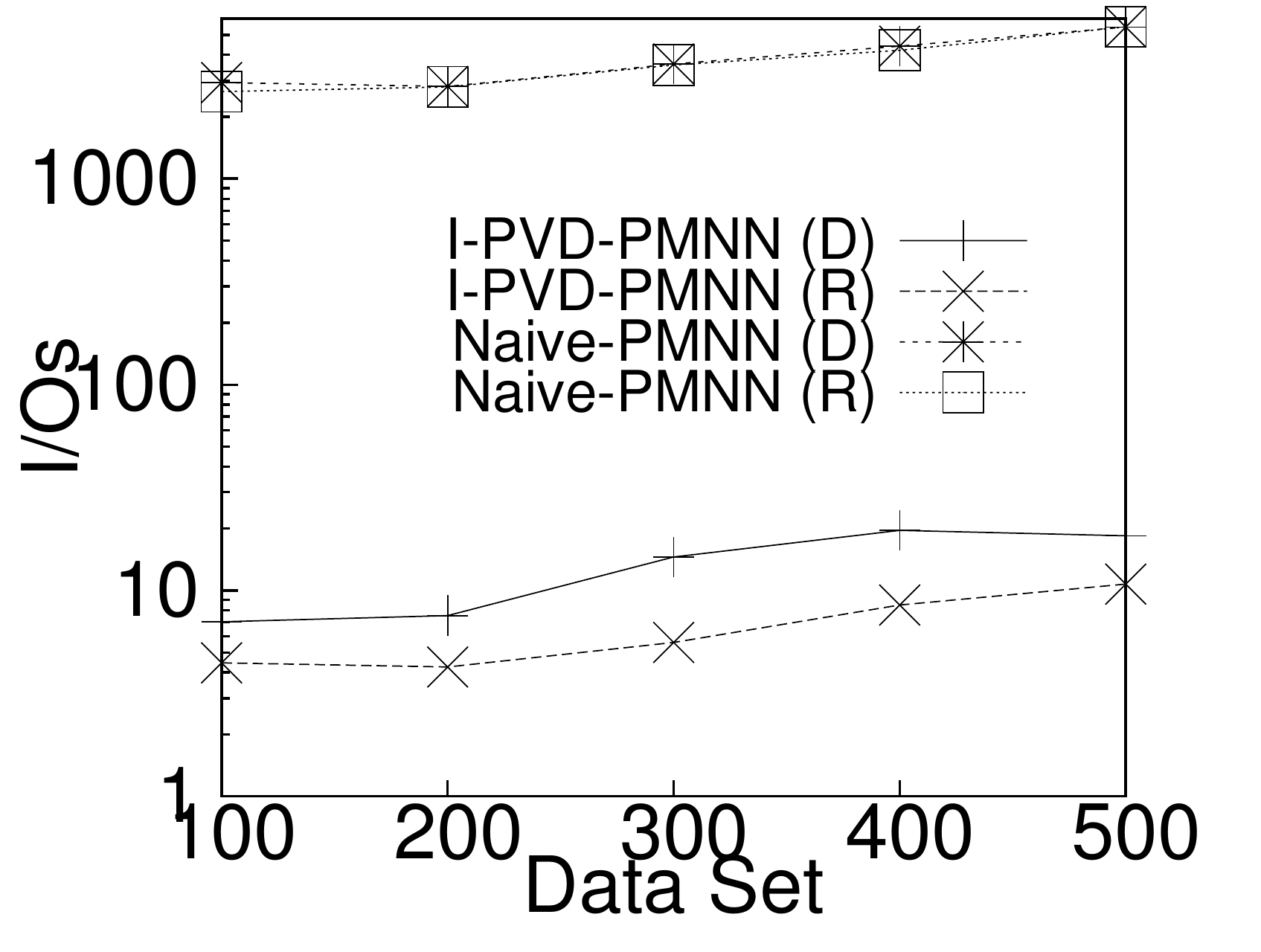}} &
         \hspace{-4mm}

        \resizebox{40mm}{!}{\includegraphics{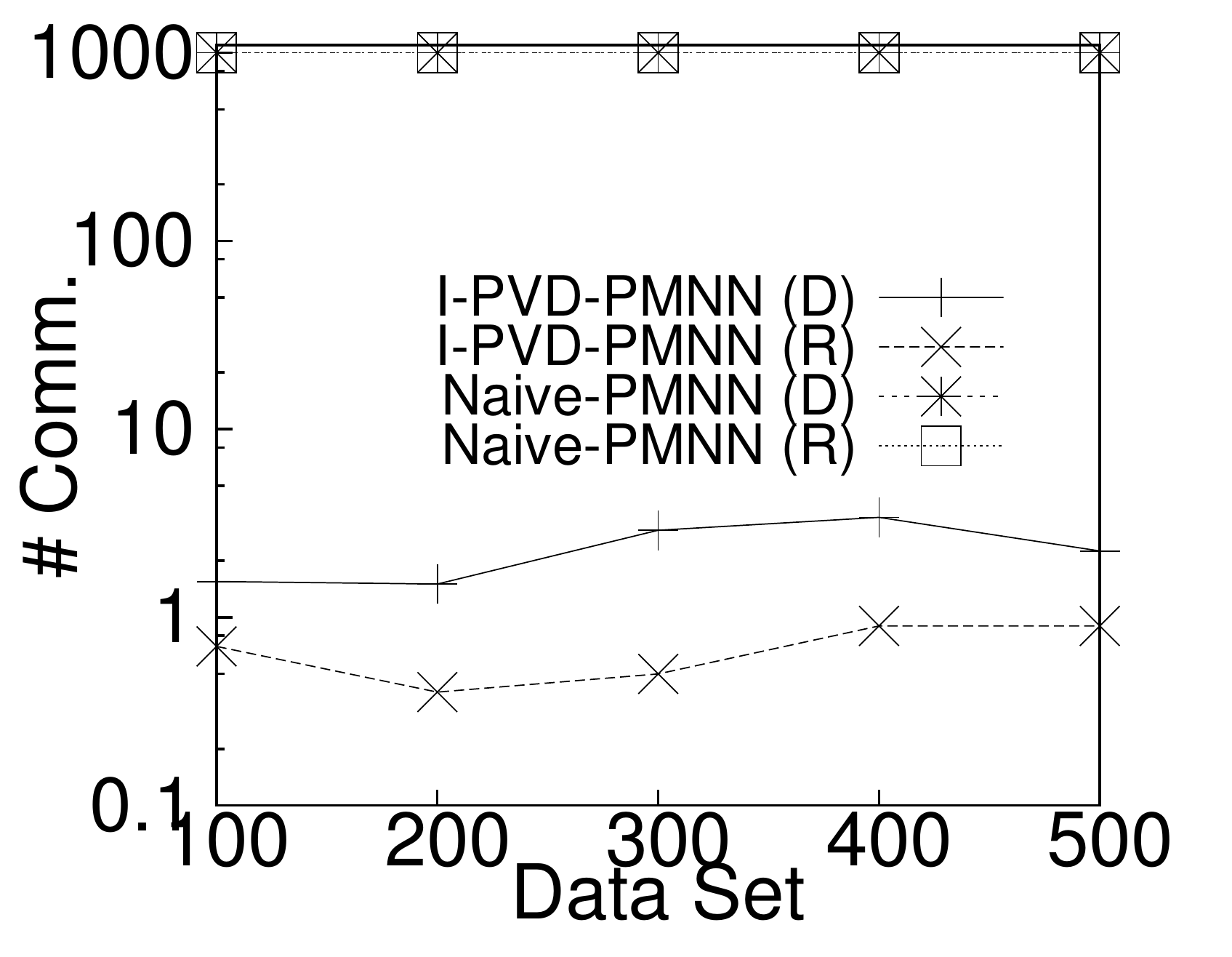}}\\
       \scriptsize{(a)\hspace{0mm}} & \scriptsize{(b)} & \scriptsize{(c)}\\
      \end{tabular}
      \begin{tabular}{cccc}
        \hspace{-5mm}
      \resizebox{40mm}{!}{\includegraphics{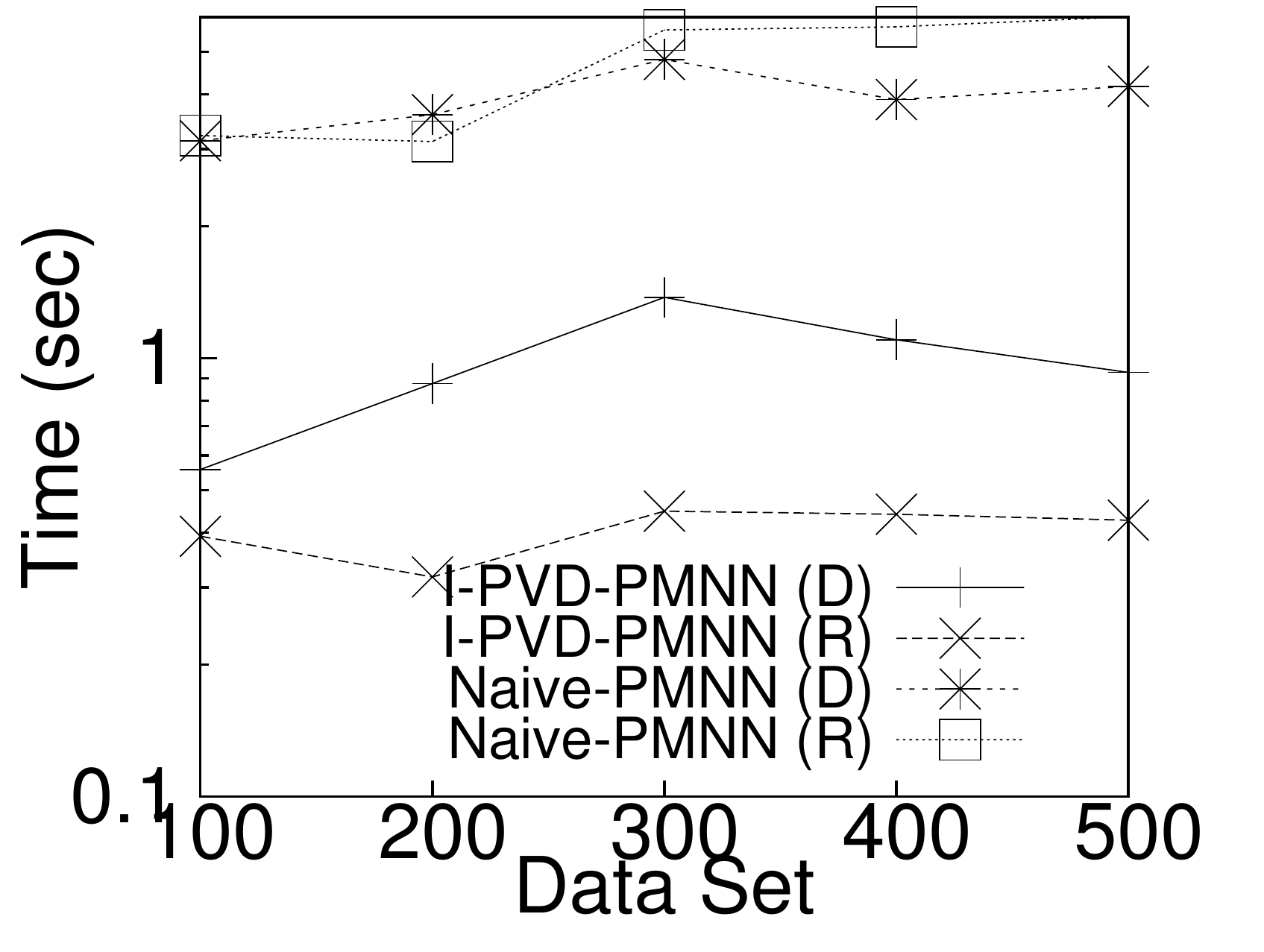}} &
        \hspace{-4mm}

        \resizebox{40mm}{!}{\includegraphics{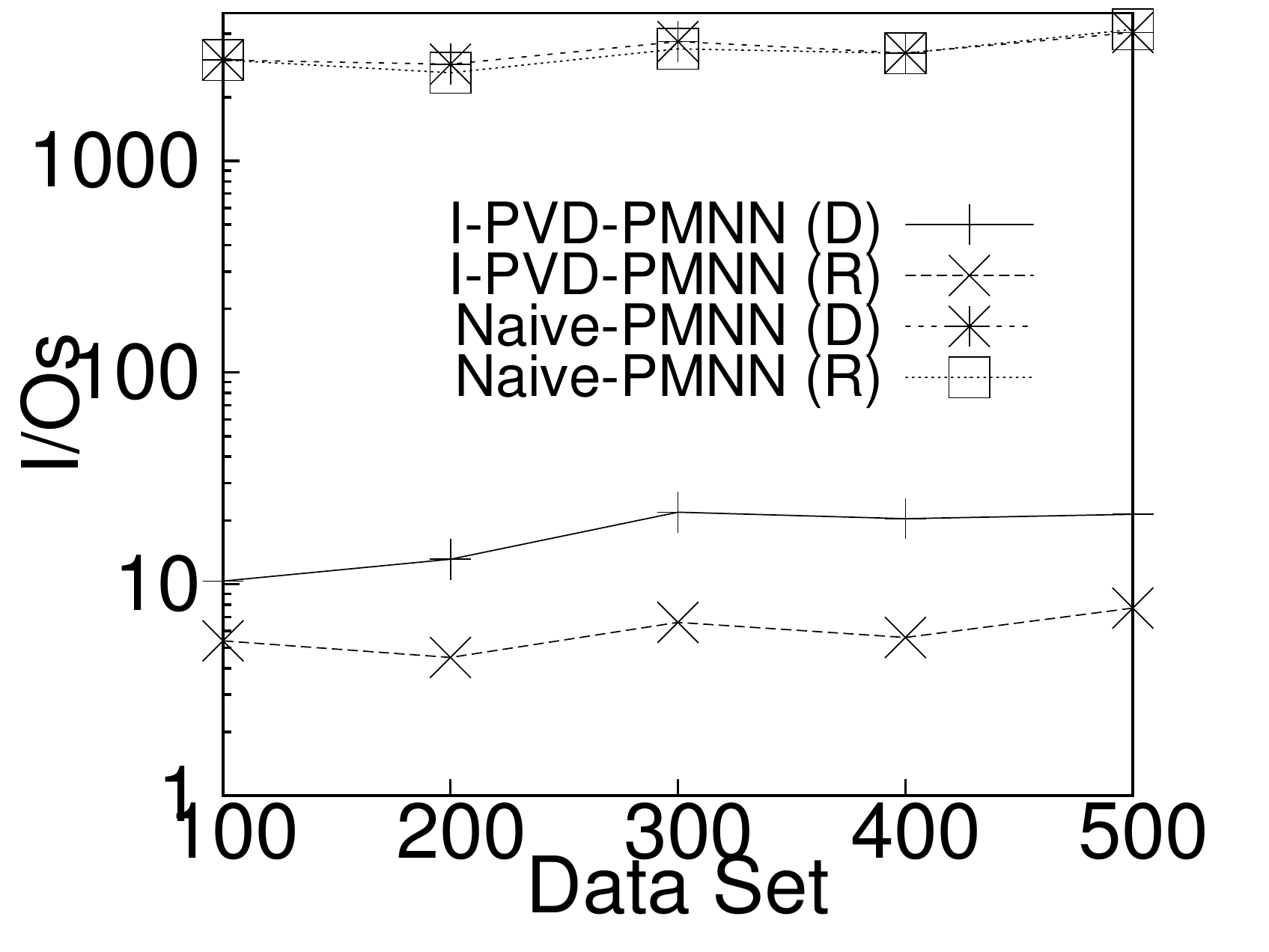}} &
        \hspace{-4mm}

        \resizebox{40mm}{!}{\includegraphics{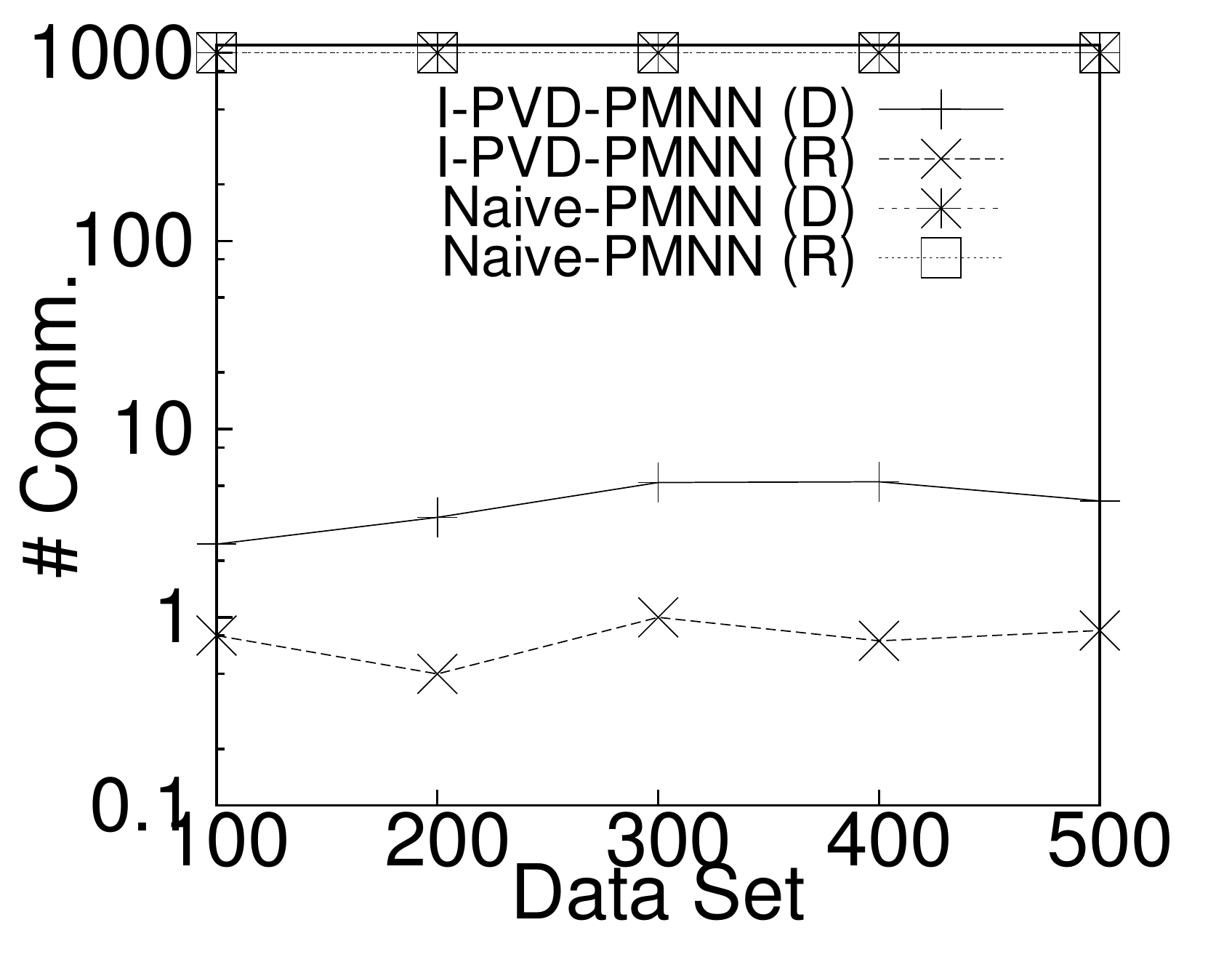}}\\
       \scriptsize{(d)\hspace{0mm}} & \scriptsize{(e)} & \scriptsize{(f)}\\
      \end{tabular}
    \caption{The effect of the data set size in U (a-c), Z (d-f) for 1D data}
    \label{fig:vd1D}
  \end{center}
\end{figure*}

\noindent\emph{\textbf{Experiments with 1D Data Sets:}}
We also evaluate our incremental approach with 1D data sets by varying the following parameters: the value of $k$, the data set size, and the length of the query trajectory.

\emph{Effect of $k$: }In this set of experiments, we study the
impact of $k$ in the performance measure of I-PVD-PMNN for 1D data sets. Figures~\ref{fig:vk1D}(a)-(e) show the results of U and Z data sets, for varying $k$ from 10 to 50. In these
experiments, we have set the data set size to 100. Figure~\ref{fig:vk1D}(a)
shows that the processing time almost remains constant for varying
$k$. Moreover, the processing time of I-PVD-PMNN is on average 6 times less for directional (D) query paths than that of Naive-PMNN, and on average 10 times less for random (R) query paths than that of Naive-PMNN.
Figures~\ref{fig:vk1D}(b)-(c) show that
the I/O costs and the number of communications decrease with the
increase of $k$. Figures also show that our I-PVD-PMNN outperforms Naive-PMNN by
2-3 orders of magnitude in terms of both I/O costs and communication costs.

Figures~\ref{fig:vk1D}(d)-(f) show the results for Z data set, which is similar to U data set.

\emph{Effect of Data Set Size: }In this set of experiments,
we vary the data set size from 100 to 500 and compare the
performance of our approach I-PVD-PMNN with Naive-PMNN. In these
experiments, we have set the value of $k$ to 30 and the trajectory length to 5000 units. Figures~\ref{fig:vd} (a)-(c) and (d)-(f) show the
processing time, I/O costs, and the number of
communications for U and Z data sets, respectively. The results reveal that the processing time, I/O costs, and the communications costs increase with the increase of the data set size.
Figures also show that our I-PVD-PMNN outperforms Naive-PMNN by
at least an order of magnitude for all data sets.

\begin{figure*}[htbp]
  \begin{center}
    \begin{tabular}{cccc}
        \hspace{-5mm}
      \resizebox{40mm}{!}{\includegraphics{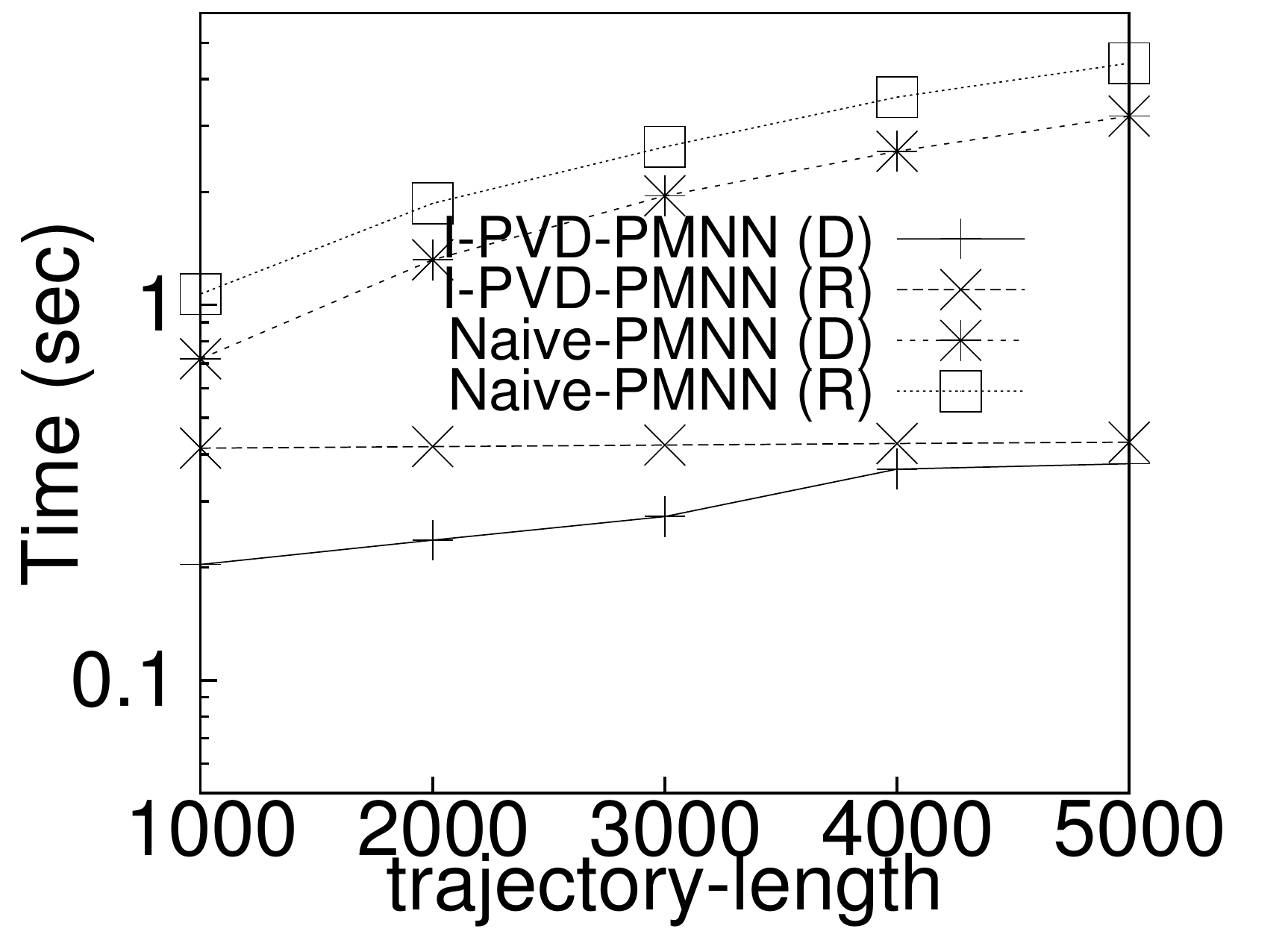}} &
        \hspace{-4mm}

        \resizebox{40mm}{!}{\includegraphics{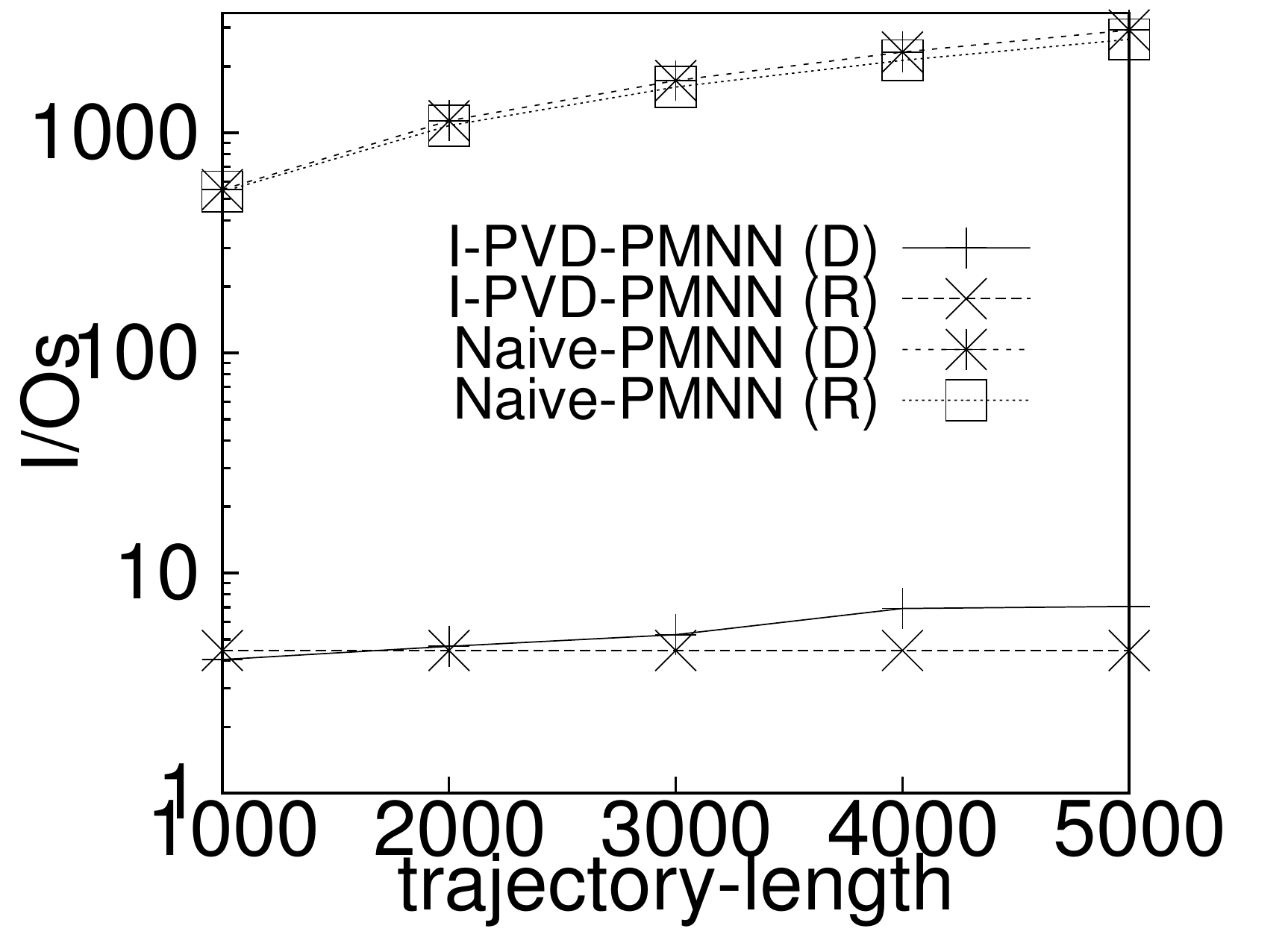}} &
         \hspace{-4mm}

        \resizebox{40mm}{!}{\includegraphics{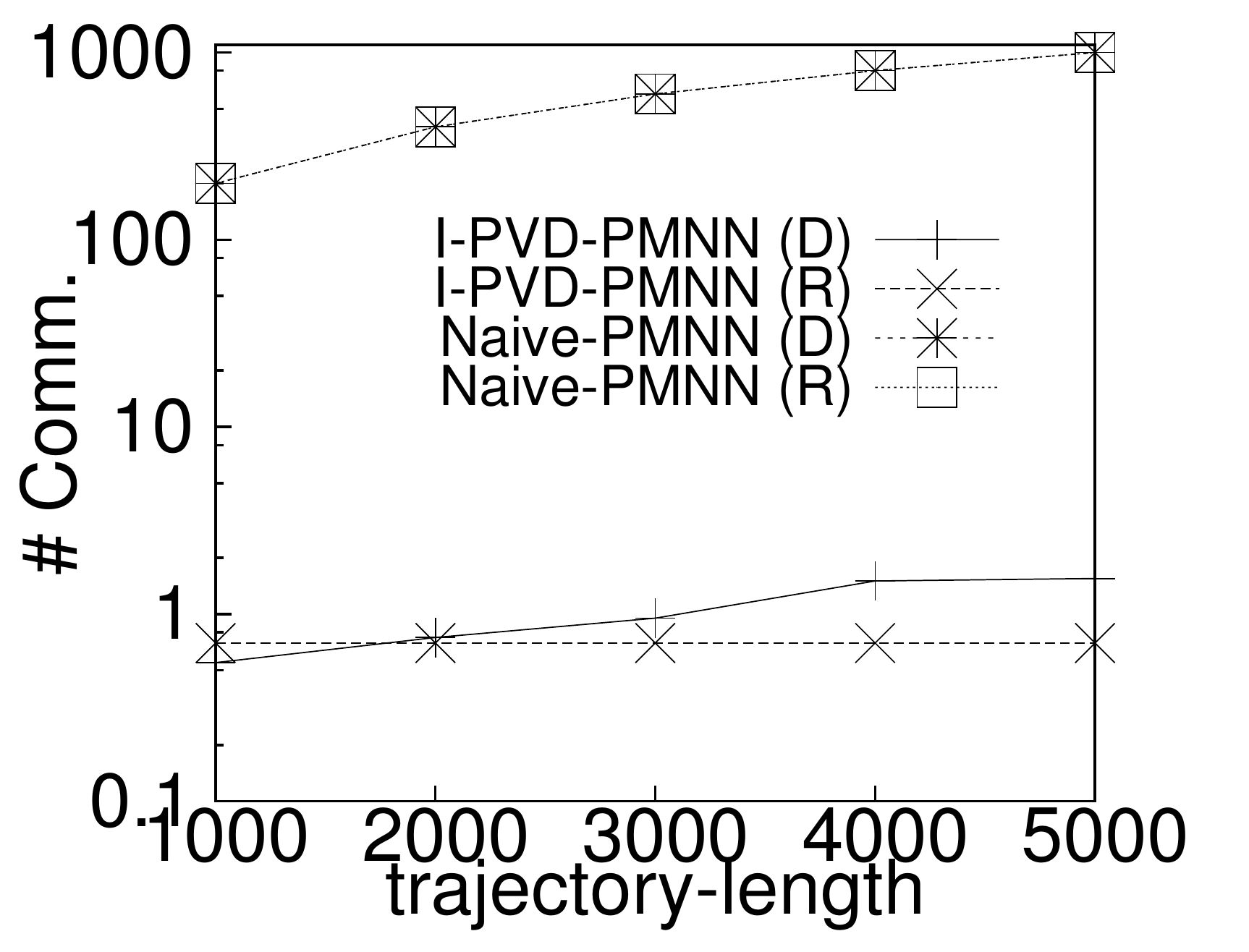}}\\
       \scriptsize{(a)\hspace{0mm}} & \scriptsize{(b)} & \scriptsize{(c)}\\
      \end{tabular}
    \begin{tabular}{cccc}
        \hspace{-5mm}
      \resizebox{40mm}{!}{\includegraphics{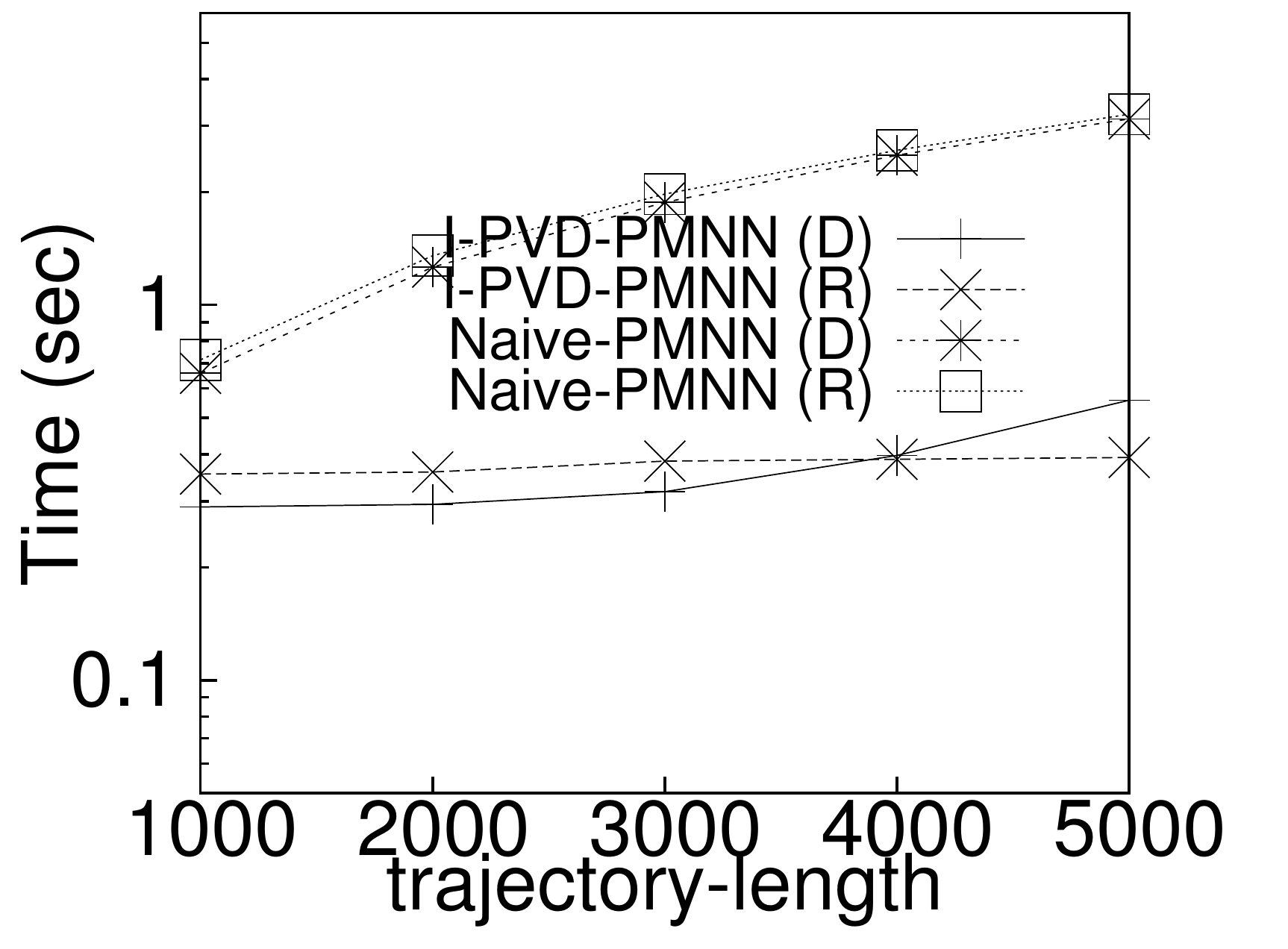}} &
        \hspace{-4mm}

        \resizebox{40mm}{!}{\includegraphics{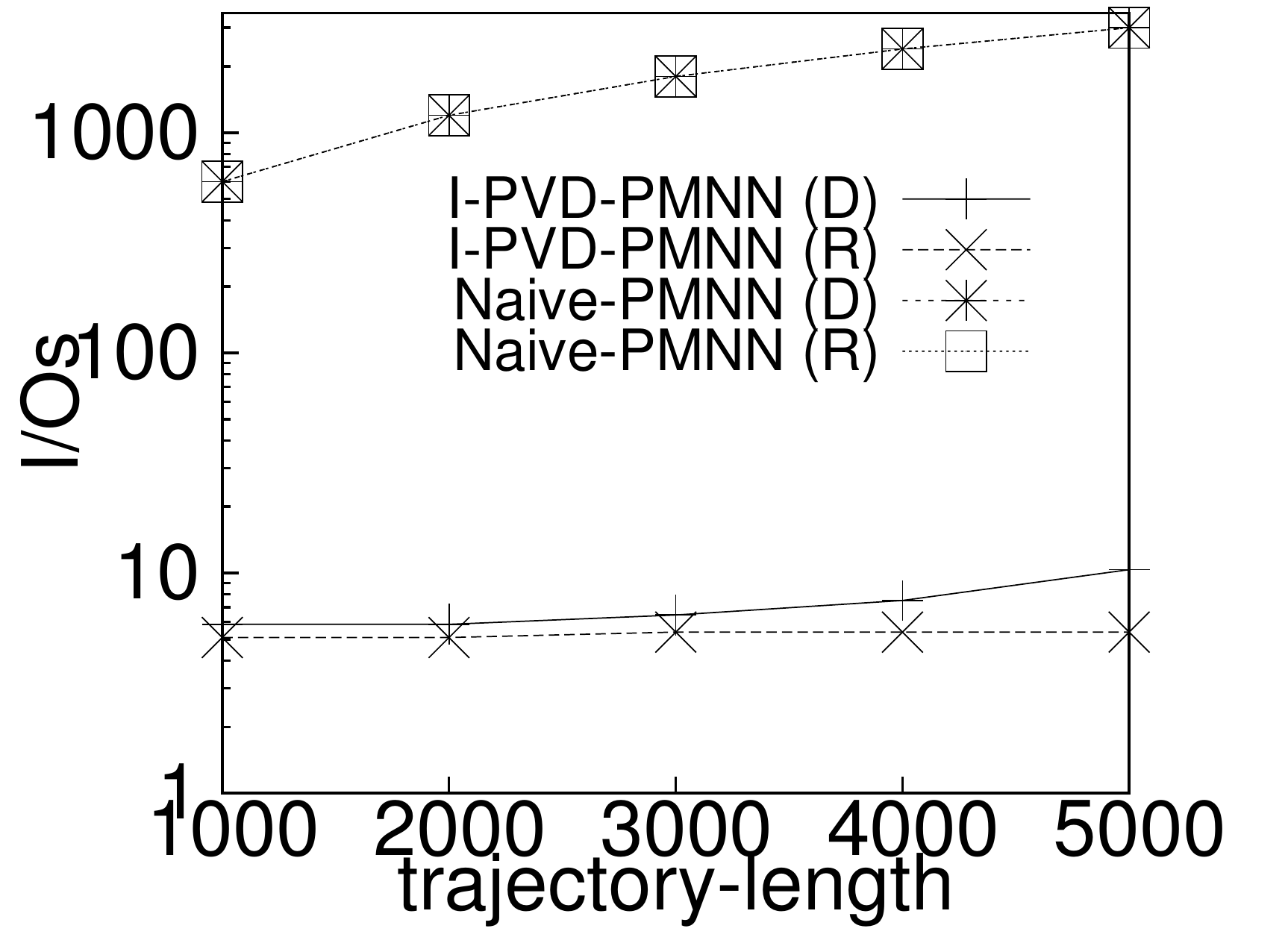}} &

        \hspace{-4mm}

        \resizebox{40mm}{!}{\includegraphics{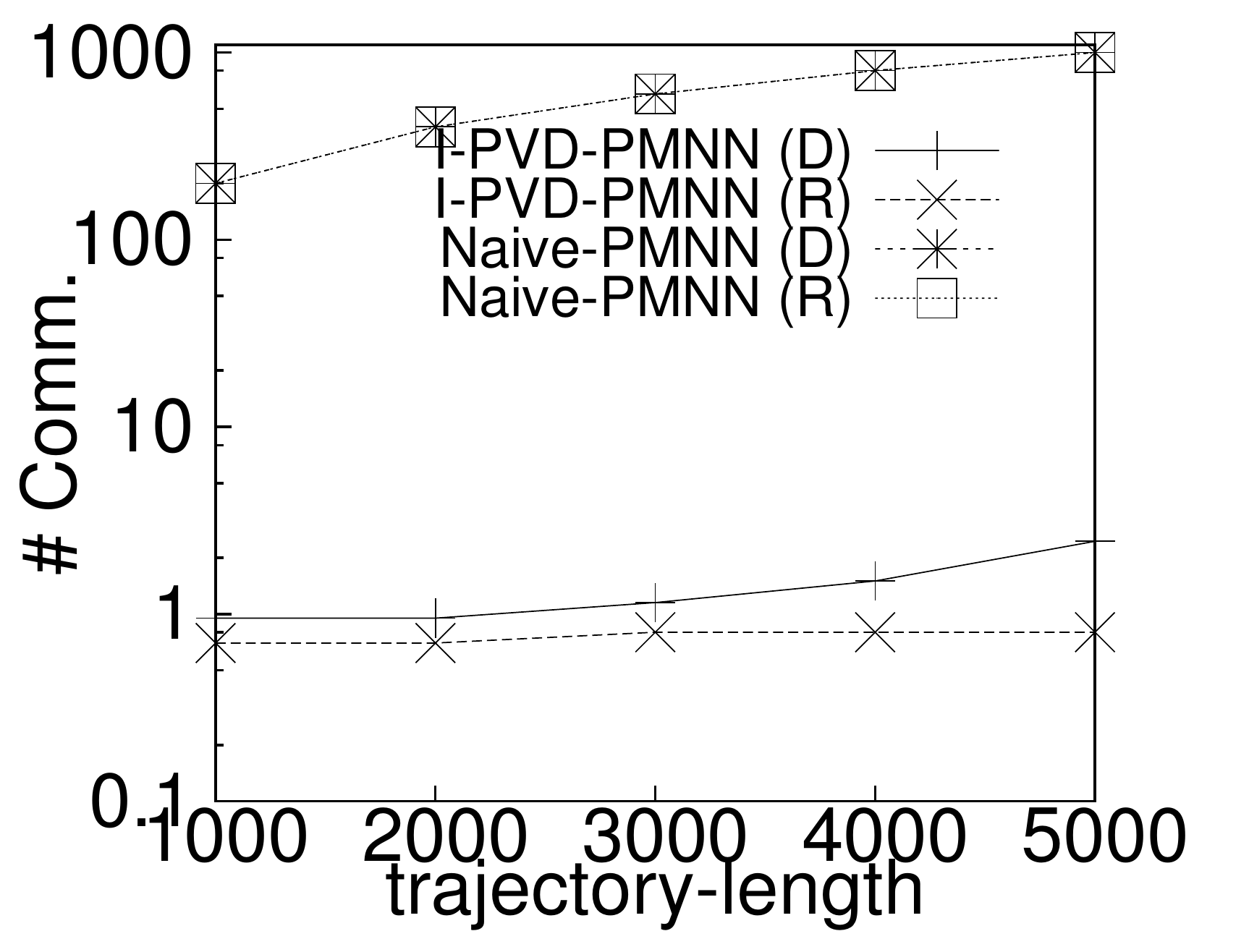}}\\
       \scriptsize{(d)\hspace{0mm}} & \scriptsize{(e)} & \scriptsize{(f)}\\
      \end{tabular}

    \caption{The effect of the query length in U (a-c), Z (d-f) for 1D data}
    \label{fig:vqi1D}
  \end{center}

\end{figure*}

\emph{Effect of the Length of a Query Trajectory: }We also vary the length of the query trajectory for 1D data sets and the results (Figures~\ref{fig:vqi1D}) for 1D data sets exhibit similar behavior to 2D data sets. In these experiments, we vary the trajectory length from 1000 to 5000 units of the data space. Also, we have set the data set size to 100, and the value of $k$ to 30. Figures~\ref{fig:vqi1D} show that for both U and Z data sets, the processing time, I/O costs, and the communication costs increase with the increase of the trajectory length. Figures also show that our I-PVD-PMNN outperforms Naive-PMNN in all evaluation metrics.

\eat{Note that we run experiments for Zipfian distribution which has
highly varied density of objects in different parts of the space.
Thus, we omit an extra set of experiments under different density
measures. Also in our experiments the length of two consecutive
query points is kept fixed. If we decrease the length between the
two consecutive query points, the response time should also reduce
as the same object will remain the most probable NN for many
consecutive points. We skip this
discussion as it is an established concept in safe region based
methods.}

\section{Summary}
\label{sec:conc}
In this paper, we have introduced the concept of Probabilistic Voronoi Diagrams
(PVDs). A PVD divides the data space using a probability measure.
Based on the PVD, we developed two different techniques: a
pre-computation approach and an incremental approach, for
efficient processing of Probabilistic Moving Nearest Neighbor
(PMNN) queries. Our experimental results show that our techniques
outperform the sampling based approach by up to two orders of
magnitude in our evaluation metrics.

Our work on PVD opens new avenues for future work. Currently our
approach finds the most probable NN for a moving query point; in
the future we aim to extend it for top-$k$ most probable NNs. PVDs
for other types of probability density functions such as normal
distribution are to be investigated. We also plan to have a
detailed investigation on PVDs of higher dimensional spaces.


\bibliographystyle{elsarticle-num}
\bibliography{wholething}







\end{document}